\documentclass[letterpaper,11pt]{article}
\usepackage{geometry}
\usepackage{amsmath}
\usepackage{commath}
\usepackage{amsthm}
\usepackage{amssymb}
\usepackage{setspace,mathrsfs}
\usepackage{algorithm,natbib}
\usepackage{algorithmic}
\usepackage{url,amsmath,amsfonts,amssymb}
\usepackage{caption,subcaption,algorithm,graphicx}
\usepackage{listings,enumerate,color,relsize,multirow,rotating}
\usepackage{bbm}
\usepackage{hyperref}
\usepackage{authblk}
\usepackage{bm}

\def\bx{\boldsymbol{X}}

\def\supm{^{\scriptscriptstyle (m)}}

\def\by{\boldsymbol{Y}}
\def\by{\boldsymbol{y}}
\def\bX{\boldsymbol{X}}
\def\bx{\boldsymbol{x}}

\def\bone{\mathbf{1}}

\newtheorem{lemma}{Lemma}

\newtheorem{theorem}{Theorem}

\newtheorem{example}{Example}

\usepackage{color}

\newcommand{\bei}{\begin{itemize}}
\newcommand{\eei}{\end{itemize}}
\newcommand{\ben}{\begin{enumerate}}
\newcommand{\een}{\end{enumerate}}
\def\0{\boldsymbol{0}}

\def\x{\boldsymbol{x}}

\def\Z{\boldsymbol{Z}}

\def\a{\boldsymbol{a}}

\def\D{\boldsymbol{D}}

\def\N{\mathcal{N}}
\def\noj{\text{--}j}

\def\independenT#1#2{\mathrel{\rlap{$#1#2$}\mkern2mu{#1#2}}}
\newcommand{\indp}{\protect\mathpalette{\protect\independenT}{\perp}}

\newcommand{\RR}{\mathbb{R}}

\newbox\TempBox \newbox\TempBoxA
\def\trans{^{\scriptscriptstyle \sf T}}

\def\bbeta{\boldsymbol{\beta}}
\def\betahat{\hat{\beta}}

\def\bepsilon{\boldsymbol{\epsilon}}

\def\bZ{\boldsymbol{Z}}

\def\bz{\boldsymbol{z}}
\def\bA{\boldsymbol{A}}
\def\ba{\boldsymbol{a}}

\def\D{\mathcal{D}_y}

\def\bu{\boldsymbol{u}}
\def\dx{\boldsymbol{d}_x}
\def\dy{\boldsymbol{d}_y}
\def\dyone{\boldsymbol{d}_{y,1}}
\def\dynone{\boldsymbol{d}_{y,\text{--}1}}

\def\suplasso{^{\sf \scriptscriptstyle lasso}}
\def\suploco{^{\sf \scriptscriptstyle loco}}

\newcommand{\One}[1]{\mathbf{1}_{\left\{#1\right\}}}

\newcommand{\rev}[1]{{#1}}
\newcommand{\dl}[1]{{\color{yellow}}}%#1}}
\newcommand{\dll}[1]{{\color{yellow}}}%#1}}

\begin{document}
%\setstretch{1.3}
%\def\spacingset#1{\renewcommand{\baselinestretch}%
%{#1}\small\normalsize} \spacingset{1}

\title{Fast and Powerful Conditional Randomization Testing\\ via Distillation}

\author[1]{Molei Liu\thanks{This paper is a post-hoc merge of two parallel works by Liu \& Janson (2020) and Katsevich \& Ramdas (2020).}}
\author[2]{Eugene Katsevich}
\author[3]{Lucas Janson}
\author[4]{Aaditya Ramdas}
\date{}
\affil[1]{Department of Biostatistics, Harvard Chan School of Public Health}%, \texttt{molei\_liu@g.harvard.edu}.}
\affil[2]{Statistics Department, Wharton School of the University of Pennsylvania}%, \texttt{ekatsevi@andrew.cmu.edu}.}
\affil[3]{Department of Statistics, Harvard University}%, \texttt{ljanson@fas.harvard.edu}.}
\affil[4]{Departments of Statistics \& Machine Learning, Carnegie Mellon University}

\maketitle

\begin{abstract}
We consider the problem of conditional independence testing: given a response $Y$ and covariates $(X,Z)$, we test the null hypothesis that $Y \indp X \mid Z$. The conditional randomization test (CRT) was recently proposed as a way to use distributional information about $X\mid Z$ to exactly (non-asymptotically) control Type-I error using {any} test statistic in {any} dimensionality without assuming anything about $Y\mid (X,Z)$. This flexibility in principle allows one to derive powerful test statistics from complex prediction algorithms while maintaining statistical validity. Yet the direct use of such advanced test statistics in the CRT is prohibitively computationally expensive, especially with multiple testing, due to the CRT's requirement to recompute the test statistic many times on resampled data. We propose the {distilled CRT}, a novel approach to using state-of-the-art machine learning algorithms in the CRT while drastically reducing the number of times those algorithms need to be run, thereby taking advantage of their power and the CRT's statistical guarantees without suffering the usual computational expense. In addition to distillation, we propose a number of other tricks like screening and recycling computations to further speed up the CRT without sacrificing its high power and exact validity. Indeed, we show in simulations that all our proposals combined lead to a test that
%\rev{saves orders of magnitude in computation relative to slower CRT implementations without sacrificing power,} 
has similar power to the \rev{most powerful existing CRT implementations} but requires orders of magnitude less computation, 
%has similar power to the CRT but requires orders of magnitude less computation, 
making it a practical tool even for large data sets.  We demonstrate these benefits on a breast cancer dataset by identifying biomarkers related to cancer stage.
\end{abstract}


%\newpage
\section{Introduction}
In our increasingly data-driven world, it has become the norm in applications from genetics and health care to policy evaluation and e-commerce to seek to understand the relationship between a response variable of interest and a high-dimensional set of potential explanatory variables or covariates. While accurately estimating this entire relationship generally would require a nearly-infinite sample size, a less-intractable but still extremely useful question is to ask, for any given covariate, whether it actually contributes to the response variable's high-dimensional conditional distribution. We address this problem by encoding a covariate's relevance as a {conditional independence test}, which requires no modeling assumptions to define. Denoting the response random variable by $Y$, a given covariate of interest by $X$, and a multidimensional set of further covariates by $Z=(Z_1,\dots,Z_p)$, the null hypothesis we seek to test is
\begin{equation*}
H_0:Y\indp X \mid Z
\end{equation*}
against the alternative $H_1:Y\not\indp X\mid Z$. Testing this hypothesis for just a single covariate is sometimes all that is needed, such as in an observational study investigating whether a particular treatment ($X$) causes a change in a response ($Y$) after controlling for a set of measured confounding variables ($Z$). But in other applications no one covariate holds {a priori} precedence over another, and a researcher seeks any and all covariates that contribute to $Y$'s conditional distribution. This variable selection objective can also be achieved by testing $H_0$ for each covariate in turn (with $Z$ containing the other covariates) and plugging the resulting $p$-values into one of the many procedures from the extensive literature on multiple testing. In addition to the considerable statistical challenge of providing a valid and powerful test of $H_0$, it is of paramount importance to also ensure that test is computationally efficient, especially, as is often the case in modern applications, when either or both the sample size and dimension are large, and even more so when a variable selection objective requires the test to be run many times. Thus the goal of this paper is to present a test for conditional independence that is provably valid, empirically powerful, and computationally efficient, when the distribution of $X|Z$ is known or can be well approximated.

\subsection{Background}\label{sec:background}
Our work builds on the {conditional randomization test} (CRT) introduced in \citet{candes2018panning}. The CRT is a  general framework for conditional independence testing that can use any test statistic one chooses and exactly (non-asymptotically) control the Type-I error regardless of the data dimensionality. The CRT's guarantees assume nothing whatsoever about $Y\mid (X,Z)$, but instead assume $X\mid Z$ is known. This so-called ``model-X" framework --- in contrast to the canonical approach of assuming a strong model for $Y\mid (X,Z)$ --- is perhaps easiest to justify when a wealth of {unlabeled data} (pairs $(X_i,Z_i)$ without corresponding $Y_i$) is available, but has also been found to be quite robust even when $X\mid Z$ is estimated using only the labeled data.

In order to define the CRT, we first need notation for our data. For $i\in\{1,\dots,n\}$, let $(Y_i,X_i,Z_i)\in\mathbb{R}^{p+2}$ be i.i.d. copies of $(Y,X,Z)$, and denote the column vector of $Y_i$'s by $\by\in\mathbb{R}^n$, the column vector of $X_i$'s by $\bx\in\mathbb{R}^n$, and the matrix whose rows are the $Z_i$'s by $\bZ\in\mathbb{R}^{n\times p}$. The CRT is given by Algorithm~\ref{alg:crt}, and its Type-I error guarantee follows below.

\begin{algorithm}[htbp]
\caption{\label{alg:crt} The conditional randomization test (CRT).}
{\bf Input:} {The distribution of $\bx\mid\bZ$,} data $(\by,\bx,\bZ)$, test statistic function $T$, and number of randomizations $M$.\\
\vspace{0.2cm}
{\bf For} $m=1,2,...,M$: Sample $\bx\supm$ from the distribution of $\bx\mid\bZ$, conditionally independently of $\bx$ and $\by$.\label{algline:cr}\\
\vspace{0.2cm}
{\bf Output:} CRT $p$-value $\frac{1}{M+1}\left(1+\sum_{m=1}^M \One{T(\by,\bx\supm,\bZ)\geq T(\by,\bx,\bZ)}\right)$.
\end{algorithm}

\begin{theorem}[\cite{candes2018panning}]
\label{thm:crt}
The CRT $p$-value $p(\bX, \by)$ satisfies \[\mathbb{P}_{H_0}(p(\bX, \by)\le \alpha)\le\alpha \text{ for all } \alpha\in [0,1].\]
\end{theorem}

For many common models of $X\mid Z$, the conditionally-independent sampling of $\bx\supm$ is straightforward. And even in more complex models it is still often easy to sample $\bx\supm$ conditionally-{exchangeably} with $\bx$ and conditionally-independently of $\by$ (for instance by conditioning on an inferred latent variable), which is sufficient for Theorem~\ref{thm:crt} to hold. Because Theorem~\ref{thm:crt} only relies on the exchangeability of the vectors $\bx,\bx^{(1)},\dots,\bx^{(M)}$ under $H_0$, it is entirely agnostic to the choice of test statistic $T$. This enables some very powerful choices, such as $T$'s derived from modern machine learning algorithms, from Bayesian inference (though neither the prior nor model for $Y\mid (X,Z)$ need be well-specified), or from highly domain-specific knowledge or intuition. Unfortunately the most powerful statistics are often particularly expensive to compute, and as can be seen from Algorithm~\ref{alg:crt}, $T$ must be applied $M+1$ times in order to compute a single $p$-value. When testing all the covariates at once, this computational problem is compounded as not only does {each} test require $M+1$ applications of $T$, but $M$ must be roughly of order $p$ to ensure the $p$-values are sufficiently high-resolution to make any discoveries with multiple testing procedures such as Benjamini--Hochberg \citep{benjamini1995controlling}. %For multiple testing, model-X knockoffs \citep{candes2018panning} provides an appealing alternative but can suffer relative to the CRT in terms of power, especially in very sparse settings. This apparent conflict between computational tractability and statistical power presents a barrier to the widespread use of the CRT.

\subsection{Our contribution}\label{sec:intro:cont}
We resolve this computational challenge in Section~\ref{sec:method:dCRT} by introducing a technique we call {distillation} that can still leverage {any} high-dimensional modeling or supervised learning algorithm, but presents dramatic computational savings by only requiring the expensive high-dimensional computation to be performed once, instead of $M+1$ times. We call our proposed method the distilled CRT, or dCRT, and show how to further improve its computation in multiple testing settings in Section~\ref{sec:speedup}. 

We demonstrate in simulations in Section~\ref{sec:sim} that there is \rev{little difference in power between} 
%generally little or no power loss when comparing 
the dCRT 
%to 
\rev{and} its more expensive CRT counterpart \rev{(in this paper we will refer to the CRT implementation as originally proposed without distillation or HRT speedup as the original CRT, or oCRT)}, \rev{and what small differences exist can be explained by factors that are separate from distillation}.
\rev{Mean}while, our proposals save orders of magnitude in computation \rev{over the oCRT} even for medium-scale problems (the savings only increase for larger data). We also show in simulations that the dCRT is comparably powerful to other state-of-the-art conditional independence tests, and is also robust to misspecification in the distribution of $X$. 

In Section~\ref{sec:realdata}, we apply the dCRT to a breast cancer dataset and discover more clinically-informative somatic mutations than competing methods, and we cite independent scientific work corroborating each of the discoveries we make. Finally, we close with a discussion in Section~\ref{sec:discussion}.

The dCRT inherits several attractive properties of the CRT: {it can be derandomized to an arbitrary extent through computation (increasing $M$) and} yields finite-sample valid $p$-values for all variables that can be used for downstream multiple testing analyses with a variety of error metrics, including not only the false discovery rate but also the family-wise error rate and others. 
\dl{This makes the dCRT an appealing alternative to model-X knockoffs \citep{candes2018panning}.}

\subsection{Related work}\label{sec:intro:com}
Our work builds upon the CRT \rev{framework} of \cite{candes2018panning}, with the goal of making it computationally tractable without sacrificing power.
%; see \cite{LJ:2017} for detailed discussion of the differences between this model-X line of work and the canonical approach that assumes a model for $Y\mid (X,Z)$. 
Our work is perhaps most similar in its goal to the HRT of \cite{tansey2018holdout}, which uses data splitting to enable the use of complex modeling in the CRT with far less computation by doing all the complex modeling on the first part of the data and testing on the second part. 
\rev{A domain-specific version of the HRT is applied by \cite{bates2020causal}} to genetic trio studies by (using causal terminology) learning a model on observational data and using it within the CRT on randomized experimental data; 
% A similar approach is adopted in \cite{bates2020causal}, who apply the CRT to genetic trio studies by (using causal terminology) learning a model on observational data and using it within the CRT on randomized experimental data; 
the power of a similar ``hybrid'' CRT approach is studied in \cite{katsevich2020theoretical}. We show in Section~\ref{sec:sim} that data splitting comes with a substantial power loss compared to the dCRT and \rev{oCRT}.
%the original (slower) CRT.
\cite{tansey2018holdout} addresses this with cross-fitting, but in doing so loses the guarantee on Type-I error control of the CRT (and dCRT). Other works have extended the CRT \citep{berrett2018conditional,bellot2019conditional} in ways that do not address its computational intractability. For the variable selection problem, model-X knockoffs \citep{candes2018panning} can simultaneously test conditional independence for each covariate, yielding a false discovery rate-controlling rejection set. {Model-X knockoffs is inherently a \textit{multiple} testing method, with power to detect groups of non-null variables without quantifying their individual significances. On the other hand, the dCRT is a \textit{single} testing method which can be paired with multiple testing procedures if desired. We elaborate further on the comparison between dCRT and model-X knockoffs in the discussion.}

We note a pair of methods, double machine learning \citep{chernozhukov2016double} and the generalized covariance measure \citep{shah2018hardness}, that both test conditional independence under assumptions that nearly (but not quite, due to moment conditions on $Y$) subsume ours, and whose test statistic resembles and can even be identical to certain special cases of the dCRT. However, their statistics only resemble a special case of the dCRT---the dCRT framework includes many other statistics which deviate substantially from double machine learning/generalized covariance measure and can be more powerful in certain settings. Furthermore, the cutoffs for their test statistics are both based on asymptotic normality, while the dCRT is non-asymptotically exact regardless of the distribution of its test statistic (see Appendix~\ref{sim:dml}). 

\subsection{Notation}
Let $I = (i_i,i_2,\dots,i_k) \subseteq\{1,\dots,n\}$ and $J=(j_1,j_2,\ldots,j_{\ell})\subseteq\{1,2,\ldots,p\}$ be subsets of samples and variables, respectively, and consider a matrix $\bA=(\ba_1,\ba_2,\ldots,\ba_p)\in\mathbb{R}^{n\times p}$ with $\ba_j=(A_{1j},A_{2j},\ldots,A_{nj})\trans$. We denote by $A_{I,J}$ the sub-matrix of $A$ with rows in $I$ and columns in $J$. We use the subscripts $j$, $\text{--}J'$, and $\bullet$ as shorthand for $J = \{j\}, \{1, \dots, p\} \setminus J'$, and $\{1, \dots, p\}$, respectively, and the same for the first index. For example, $A_{\bullet, \noj}$ represents the matrix $A$ with the $j$th column removed. For any two vectors $\ba_j$ and $\ba_{\ell}$, let $\ba_j\odot\ba_{\ell}=(A_{1j}A_{1\ell},A_{2j}A_{2\ell},\ldots,A_{nj}A_{n\ell})\trans$ denote their elementwise product, and for $L=\{\ell_1,\ell_2,\ldots,\ell_k\}$ let $\ba_j\odot\bA_L=(\ba_j\odot\ba_{\ell_1},\ldots,\ba_j\odot\ba_{\ell_k})$; these will be used when fitting \dl{first-order} interaction effects. 
\section{The distilled conditional randomization test}\label{sec:method:dCRT}  

\subsection{Main idea}\label{sec:method:mot}  
It is natural to derive CRT test statistics from machine learning methods with high predictive and estimation accuracy. Indeed the original paper proposing the CRT \citep{candes2018panning} used
%\dl{the lasso \citep{tibshirani1996regression} to derive a test statistic and found it to be quite powerful. Specifically, the test statistic was chosen to be} 
as test statistic {$T_{\mathrm{\rev{o}CRT}}(\by,\bx,\bZ) := |\hat{\beta}\suplasso_x|$}, the absolute value of the fitted coefficient on $\bx$ from the lasso {\citep{tibshirani1996regression}} of $\by$ on $(\bx,\bZ)$ with penalty parameter chosen by cross-validation. Although powerful and computationally much faster than many other machine learning algorithms, it is still expensive to repeatedly run the lasso on large data sets hundreds or more times just to compute a single CRT $p$-value, and many times more than that in multiple-testing scenarios when a CRT $p$-value for each covariate is needed.

Consider now the following alternative test statistic which captures the essence of our proposal. First fit a cross-validated lasso of $\by$ on only $\bZ$ to obtain the $p$-dimensional coefficient vector {$\hat{\beta}_z\suploco$}. Then fit a least-squares regression of the residual {$(\by-\bZ\hat{\beta}_z\suploco)$} on $\bx$ to obtain the scalar coefficient {$\hat{\beta}\suploco_x$} and take its absolute value {$T_{\mathrm{dCRT}}(\by,\bx,\bZ) := |\hat{\beta}\suploco_x|$} as the test statistic. {Here, the superscription ``${\sf loco}$" represents ``leave-one-covariate-out" regression as $\x$ is left out when regressing $\by$ solely on $\Z$. We introduce this notation to distinguish the leave-one-covariate-out construction from the \rev{o}CRT lasso statistics when needed, although in the remaining paper, we will just use $(\widehat\beta_x,\widehat\beta_z)$ to represent $(\widehat\beta_x\suploco,\widehat\beta_z\suploco)$ when there is no need to distinguish them from $(\widehat\beta_x\suplasso,\widehat\beta_z\suplasso)$.} It may seem as though little has changed from the preceding paragraph---we would expect $T_{\mathrm{\rev{o}CRT}}$ and $T_{\mathrm{dCRT}}$ to have similar statistical properties and require nearly the same computation. Although the statistical properties of $T_{\mathrm{\rev{o}CRT}}$ and $T_{\mathrm{dCRT}}$ are indeed very similar and they do require nearly the same time to compute {once}, they require dramatically different computation within the CRT. The key difference is that the expensive $(p+1)$-dimensional lasso fit in $T_{\mathrm{\rev{o}CRT}}$ must be recomputed for {each} resample of $\bx$, while the expensive $p$-dimensional lasso fit in $T_{\mathrm{dCRT}}$ must only be computed once, since that lasso does not depend on $\bx$ and hence is identical for all its resamples. In the CRT, neither $\by$ nor $\bZ$ change during the resampling procedure, and we take advantage of this by applying our expensive computation to only $\by$ and $\bZ$ so it only has to be done once. All that is required for each resample's computation of $T_{\mathrm{dCRT}}$ is a {univariate} regression, whose computational expense is much lower than a $p$-dimensional lasso.

We can generalize this idea far beyond the lasso or linear regressions. The core proposal is to {distill} all the high-dimensional information in $\bZ$ about $\by$ into a low-dimensional representation, without looking at $\bx$. Then the test statistic estimates a relationship between $\bx$ and the {leftover} information in $\by$ by only looking at $\bx$, $\by$, and the distilled (low-dimensional) function of $\bZ$. Thus all the computation on high-dimensional data, namely the distillation, only needs to be performed once, while the computation that is repeatedly applied to the resampled data is low-dimensional and hence relatively fast. \dl{It will often be advantageous to also distill the high-dimensional information in $\bZ$ about $\bx$ and include this in the test statistic as well, but we will see this can be done without looking at $\bx$ and hence does not require any repeated computation on the resampled $\bx\supm$.}

\subsection{Formal presentation of dCRT}\label{sec:method:dCRT:frame} 
We now formalize the idea from the previous subsection in Algorithm~\ref{alg:frame}\dl{, the distilled conditional randomization test (dCRT)}.

\begin{algorithm}[htbp]
\caption{The distilled conditional randomization test (dCRT).}
\label{alg:frame}
{\bf Input:} {The distribution of $\bx\mid\bZ$,} data $(\by,\bx,\bZ)$, $\by$-distillation-fitting function $\D$, $\bx$-distillation function $d_x$, test statistic function $T$, and number of randomizations $M$.\\
\vspace{0.2cm}
Distill $\bZ$'s information about $\by$ into $\dy = \D(\by,\bZ)$ and about $\bx$ into $\dx = d_x(\bZ)$.\label{algline:ydist}\\
\vspace{0.2cm}
{\bf For} $m=1,2,...,M$: Sample $\bx\supm$ from the distribution of $\bx\mid\bZ$, conditionally independently of $\bx$ and $\by$.\label{algline:dcrt}\\
\vspace{0.2cm}
{\bf Output:} dCRT $p$-value $\frac{1}{M+1}\left(1+\sum_{m=1}^M \One{T(\by,\bx\supm,\dy,\dx)\geq T(\by,\bx,\dy,\dx)}\right)$.
\end{algorithm}

The key difference from the more general CRT in Algorithm~\ref{alg:crt} is that the test statistic function $T$ in Algorithm~\ref{alg:frame} only sees information about the high-dimensional $\bZ$ through its $\by$- and $\bx$-distillations $\dy$ and $\dx$, which are both computed just once in the first line of the algorithm. $\D$ and $d_x$ should be chosen such that the distillation step produces $\dy$ and $\dx$ with dimension much less than $p$, so that $T$'s inputs are low-dimensional. Then since $T$ is the only repeatedly-applied function and its computation does not suffer from the high-dimensionality of the original data, the dCRT's computation will be dominated by the single application of $\D$. For instance, in the dCRT example in Section~\ref{sec:method:mot}, $\dx$ is not used and $\D$ fits a\dl{cross-validated} lasso of $\by$ on $\bZ$ and returns $\dy = \bZ\hat{\beta}_z$, while $T(\by,\bx,\dy) = |(\by-\dy)^\top\bx|/\|\bx\|^2$ 
%{\bf Molei: should we make it consistent with the example in Section 2.3?} 
requires negligible computation by comparison. 

\dl{Note that $\D$ and $d_x$, despite both producing distillations, operate quite differently. In particular, although $d_x$ distills $\bZ$'s information about $\bx$, it does not take $\bx$ as an argument. This is because the distillation function $d_x$ can be chosen purely based on $X\mid Z$ which is assumed known, and thus can bypass looking at $\bx$ and distill $\bZ$ directly; for example, we will soon discuss dCRTs with $d_x(\bz) = \mathbb{E}[\bx\, |\, \bZ=\bz]$. In contrast, $\D$ needs to internally {fit} a distillation function which we could think of as ``$d_y$" (this is the expensive step) and then apply it to $\bZ$ to compute $\dy$. \dl{This distinction between $\D$ and $d_x$ is important since if $d_x$ performed a complicated fitting step that depended on $\bx$, then that complicated computation would have to be repeated for each $\bx\supm$ in order to maintain the exchangeability of $\bx,\bx^{(1)},\dots,\bx^{(M)}$ under $H_0$ used in the proof of Theorem~\ref{thm:crt}.} Importantly, Theorem~\ref{thm:crt}'s validity guarantee indeed applies to the dCRT because it is a computationally efficient instantiation of (and thus a special case of) the CRT.}

We emphasize that $\D$ can really be {any} regression algorithm and Theorem~\ref{thm:crt} still holds{, since for any choice of $\D$ the dCRT is still a special case of the CRT}. Thus it can take advantage of the predictive power of state-of-the-art machine learning algorithms, precise knowledge in the form of a Bayesian prior, or even imprecise domain expertise or intuition applied by trying many different regressions of $\by$ on $\bZ$ and choosing whichever ``feels" best (as long as $\bx$ is not factored into that decision). In the sequel we provide some suggestions and default choices. 

\dl{Algorithm~\ref{alg:frame} provides a framework for fast and powerful CRTs, but leaves much unspecified. In the next two sections, we provide more detail on some ways that the dCRT can be implemented in different scenarios, and discuss their associated advantages and disadvantages.}

\subsection{The d$_0$CRT: fast, powerful, and intuitive}\label{sec:power}
The most computationally-efficient and intuitive class of dCRT procedures has both $\by$- and $\bx$-distillations reduce $\bZ$ to an output with a single column. We label this subclass of dCRT procedures as d$_0$CRT because it represents the choice to maximally-distill each row of $\bZ$ down to a single scalar. Assuming $T$'s computation generally increases with the dimension of its inputs, the d$_0$CRT also represents a particularly computationally-efficient class of dCRTs.

A natural approach to constructing a d$_0$CRT, especially when $Y$ is continuous, is to have distillation take the form of conditional mean functions. That is, let $d_x(\bZ) = \mathbb{E}[\bx\,|\,\bZ]$ and have $\D$ fit an estimate of the analogous regression function for $\by$, i.e., $\D(\by,\bZ)\approx \mathbb{E}[\by\,|\,\bZ]$. Then $T$ can be chosen as an empirical measure of dependence between the residuals $\by-\dy$ and $\bx-\dx$, such as the square of the fitted coefficient when regressing the former on the latter. This approach is also easy to understand and implement since it just requires choosing $\D$ and $T$, with $\D$ just performing a (possibly nonparametric) regression while $T$ can be thought of as computing a test statistic for testing the independence between two scalar random variables from a paired sample of size $n$: $(\by-\dy,\bx-\dx)$. As both regression and bivariate independence testing are highly-studied topics, users can easily draw from their statistical training, domain expertise, and a rich literature in order to design an appropriate d$_0$CRT for their particular problem. The following is a generic example we found to be computationally efficient and powerful in our simulations.

\begin{example}[lasso-based d$_0$CRT]\label{ex:1}
\dl{Let} $\dy=\bZ\hat{\beta}_z$ \dl{be} is the fitted predictions from a cross-validated lasso of $\by$ on $\bZ$, \dl{let} $\dx = \mathbb{E}\left[\bx\,|\,\bZ\right]$, and \dl{let} {$T(\by,\bx,\dy,\dx) = |\hat{\beta}_x|:= \frac{|(\by-\dy)^\top(\bx-\dx)|}{\|\bx-\dx\|^2}$}.
\end{example}

{More generally, the d$_0$CRT's distillation need not be couched in terms of finding conditional means.} \dl{In spite of the intuitive appeal of couching distillation in terms of finding conditional means, in some problem instances an alternative d$_0$CRT may be more appropriate.} For instance, an appealing analogue of Example~\ref{ex:1} for binary $Y$ might fit $\hat{\beta}_z$ by a cross-validated $L_1$-penalized logistic regression of $\by$ on $\bZ$ and otherwise leave $\D$ and $\dx$ unchanged (note $\bZ\hat{\beta}_z$ no longer approximates $\mathbb{E}[\by\, |\, \bZ]$), and take $T(\by,\bx,\dy,\dx)$ to be {absolute value of the} fitted coefficient from a logistic regression of $\by$ on $\bx-\dx$ with offset $\dy$. \dl{The substantial flexibility of the d$_0$CRT allows it to detect many kinds of nonlinear relationships between $Y$ and $X$, but the stringent distillation inherently limits its ability to detect most types of {interactions} between $X$ and $Z$. This shortcoming can be important and in the next subsection we discuss how interactions can be incorporated by moving beyond the d$_0$CRT.}

\subsection{The d$_\mathrm{I}$CRT: accounting for interactions}\label{sec:dIcrt}
Of the three functions applied in Algorithm~\ref{alg:frame}, only $T$ takes both $\by$ and $\bx$ as arguments and hence the choice of $T$ is how a user can encode the kinds of non-null relationships between $Y$ and $X$ that are deemed plausible. But because $T$ only sees $\bZ$ through $\dy$ and $\dx$, any plausible models for $Y$ must be expressed using only $\bx$, $\dy$, and $\dx$. This means that the d$_0$CRT has almost no capacity to model even first-order interactions between $X$ and $Z$. 
For instance, suppose $p=3$ and $Z_j\stackrel{i.i.d.}{\sim}\N(0,1)$, $X \sim Z_1+\N(0,1)$, and $Y \sim Z_2+XZ_3 + \N(0,1)$. Then the best possible distillations of $\bx$ and $\by$ are $\dx=\bZ_1$ and $\dy=\bZ_2+\bZ_1\odot\bZ_3$, making it impossible for $T$ to encode the true conditional mean of $\by$, namely, $\bZ_2+\bx\odot\bZ_3$, from just $\bx$, $\dx$, and $\dy$.

To address this limitation of the d$_0$CRT, one can simply increase the dimension of $\dy$ and $\dx$ to explicitly include possible columns of $\bZ$ with which $\bx$ might be expected to interact. But of course increasing the dimension of $\dy$ and $\dx$ tends to come at a computational cost, since their low-dimensionality is exactly what makes the dCRT fast in the first place. Thus one needs some sort of prior, domain knowledge, or heuristic for choosing based on either the pair $(\by,\bZ)$ or $(\bx,\bZ)$ (but not based on $(\by,\bx,\bZ)$ together) a small subset of columns of $\bZ$ that $\bx$ might plausibly interact with. One option is to split the data into two independent parts and use one part in an unconstrained way to select columns of $\bZ$ that are likely to interact with $\bx$, and then to leverage these selections in a dCRT run only on the other part. We propose here an alternative that avoids sample splitting, based on the common statistical practice of only allowing for interactions between variables with strong main effects. This practice of enforcing hierarchy in interactions has a long history in applied and theoretical statistics under many different names \citep{nelder1977reformulation,cox1984interaction,peixoto1987hierarchical,hamada1992analysis,chipman1996bayesian,bien2013lasso}. 

Our proposed method for incorporating interactions, which we call the d$_\mathrm{I}$CRT, is to have $\D$ still distill $\bZ$ into one column to best-capture the relationship between $\by$ and $\bZ$, but then to additionally return (as further columns of $\dy$) a limited subset of columns of $\bZ$ whose contributions to that fitted relationship are strongest. Then $T$ can be chosen as a test statistic that allows $\bx$ to interact with those columns of $\bZ$ contained in $\dy$, while still prioritizing the main effect of $\bx$. As a generic example we found to be powerful to detect hierarchical interactions without losing much power in the absence of interactions, consider the following.

\begin{example}[lasso-based d$_\mathrm{I}$CRT]\label{ex:2}
\dl{Let} $\dy=(\bZ\hat{\beta}_z,\bZ_{\bullet,\mathrm{top}(k)}):=(\dyone,\dynone)$ \dl{be} is the fitted predictions from a cross-validated lasso of $\by$ on $\bZ$ concatenated with the columns of $\bZ$ corresponding to the $k$ largest entries of $|\hat{\beta}_z|$, \dl{let} $\dx = \mathbb{E}\left[\bx\,|\,\bZ\right]$, and \dl{let} $T(\by,\bx,\dy,\dx) = \hat{\beta}_{x,1}^2+\frac{1}{k}\sum_{j=2}^{k+1}\hat{\beta}_{x,j}^2$, where $\hat{\beta}_x\in\mathbb{R}^{k+1}$ are the fitted coefficients from a least-squares fit of $(\by-\dyone)$ on $(\bx-\dx)$ and $(\bx-\dx)\odot\dynone$.
\end{example}

The normalization by $1/k$ of $\sum_{j=2}^{k+1}\hat{\beta}_{x,j}^2$ encodes our hierarchical prioritization of the main effect $\hat{\beta}_{x,1}$ over the interaction effects. For small $k$ we still expect the computation to be dominated by $\D$, but it also represents a statistical trade-off in how widely to search for interactions; we found the performance to be quite stable to $k$ in our simulations, but set as a default $k=\lceil 2\log(p)\rceil$. Note that $k$ could also be chosen after looking at $(\by,\bZ)$, and more generally, one can construct many different types of d$_\mathrm{I}$CRT. For instance, one can adapt Example~\ref{ex:2} to binary $Y$ in an analogous way as was done for Example~\ref{ex:1} by replacing linear regressions with logistic regressions and using $\dyone$ as an offset in $T$. Or one could have $\D$ and/or $T$ use the predictions and default variable importance measures from a random forest. We explore some of these options in simulations in Section~\ref{sec:sim}.

\subsection{Running the dCRT without resampling} \label{sec:main:mcf}
Distillation \rev{provides massive computational savings within}
%massively reduces the computation time of 
the CRT by only requiring a single evaluation of the by-far-most-expensive function $\D$. But it still requires $M+1$ evaluations of $T$, which can sometimes still contribute nontrivially to the computation time, and requires the user to choose the tuning parameter $M$ which trades off computation and statistical power. It turns out that in certain cases the simplicity of $T$ in the dCRT can be leveraged to remove the resampling of $\bx\supm$ entirely and compute an exact $p$-value directly from the single function evaluation $T(\by,\bx,\dy,\dx)$.

For intuition, suppose $X\mid Z\sim\N(Z^\top\gamma,\sigma^2)$, and consider the d$_0$CRT with $T$ as in Example~\ref{ex:1}, \[{T(\by,\bx,\dy,\dx) = \frac{|(\by-\dy)^\top(\bx-\dx)|}{\|\bx-\dx\|^2}.}\] 
Then since the (d)CRT conditions on $\by$ and $\bZ$ (and hence also $\dy$ and $\dx=\bZ\beta$), 
\begin{equation}
(\by-\dy)^\top(\bx-\dx)\sim\N\left(0,\sigma^2\|\by-\dy\|^2\right).
\label{centered-null}
\end{equation}
The denominator of $T$ makes things a bit more complicated, but the nature of the statistic does not change much if we replace the denominator by its expectation or, equivalently (since multiplying $T$ by a fixed constant has no effect on its resulting $p$-value), simply replace it by  {$T'(\by,\bx,\dy,\dx) = |(\by-\dy)^\top(\bx-\dx)|$}. We then get immediately that the {exact} $p$-value --- i.e., the $p$-value that would result from taking the limit as $M\rightarrow\infty$ --- can be computed as {$2\left(1-\Phi\left(\frac{T'(\by,\bx,\dy,\dx)}{\sigma\|\by-\dy\|}\right)\right)$} without ever resampling $\bx\supm$ or recomputing $T'$, where $\Phi$ is the standard normal cumulative distribution function.

The same principle can be applied to non-Gaussian $X$: since the distribution of $(\bx-\dx)\mid\bZ$ is known and the rows are independent, $(\bx-\dx)$ can be element-wise transformed via scalar monotone functions to be i.i.d. $\N(0,1)$ given $\bZ$. For conditionally-continuously-distributed $(\x-\dx)$, this can be done via the probability inverse transform, while for distributions with atoms the atoms need to be carefully randomized (though just once); see Appendix~\ref{sec:app:mcf} for details.

As long as $(\bx-\dx)$ is independent Gaussian or transformed to be, the same principle can also be applied to some more complex $T$ functions. For instance, in Example~\ref{ex:2}, we can again replace the random ``denominator" (in this case the matrix inverse in the least-squares formula for $\hat{\beta}_x$) with its conditional expectation given $\bZ$, and end up with a quadratic form in Gaussian random variables. Efficient algorithms for computing the quantiles of a quadratic form in Gaussian random variables exist \citep{duchesne2010computing} and can be applied to again compute the exact dCRT $p$-value without any resampling; see Appendix~\ref{sec:app:mcf} for details.

\section{Variable selection and multiple testing via the dCRT}\label{sec:speedup}

Conditional independence testing is often done in the context of a variable selection problem. Given $p$ covariates $X_1, \dots, X_p$ and a response $Y$, the goal is to discover the covariates $X_j$ that are conditionally associated with the response, i.e., $Y \not \indp X_j \mid X_{\noj}$. For a given $j$, we arrive at the problem formulation from the previous two sections by setting $X = X_j$ and $Z = X_{\noj}$. This change of notation highlights the fact that the effects of all variables are of interest, rather than that of one special variable. Given a design matrix $\bX \in \RR^{n \times p}$ and a response vector $\by$, we propose to approach the variable selection problem by applying the dCRT to $(\by, \bx, \bZ) = (\by, \bX_{\bullet, j}, \bX_{\bullet,\noj})$ for each covariate $j$, followed by a multiple testing procedure on the resulting $p$-values. 
% This solution is conceptually straightforward, but involves statistical and computational subtleties. We address the former in Section~\ref{sec:multiple-testing} and the latter in Section~\ref{sec:computational}.
%\subsection{Family-wise error rate, false discovery rate, and knockoffs} \label{sec:multiple-testing}
Two common error rates to control are the family-wise error rate and the false discovery rate. The former can be easily achieved based on the Bonferroni correction, which works under arbitrary $p$-value dependence. The latter is usually done via the Benjamini--Hochberg procedure. Even though the $p$-values are technically not positively dependent in the sense required for mathematical false discovery rate control \citep{BY01}, the Benjamini--Hochberg procedure is known to be very robust to dependent $p$-values in all but adversarially-constructed settings, as confirmed in our simulations.

Regardless of error rate, the straightforward application of the dCRT to the variable selection problem requires computing $\D$ a total of $p$ times, once for each variable. Note that these are entirely parallel computations, so for certain problem dimensionalities and parallel computing resources, this is entirely feasible. However, in large-scale variable selection applications such as genome-wide association studies, there may be too many covariates for the direct application of dCRT to each. 
In the following subsections we present two computational shortcuts that make variable selection via the dCRT feasible for large-scale applications. %but first

\subsection{Data-dependent screening of variables}\label{sec:screening}

A natural acceleration of the dCRT for variable selection is to first use the data to identify a preliminary subset $\mathcal S \subseteq \{1,\dots, p\}$ of promising covariates via a screening function $S: (\bX,\by) \mapsto \mathcal S$. We can then compute (d)CRT $p$-values $p_j(\bX, \by)$ {(via Algorithms~\ref{alg:crt} or \ref{alg:frame})} for only $j \in \mathcal S$ while setting the $p$-values for all the other covariates to 1, yielding the screened $p$-values 
\begin{equation}
p'_j(\bX, \by) = 
\begin{cases}
p_j(\bX, \by) &\text{if } j \in S(\bX,\by); \\
1 &\text{if } j \not \in S(\bX,\by).
\end{cases}
\label{screening}
\end{equation}
For instance, $\mathcal S$ could be the active set of a cross-validated lasso fit of $\by$ on all the covariates.

In general, a screening step like this applied before the (d)CRT breaks the exchangeability between the original and resampled test statistics which Theorem~\ref{thm:crt} relies on to guarantee $p$-value validity. \dll{The intuitive reason for the failure of exchangeability is that the screening, when it discards covariates, takes a (data-dependent, and thus random) subset of covariates and implicitly changes their (d)CRT test statistics to ensure a $p$-value of 1 is returned. Hence, the screening implies an $\bx$-dependent choice of $T$, whose distribution under the null will then be different when its argument is $\bx$ versus $\bx\supm$.} Despite this failure of exchangeability, the screening can only inflate a $p$-value and thus does not affect its validity. \dll{Indeed, since $p_j(\bX, \by) \leq p_j'(\bX, \by)$, the validity of the CRT $p$-value $p_j(\bX, \by)$ implies that
\[ 
\mathbb{P}(p'_j(\bX, \by) \leq u) \le \mathbb{P}(p_j(\bX, \by)\le u)\le u,
\]
making $p'_j$ also a valid $p$-value. The above discussion is summarized as the following theorem.}

\begin{theorem}\label{thm:valid}
Let $j$ be a null variable. For any screening rule $S$, the screened $p$-value $p_j'(\bX, \by)$ obtained from equation~\eqref{screening} is stochastically larger than uniform.\dll{, i.e.,
\[ \mathbb{P}(p_j'(\bX, \by)\le u)\le u \text{ for all } u\in [0,1].\]}
\end{theorem}
{
\begin{proof}
By equation~\eqref{screening}, for any $u\in [0,1]$, 
$\mathbb{P}(p'_j(\bX, \by) \leq u) \le \mathbb{P}(p_j(\bX, \by)\le u)\le u$.
\end{proof}
}
Thus, with the small computational overhead of a single well-chosen screening function, we can expect to dramatically cut the computation time of using the (d)CRT for variable selection. \dll{Note that the screening procedure increases the $p$-values relative to their unscreened counterparts, but it nevertheless has no impact on power as long as it does not screen away any non-null $p$-values {that would have been rejected by a multiple testing procedure}, which is far less-stringent and more achievable than requiring the screening not to screen away {any} non-null $p$-values. In other words, imagine applying a Bonferroni correction (or the Benjamini--Hochberg procedure) to the screened $p$-values: there would only be a loss in power if the screening procedure failed to select a variable whose CRT $p$-value was less than $\alpha/p$---however, any reasonable screening procedure will not fail to pick up such a strong signal.} Indeed we found in our simulations that simple screenings substantially decreased computation time without affecting the power.

\subsection{Recycling computation for L$_1$-regularized M-estimators} \label{sec:recycling}

In some cases, we may want to compute $p$-values for all variables under consideration, even if only a small fraction of these are statistically significant. For instance, these may be needed for downstream analysis tasks like calibration assessment or meta-analysis. In such cases, we must look beyond the screening approach. In this section, we present a way of recycling computation for $L_1$-regularized $M$-estimators including the lasso (recall Examples~\ref{ex:1} and~\ref{ex:2}). This reduces the number of $\mathcal D_y$ computations from $p$ to $|\mathcal A|$, where $\mathcal A$ is the active set of the lasso on $(\bX, \by)$.

Let $\mathcal D_y$ be the cross-validated lasso with strictly convex and differentiable loss function $\mathcal \ell$. Variable selection via the dCRT based on this distillation function requires computing
\begin{equation}
\widehat \beta(\bX_{\bullet,\noj}, \by; \lambda) := \underset{\beta \in \RR^{p-1}}{\arg \min}\ \sum_{i = 1}^n \ell(Y_i, X_{i,\noj}\beta) + \lambda \|\beta\|_1
\label{lasso}
\end{equation}
for each $j = 1, \dots, p$, along a grid of regularization parameters.
%To avoid degeneracies, we assume throughout that the columns of $\bX$ are in general position. This, together with the strict convexity and differentiability of $\mathcal \ell$, guarantees that $\widehat \beta$ is unique~\citep{tibshirani2013lasso}. 
There is redundancy among these $p$ lasso problems; they all differ from the the full lasso problem on $(\bX, \by)$ by just one variable. We may therefore expect that we can save computation by somehow recycling computation across these lasso problems. The next lemma suggests a means to this end:

\begin{lemma} \label{lem:lime-speedup}
Suppose the columns of $\bX$ are in general position and that the loss $\mathcal \ell$ is differentiable and strictly convex. Then, for any $\lambda > 0$,
\begin{equation}
\widehat \beta_{j}(\bX, \by; \lambda) = 0 \quad \Longrightarrow \quad \widehat \beta(\bX_{\bullet,\noj}, \by; \lambda) = \widehat \beta_{\noj}(\bX, \by; \lambda).
\label{lime-speedup}
\end{equation}
\end{lemma}
%The first part of this result was proved by Tibshirani~\cite[Lemma 5]{Tibshirani2013}, and the second half is proved in Appendix~\ref{sec:proofs}. 
In words, Lemma~\ref{lem:lime-speedup} states that removing an inactive variable from the lasso does not change the fitted coefficient vector. This has important computational implications (potentially even outside the scope of this paper)---it suggests that we can avoid refitting the lasso~\eqref{lasso} for most variables $j$, instead recycling the lasso fit on the full design matrix. Of course, the parameter $\lambda$ is usually tuned via cross-validation, which introduces extra complications. However, we claim that if $\lambda$ is chosen in an appropriate data-dependent way, then an analogous result will still hold.

To make this precise, consider a grid of regularization parameters
\begin{equation}
\lambda(1) > \lambda(2) > \cdots > \lambda(G) > 0
\label{grid}
\end{equation}
and a corresponding set of cross-validation errors $\mathcal E_1, \dots, \mathcal E_G$. Define a rule $\widehat g$ to select the penalty parameter $\lambda$ based on cross-validation errors $\mathcal E_1, \dots, \mathcal E_G$ to be \textit{sequential} if these values  are traversed in this order, and at some stopping time $\widetilde g$, the algorithm terminates and chooses $\lambda({\widehat g})$ for some $\widehat g \leq \widetilde g$. For example, for any integer $\Delta \geq 1$, the following rule is sequential: 
\begin{equation*}
\widehat g \equiv \min\{g: \mathcal E_g \leq \min(\mathcal E_{g+1}, \dots, \mathcal E_{g + \Delta})\},
\end{equation*}
which is the first time along the regularization path that the cross-validation error is smaller than the following $\Delta$ steps (the first `local minimum' on the cross-validation path, and the `sparsest' of all such local minima). In this case the stopping time is $\widetilde g = \widehat g + \Delta$. The lasso with any sequential rule $\widehat g$ has the property~\eqref{lime-speedup}. 
\begin{theorem}\label{prop:lasso-shortcut}
Fix a grid of regularization parameters~\eqref{grid}. Consider applying $L_1$-regularized regression with loss $\ell$ on the whole data $(\bX, \by)$, with $\lambda$ selected by $K$-fold cross-validation and a sequential stopping rule $\widehat g$. Let $\widehat g(\bX, \by)$ and $\widetilde g(\bX, \by)$ be the resulting grid point and stopping time, respectively. Letting $\{1, \dots, n\} = D_1 \cup \cdots \cup D_K$ denote the split of the data into non-overlapping folds, define the active set
\begin{equation}
\begin{split}
\mathcal A = \{j \in \{1, \dots, p\}:\ &\widehat \beta_j(\bX, \by; \lambda(\widehat g(\bX, \by))) \neq 0 \text { or } \\
&\widehat \beta_j(\bX_{\text{--}D_k,\bullet}, \by_{\text{--}D_k}; \lambda(g)) \neq 0 \text{ for some } k, g \leq \widetilde g(\bX, \by)\}.
\end{split}
\label{active-set}
\end{equation}
If the loss $\ell$ is differentiable and strictly convex, and the columns of $\bX$ and $\bX_{\text{--}D_k,\bullet}$ are in general position for each $k$, then excluding non-active variables $j$ does not alter the fitted coefficients: for each $j \not \in \mathcal A$,
\begin{equation}
\widehat g(\bX_{\bullet,\noj}, \by) = \widehat g(\bX, \by) \ \text{and} \ \widehat \beta(\bX_{\bullet,\noj}, \by; \lambda(\widehat g(\bX_{\bullet,\noj}, \by))) = \widehat \beta_{\noj}(\bX, \by; \lambda(\widehat g(\bX, \by))).
\label{conclusion}
\end{equation}
\end{theorem}
Theorem~\ref{prop:lasso-shortcut} states that for each variable $j$ \textit{not} in the active set, we need not re-run the lasso holding out variable $j$; we can instead fit the full lasso once and then read off the coefficient vector. This computational shortcut, summarized in Algorithm~\ref{alg:loco-crt-lasso}, reduces the number of lasso applications required by the dCRT from $p$ to $|\mathcal A|$. Depending on the sparsity of the problem, this reduction can save several orders of magnitude of computation. It is known that at most, the lasso solution has $\min(p,n)$ nonzero entries \citep{tibshirani2013lasso}, though often it is much sparser.

\section{Statistical performance of the dCRT}\label{sec:sim}

{
\subsection{Implications of distillation for power} \label{sec:dcrt-vs-crt}

Our motivation for proposing the dCRT is computational; using distilled test statistics accelerates the CRT by orders of magnitude compared to the originally-proposed lasso coefficient test statistic. In this section, we discuss the statistical implications of this computational acceleration. While distillation is a flexible framework that can encompass a variety of test statistics, for concreteness in this section we narrow our focus to the d$_0$CRT. Our goal is to carefully compare the d$_0$CRT to its ``undistilled" counterpart, the \rev{o}CRT based on the absolute lasso coefficient.
%(henceforth referred to simply as ``CRT"). 
Our main conclusion is that, perhaps surprisingly, distillation does not have much effect on the the power of the CRT. We present the main reasoning behind this conclusion here and defer the details to Appendix~\ref{sec:dCRT-versus-CRT}.

To emphasize the exclusion of $\bm x$ from the lasso regression, let $\widehat \beta\suploco_z$  be the fitted coefficients in the lasso regression of $\bm y$ on $\bm Z$ and let $\widehat \beta\suploco_x = (\bm x - \bm d_x)^{\top} (\bm y - \bm Z \widehat\beta\suploco_z)/\|\bm x - \bm d_x\|^2$ as in the definition of the d$_0$CRT. By contrast, let $(\widehat \beta_x\suplasso, \widehat \beta_z\suplasso)$ denote the fitted coefficients in the lasso regression of $\bm y$ on $\bm x$ and $\bm Z$, so the \rev{o}CRT is based on the test statistic $|\widehat \beta_x\suplasso|$. Let us also assume in this section, as we did in Section~\ref{sec:main:mcf}, that $X\mid Z \sim N(Z^{\top} \gamma, s^2)$. Finally, suppose (for intuition) that $Y\mid X,Z$ follows a Gaussian linear model with coefficients $\beta_x$ and $\beta_z$.

The obvious difference between \rev{o}CRT and dCRT is that the latter is based on a lasso regression excluding the variable of interest while the former is based on a lasso regression on all variables. Thus, $\widehat \beta_z\suploco \neq \widehat\beta\suplasso_z$ and so of course $\widehat \beta_x\suploco \neq \widehat \beta_x\suplasso$. However, there are two additional differences that must be accounted for in order to  understand the relationship between the two methods. First of all, $\widehat \beta_x\suplasso$ has a nonzero probability of being equal to zero, while $\widehat \beta_x\suploco$ is almost surely nonzero. Secondly, $\widehat \beta\suplasso_x$ does not necessarily have a null distribution centered on zero, whereas $\widehat \beta_x\suploco$ does (recall equation~\eqref{centered-null}). 

In Appendix~\ref{sec:dCRT-versus-CRT}, we examine the impacts of these two properties of the \rev{o}CRT. We find that the sparsity that $\widehat \beta_x\suplasso$ inherits from the lasso can only hurt the power of the \rev{o}CRT and propose a simple alternative based on removing the soft threshold operator. Furthermore, we show that, depending on the distribution of $X\mid Z$ and on the locations and signs of the nonzero elements of $(\beta_x, \beta_z)$, the null distribution of $\widehat \beta_x\suplasso$ can either be centered at the origin, or left or right of the origin. By using the absolute value of the potentially off-center test statistic $\widehat \beta_x\suplasso$, the \rev{o}CRT gains or loses power to the extent that the null distribution is shifted to the right or left, respectively. This motivates us to propose a centered and non-soft-thresholded version of the \rev{o}CRT test statistic.

Using numerical simulations, we found essentially no difference between the performance of the dCRT and the centered, non-soft-thresholded version of the \rev{o}CRT. In other words, after accounting for the aforementioned two differences, the distillation step has little or no impact on the power of the CRT. This conclusion may seem surprising, since on first glance leaving $\bm x$ out appears to cause some of the signal (namely the contribution of $\bm x$ to $\bm y$ that can instead be explained by $\bm Z$) to be regressed out. One may expect this effect to decrease the power of the dCRT. However, this is not the case because it is precisely the component of $\bm x$ that \textit{cannot} be explained by $\bm Z$ that carries signal. Therefore, dropping $\bm x$ and regressing $\bm Z$ out of $\bm y$ first does not have much effect on the power of the dCRT. This intuition would be precise if $\widehat \beta_z\suploco$ were obtained from (unpenalized) linear regression of $y$ on $\bm Z$. Indeed, it is a well-known property of linear regression that the coefficient of $\bm x$ can be obtained by first regressing $\bm Z$ out of $\bm x$ and $\bm y$, and then regressing the residual of $\bm y$ onto that of $\bm x$.
}

\subsection{\rev{Numerical comparisons of power, speed, robustness, and stability}} \label{sec:more-simulations}

\rev{Going beyond the numerical simulations in the previous section,} we designed an extensive simulation suite to systematically assess the power \rev{and other operating characteristics} of dCRT and compare it to that of several alternative methods. Preferring to compare the dCRT to existing methods, we chose to benchmark it against the originally proposed \rev{o}CRT instead of the modified version considered above. We must keep in mind, however, that this choice also complicates the comparison for the aforementioned reasons. Furthermore, we suspect that the centering and soft-thresholding issues may impact the performance of knockoffs as well. Unlike for the CRT, however, it is harder to pull these aspects apart for knockoffs. The soft thresholding affects both the one-bit $p$-values and the ordering of the variables, so removing it may not result in a uniform improvement like it did for \rev{o}CRT. Regarding centering, it is not obvious how to recenter the knockoffs null distribution because knockoffs does not really use a null distribution. We leave the study of these phenomena for knockoffs to future work, and in the meantime compare dCRT to published implementations of the latter.

In the interest of space we defer the details of our simulations to the appendix and present here a detailed summary of the takeaways of those simulations, directly linking each takeaway to the figure and section of the appendix with the corresponding simulation(s) supporting it. The main focus of our simulations is examining the performance of the dCRT through the d$_0$CRT and d$_\mathrm{I}$CRT given by Examples~\ref{ex:1} and~\ref{ex:2}, respectively. Except where explicitly stated otherwise, we apply them in a resampling-free manner per Section~\ref{sec:main:mcf} and, when simulating a variable selection task, with screening (using the cross-validated lasso for selection) per Section~\ref{sec:screening}. For variable selection simulations, we take each of the $p$-value methods (\rev{o}CRT, dCRT, HRT) and apply the Benjamini--Hochberg procedure when targeting false discovery rate control and the Bonferroni correction when targeting family-wise error rate control. Source code for running the dCRT and reproducing our results can be found along with example scripts for illustration at \url{https://github.com/moleibobliu/Distillation-CRT}.

{
We compared the dCRT to the \rev{o}CRT in a broader set of simulations than those referenced in Section~\ref{sec:dcrt-vs-crt}, including linear and logistic regression models and d$_I$CRT as well as d$_0$CRT. We chose the smaller problem size of $n=p=300$ to accommodate the computational burden of the \rev{o}CRT. We found that distillation dramatically reduces CRT computation; both the d$_0$CRT and d$_\mathrm{I}$CRT conferred a computational savings of approximately 500 times over the \rev{o}CRT (Table~\ref{tab:com}). The relative powers of dCRT and \rev{o}CRT (Figures~\ref{fig:smc}, \ref{fig:smc:fwer}) were consistent with what we found in Section~\ref{sec:dcrt-vs-crt}, with dCRT sometimes more powerful and sometimes less powerful than the \rev{o}CRT. The \rev{o}CRT was more powerful when signals were equally spaced while the dCRT was more powerful when signals were adjacent to each other. We suspect these differences to be caused mainly by the discussed soft thresholding and centering issues. See Appendix~\ref{sec:sim:smc} for details.}

%\paragraph{The dCRT is more powerful than other fast model-X methods.} In both the aforementioned $n=p=300$ simulations and a larger simulation with $n=p=800$, the dCRT computation times were mostly within an order of magnitude of the HRT and knockoffs (Tables~\ref{tab:com} and~\ref{tab:com:large}). But across settings that included a range of $n$ up to 1400, a range of $p$ up to 3200, a range of signal magnitudes, a range of sparsities, a range of covariance structures for $X$, and a range of models for $Y\mid X$, both dCRT methods had consistently and substantially higher power than both the HRT (up to about 70 percentage points higher) and knockoffs (up to about 30 percentage points higher in less-sparse settings; in very sparse settings knockoffs becomes powerless and the dCRT can have arbitrarily-higher power). See Figures~\ref{fig:smc},~\ref{fig:diffnp}, and~\ref{fig:diff:design}, and Appendices~\ref{sec:sim:smc} and~\ref{sec:sim:hrt} for details.

The dCRT is more powerful than the HRT: In both the aforementioned $n=p=300$ simulations and a larger simulation with $n=p=800$, the dCRT computation times were mostly within an order of magnitude of the HRT (Tables~\ref{tab:com} and~\ref{tab:com:large}). But across settings that included a range of $n$ up to 1400, a range of $p$ up to 3200, a range of signal magnitudes, a range of sparsities, a range of covariance structures for $X$, and a range of models for $Y\mid X$, both dCRT methods had consistently and substantially higher power than the HRT (up to about 50 percentage points higher); see Figures~\ref{fig:smc},~\ref{fig:smc:fwer},~\ref{fig:diffnp},~\ref{fig:diffnp:equal},~\ref{fig:diffnp:fwer},~\ref{fig:diffnp:equal:fwer},~\ref{fig:diff:design}. See Appendices~\ref{sec:sim:smc} and~\ref{sec:sim:hrt} for details.

{When controlling false discovery rate, the relative performance of dCRT and knockoffs varies across simulation settings, similar to the relative performance of the dCRT and \rev{o}CRT. The dCRT methods tend to have higher power than knockoffs when signal variables are adjacent and lower power than knockoffs when the signal variables are equally spaced. The power comparison between dCRT and knockoffs is a subtle one, and we leave its further investigation for future work.} In very sparse settings, dCRT still has power, while knockoffs does not (due to its reliance on the Selective SeqStep+ procedure \citep{barber2015controlling}). In such regimes, the family-wise error rate may be more appropriate, and the dCRT can be used to control this error rate as well. See Figures~\ref{fig:smc},~\ref{fig:smc:fwer},~\ref{fig:diffnp},~\ref{fig:diffnp:equal},~\ref{fig:diffnp:fwer},~\ref{fig:diffnp:equal:fwer},~\ref{fig:diff:design}. The dCRT is more computationally expensive than knockoffs, but usually within an order of magnitude (Tables~\ref{tab:com} and~\ref{tab:com:large}). {Finally, the dCRT has substantially less algorithmic variability than knockoffs, as measured by the expected Jaccard similarity between two rejection sets obtained by re-running the methods with different seeds (Figure~\ref{fig:variability})}. See Appendices~\ref{sec:sim:smc}, ~\ref{sec:sim:hrt}, and \ref{sec:knockoffs-variability} for details.

%\paragraph{Double machine learning and the generalized covariance measure perform nearly identically to the d$_0$CRT except when $Y\mid X,Z$ has heavy tails, in which case double machine learning and generalized covariance measure fail to control the Type-I error.} We applied double machine learning and generalized covariance measure in nearly all the simulations we ran, and found their performances so similar to the d$_0$CRT that we did not bother plotting them. Their computation times were also similar to the dCRT (with double machine learning's somewhat higher due to cross-fitting). Recall however that both double machine learning and generalized covariance measure rely on asymptotic normality and hence require both large samples and well-behaved tails (unlike the dCRT which is exact for any $n$ and any data distribution); indeed when we simulated $Y\mid X,Z$ to have a Laplace distribution and set $n=30$, double machine learning's and generalized covariance measure's Type-I error was inflated to about 150\% of nominal (Figure~\ref{fig:dml}). And, of course, in settings in which the d$_\mathrm{I}$CRT has higher power than the d$_0$CRT, it similarly outperforms double machine learning and generalized covariance measure (Figures~\ref{fig:rf},~\ref{fig:dglm}). See Appendices~\ref{sim:dml},~\ref{sec:sim:dk} and~\ref{sec:rf} for details.

The d$_\mathrm{I}$CRT is stable to the choice of $k$ and has slightly less power than the d$_0$CRT in additive models but can have substantially higher power in the presence of interactions: In a simulation with an additive model, the power of the d$_\mathrm{I}$CRT was identical as $k$ ranged from 2--22 (the default value of $k=2\log(p)$ would have been 13), while in a model with five true interactions, the power only varied from about 50\% to about 40\% over the same range of $k$ (Figure~\ref{fig:choose:k}). Throughout all our simulations in additive models we found the d$_0$CRT to be slightly but consistently more powerful than the d$_\mathrm{I}$CRT (e.g., Figures~\ref{fig:smc},~\ref{fig:smc:fwer},~\ref{fig:diffnp},~\ref{fig:diffnp:equal},~\ref{fig:diffnp:fwer},~\ref{fig:diffnp:equal:fwer},~\ref{fig:diff:design},~\ref{fig:robust:sec},~\ref{fig:robust:X:power}), but in the presence of interactions obeying the hierarchy principle discussed in Section~\ref{sec:dIcrt}, we found that the d$_\mathrm{I}$CRT could be quite a bit more powerful (up to about 25 percentage points) than the d$_0$CRT (Figure~\ref{fig:dglm}).
See Appendices~\ref{sec:sim:dk} and~\ref{app:dIcrt:k} for details.

The dCRT can leverage nonparametric machine learning algorithms for substantial power gains in highly-nonlinear models: In a simulation in which $X$'s relationship with $Y$ was highly-nonlinear and interacted with five $Z_j$'s, our default (lasso-based) d$_\mathrm{I}$CRT had somewhat higher power than d$_0$CRT (as much as about 20 percentage points), but a different, random-forest-based d$_\mathrm{I}$CRT had far higher power than the lasso-based d$_\mathrm{I}$CRT (as much as about 50 percentage points) (Figure~\ref{fig:rf}). See Appendix~\ref{sec:rf} for details.

The dCRT is quite robust to misspecification of $X$'s distribution: When the distribution of $X\mid Z$ is Poisson even with a very small mean parameter (making it highly discrete and heavily skewed) but approximated by a Gaussian with matching mean and variance, both the d$_0$CRT and d$_\mathrm{I}$CRT maintain Type-I error control and high power (Figure~\ref{fig:robust:sec}). Furthermore, when the covariates are jointly Gaussian {and the $X\mid Z$ distributions are estimated in-sample using any of three standard methods detailed in Appendix~\ref{app:sim:robustestmom}, the Type-I error of both dCRT methods always remains close to the nominal level (Figure~\ref{fig:robust:X:fdr}).} See Appendices~\ref{app:sim:robustfirstsecond}~and~\ref{app:sim:robustestmom} for details.

The resampling-free versions of the dCRT are faster and just as powerful as the non-resampling-free dCRT except when $X\mid Z$ is highly discrete: The resampling-free modification sped up the d$_0$CRT by 2.5 times in an $n=p=800$ simulation and sped up the d$_\mathrm{I}$CRT by 11 times in an $n=p=800$ simulation, even after applying screening (Table~\ref{tab:com:resample}). When $X\mid Z$ is Gaussian, changing the form of the test statistics of the d$_0$CRT and d$_\mathrm{I}$CRT as proposed in paragraphs 2 and 4, respectively, of Section~\ref{sec:main:mcf} had a negligible effect on their power (Figure~\ref{fig:resample:free:gauss}). When $X\mid Z$ is non-Gaussian and must be transformed to Gaussian as described in paragraph 3 of Section~\ref{sec:main:mcf}, we found essentially no power loss for the resampling-free d$_0$CRT and d$_\mathrm{I}$CRT relative to their non-resampling-free counterparts when $X\mid Z$ was Gamma-distributed (with shape $=3$ and rate $=0.5$, so that skew $>1$ and excess kurtosis $=2$), while there was substantial power loss (up to about 40 percentage points) when $X\mid Z$ was binary and hence required substantial exogenous randomization to be transformed to Gaussian, though the resampling-free dCRTs were still more powerful (up to about 10 percent) than the HRT (Figure~\ref{fig:nmc}). See Appendices~\ref{sim:rfgauss}~and~\ref{sim:rfnongauss}.

Screening makes the dCRT faster without affecting its power: In a simulation with $n=p=800$, screening reduced the computation time by a factor of about 5 for both d$_0$CRT and d$_\mathrm{I}$CRT (Table~\ref{tab:com:screening}) without perceptibly hurting power (Figure~\ref{fig:power:screening}). See Appendix~\ref{sec:sim:screening} for details.

\section{Identifying biomarkers for breast cancer}\label{sec:realdata}

As a final demonstration of the effectiveness of the dCRT, we apply it to the data set from \cite{curtis2012genomic}, consisting of $n=1,396$ staged oestrogen-receptor-positive cases of breast cancer, each with expression level (mRNA) and copy number aberration (CNA) measured for $p=164$ genes, which were studied in \cite{pereira2016somatic}. Our goal is to find genes on which cancer stage depends, conditional on the remaining genes and all CNAs, while controlling either the false discovery rate or family-wise error rate at level 0.1. Discovering such biomarkers for cancer can reveal new pathways and mechanisms for cancer progression; see \cite{shen2019false} for a recent application of model-X knockoffs to the same end.

After log-transforming the gene expressions, we adjust to them using the CNA data with linear model as in \cite{solvang2011linear,lahti2012cancer,leday2013modeling} and modeled the processed gene expressions jointly as multivariate Gaussian similar to \cite{shen2019false}. We applied the d$_0$CRT, the d$_{\mathrm{I}}$CRT, the \rev{o}CRT, the HRT, and model-X knockoffs and compared the results. See Appendix~\ref{app:data} for details of the data pre-processing, covariate-modeling, and method implementations. {Each method is run 300 times. Table~\ref{tab:comp:data} contains average runtimes (in $\mathsf{R}$) for all methods, showing that the dCRTs are quite fast compared to the \rev{o}CRT. In particular, the \rev{o}CRT takes over 7 hours to run while the dCRTs take under a minute. 

Figure~\ref{fig:num:data} presents the distribution of the numbers of discoveries among the 300 repetitions for all the methods. Methods including dCRT, \rev{o}CRT and HRT have stable outputs about the number of detected genes. In terms of false discovery rate control, d$_0$CRT and d$_{\mathrm{I}}$CRT detect exactly 5 genes in more than 80\% repetitions and $\geq$5 genes at all times. \rev{o}CRT and HRT detect exactly 3 genes in more than 70\% repetitions and always have fewer discoveries than dCRT. While the knockoffs have 0 discoveries in about $45\%$ repetitions but $\geq$10 discoveries in the remaining times, which implies that knockoffs fail to produce stable output. Knockoffs' instablility and lack of power is due to the sparsity of discoverable genes. In terms of family-wise error rate control, d$_0$CRT and HRT have 3 discoveries in most runs, d$_{\mathrm{I}}$CRT has 4 discoveries and \rev{o}CRT has 2.
}

\begin{table}[htbp]
\centering
\begin{tabular}{c|c|c|c|c}
\multicolumn{5}{c}{{\bf Average computation times (minutes)}} \\
\hline
      \mbox{d$_0$CRT} &  d$_{\mathrm{I}}$CRT & \rev{o}CRT & knockoffs  & HRT  \\ \hline
     $0.8$  & $0.8$ & $443.3$ & $0.3$ & $3.1$\\ \hline
\end{tabular}
\caption{\label{tab:comp:data} Average computation times (300 repetitions in $\mathsf{R}$) in the breast cancer application. Note that our use of the resampling-free version of the dCRT makes it faster than the HRT in this case.}
\end{table}

\begin{figure}
    \centering
    \includegraphics[width = 0.45\textwidth]{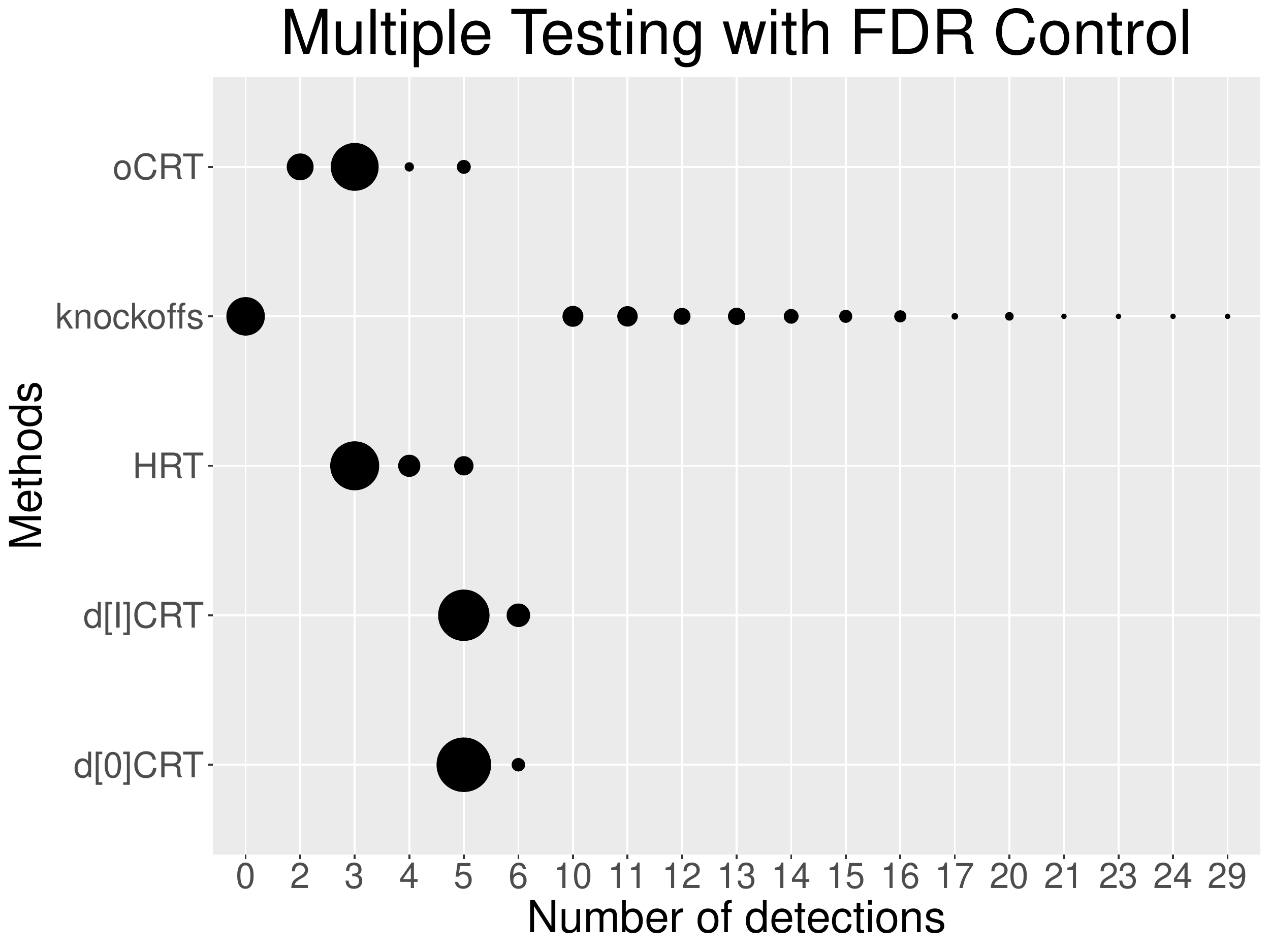}
    \includegraphics[width = 0.45\textwidth]{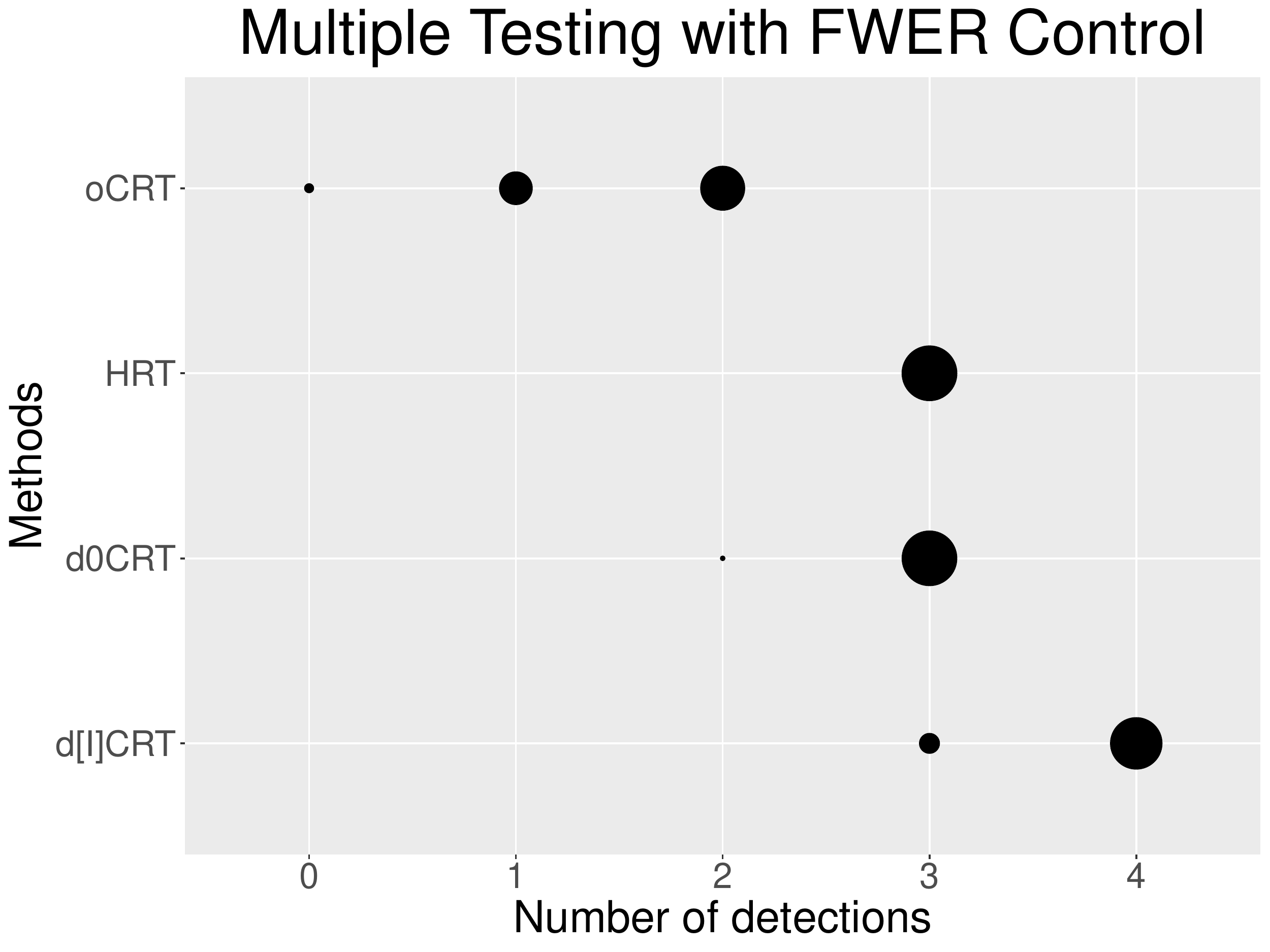}
    \caption{\label{fig:num:data} Summary of the numbers of discoveries over 300 repetitions, with false discovery rate and family-wise error rate control, in the breast cancer application. Area of each black point is proportional to the frequency that the corresponding method makes this number of discoveries in the 300 repetitions. The dCRT approaches are more powerful than the competitors \rev{o}CRT and HRT. The knockoffs have no discoveries in around 45\% of the experiments.}
\end{figure}

When used to control the false discovery rate, it turns out that all five genes discovered by the dCRT --- {\em FBXW7}, {\em MAP3K13}, {\em HRAS}, {\em GPS2}, and {\em RUNX1}; see  Appendix~\ref{tab:1} for their corresponding {average $p$-values and frequencies of being detected} --- have been linked in independent research to cancer, suggesting the dCRT makes promising discoveries. In particular, {\em FBXW7} encodes a member of the F-box protein family and its mutations are detected in ovarian and breast cancer cell lines \citep{liu2019fbxw7,kirzinger2019fbxw7}; {\em MAP3K13} belongs to the serine/threonine protein kinase family acting as a regulator for cancer \citep{han2016microrna}; {\em HRAS} belongs to the RAS oncogene family which is related to the transforming of genes of mammalian sarcoma retroviruses, and defects in this gene have been implicated in a variety of cancers \citep{geyer2018hras}; over-expression of {\em GPS2} in mammalian cells may suppress signals mediated by RAS/MAPK and interfere with JNK activity, all of which are cancer-related \citep{jarmalavicius2010gps2,huang2016gps2}; {\em RUNX1} has been found to activate certain signaling pathways that promote tumor metastasis \citep{li2019runx1}.

\section{Discussion} \label{sec:discussion}

The HRT provided the first indication that a variant of the CRT could be computationally tractable, albeit at the cost of statistical performance. In this paper, we demonstrate that leaving out \textit{variables} instead of \textit{samples} creates a procedure that is not quite as fast (though still a tiny fraction of the \rev{o}CRT's computational cost) but much more powerful.  This brings the dCRT into the realm of fast and powerful model-X methods, where knockoffs is currently the methodology of choice. Knockoffs and dCRT have complementary strengths, which we discuss briefly below.

Model-X knockoffs addresses the variable selection problem, targeting false discovery rate control. It is very computationally efficient, requiring just one high-dimensional model fit. Furthermore, our simulations confirm that knockoffs is quite powerful in several settings. These advantages have led to the successful application of knockoffs to genome-wide association studies  \citep{SetC17, SetS19}. By comparison, the dCRT still requires several high-dimensional model fits and is therefore more computationally costly. On the other hand, dCRT computation benefits from being embarrassingly parallelizable, so modern parallel computing resources can greatly reduce its runtime. As far as power goes, the relative performance of the two methods varies with simulation setting (see Section~\ref{sec:sim} and Appendix~\ref{sec:app:sim}); neither procedure uniformly dominates the other (when controlling the false discovery rate).

Aside from these considerations, the dCRT provides a few important advantages over knockoffs. The first is that, unlike knockoffs, the dCRT provides $p$-values (arbitrarily fine-grained and essentially exact)  for each conditional independence hypothesis. In addition to providing an interpretable measure of significance, this decoupling of statistical significance quantification from downstream analyses such as multiple testing brings great versatility. Indeed, dCRT $p$-values can be used for single hypothesis testing, multiple hypothesis testing with a variety of error rates, and any number of other tasks that take $p$-values as input. While the knockoffs framework has gradually been extended to handle analysis tasks beyond false discovery rate control --- e.g. $k$-family-wise error rate control by \cite{Janson2016}, and simultaneous false discovery probability control by \cite{Katsevich2020} --- such extensions require custom solutions and some are currently out of reach (such as single testing or family-wise error rate control).  Another advantage of the dCRT is that it has little or no variability across runs. On the other hand, knockoffs is a randomized procedure; this randomization can lead to variability in the performance of the procedure on a given data set; see Appendix~\ref{sec:knockoffs-variability} and Figure 4 of \cite{SetC17}.

The dCRT is therefore a useful addition to the model-X methodology toolbox. Much work still remains to refine this new tool for better power and even faster computation. Indeed, many degrees of freedom in the construction of the dCRT test statistic remain to be explored. For example, should the statistic be based on the fitted coefficient of a variable or on the loss function? What is the best way to test groups of variables? The recent theoretical exploration of the CRT~\citep{katsevich2020theoretical} may help guide the search for powerful test statistics. Another open question is whether there are efficient resampling-free dCRT variants for highly discrete covariates. Finally, the dependence structure of (d)CRT $p$-values is an important subject for further exploration. We may not always be able to plug-and-play (d)CRT $p$-values in multiple testing procedures, since their dependency structure is currently unknown. In a related development, \cite{bates2020causal} recently proposed a clever method of generating independent HRT $p$-values for groups of linearly-structured covariates. 

Despite these open questions, our initial demonstrations of the dCRT on simulated and real data are quite promising. We are therefore optimistic about the prospects of the dCRT for fruitful practical applications, and look forward to continued improvements in the computational and statistical efficiency of model-X methodology.

\subsection*{Acknowledgements}
We would like to thank Siyuan Ma, Wenshuo Wang, Dae Woong Ham, Lu Zhang, \rev{Shuangning Li and Emmanuel Cand\`es} for helpful discussions and feedback on the paper. \rev{We would also like to thank two referees and an associate editor for constructive feedback that helped improve our paper.} This work used the Extreme Science and Engineering Discovery Environment (XSEDE)~\citep{XSEDE}, supported by National Science Foundation grant ACI-1548562. Specifically, it used the Bridges system~\citep{Bridges}, which is supported by NSF award ACI-1445606, at the Pittsburgh Supercomputing Center (PSC).

\bibliographystyle{biometrika}
\bibliography{library}

\begin{thebibliography}{47}
\expandafter\ifx\csname natexlab\endcsname\relax\def\natexlab#1{#1}\fi

\bibitem[{Barber \& Cand{\`e}s(2015)}]{barber2015controlling}
\textsc{Barber, R.~F.} \& \textsc{Cand{\`e}s, E.~J.} (2015).
\newblock Controlling the false discovery rate via knockoffs.
\newblock \textit{The Annals of Statistics} \textbf{43}, 2055--2085.

\bibitem[{Bates et~al.(2020)Bates, Sesia, Sabatti \&
  Cand{\`e}s}]{bates2020causal}
\textsc{Bates, S.}, \textsc{Sesia, M.}, \textsc{Sabatti, C.} \&
  \textsc{Cand{\`e}s, E.} (2020).
\newblock Causal inference in genetic trio studies.
\newblock \textit{Proceedings of the National Academy of Sciences}
  \textbf{117}, 24117--24126.

\bibitem[{Bellot \& van~der Schaar(2019)}]{bellot2019conditional}
\textsc{Bellot, A.} \& \textsc{van~der Schaar, M.} (2019).
\newblock Conditional independence testing using generative adversarial
  networks.
\newblock In \textit{Advances in Neural Information Processing Systems}.

\bibitem[{Benjamini \& Hochberg(1995)}]{benjamini1995controlling}
\textsc{Benjamini, Y.} \& \textsc{Hochberg, Y.} (1995).
\newblock Controlling the false discovery rate: a practical and powerful
  approach to multiple testing.
\newblock \textit{Journal of the Royal Statistical Society: Series B}
  \textbf{57}, 289--300.

\bibitem[{Benjamini \& Yekutieli(2001)}]{BY01}
\textsc{Benjamini, Y.} \& \textsc{Yekutieli, D.} (2001).
\newblock {The control of the false discovery rate in multiple testing under
  dependency}.
\newblock \textit{The Annals of Statistics} \textbf{29}, 1165--1188.

\bibitem[{Berrett et~al.(2020)Berrett, Wang, Barber \&
  Samworth}]{berrett2018conditional}
\textsc{Berrett, T.~B.}, \textsc{Wang, Y.}, \textsc{Barber, R.~F.} \&
  \textsc{Samworth, R.~J.} (2020).
\newblock The conditional permutation test for independence while controlling
  for confounders.
\newblock \textit{Journal of the Royal Statistical Society: Series B
  (Statistical Methodology)} \textbf{82}, 175--197.

\bibitem[{Bien et~al.(2013)Bien, Taylor \& Tibshirani}]{bien2013lasso}
\textsc{Bien, J.}, \textsc{Taylor, J.} \& \textsc{Tibshirani, R.} (2013).
\newblock A lasso for hierarchical interactions.
\newblock \textit{The Annals of Statistics} \textbf{41}, 1111.

\bibitem[{Cand\`es et~al.(2018)Cand\`es, Fan, Janson \& Lv}]{candes2018panning}
\textsc{Cand\`es, E.}, \textsc{Fan, Y.}, \textsc{Janson, L.} \& \textsc{Lv, J.}
  (2018).
\newblock Panning for gold: model-$\mbox{X}$ knockoffs for high dimensional
  controlled variable selection.
\newblock \textit{Journal of the Royal Statistical Society: Series B}
  \textbf{80}, 551--577.

\bibitem[{Chernozhukov et~al.(2018)Chernozhukov, Chetverikov, Demirer, Duflo,
  Hansen, Newey \& Robins}]{chernozhukov2016double}
\textsc{Chernozhukov, V.}, \textsc{Chetverikov, D.}, \textsc{Demirer, M.},
  \textsc{Duflo, E.}, \textsc{Hansen, C.}, \textsc{Newey, W.} \&
  \textsc{Robins, J.} (2018).
\newblock Double/debiased machine learning for treatment and structural
  parameters.
\newblock \textit{The Econometrics Journal} \textbf{21}, C1--C68.

\bibitem[{Chipman(1996)}]{chipman1996bayesian}
\textsc{Chipman, H.} (1996).
\newblock Bayesian variable selection with related predictors.
\newblock \textit{Canadian Journal of Statistics} \textbf{24}, 17--36.

\bibitem[{Cox(1984)}]{cox1984interaction}
\textsc{Cox, D.~R.} (1984).
\newblock Interaction.
\newblock \textit{International Statistical Review} , 1--24.

\bibitem[{Curtis et~al.(2012)Curtis, Shah, Chin, Turashvili, Rueda, Dunning,
  Speed, Lynch, Samarajiwa, Yuan et~al.}]{curtis2012genomic}
\textsc{Curtis, C.}, \textsc{Shah, S.~P.}, \textsc{Chin, S.-F.},
  \textsc{Turashvili, G.}, \textsc{Rueda, O.~M.}, \textsc{Dunning, M.~J.},
  \textsc{Speed, D.}, \textsc{Lynch, A.~G.}, \textsc{Samarajiwa, S.},
  \textsc{Yuan, Y.} et~al. (2012).
\newblock The genomic and transcriptomic architecture of 2,000 breast tumours
  reveals novel subgroups.
\newblock \textit{Nature} \textbf{486}, 346--352.

\bibitem[{Davies(1980)}]{davies1980algorithm}
\textsc{Davies, R.~B.} (1980).
\newblock Algorithm as 155: The distribution of a linear combination of
  $\chi^2$ random variables.
\newblock \textit{Journal of the Royal Statistical Society: Series C}
  \textbf{29}, 323--333.

\bibitem[{Duchesne \& De~Micheaux(2010)}]{duchesne2010computing}
\textsc{Duchesne, P.} \& \textsc{De~Micheaux, P.~L.} (2010).
\newblock Computing the distribution of quadratic forms: Further comparisons
  between the $\mbox{Liu}$--$\mbox{Tang}$--$\mbox{Zhang}$ approximation and
  exact methods.
\newblock \textit{Computational Statistics $\&$ Data Analysis} \textbf{54},
  858--862.

\bibitem[{Friedman et~al.(2008)Friedman, Hastie \&
  Tibshirani}]{friedman2008sparse}
\textsc{Friedman, J.}, \textsc{Hastie, T.} \& \textsc{Tibshirani, R.} (2008).
\newblock Sparse inverse covariance estimation with the graphical lasso.
\newblock \textit{Biostatistics} \textbf{9}, 432--441.

\bibitem[{Geyer et~al.(2018)Geyer, Li, Papanastasiou, Smith, Selenica, Burke,
  Edelweiss, Wen, Piscuoglio \& Schultheis}]{geyer2018hras}
\textsc{Geyer, F.~C.}, \textsc{Li, A.}, \textsc{Papanastasiou, A.~D.},
  \textsc{Smith, A.}, \textsc{Selenica, P.}, \textsc{Burke, K.~A.},
  \textsc{Edelweiss, M.}, \textsc{Wen, H.~C.}, \textsc{Piscuoglio, S.} \&
  \textsc{Schultheis, A.~M.} (2018).
\newblock Recurrent hotspot mutations in $\mbox{HRAS-Q61}$ and
  $\mbox{PI3K-AKT}$ pathway genes as drivers of breast adenomyoepitheliomas.
\newblock \textit{Nature Communications} \textbf{9}, 1--16.

\bibitem[{Hamada \& Wu(1992)}]{hamada1992analysis}
\textsc{Hamada, M.} \& \textsc{Wu, C.~J.} (1992).
\newblock Analysis of designed experiments with complex aliasing.
\newblock \textit{Journal of Quality Technology} \textbf{24}, 130--137.

\bibitem[{Han et~al.(2016)Han, Chen, Cheng, Prochownik \& Li}]{han2016microrna}
\textsc{Han, H.}, \textsc{Chen, Y.}, \textsc{Cheng, L.}, \textsc{Prochownik,
  E.~V.} \& \textsc{Li, Y.} (2016).
\newblock micro$\text{RNA}$-206 impairs c-$\text{M}$yc-driven cancer in a
  synthetic lethal manner by directly inhibiting $\text{MAP3K13}$.
\newblock \textit{Oncotarget} \textbf{7}, 16409.

\bibitem[{Huang et~al.(2008)Huang, Ma \& Zhang}]{huang2008adaptive}
\textsc{Huang, J.}, \textsc{Ma, S.} \& \textsc{Zhang, C.-H.} (2008).
\newblock Adaptive lasso for sparse high-dimensional regression models.
\newblock \textit{Statistica Sinica} \textbf{18}, 1603--1618.

\bibitem[{Huang et~al.(2016)Huang, Xiao, Wang, Yin, Lu, Li, Liu, Wang \&
  Li}]{huang2016gps2}
\textsc{Huang, X.}, \textsc{Xiao, F.}, \textsc{Wang, S.}, \textsc{Yin, R.},
  \textsc{Lu, C.}, \textsc{Li, Q.}, \textsc{Liu, N.}, \textsc{Wang, L.} \&
  \textsc{Li, P.} (2016).
\newblock G protein pathway suppressor 2 $\mbox{(GPS2)}$ acts as a tumor
  suppressor in liposarcoma.
\newblock \textit{Tumor Biology} \textbf{37}, 13333--13343.

\bibitem[{Imhof(1961)}]{imhof1961computing}
\textsc{Imhof, J.~P.} (1961).
\newblock Computing the distribution of quadratic forms in normal variables.
\newblock \textit{Biometrika} \textbf{48}, 419--426.

\bibitem[{Janson \& Su(2016)}]{Janson2016}
\textsc{Janson, L.} \& \textsc{Su, W.} (2016).
\newblock {Familywise error rate control via knockoffs}.
\newblock \textit{Electronic Journal of Statistics} \textbf{10}, 960--975.

\bibitem[{Jarmalavicius et~al.(2010)Jarmalavicius, Trefzer \&
  Walden}]{jarmalavicius2010gps2}
\textsc{Jarmalavicius, S.}, \textsc{Trefzer, U.} \& \textsc{Walden, P.} (2010).
\newblock Differential arginine methylation of the $\mbox{G}$-protein pathway
  suppressor $\mbox{GPS}$-2 recognized by tumor-specific $\mbox{T}$-cells in
  melanoma.
\newblock \textit{The FASEB Journal} \textbf{24}, 937--946.

\bibitem[{Katsevich \& Ramdas(2020{\natexlab{a}})}]{Katsevich2020}
\textsc{Katsevich, E.} \& \textsc{Ramdas, A.} (2020{\natexlab{a}}).
\newblock {Simultaneous high-probability bounds on the false discovery
  proportion in structured, regression, and online settings}.
\newblock \textit{The Annals of Statistics, to appear} .

\bibitem[{Katsevich \& Ramdas(2020{\natexlab{b}})}]{katsevich2020theoretical}
\textsc{Katsevich, E.} \& \textsc{Ramdas, A.} (2020{\natexlab{b}}).
\newblock A theoretical treatment of conditional independence testing under
  model-$\mbox{X}$.
\newblock \textit{arXiv preprint arXiv:2005.05506} .

\bibitem[{Kirzinger et~al.(2019)Kirzinger, Vizeacoumar, Haave, Gonzalez~Lopez,
  Bonham, Kusalik \& Vizeacoumar}]{kirzinger2019fbxw7}
\textsc{Kirzinger, M.~W.}, \textsc{Vizeacoumar, F.~S.}, \textsc{Haave, B.},
  \textsc{Gonzalez~Lopez, C.}, \textsc{Bonham, K.}, \textsc{Kusalik, A.} \&
  \textsc{Vizeacoumar, F.~J.} (2019).
\newblock Humanized yeast genetic interaction mapping predicts synthetic lethal
  interactions of $\mbox{FBXW7}$ in breast cancer.
\newblock \textit{BMC medical genomics} \textbf{12}, 112.

\bibitem[{Lahti et~al.(2012)Lahti, Sch{\"a}fer, Klein, Bicciato \&
  Dugas}]{lahti2012cancer}
\textsc{Lahti, L.}, \textsc{Sch{\"a}fer, M.}, \textsc{Klein, H.~U.},
  \textsc{Bicciato, S.} \& \textsc{Dugas, M.} (2012).
\newblock Cancer gene prioritization by integrative analysis of m$\mbox{RNA}$
  expression and dna copy number data: a comparative review.
\newblock \textit{Briefings in Bioinformatics} \textbf{14}, 27--35.

\bibitem[{Leday et~al.(2013)Leday, van~der Vaart, van Wieringen \& van~de
  Wiel}]{leday2013modeling}
\textsc{Leday, G.~G.}, \textsc{van~der Vaart, A.~W.}, \textsc{van Wieringen,
  W.~N.} \& \textsc{van~de Wiel, M.~A.} (2013).
\newblock Modeling association between $\mbox{DNA}$ copy number and gene
  expression with constrained piecewise linear regression splines.
\newblock \textit{The Annals of Applied Statistics} \textbf{7}, 823--845.

\bibitem[{Ledoit \& Wolf(2004)}]{ledoit2004well}
\textsc{Ledoit, O.} \& \textsc{Wolf, M.} (2004).
\newblock A well-conditioned estimator for large-dimensional covariance
  matrices.
\newblock \textit{Journal of multivariate analysis} \textbf{88}, 365--411.

\bibitem[{Li et~al.(2019)Li, Lai, He, Fang, Yan, Zhang, Wang, Gu, Wang, Ye,
  Han, Lin, Chen, Cai, Li \& Liu}]{li2019runx1}
\textsc{Li, Q.}, \textsc{Lai, Q.}, \textsc{He, C.}, \textsc{Fang, Y.},
  \textsc{Yan, Q.}, \textsc{Zhang, Y.}, \textsc{Wang, X.}, \textsc{Gu, C.},
  \textsc{Wang, Y.}, \textsc{Ye, L.}, \textsc{Han, L.}, \textsc{Lin, X.},
  \textsc{Chen, J.}, \textsc{Cai, J.}, \textsc{Li, A.} \& \textsc{Liu, S.}
  (2019).
\newblock $\mbox{RUNX1}$ promotes tumour metastasis by activating the
  $\mbox{W}$nt/$\beta$-catenin signalling pathway and $\mbox{EMT}$ in
  colorectal cancer.
\newblock \textit{Journal of Experimental $\&$ Clinical Cancer Research}
  \textbf{38}, 334.

\bibitem[{Liu et~al.(2019)Liu, Zou, Wang, Yang, Zhang, Luo, Liang, Zhou \&
  Huang}]{liu2019fbxw7}
\textsc{Liu, F.}, \textsc{Zou, Y.}, \textsc{Wang, F.}, \textsc{Yang, B.},
  \textsc{Zhang, Z.}, \textsc{Luo, Y.}, \textsc{Liang, M.}, \textsc{Zhou, J.}
  \& \textsc{Huang, O.} (2019).
\newblock $\mbox{FBXW7}$ mutations promote cell proliferation, migration, and
  invasion in cervical cancer.
\newblock \textit{Genetic testing and molecular biomarkers} \textbf{23},
  409--417.

\bibitem[{Liu et~al.(2009)Liu, Tang \& Zhang}]{liu2009new}
\textsc{Liu, H.}, \textsc{Tang, Y.} \& \textsc{Zhang, H.~H.} (2009).
\newblock A new chi-square approximation to the distribution of non-negative
  definite quadratic forms in non-central normal variables.
\newblock \textit{Computational Statistics $\&$ Data Analysis} \textbf{53},
  853--856.

\bibitem[{Nelder(1977)}]{nelder1977reformulation}
\textsc{Nelder, J.} (1977).
\newblock A reformulation of linear models.
\newblock \textit{Journal of the Royal Statistical Society: Series A}
  \textbf{140}, 48--63.

\bibitem[{Nystrom et~al.(2015)Nystrom, Levine, Roskies \& Scott}]{Bridges}
\textsc{Nystrom, N.~A.}, \textsc{Levine, M.~J.}, \textsc{Roskies, R.~Z.} \&
  \textsc{Scott, J.~R.} (2015).
\newblock {Bridges: A Uniquely Flexible HPC Resource for New Communities and
  Data Analytics}.
\newblock In \textit{Proceedings of the 2015 XSEDE Conference: Scientific
  Advancements Enabled by Enhanced Cyberinfrastructure}, XSEDE '15. New York,
  NY, USA: ACM.

\bibitem[{Peixoto(1987)}]{peixoto1987hierarchical}
\textsc{Peixoto, J.~L.} (1987).
\newblock Hierarchical variable selection in polynomial regression models.
\newblock \textit{The American Statistician} \textbf{41}, 311--313.

\bibitem[{Pereira et~al.(2016)Pereira, Chin, Rueda, Vollan, Provenzano,
  Bardwell, Pugh, Jones, Russell, Sammut, Tsui, Liu, Dawson, Abraham, Northen,
  Peden, Mukherjee, Turashvili, Green, McKinney, Oloumi, Shah, Rosenfeld,
  Murphy, Bentley, Ellis, Purushotham, Pinder, Børresen-Dale, Earl, Pharoah,
  Ross, Aparicio \& Caldas}]{pereira2016somatic}
\textsc{Pereira, B.}, \textsc{Chin, S.-F.}, \textsc{Rueda, O.~M.},
  \textsc{Vollan, H.-K.~M.}, \textsc{Provenzano, E.}, \textsc{Bardwell, H.~A.},
  \textsc{Pugh, M.}, \textsc{Jones, L.}, \textsc{Russell, R.}, \textsc{Sammut,
  S.-J.}, \textsc{Tsui, D. W.~Y.}, \textsc{Liu, B.}, \textsc{Dawson, S.-J.},
  \textsc{Abraham, J.}, \textsc{Northen, H.}, \textsc{Peden, J.~F.},
  \textsc{Mukherjee, A.}, \textsc{Turashvili, G.}, \textsc{Green, A.~R.},
  \textsc{McKinney, S.}, \textsc{Oloumi, A.}, \textsc{Shah, S.},
  \textsc{Rosenfeld, N.}, \textsc{Murphy, L.}, \textsc{Bentley, D.~R.},
  \textsc{Ellis, I.~O.}, \textsc{Purushotham, A.}, \textsc{Pinder, S.~E.},
  \textsc{Børresen-Dale, A.-L.}, \textsc{Earl, H.~M.}, \textsc{Pharoah,
  P.~D.}, \textsc{Ross, M.~T.}, \textsc{Aparicio, S.} \& \textsc{Caldas, C.}
  (2016).
\newblock The somatic mutation profiles of 2,433 breast cancers refine their
  genomic and transcriptomic landscapes.
\newblock \textit{Nature communications} \textbf{7}, 11479.

\bibitem[{Sesia et~al.(2020)Sesia, Katsevich, Bates, Cand{\`{e}}s \&
  Sabatti}]{SetS19}
\textsc{Sesia, M.}, \textsc{Katsevich, E.}, \textsc{Bates, S.},
  \textsc{Cand{\`{e}}s, E.} \& \textsc{Sabatti, C.} (2020).
\newblock {Multi-resolution localization of causal variants across the genome}.
\newblock \textit{Nature Communications} \textbf{11}, 1093.

\bibitem[{Sesia et~al.(2019)Sesia, Sabatti \& Cand{\`{e}}s}]{SetC17}
\textsc{Sesia, M.}, \textsc{Sabatti, C.} \& \textsc{Cand{\`{e}}s, E.~J.}
  (2019).
\newblock {Gene hunting with hidden Markov model knockoffs}.
\newblock \textit{Biometrika} \textbf{106}, 1--18.

\bibitem[{Shah \& Peters(2018)}]{shah2018hardness}
\textsc{Shah, R.~D.} \& \textsc{Peters, J.} (2018).
\newblock The hardness of conditional independence testing and the generalised
  covariance measure.
\newblock \textit{The Annals of Statistics, to appear} .

\bibitem[{Shen et~al.(2019)Shen, Fu, He \& Jiang}]{shen2019false}
\textsc{Shen, A.}, \textsc{Fu, H.}, \textsc{He, K.} \& \textsc{Jiang, H.}
  (2019).
\newblock False discovery rate control in cancer biomarker selection using
  knockoffs.
\newblock \textit{Cancers} \textbf{11}, 744.

\bibitem[{Solvang et~al.(2011)Solvang, Lingj{\ae}rde, Frigessi,
  B{\o}rresen~Dale \& Kristensen}]{solvang2011linear}
\textsc{Solvang, H.~K.}, \textsc{Lingj{\ae}rde, O.~C.}, \textsc{Frigessi, A.},
  \textsc{B{\o}rresen~Dale, A.~L.} \& \textsc{Kristensen, V.~N.} (2011).
\newblock Linear and non-linear dependencies between copy number aberrations
  and $\mbox{mRNA}$ expression reveal distinct molecular pathways in breast
  cancer.
\newblock \textit{BMC bioinformatics} \textbf{12}, 197.

\bibitem[{Tansey et~al.(2018)Tansey, Veitch, Zhang, Rabadan \&
  Blei}]{tansey2018holdout}
\textsc{Tansey, W.}, \textsc{Veitch, V.}, \textsc{Zhang, H.}, \textsc{Rabadan,
  R.} \& \textsc{Blei, D.~M.} (2018).
\newblock The holdout randomization test: Principled and easy black box feature
  selection.
\newblock \textit{arXiv preprint arXiv:1811.00645} .

\bibitem[{Tibshirani(1996)}]{tibshirani1996regression}
\textsc{Tibshirani, R.} (1996).
\newblock Regression shrinkage and selection via the lasso.
\newblock \textit{Journal of the Royal Statistical Society: Series B}
  \textbf{58}, 267--288.

\bibitem[{Tibshirani(2013)}]{tibshirani2013lasso}
\textsc{Tibshirani, R.~J.} (2013).
\newblock The lasso problem and uniqueness.
\newblock \textit{Electronic Journal of Statistics} \textbf{7}, 1456--1490.

\bibitem[{Towns et~al.(2014)Towns, Cockerill, Dahan, Foster, Gaither, Grimshaw,
  Hazlewood, Lathrop, Lifka, Peterson, Roskies, Scott \& Wilkins-Diehr}]{XSEDE}
\textsc{Towns, J.}, \textsc{Cockerill, T.}, \textsc{Dahan, M.}, \textsc{Foster,
  I.}, \textsc{Gaither, K.}, \textsc{Grimshaw, A.}, \textsc{Hazlewood, V.},
  \textsc{Lathrop, S.}, \textsc{Lifka, D.}, \textsc{Peterson, G.~D.},
  \textsc{Roskies, R.}, \textsc{Scott, J.} \& \textsc{Wilkins-Diehr, N.}
  (2014).
\newblock {XSEDE: Accelerating Scientific Discovery}.
\newblock \textit{Computing in Science {\&} Engineering} \textbf{16}, 62--74.

\bibitem[{Zou(2006)}]{zou2006adaptive}
\textsc{Zou, H.} (2006).
\newblock The adaptive lasso and its oracle properties.
\newblock \textit{Journal of the American Statistical Association}
  \textbf{101}, 1418--1429.

\bibitem[{Zou \& Hastie(2005)}]{zou2005regularization}
\textsc{Zou, H.} \& \textsc{Hastie, T.} (2005).
\newblock Regularization and variable selection via the elastic net.
\newblock \textit{Journal of the Royal Statistical Society: Series B}
  \textbf{67}, 301--320.

\end{thebibliography}

\clearpage
\newpage
\setcounter{page}{1}

\appendix
\setcounter{lemma}{0}
\setcounter{theorem}{0}
\renewcommand{\thelemma}{A\arabic{lemma}}
\renewcommand{\thetheorem}{A\arabic{theorem}}

\setcounter{definition}{0}
\renewcommand{\thedefinition}{A\arabic{definition}}

\section*{Appendix}

% \section{The power of model-X knockoffs under strong sparsity}\label{app:knockoffspower}
% {TBD: move this discussion to Section~\ref{sec:multiple-testing}.} Model-X knockoffs uses the Selective SeqStep+ procedure with threshold $c=0.5$ \citep{barber2015controlling} in its final step, which cannot make a positive number of rejections fewer than $1/q$, where $\alpha$ is the nominal false discovery rate level. Thus, if model-X knockoffs cannot make at least 10 discoveries while controlling the false discovery rate at level $\alpha=0.1$, then it must make zero discoveries. So, for instance, if there are only 5 non-null covariates, model-X knockoffs will be nearly powerless to discover them, no matter how strong their relationships with $Y$.

\section{Resampling-free distilled CRT}\label{sec:app:mcf}

\subsection{Resampling-free lasso-based d$_\mathrm{I}$CRT}\label{sec:dI:pvl}

In this section, we describe the resampling-free version of the lasso-based d$_\mathrm{I}$CRT of Example~\ref{ex:2} for gaussian $X$, in analog to the resampling-free d$_0$CRT detailed in Section~\ref{sec:main:mcf}. We follow the notation of Example~\ref{ex:2} and for any $\a=(a_1,\ldots,a_n)\trans\in\mathbb{R}^n$, let ${\rm diag}(\a)$ denote the diagonal matrix with its $(i,i)$-th entry being $a_i$ for $i=1,2,\ldots,n$. Then in Example~\ref{ex:2},
\begin{align*}
\betahat_x=&\left[(\bone,{\bZ}_{\bullet,\mathrm{top}(k)})\trans{\rm diag}^2(\bepsilon_x)(\bone,{\bZ}_{\bullet,\mathrm{top}(k)})\right]^{-1}(\bone,{\bZ}_{\bullet,\mathrm{top}(k)})\trans{\rm diag}(\bepsilon_x)(\by-\dy)\\
=&\hat{\mathbb{H}}^{-1}(\bone,{\bZ}_{\bullet,\mathrm{top}(k)})\trans{\rm diag}(\by-\dy)\bepsilon_x,
\end{align*}
where $\hat{\mathbb{H}}=(\bone,{\bZ}_{\bullet,\mathrm{top}(k)})\trans{\rm diag}^2(\bepsilon_x)(\bone,{\bZ}_{\bullet,\mathrm{top}(k)})$ and $\bepsilon_x=\x-\dx$. And the test statistics
\[
T(\by,\bx,\dy,\dx) = \hat{\beta}_{x,1}^2+\frac{1}{k}\sum_{j=2}^{k+1}\hat{\beta}_{x,j}^2=\|\hat{\mathbb{H}}^{-1}\tilde{\bZ}_{\bullet,\mathrm{top}(k)}\bepsilon_x\|_2^2
\]
where $\tilde{\bZ}_{\bullet,\mathrm{top}(k)}=(\bone,k^{-1/2}{\bZ}_{\bullet,\mathrm{top}(k)})\trans{\rm diag}(\by-\dy)$. In analog to the resampling-free d$_0$CRT introduced in Section~\ref{sec:main:mcf}, we replace $\hat{\mathbb{H}}$ with its conditional expectation given $(\by,\bZ)$, i.e., $\mathbb{H}=\sigma^2_x(\bone,{\bZ}_{\bullet,\mathrm{top}(k)})\trans(\bone,{\bZ}_{\bullet,\mathrm{top}(k)})$ with $\sigma_x^2$ being the conditional variance of $X$ given $Z$. Then the test statistics of the resampling-free version of d$_{\mathrm{I}}$CRT can be constructed as $\|\mathbb{H}^{-1}\tilde{\bZ}_{\bullet,\mathrm{top}(k)}\bepsilon_x\|_2^2$. Conditional on $(\by,\bZ)$, it is a quadratic form of the gaussian vector $\bepsilon_x$ under the null. Accurate and efficient computational methods have been proposed to handle such problems (see, e.g., \citet{imhof1961computing,davies1980algorithm,liu2009new}). We use the method proposed by \cite{imhof1961computing} and realized by $\mathsf{R}$ package ${\sf CompQuadForm}$ \citep{duchesne2010computing} to compute the $p$-value of $\|\mathbb{H}^{-1}\tilde{\bZ}_{\bullet,\mathrm{top}(k)}\bepsilon_x\|_2^2$.

\subsection{Resampling-free dCRT with non-Gaussian $X$}\label{sec:gauss:trans}

Let $\Phi(\cdot)$ denote the cumulative distribution function of the standard normal distribution and denote by $\sigma^2_i={\rm Var}(X_i\,|\,Z_{i\cdot})$. In Algorithm~\ref{alg:mcf:crt}, we describe how to transform non-Gaussian $X$ to be Gaussian with the same conditional variance, so that the resampling-free dCRT (for certain statistics) can be applied.
\begin{algorithm}[htbp]
\caption{\label{alg:mcf:crt} Gaussian transformation.}
{\bf If} $X$ is continuous with conditional cumulative distribution function $F(x\,|\,Z),$, let $U_i = \sigma_i\Phi^{-1}(F(X_i\,|\,Z_i))$ for $i=1,2,\ldots,n$.\\
\vspace{0.2cm}
{\bf If} $X$ is discrete and supported on $\mathcal{X}=\{a_k:k\in K\}$ where $K\subseteq\mathbb{N}$ is some set of indices, $a_{k_1}<a_{k_2}$ for all $k_1<k_2$ and $\mathbb{P}(X_i=a_k\,|\,Z_{i})\neq 0$ for all $k\in K$: for $X_i=a_k$, draw $V_{i}$ uniformly from $\left[\mathbb{P}(X_i<a_k\,|\,Z_{i}),\mathbb{P}(X_{i}\leq a_{k}\,|\,Z_{i})\right]$ and $U_i = \sigma_i\Phi^{-1}(V_i)$.\\
\vspace{0.2cm}
{\bf Output:} $\bu=(U_1,U_2,\ldots,U_n)\trans$.
\end{algorithm}
Lemma~\ref{lem:gaustrans} establishes the properties that make $\bu$ a good Gaussian stand-in for $\x-\dx$, so that it can be used in a test statistic $T$ in the same way as $\x-\dx$ while being amenable to the resampling-free speedup. 

\begin{lemma}\label{lem:gaustrans}
The $U_i$ ouput by Algorithm~\ref{alg:mcf:crt} are (i) monotonically increasing in $X_i$, (ii) distributed as $\N(0,\sigma_i^2)$ given $Z_i$, and (iii) independent from $Z_i$.
\end{lemma}

\begin{proof}
For (i), when $X_{i}$ is continuous, we note that both $\Phi(x)$ and $F(x\,|\,Z_{i})$ are increasing and this implies $U_{i}$ is unique and monotonically increasing with $X_{i}$. When $X_{i}$ is discrete, noting that the range of $V_i$ does not intersect as $X_{i}$ takes different values and the range of $V_i$ is increasing with $X_{i}$, we again have that $U_{i}$ is monotonically increasing with $X_{i}$.

For (ii), when $X_{i}$ is continuous, let $V_i=F(X_{i}\,|\,Z_{i})$ and when $X_{i}$ is discrete, let $V_i$ be defined as in Algorithm~\ref{alg:mcf:crt}. Since $V_i$ is uniformly distributed on $[0,1]$ conditional on $Z_{i}$, we have that for any $u\in\mathbb{R}$,
\begin{equation*}
\begin{split}
\mathbb{P}(U_{i}\leq u\,|\,Z_{i})&=\mathbb{P}(\Phi(U_{i}/\sigma_{i})\leq \Phi(u/\sigma_{i})\,|\,Z_{i})\\
&=\mathbb{P}\left(V_i\leq \Phi(u/\sigma_{i})\,|\,Z_{i}\right)=\Phi(u/\sigma_{i}),
\end{split}    
\end{equation*}
which indicates that $\mathbb{P}(U_{i}/\sigma_{i}\leq v\,|\,Z_{i})=\Phi(v)$ and $U_{i}\sim\mathcal{N}(0,\sigma_{i}^2)$. Also, we have $\mathbb{P}(U_{i}/\sigma_{i}\leq v)=\Phi(v)=\mathbb{P}(U_{i}/\sigma_{i}\leq v\,|\,Z_{i})$ for all $v\in\mathbb{R}$, which indicates that $U_{i}/\sigma_{i}\indp Z_{i}$.
\end{proof}

\section{Recycling computation for lasso-based distillation}

Here, prove Lemma~\ref{lem:lime-speedup} and Theorem~\ref{prop:lasso-shortcut} from Section~\ref{sec:recycling}. These lead to Algorithm~\ref{alg:loco-crt-lasso} below.

\begin{algorithm}[htbp]
\caption{\label{alg:loco-crt-lasso} dCRT for $L_1$-regularized M-estimators.}	

{\bf Input:} $\bX, \by$, sequence $\lambda(1) > \cdots > \lambda(G) > 0$, loss $\ell$, sequential rule $\widehat g$.\\
\vspace{0.2cm}
Fit a cross-validated lasso on $(\bX, \by)$ to obtain $\widehat \beta(\bX, \by; \lambda(\widehat g(\bX, \by)))$, and record the active set $\mathcal A$ as defined in \eqref{active-set}.\\
\vspace{0.2cm}
{\bf For} $j \in \mathcal A$: Refit the lasso on $(\bX_{\bullet,\noj}, \by)$ to obtain $\widehat \beta_{\noj} \equiv \widehat \beta(\bX_{\bullet,\noj}, \by; \lambda(\widehat g(\bX_{\bullet,\noj}, \by)))$.\\
\vspace{0.2cm}
{\bf For} $j \not\in \mathcal A$: Set $\widehat \beta_{\noj} = \widehat \beta_{\noj}(\bX, \by; \lambda(\widehat g(\bX, \by)))$. \\
\vspace{0.2cm}
For each $j$, let $d_{y,j} = \bX_{\bullet,\noj}\widehat \beta_{\noj}$.\\
\vspace{0.2cm}
{\bf Output:} Distillations $d_{y,j}$ for each variable $j$.
\end{algorithm} 		

\begin{proof}[~of Lemma~\ref{lem:lime-speedup}]
	% The fact that the solutions to \eqref{eq:beta-hat} and \eqref{lasso} are unique, as long as $\ell$ is differentiable and strictly convex and the columns of $X$ are in general position, is proved in Ryan Tibshirani's lecture notes \cite{}.
	% but since this is not a permanent link, we prove it below for completeness. We prove this fact only for $\widehat \beta(\bX, \by; \lambda)$, the other claim is entirely analogous.
	%To summarize the uniqueness proof from Tibshirani~\cite{Tibshirani2013}, note that any solution $\widehat \beta(\bX, \by; \lambda)$ must satisfy the KKT conditions for \eqref{eq:beta-hat}:
	Since $\widehat \beta(\bX, \by; \lambda)$ is a minimizer of the convex objective 
	\begin{equation}\label{eq:beta-hat}
	\sum_{i = 1}^n \ell(\by_i, \bX_{i,\bullet}\beta) + \lambda \|\beta\|_1,
	\end{equation}
	0 must belong to its subgradient at this point. This means that there exists an $\widehat s \in \mathbb R^p$ such that
	\begin{equation}
	\label{eq:KKT-all}
	\sum_{i=1}^n {\bX_{i,j'}} \dot \ell(\by_i, \bX_{i,\bullet} \widehat \beta(\bX, \by; \lambda)) + \lambda \cdot \widehat s_{j'} = 0 \quad \text{for all } j' = 1, \dots, p,
	\end{equation}
	where $\widehat s_{j'} = \text{sign}(\widehat \beta_{j'}(\bX, \by; \lambda))$ if $\widehat \beta_{j'}(\bX, \by; \lambda) \neq 0$ and $\widehat s_{j'} \in [-1,1]$ otherwise. If $\widehat \beta_{j}(\bX, \by; \lambda) = 0$, then we have
	\begin{equation*}
	\bX_{i,\text{--}j}\widehat \beta_{\noj}(\bX, \by; \lambda) = \bX_{i,\bullet}\widehat \beta(\bX, \by; \lambda) \text{ for every } i = 1, \dots, n,
	\end{equation*}
	which together with equation~\eqref{eq:KKT-all} implies that
	\begin{equation*}
	\text{for all } j' \neq j, \quad \sum_{i=1}^n {\bX_{i,j'}} \dot \ell(\by_i, \bX_{i,\noj}\widehat \beta_{\noj}(\bX, \by; \lambda)) + \lambda \cdot s_{j'} = 0.
	\end{equation*}
	%or in matrix notation
	%\[
	%X_{\noj}^{\top} \nabla \ell(X_{\noj}\widehat \beta_{\noj}) + \lambda \widehat s_{\noj} = 0.
	%\]
	Therefore, $\widehat \beta_{\noj}(\bX, \by; \lambda)$ satisfies the first-order optimality condition in the convex problem~\eqref{lasso}, so it must be a minimizer. Given the assumed general position of the columns of $\bX$ and the assumptions on $\ell$, the minimizer of the problem~\eqref{lasso} is unique \citep{tibshirani2013lasso}. Therefore, $\widehat \beta(\bX_{\bullet,\noj}, \by; \lambda) = \widehat \beta_{\noj}(\bX, \by; \lambda)$, as desired.
\end{proof}

\vspace{0.1in}

\begin{proof}[~of Theorem~\ref{prop:lasso-shortcut}]
	Fix $j \not \in \mathcal A$. For this variable, $\widehat \beta_j(\bX_{\text{--}D_k,\bullet}, \by_{\text{--}D_k}; \lambda(g)) = 0$ for each fold $k$ and for all $g \leq \widetilde g(\bX, \by)$. By Lemma~\ref{lem:lime-speedup} applied to $(\bX_{\text{--}D_k,\bullet}, \by_{\text{--}D_k})$, it follows that 
	\begin{equation*}
	\widehat \beta(\bX_{\text{--}D_k,\noj}, \by_{\text{--}D_k}; \lambda(g)) = \widehat \beta_{\noj}(\bX_{\text{--}D_k,\bullet}, \by_{\text{--}D_k}; \lambda(g)) \quad \text{ for each fold } k \text{ and all } g \leq \widetilde g(\bX, \by).
	\end{equation*}
	Therefore, the lasso cross-validation errors for $(\bX, \by)$ and $(\bX_{\bullet,\noj}, \by)$ coincide for each $g \leq \widetilde g(\bX, \by)$:
	\begin{equation*}
	\mathcal E_{\noj,g} = \sum_{i = 1}^n \ell(\by_i, \bX_{i,\noj} \widehat \beta(\bX_{\text{--}D_k,\noj}, \by_{\text{--}D_k}; \lambda(g))) = \sum_{i = 1}^n \ell(\by_i, \bX_{i,\bullet} \widehat \beta(\bX_{\text{--}D_k,\bullet}, \by_{\text{--}D_k}; \lambda(g))) = \mathcal E_g.
	\end{equation*}
	Because the rule to choose $\lambda$ is sequential, we conclude that $\widetilde g(\bX_{\bullet,\noj},\by) = \widetilde g(\bX,\by)$ and also $\widehat g(\bX_{\bullet,\noj}, \by) = \widehat g(\bX, \by)$. The conclusion~\eqref{conclusion} now follows from another application of Lemma~\ref{lem:lime-speedup}, this time with the full data $(\bX, \by)$ and the regularization parameter $\lambda(\widehat g(\bX_{\bullet,\noj}, \by)) = \lambda(\widehat g(\bX, \by))$.
\end{proof}

{
\section{Comparing dCRT to \rev{o}CRT} \label{sec:dCRT-versus-CRT}

\subsection{Setup}

To draw a distinction with the \rev{o}CRT, let us denote by
\begin{equation*}
(\widehat \beta_x\suplasso, \widehat \beta_z\suplasso) = \underset{\beta_x, \beta_z}{\arg \min}\ \frac12\|y - x\beta_x - Z\beta_z\|^2 + \lambda |\beta_x| + \lambda \|\beta_z\|_1
\end{equation*}
the solution to the full lasso problem, while denoting by
\begin{equation*}
\widehat \beta\suploco_z = \underset{\beta_z}{\arg \min}\ \frac12\|y - Z\beta_z\|^2 + \lambda \|\beta_z\|_1; \quad \widehat \beta\suploco_x =  \frac{(x - d_x)^{\top} (y - Z \widehat\beta\suploco_z)}{\|x - d_x\|^2}
\end{equation*}
the solution to the distilled lasso problem. In this notation, we have
\begin{equation}
T_{\text{dCRT}} = |\widehat \beta\suploco_x| \quad \text{and} \quad T_{\text{\rev{o}CRT}} = |\widehat \beta_x\suplasso|.
\label{dCRT-versus-CRT-1}
\end{equation}

To contrast $\widehat \beta_x\suplasso$ to $\widehat \beta\suploco_x$, it is helpful to rewrite the former in a way that parallels the definition of the latter. To this end, note that the KKT conditions for the full lasso problem state that
\begin{equation}
x^{\top}(y - x \widehat \beta_x\suplasso - Z \widehat \beta_z\suplasso) = \lambda \cdot \text{sign}(\widehat \beta_x\suplasso),
\end{equation}
where $\text{sign}(\widehat \beta_x\suplasso)$ is defined as 
% usual for $\widehat \beta_x\suplasso \neq 0$ and as 
any number in $[-1,1]$ when $\widehat \beta_x\suplasso = 0$. Rearranging yields
\begin{equation}
\widehat \beta_x\suplasso = S_{\lambda/\|x\|^2}\left(\frac{x^{\top}(y - Z \widehat \beta_z\suplasso)}{\|x\|^2}\right),
\end{equation}
where $S_{\lambda}$ is the soft-threshold operator. Using this identity, let us rewrite definition~\eqref{dCRT-versus-CRT-1} as follows:
\begin{equation}
	T_{\text{dCRT}} = \left|\frac{(x - d_x)^{\top} (y - Z \widehat\beta\suploco_z)}{\|x - d_x\|^2}\right|; \quad T_{\text{\rev{o}CRT}} = \left|S_{\lambda/\|x\|^2}\left(\frac{x^{\top}(y - Z \widehat \beta_z\suplasso)}{\|x\|^2}\right)\right|.
	\label{dCRT-vs-CRT}
\end{equation}

Having written these test statistics side by side in this way, we observe three differences:
\begin{enumerate}
	\item $T_{\text{\rev{o}CRT}}$ includes a soft thresholding operation while $T_{\text{dCRT}}$ does not.
	\item $T_{\text{dCRT}}$ uses $x - d_x$ where $T_{\text{\rev{o}CRT}}$ uses $x$.
	\item $T_{\text{\rev{o}CRT}}$ is based on the full lasso coefficients $\widehat \beta_z\suplasso$, while $T_{\text{dCRT}}$ is based on the lasso coefficients $\widehat\beta\suploco_z$ resulting from holding out the variable $x$.
\end{enumerate}
We may think of the third difference as being the main one between \rev{o}CRT and dCRT, but the first and second differences must also be accounted for. In fact, we claim that the first and second differences actually account for the majority of the difference in power between \rev{o}CRT and dCRT. In the remainder of this section, we examine each of these differences and their implications one at a time. 

We illustrate these differences in the context of the following simple numerical simulation setup. Consider $n = p = 800$, with the rows of the design matrix distributed as $N(0,\Sigma)$ for $\Sigma_{ij} = \rho^{|i-j|}$ and the response generated from the Gaussian linear model with $s = 50$ nonzero coefficients. The nonzero coefficients have the same magnitude but alternating signs. We let $x$ represent one of the non-null variables and $Z$ represent all other variables. We consider $\rho = 0$ or $\rho = 0.5$, and we consider the non-null variables either being adjacent to each other or equally spaced. 

\subsection{The effect of soft thresholding on the \rev{o}CRT}

Note that the soft-thresholding operation in the \rev{o}CRT can cause the test statistic to be exactly equal to zero sometimes, leading to a $p$-value of one. We argue that this can only decrease the power of the \rev{o}CRT. 

Let $\lambda' = \lambda/\|x\|^2$ be the value of $\lambda$ selected by cross-validation and scaled by $\|x\|^2$ and let $W = \frac{x^{\top}(y - Z \widehat \beta_z\suplasso)}{\|x\|^2}$. Furthermore, let $\widetilde \lambda '$ and $\widetilde W$ be the corresponding quantities obtained from a resampled dataset $(\widetilde x, y, Z)$. Then, we find that
\begin{equation*}
\begin{split}
p_{\text{\rev{o}CRT}} &= \mathbb P\left[|S_{\widetilde \lambda'}(\widetilde W)| \geq |S_{\lambda'}(W)| \ \mid \ x, y, Z\right] \\
&\approx \mathbb P\left[|S_{\lambda'}(\widetilde W)| \geq |S_{\lambda'}(W)| \ \mid \ x, y, Z\right] \\
&\geq \mathbb P\left[|\widetilde W| \geq |W| \ \mid \ x, y, Z\right].
\end{split}
\end{equation*}
The second line rests on the approximation $\widetilde \lambda ' \approx \lambda'$; i.e. that resampling does not change the scaled penalty parameter by too much. This is a plausible approximation whose rigorous study we defer to future work. The inequality in the third line follows from the observation that for any $w_1, w_2 \in \mathbb R$ and $\lambda \geq 0$,
\begin{equation*}
	|w_1| \geq |w_2| \quad \Longrightarrow \quad |S_\lambda(w_1)| \geq |S_\lambda(w_2)|.
\end{equation*}

In other words, removing the soft threshold from the \rev{o}CRT test statistic can only decrease the \rev{o}CRT $p$-value. We illustrate this in Figure~\ref{fig:soft-threshold}, considering the situation when $\rho = 0.5$. We observed that the \rev{o}CRT version without soft threshold uniformly dominates the original across all our simulations.
\begin{figure}[h!]
	\centering
	\includegraphics[width = 0.4\textwidth]{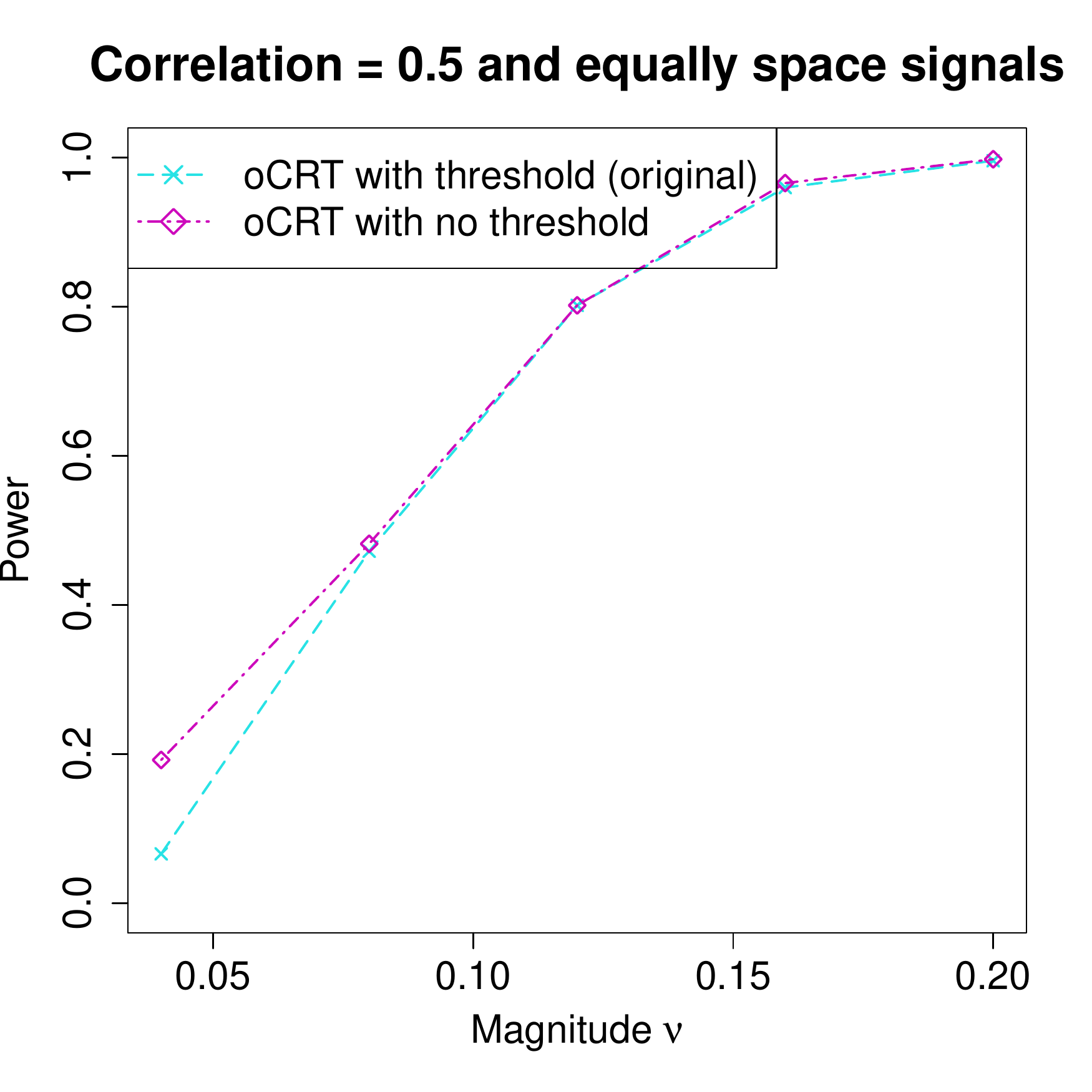}
	\includegraphics[width = 0.4\textwidth]{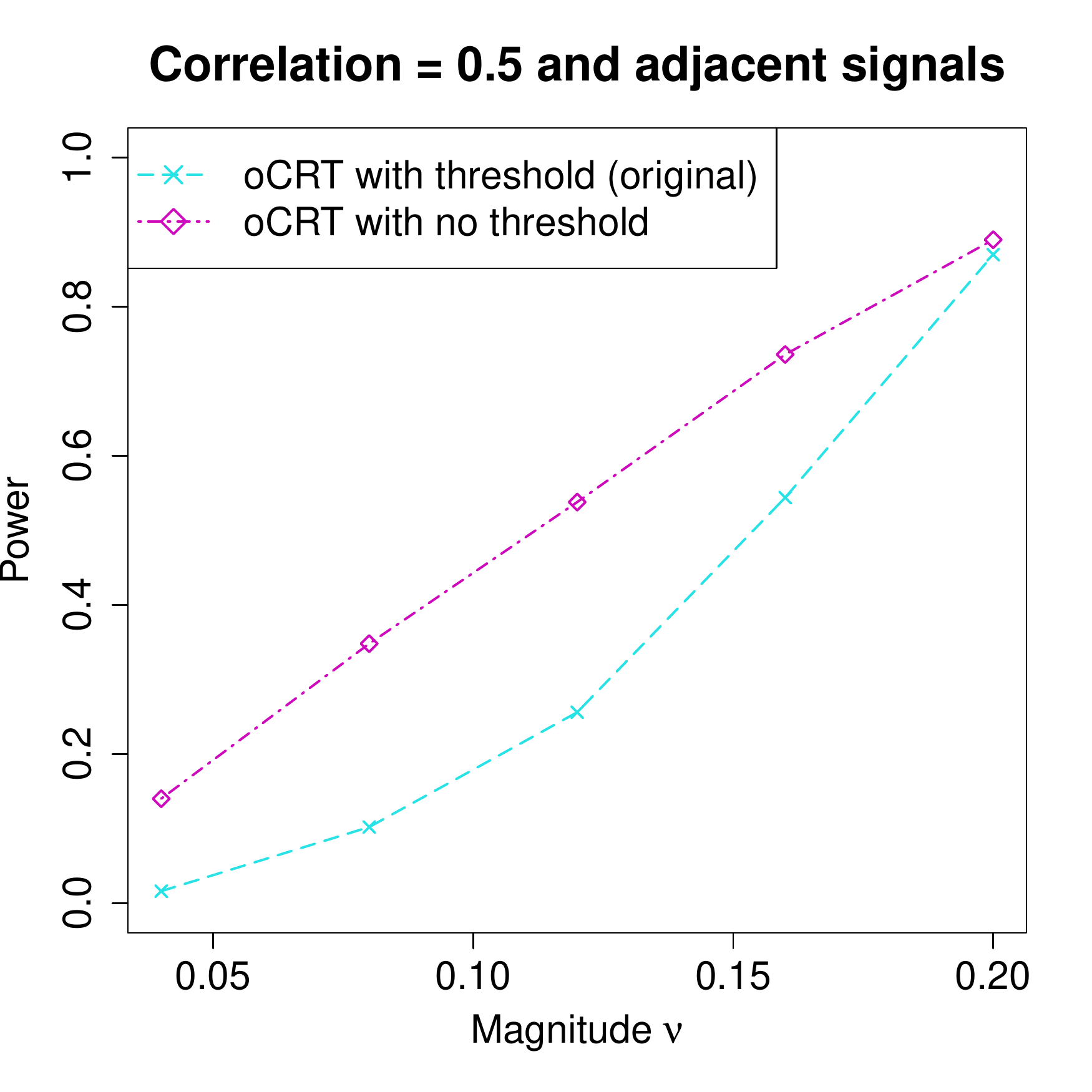}
	\caption{We compare the power of the \rev{o}CRT with no soft threshold with that of \rev{o}CRT with soft threshold (used in \cite{candes2018panning}) in the simulation with the auto-regressive correlation coefficient $\rho = 0.5$, equally spaced or adjacent signals, and the magnitude of the signals varying, as described in Section \ref{sec:dcrt-vs-crt}.}
	\label{fig:soft-threshold}
\end{figure}

\subsection{The effect of distilling $x$}

For the purposes of this subsection, we ignore the normalizations in the denominators of $T_{\text{\rev{o}CRT}}$ and $T_{\text{dCRT}}$ in equation~\eqref{dCRT-vs-CRT} and remove the soft threshold from $T_{\text{\rev{o}CRT}}$, leaving us with
\begin{equation}
	T_{\text{dCRT}} = |(x-d_x)^{\top}(y-Z\widehat\beta\suploco_z)| \quad \text{and} \quad T_{\text{\rev{o}CRT}} = |x^{\top}(y-Z\widehat \beta_z\suplasso)|.
\label{modified-test-stats}
\end{equation}
Now, recall from Section 2.5 that under the null we have
\begin{equation*}
	(x-d_x)^{\top}(y -d_y) \sim N(0, \sigma^2 \|y-d_y\|^2).
\end{equation*}
In particular, the null distribution of $(x-d_x)^{\top}(y -d_y)$ is centered on the origin. This means that taking absolute values results in a two-tailed test, as usual. However, we observe that the quantity $x^{\top}(y - Z\widehat \beta_z\suplasso)$ in the \rev{o}CRT is not necessarily centered on the origin under the null hypothesis. Therefore, taking an absolute value of this quantity to define the \rev{o}CRT may not result in a usual two-tailed test. This can lead to unintended consequences, as we demonstrate next in the context of our simulation example. 

In Figure~\ref{fig:centering}, we compare the dCRT and \rev{o}CRT obtained from definition~\eqref{modified-test-stats} in three scenarios: $\rho = 0$ and equally spaced signals, $\rho = 0.5$ and equally spaced signals, and $\rho = 0.5$ and adjacent signals. Across the top row are the powers of these two methods to reject the null hypothesis at level $0.05$, and we see that in the first scenario the powers are about equal, in the second scenario the power of the \rev{o}CRT is larger, and in the third scenario the power of the dCRT is larger. What does this have to do with centering? In the bottom row, we show the test statistics and resampling-based null distributions (prior to taking the absolute value) for a typical run. We observe that the \rev{o}CRT test statistic and null distribution look very similar to that of the dCRT, up to a horizontal shift (a shift right in the second scenario and a shift left in the third scenario). This horizontal shift would not make any difference for one-sided tests, but in this case we are conducting two-sided tests by applying the absolute value. Therefore, a right-shifted test statistic and null distribution has the effect of decreasing the $p$-value by reducing the contribution of the left tail. On the other hand, a  left-shifted test statistic and null distribution can drastically increase the $p$-value by exaggerating the contribution of the left tail. The differences in power observed in the top row of Figure~\ref{fig:centering} seem to reflect these differences in centering.
\begin{figure}[h!]
	\centering
	\includegraphics[width = 0.32\textwidth]{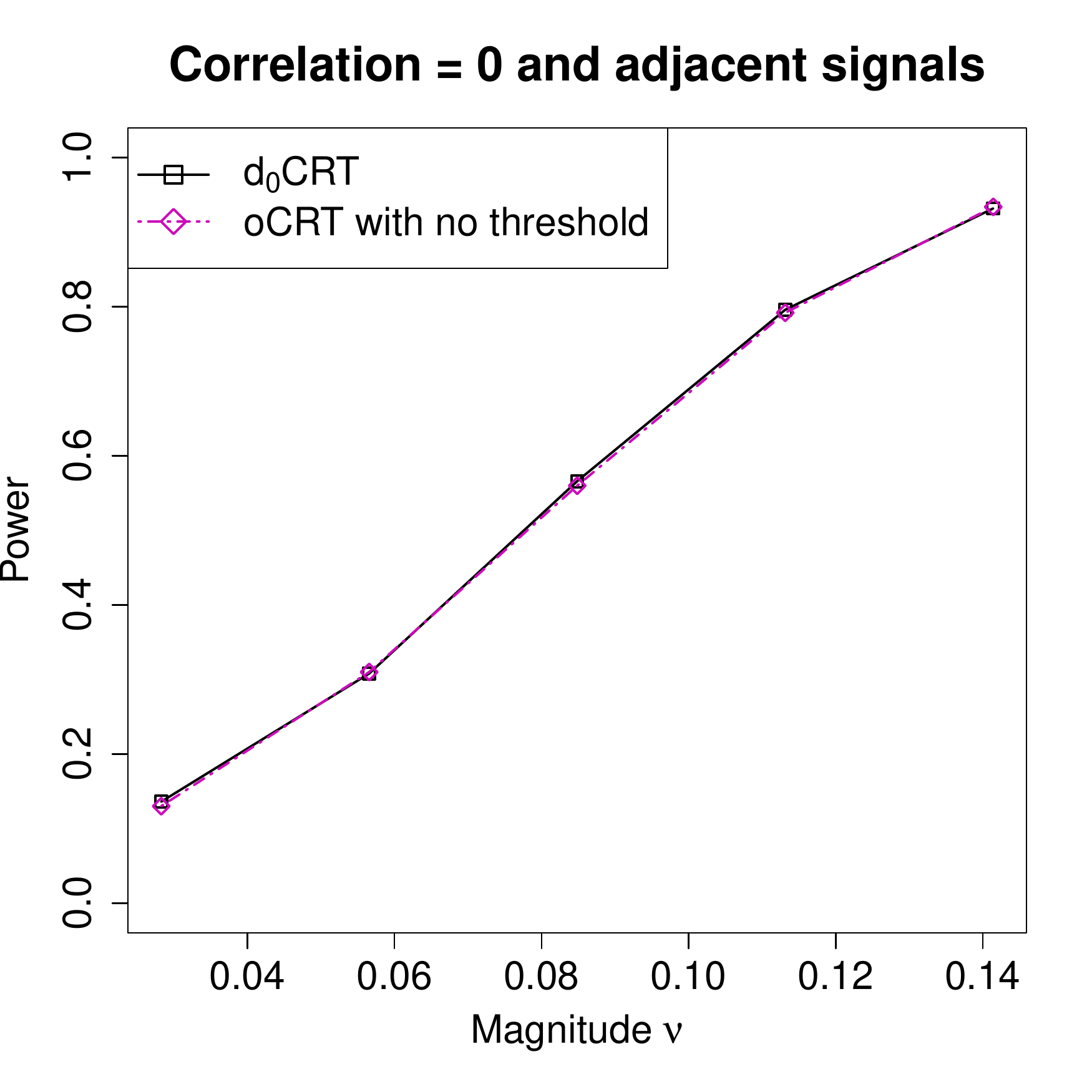}
	\includegraphics[width = 0.32\textwidth]{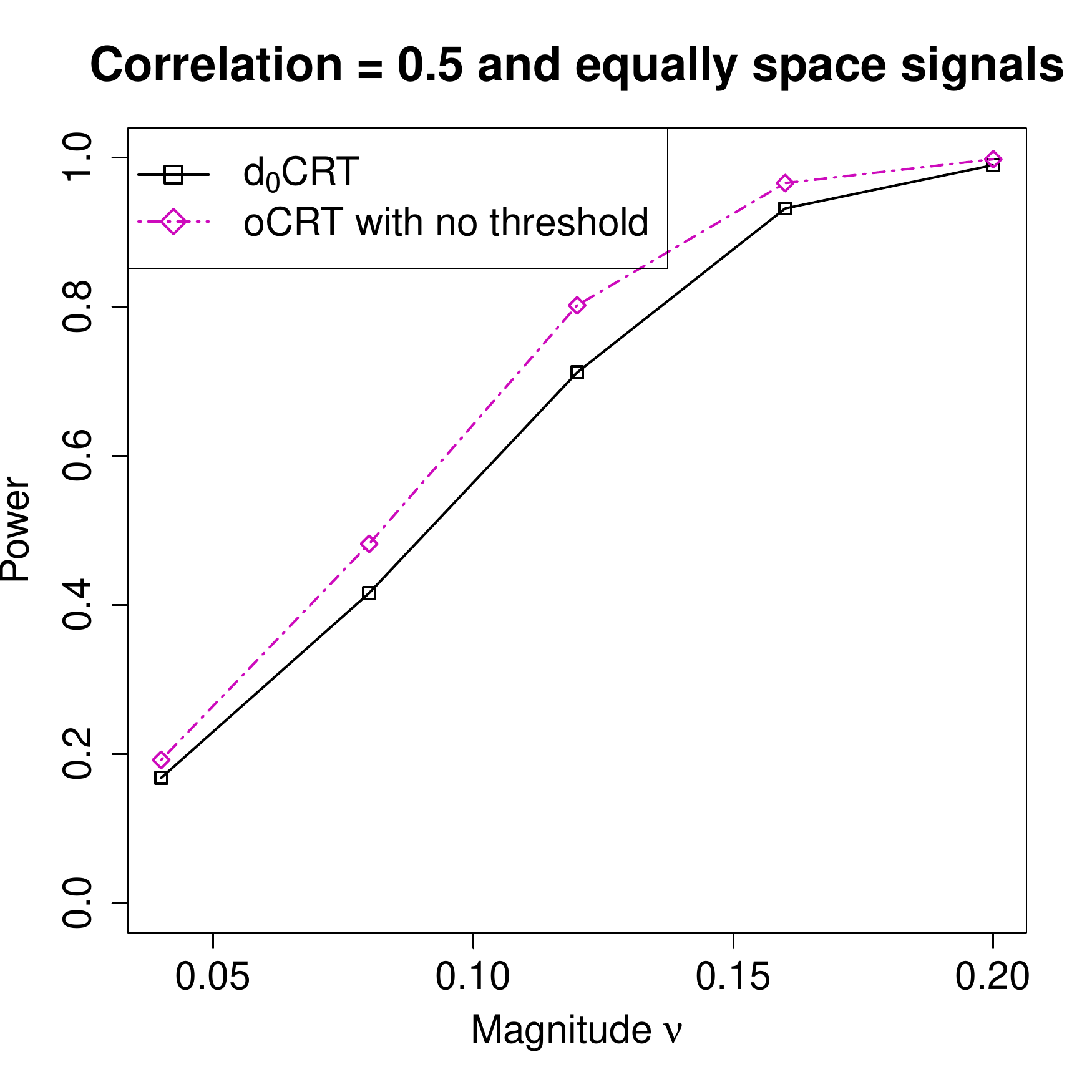}
	\includegraphics[width = 0.32\textwidth]{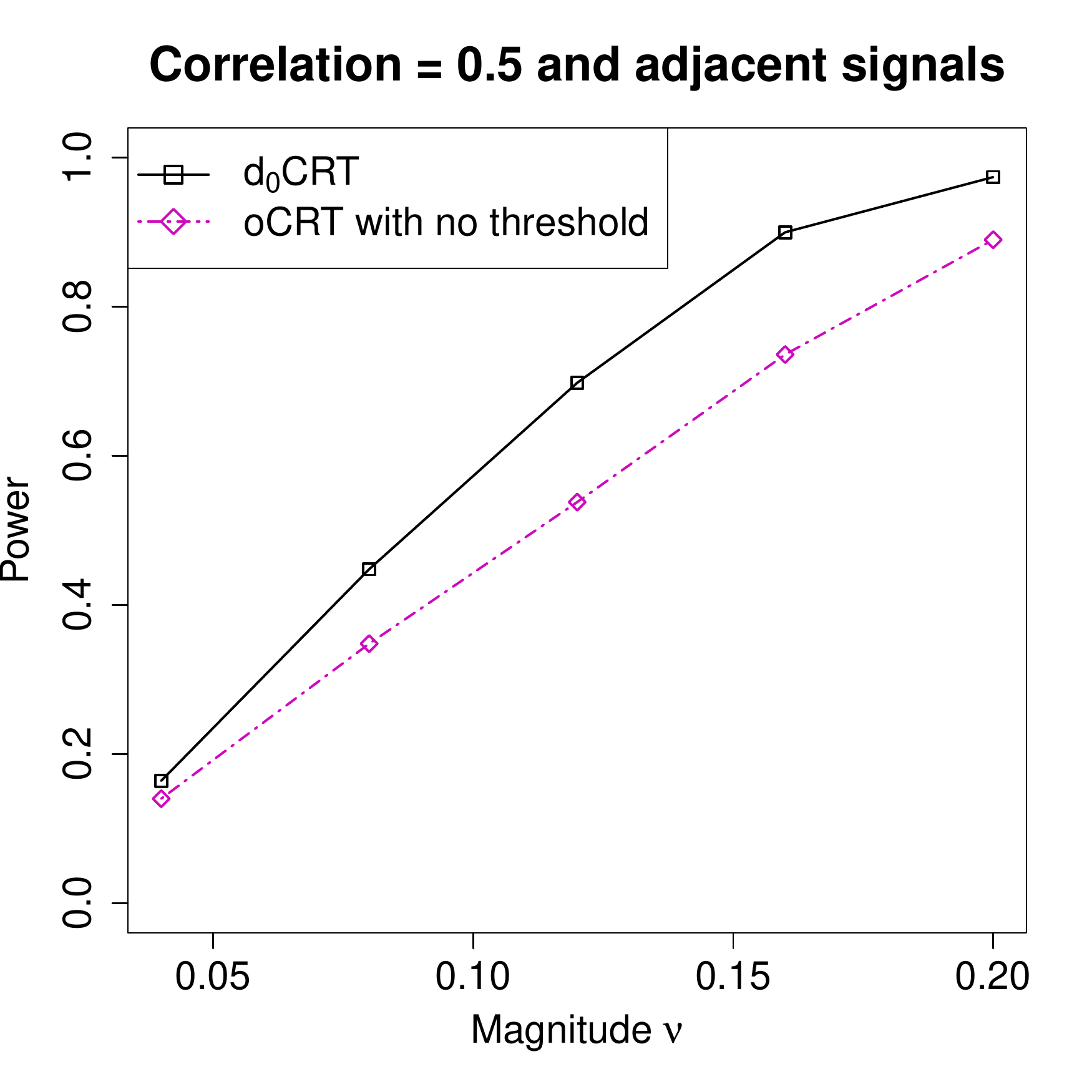}
	\includegraphics[width = 0.32\textwidth]{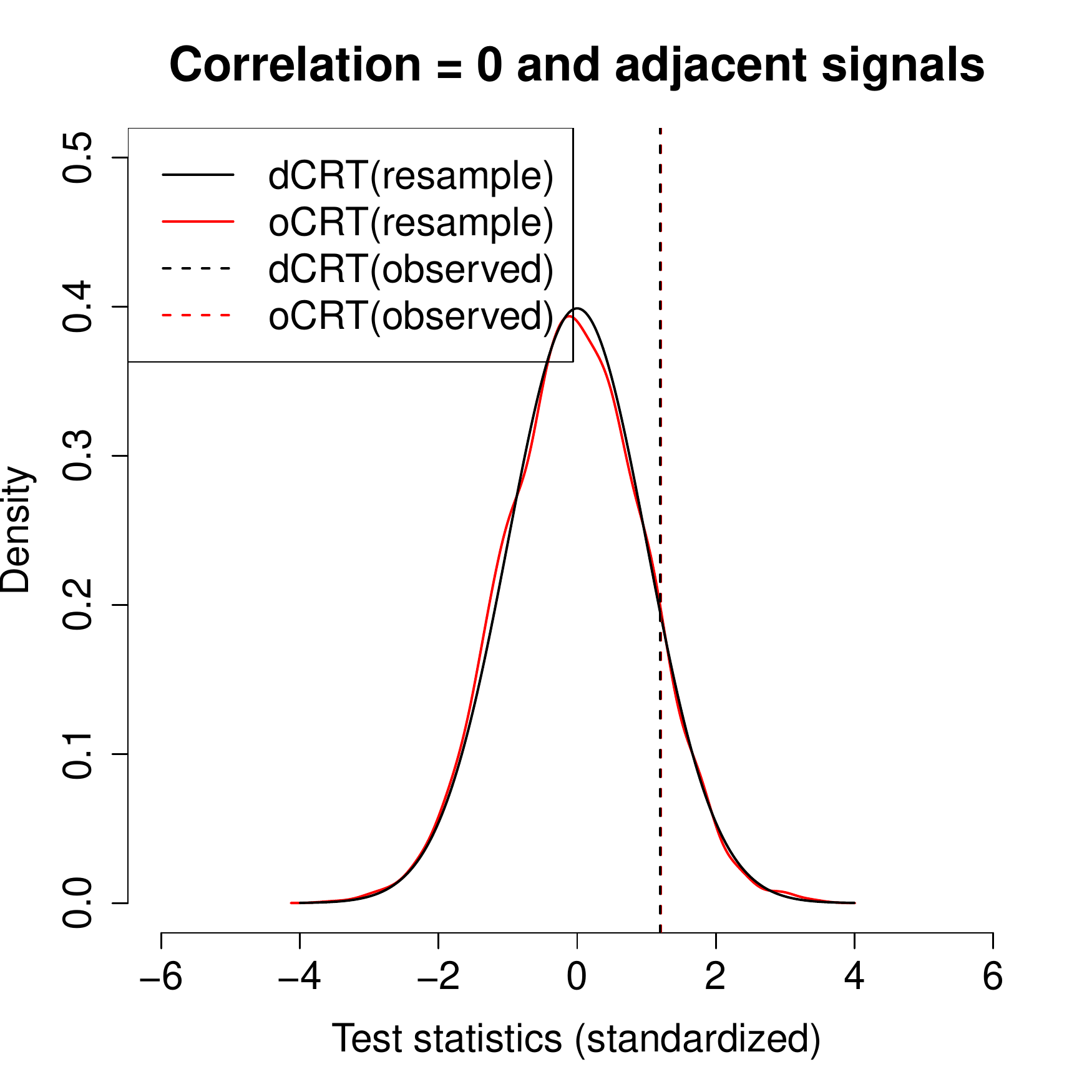}
	\includegraphics[width = 0.32\textwidth]{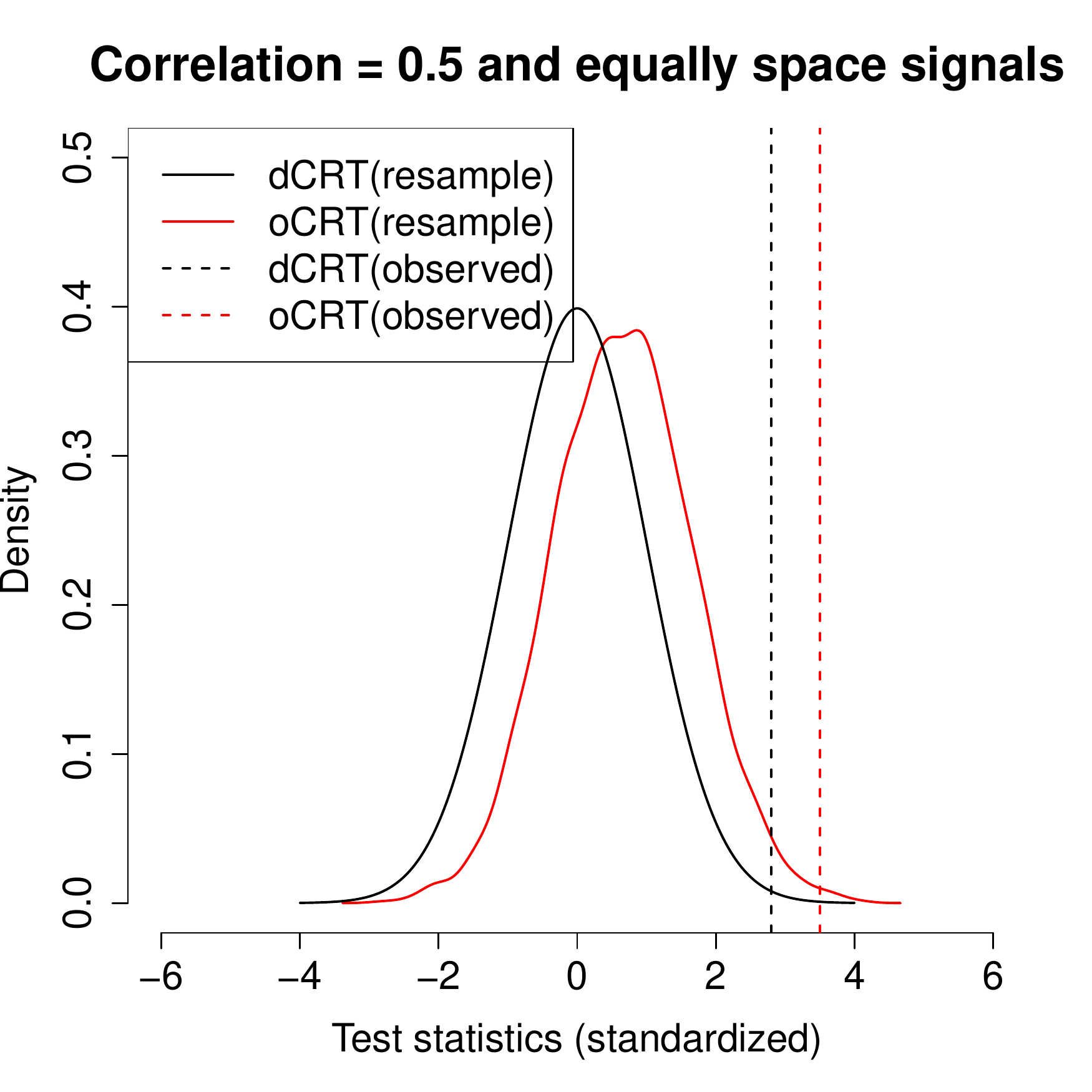}
	\includegraphics[width = 0.32\textwidth]{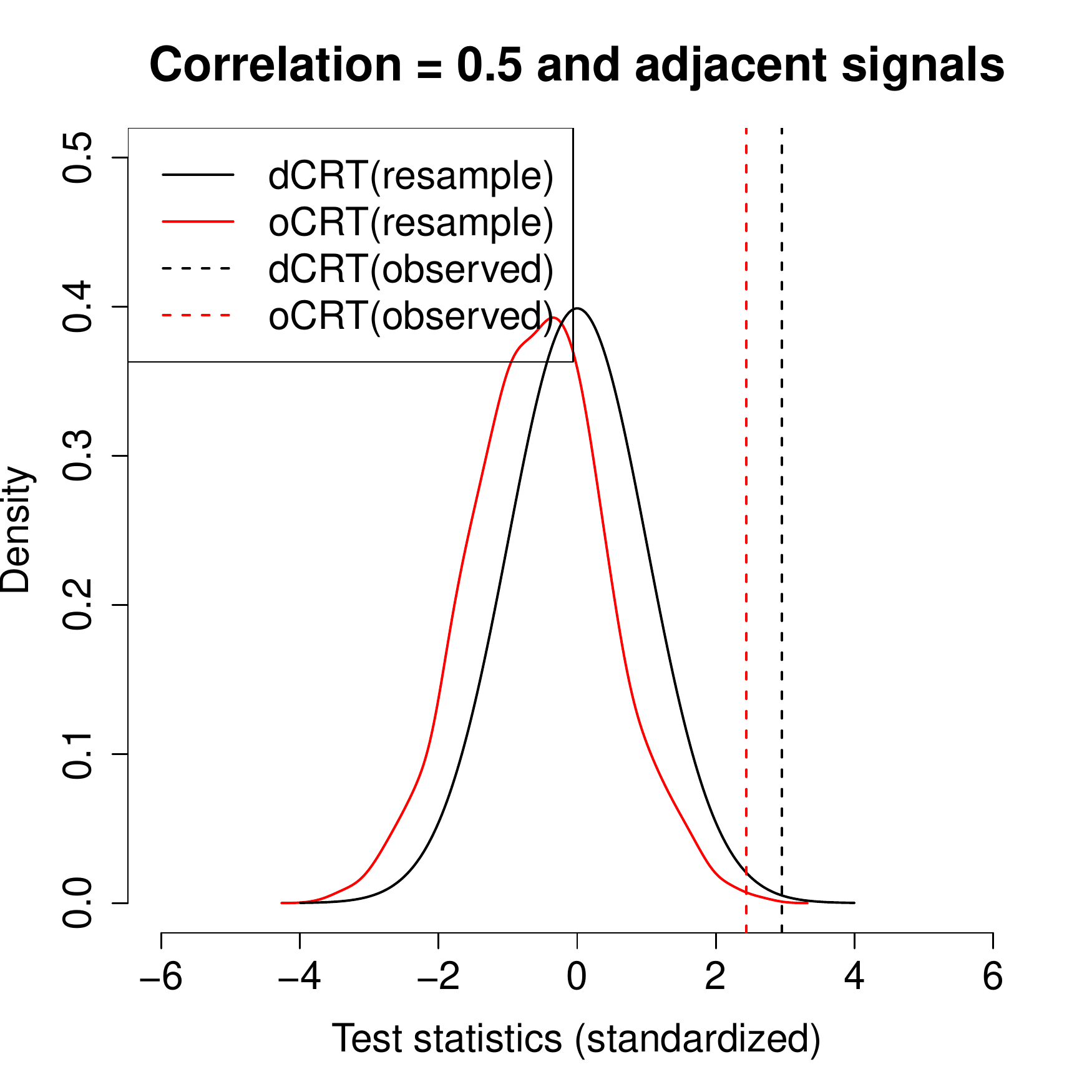}
	\caption{On the top panel, we compare the power of d$_0$CRT with that of \rev{o}CRT with no soft threshold, under the different settings described in Section \ref{sec:dcrt-vs-crt} including the correlation coefficient $\rho=0$, $\rho=0.5$ and equally spaced signals, and $\rho=0.5$ and adjacent signals, with the magnitude of the signals varying. On the bottom panel, we randomly pick one repetition under each of the three settings, smooth and plot the empirical densities of the resampled d$_0$CRT and \rev{o}CRT test statistics, with their observed test statistics indicated by the vertical lines.}
	\label{fig:centering}
\end{figure}

In sum, this centering issue for the \rev{o}CRT can either increase or decrease its power, but it is unclear on any given dataset which direction the effect will be. This lack of centering may therefore be undesirable, even if it does in some cases boost power. To remedy this, one can construct a bona fide two-tailed $p$-value:
\begin{equation}
	p_{\text{\rev{o}CRT}} = 2 \cdot \min\left(\mathbb P\left[\widetilde T_{\text{\rev{o}CRT}} \geq T_{\text{\rev{o}CRT}}\right], \mathbb P\left[-\widetilde T_{\text{\rev{o}CRT}} \geq -T_{\text{\rev{o}CRT}}\right]\right).
\end{equation}

\subsection{Comparing the dCRT to the \rev{o}CRT}

We therefore arrive at three versions of the \rev{o}CRT: original, original but without soft-threshold, and original but without soft-threshold and with centering. In Figure~\ref{fig:final-CRT-comparison}, we consider how these three \rev{o}CRT variants compare to the dCRT.

\begin{figure}[h!]
	\centering
	\includegraphics[width = 0.32\textwidth]{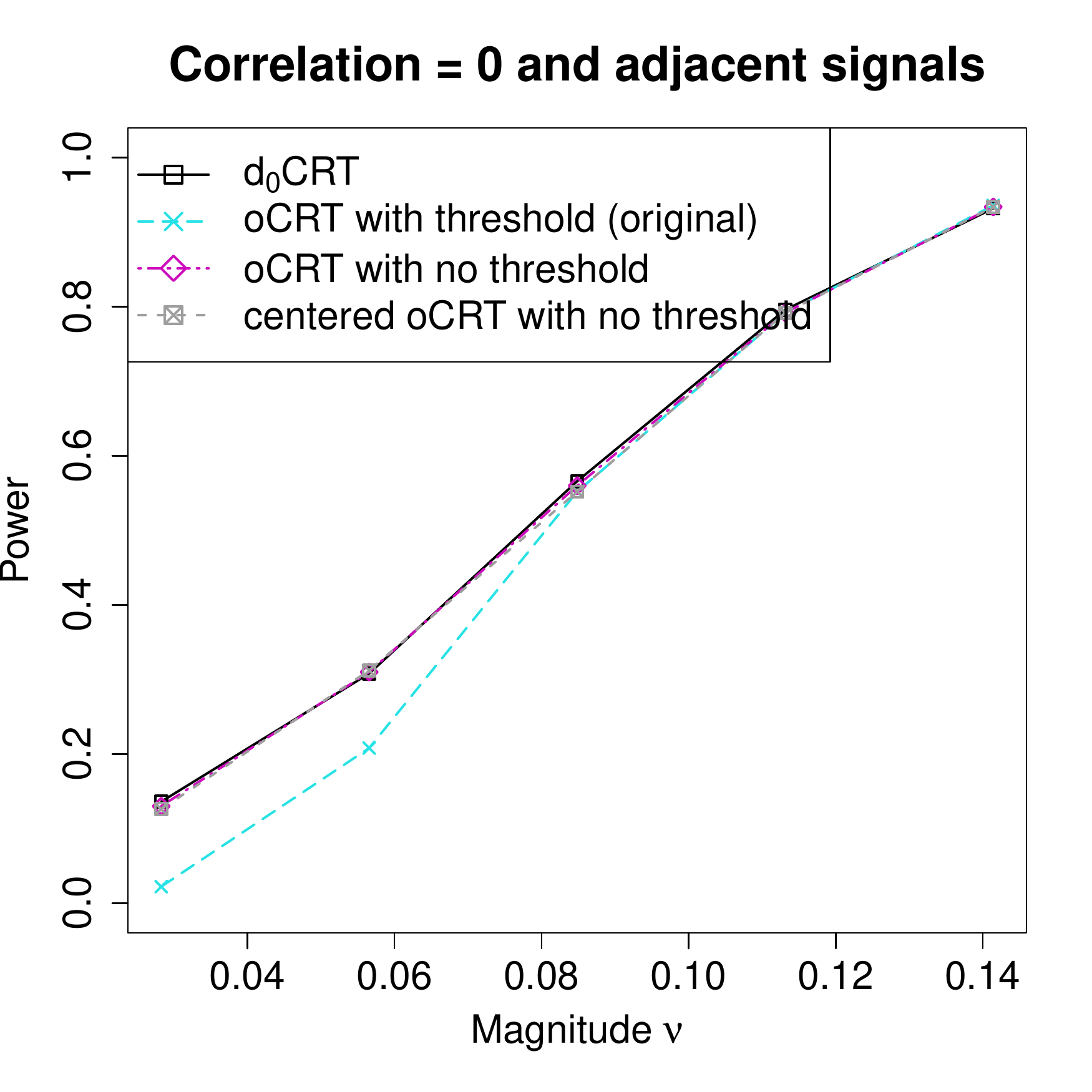}
	\includegraphics[width = 0.32\textwidth]{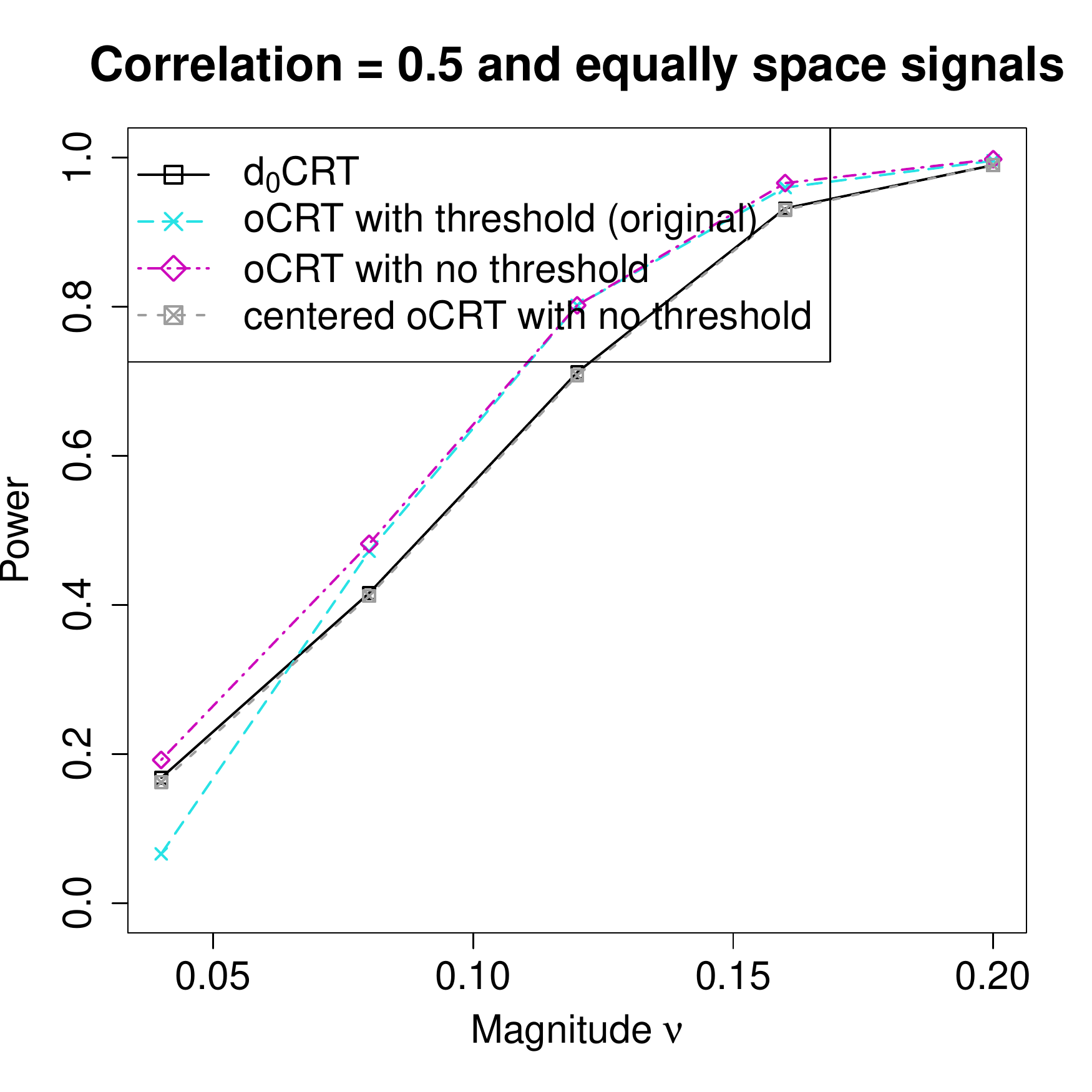}
	\includegraphics[width = 0.32\textwidth]{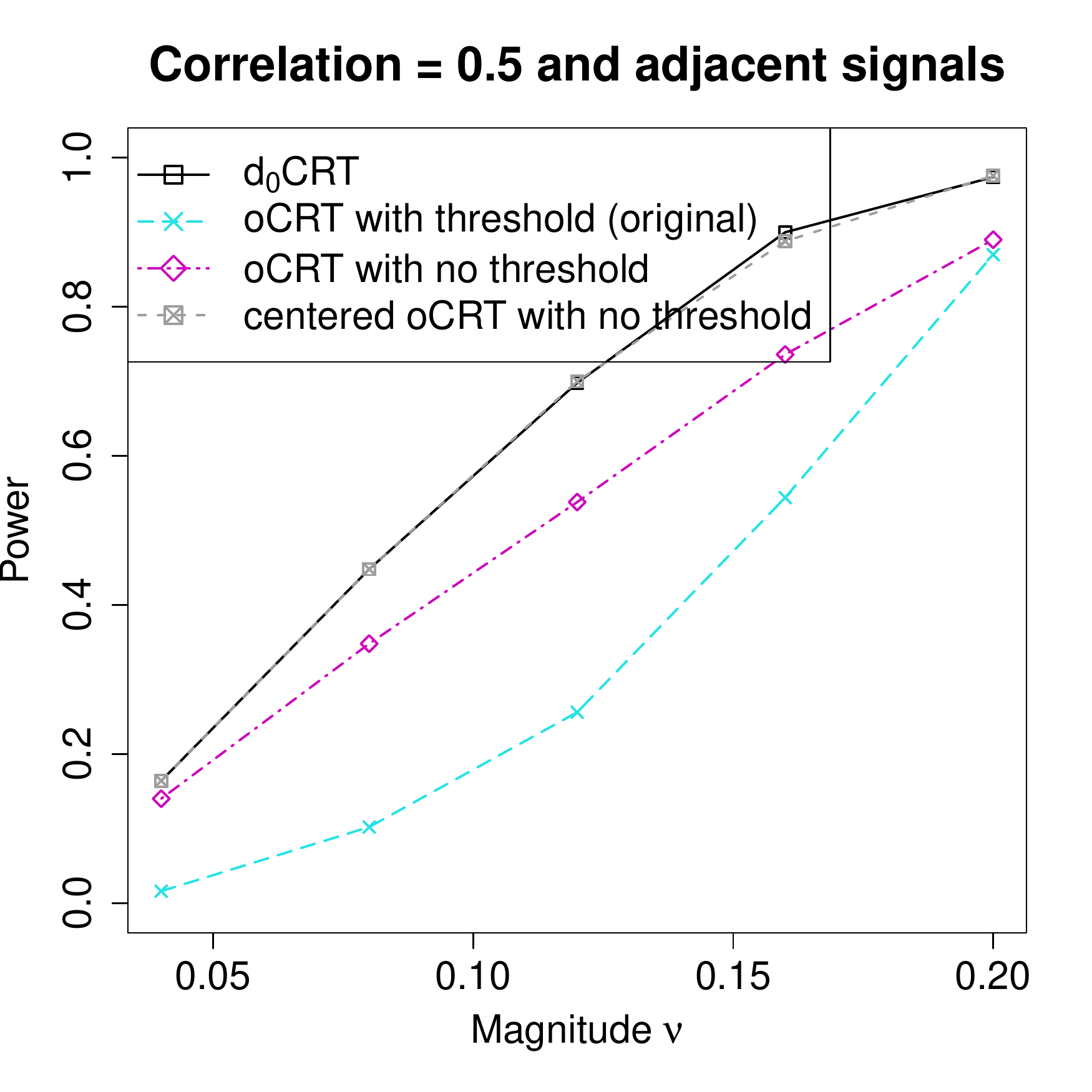}
	\caption{We compare the power of d$_0$CRT, the \rev{o}CRT with soft threshold, \rev{o}CRT without soft threshold, and \rev{o}CRT without soft threshold and with centering, under the different settings described in Section \ref{sec:dcrt-vs-crt} including the correlation coefficient $\rho=0$, $\rho=0.5$ and equally spaced signals, and $\rho=0.5$ and adjacent signals, with the magnitude of the signals varying.}
	\label{fig:final-CRT-comparison}
\end{figure}

The most important observation in this figure is that the dCRT performs almost identically to the centered and non-soft-thresholded \rev{o}CRT. Recalling the three discrepancies between the \rev{o}CRT and dCRT from the beginning of this section, only the third---the effect on the lasso coefficients of leaving one covariate out---remains. We conclude that the effect of leaving $x$ out of the lasso on the power of the CRT is minimal. In conclusion, the computational acceleration facilitated by distillation comes at little or no cost in power, after properly accounting for the effects of soft thresholding and centering.
}

\newpage
\section{Simulation results} \label{sec:app:sim}
In this section, we present the details of the simulations summarized in Section~\ref{sec:sim}. Source code for conducting dCRT and other benchmark methods in our simulation studies can be found at \url{https://github.com/moleibobliu/Distillation-CRT}.

\subsection{Method implementation details}\label{sec:sim:imp}
We describe here the implementation choices and tuning parameters used for the main methods employed in our simulations; these descriptions apply everywhere to the simulated methods unless specifically stated otherwise. For many of our methods we use the lasso, which is implemented in the $\mathsf{R}$ package {\sf glmnet} with family$=$``binomial" if $Y$ is binary and family=``gaussian" otherwise, and penalty parameter selected by 10-fold cross-validation. 

The d$_0$CRT and d$_{\mathrm{I}}$CRT are the resampling-free versions of Examples~\ref{ex:1} and~\ref{ex:2}, respectively. Note that the resampling-free dCRT may not be the most powerful choice for binary responses; we could have used a resampling-based univariate logistic regression statistic instead, but choose not to for computational purposes.

The dimension $k$ in Examples~\ref{ex:2} is set as $k=2\log (p)$. When we combine the dCRT methods with screening from Section~\ref{sec:screening}, the screening is done by running the 10-fold-cross-validated lasso and keeping only the covariates with nonzero fitted coefficients.

The other methods we include are HRT and knockoffs (the last only in simulations targeting false discovery rate control). We implement the HRT of Algorithm 1 in \cite{tansey2018holdout} with linear model fitted by the lasso, empirical risk function set to logistic loss for binary $Y$ and sum of squared error otherwise, and a data split of 50\%-50\%. Due to data-splitting, the fitted lasso of the HRT is independent of the data used for hypothesis testing. Thus, the Benjamini--Hochberg and Bonferroni correction procedures can be used on the $p$-values of the variables selected by the lasso, instead of on the full set of variables. This reduction of the multiplicity burden improves the power of the HRT, and we apply this screening step in all simulations and data analysis. In multiple testing simulations, Benjamini--Hochberg is applied to the $p$-values of all methods except knockoffs. As we set $p\leq 800$ and false discovery rate level $\alpha=0.1$ for multiple testing, we set of the number of resamples $M=50,000$ for the CRT approaches (\rev{o}CRT, HRT, non-resampling-free dCRTs). This choice was made to ensure these methods' powers are not affected by $M$, since $M/5=10,000>p/\alpha$, the smallest possible Benjamini--Hochberg cutoff in our simulations. The only exception is in Appendix~\ref{sec:sim:hrt} where $p$ reached as high as $3,200$, and there we choose $M=200,000$ to ensure $M/5>p/\alpha$. For single hypothesis testing simulations where the significance level was $0.05$, we set $M=2,000$ to ensure that we still have $M>1/0.05$. For knockoffs, we use the lasso coefficient difference statistics as defined in the (3.6) of \cite{candes2018panning}, the ``knockoffs+" threshold, and the SDP knockoff construction when $p<1000$ and the equicorrelated construction otherwise. We note that for both autocorrelated and equicorrelated variables, SDP and equicorrelated knockoffs are quite similar.

\subsection{Moderate size data simulation}\label{sec:sim:smc}
We first compare the dCRT with the \rev{o}CRT procedure in \cite{candes2018panning}. We generate Gaussian covariates as AR(1), with autocorrelation coefficient $0.5$. The true model for $Y$ is chosen as either a Gaussian linear model with unit residual variance or a logistic regression model, and in either case the coefficient vector was set to have $s$ nonzero entries of equal magnitude $\nu$ and random signs (each independently having equal probability of being positive or negative). Two types of structures of the coefficient support are considered separately: adjacent support with the first $s$ coefficients being non-zero, and equally-spaced support with the non-zero coefficients indexed by $\{1,[p/s]+1,2[p/s]+1,\ldots,(s-1)[p/s]+1\}$. We pursue two multiple testing goals of selecting non-null variables while controlling the false discovery rate or family-wise error rate, both at level $\alpha=0.1$.

In addition to the methods described in Appendix~\ref{sec:sim:imp}, we implement the \rev{o}CRT with three different test statistics: the fitted coefficients of a linear or logistic lasso regression, elastic net \citep{zou2005regularization} regression (penalty $\lambda(\|\bbeta\|_1+\|\bbeta\|_2^2/2)$), and adaptive lasso regression \citep{zou2006adaptive,huang2008adaptive}, each tuned with 10-fold cross-validation. Due to the high computational burden of these \rev{o}CRTs, we focus on moderate size data with $n=300$, $p=300$ and the sparsity level $s=30$ and vary $\nu$ to observe a range of powers.

The resulting average power in the linear and logistic settings is plotted against the signal magnitude $\nu$ in Figure~\ref{fig:smc} for false discovery rate control and Figure~\ref{fig:smc:fwer} for family-wise error rate control, and the false discovery rate plots are presented in Appendix~\ref{sec:sim:fdr}. All methods control the false discovery rate. For both false discovery rate and family-wise error rate control, the d$_0$CRT and d$_{\mathrm{I}}$CRT significantly outperform the HRT for both types of support, perform better than knockoffs and all the \rev{o}CRT methods for adjacent support but worse than them for equally spaced support. The d$_{\mathrm{I}}$CRT has slightly less power than the d$_0$CRT due to its allowance of interaction effects, since the true model is exactly additive. 
\begin{figure}[htpb]
\centering
  \includegraphics[width=0.4\textwidth]{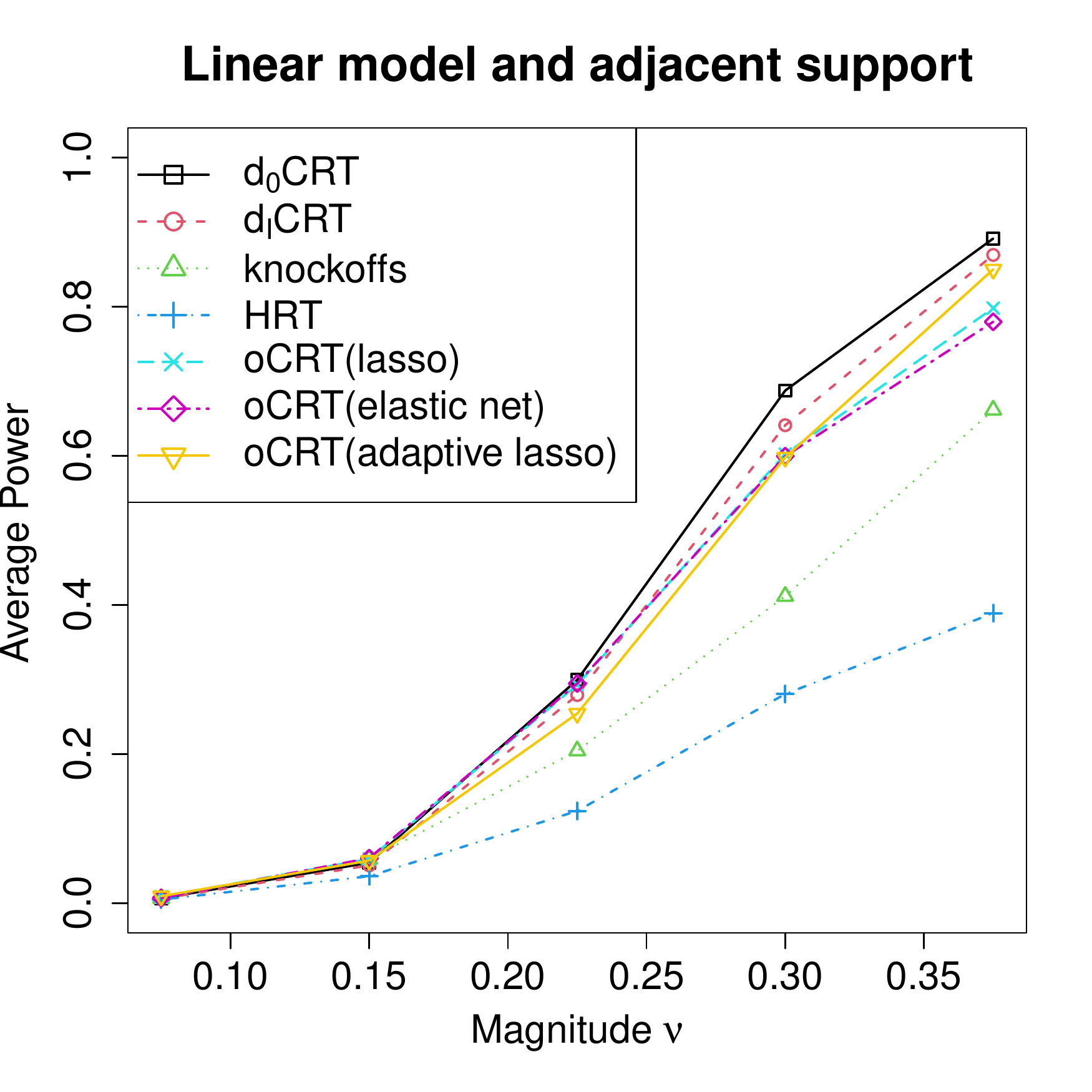}
  \includegraphics[width=0.4\textwidth]{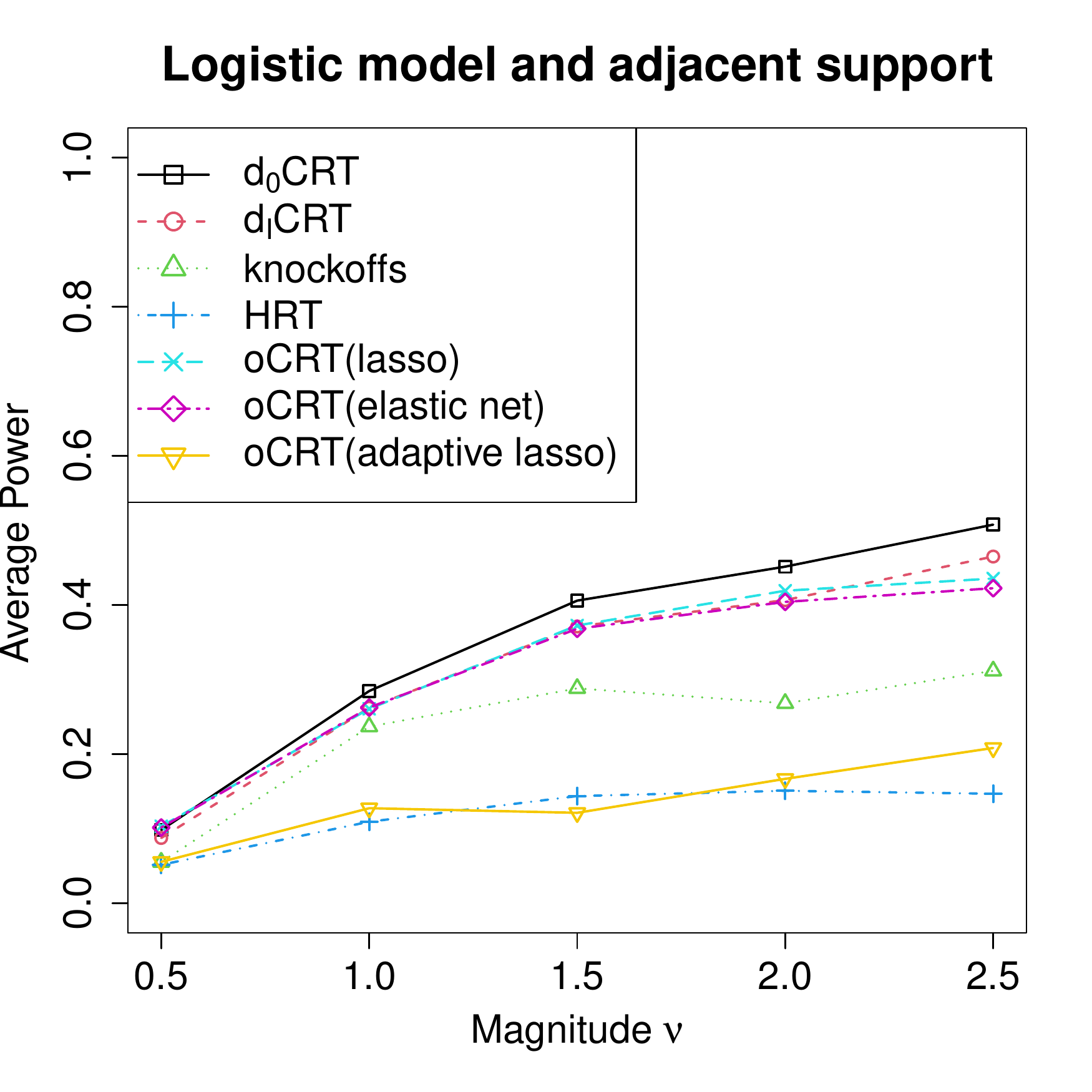}
  \includegraphics[width=0.4\textwidth]{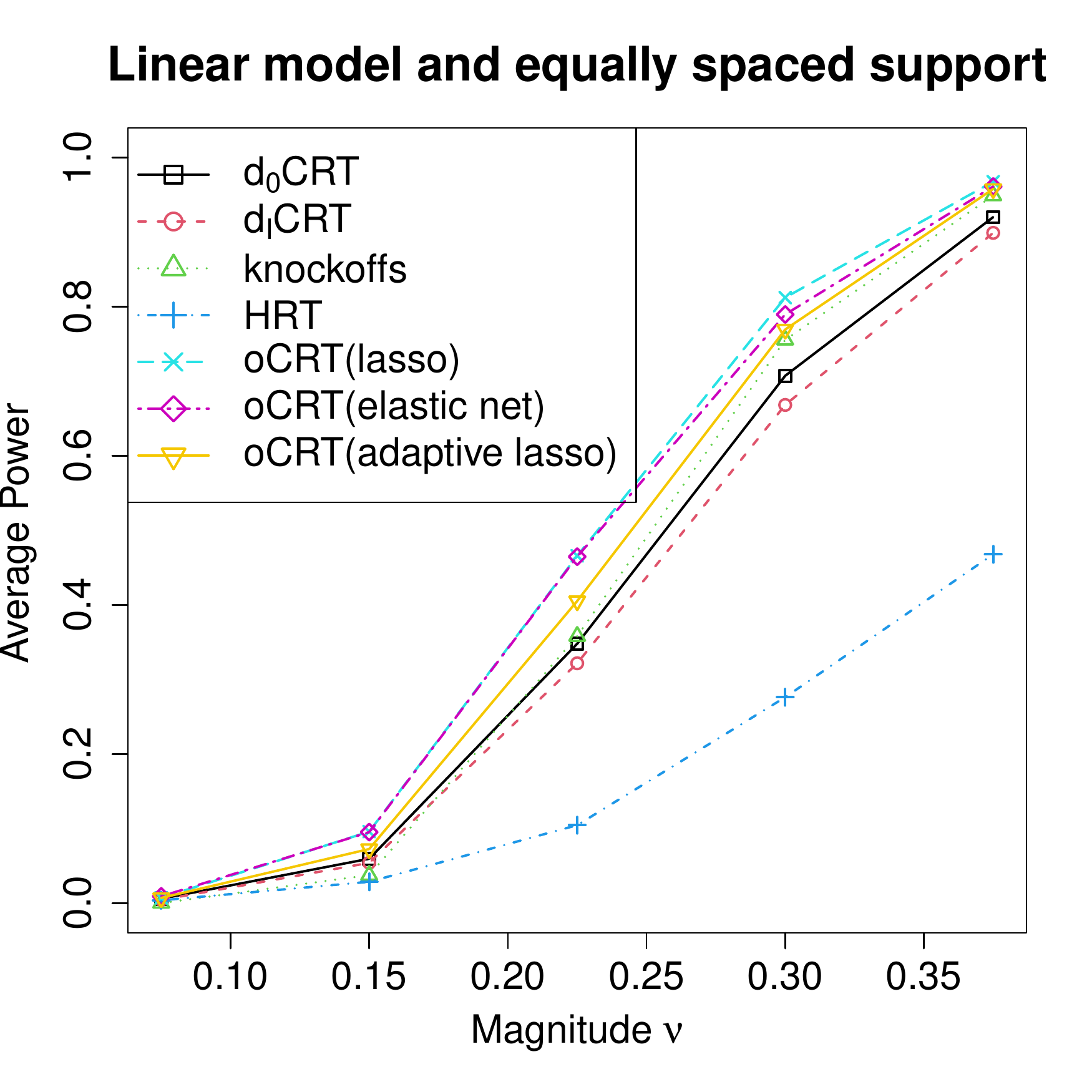}
  \includegraphics[width=0.4\textwidth]{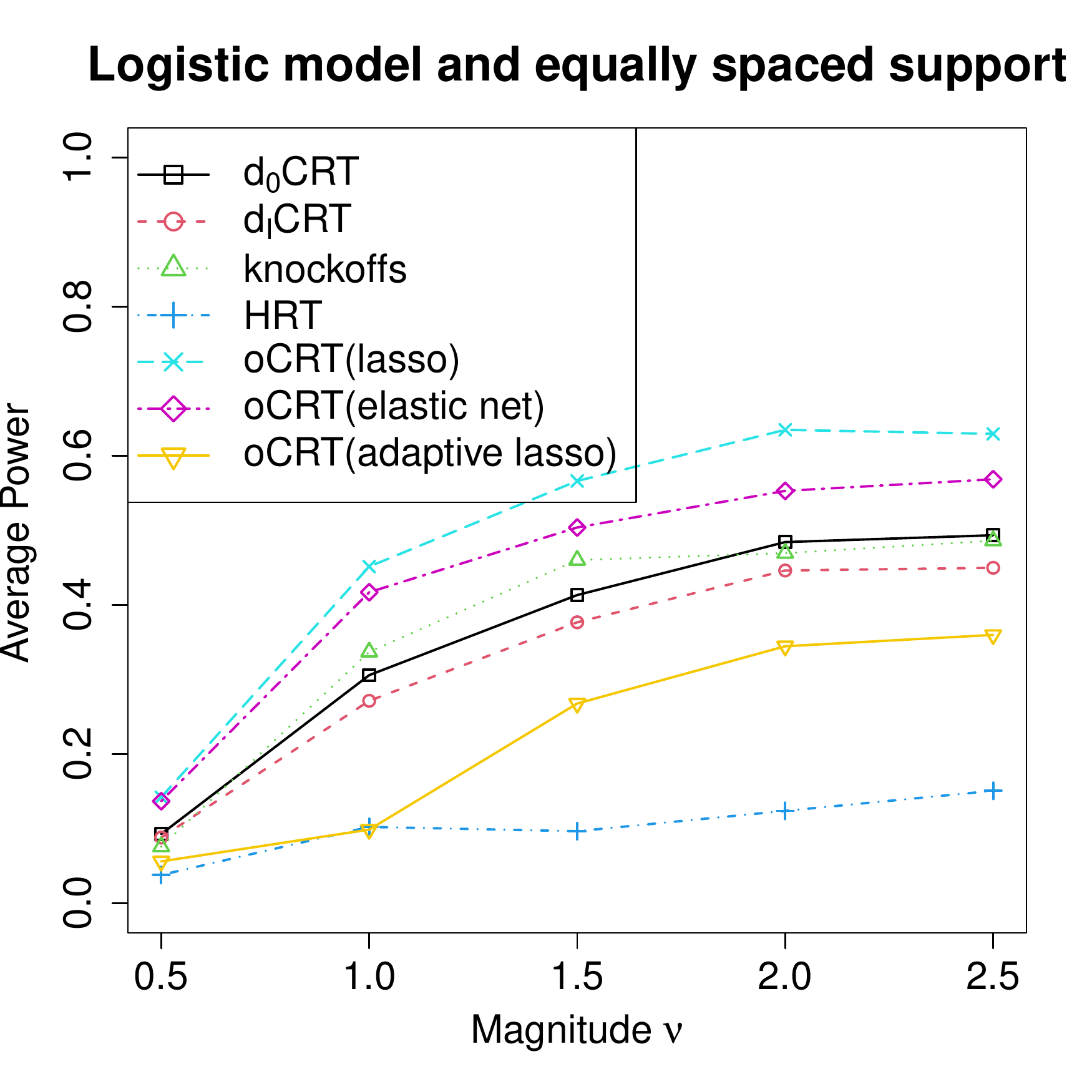}
\caption{\label{fig:smc} Average powers of false discovery rate control of the $n=p=300$ simulation of Appendix~\ref{sec:sim:smc}. All standard errors are below $0.01$ and are hence excluded. The dCRT approaches are the most powerful for adjacent support but slightly worse than \rev{o}CRT methods and knockoffs for equally-spaced support.}
\end{figure}

\begin{figure}[htpb]
\centering
  \includegraphics[width=0.4\textwidth]{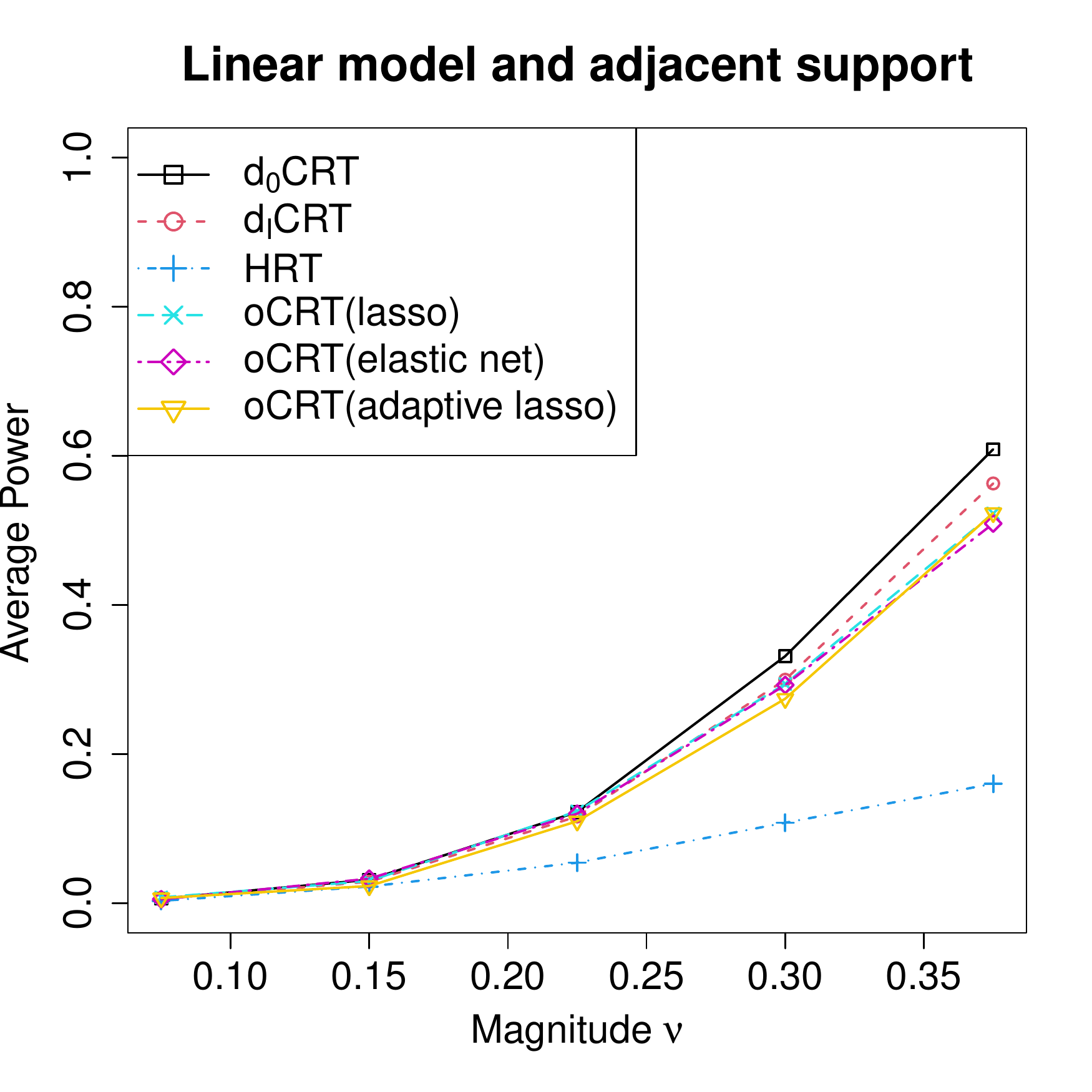}
  \includegraphics[width=0.4\textwidth]{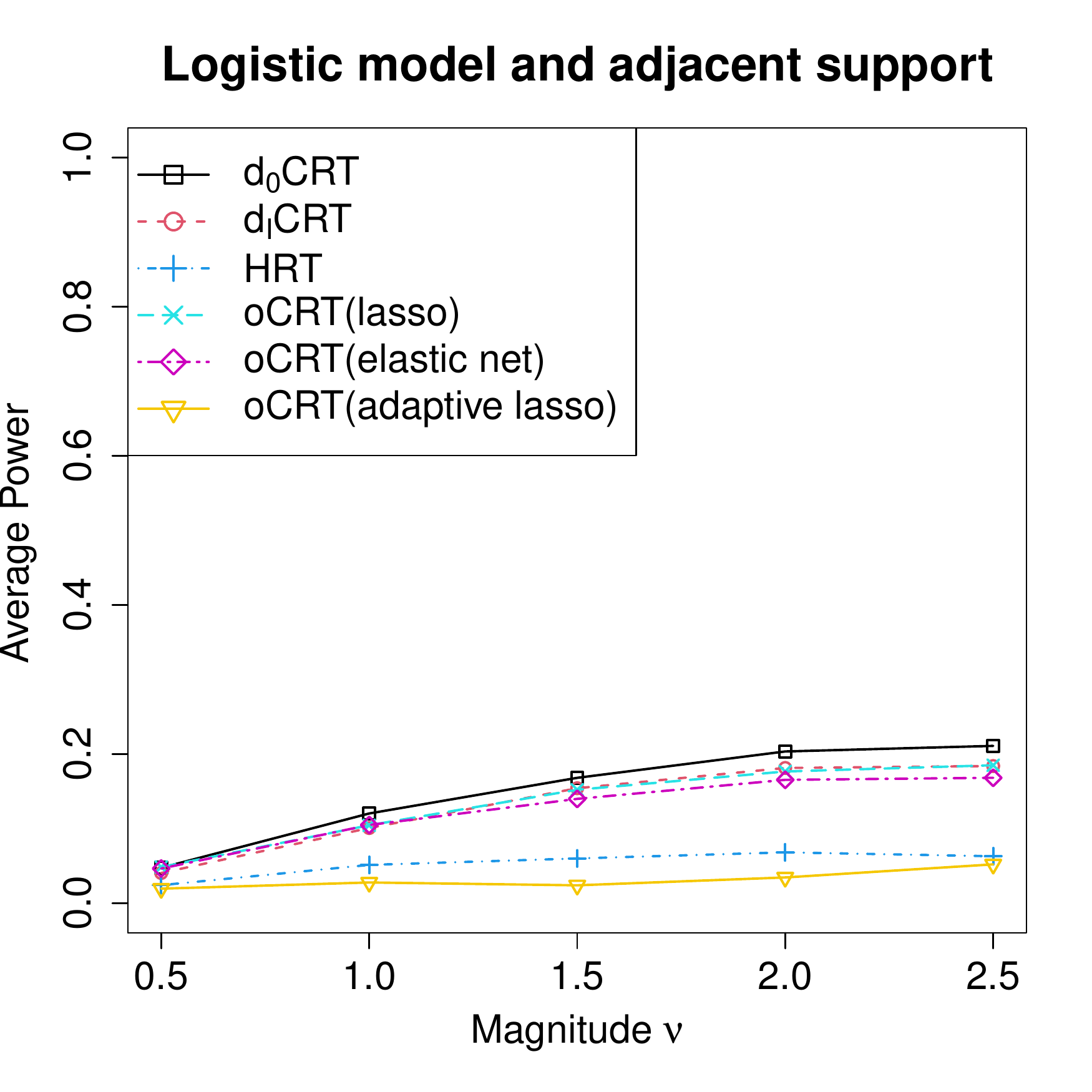}
  \includegraphics[width=0.4\textwidth]{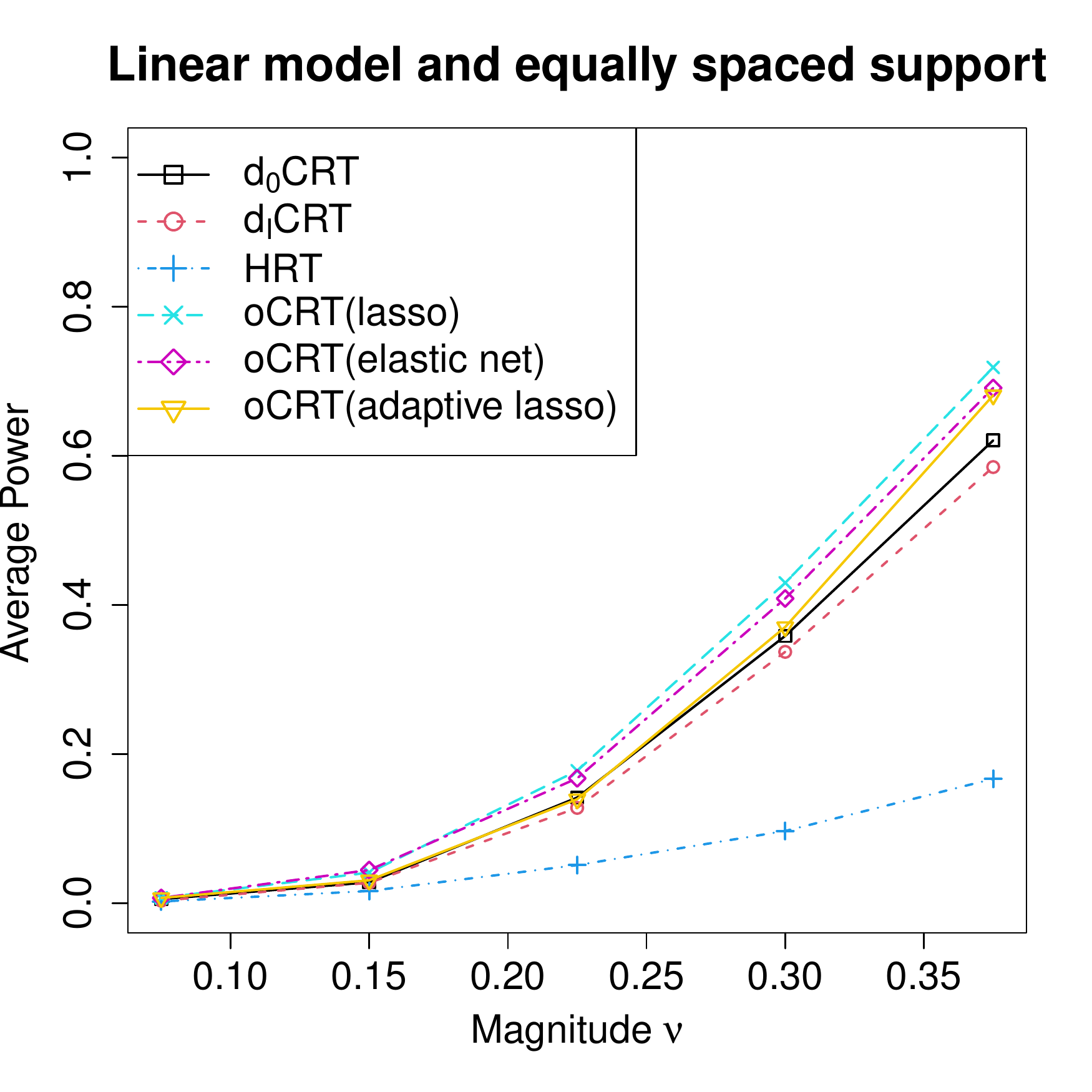}
  \includegraphics[width=0.4\textwidth]{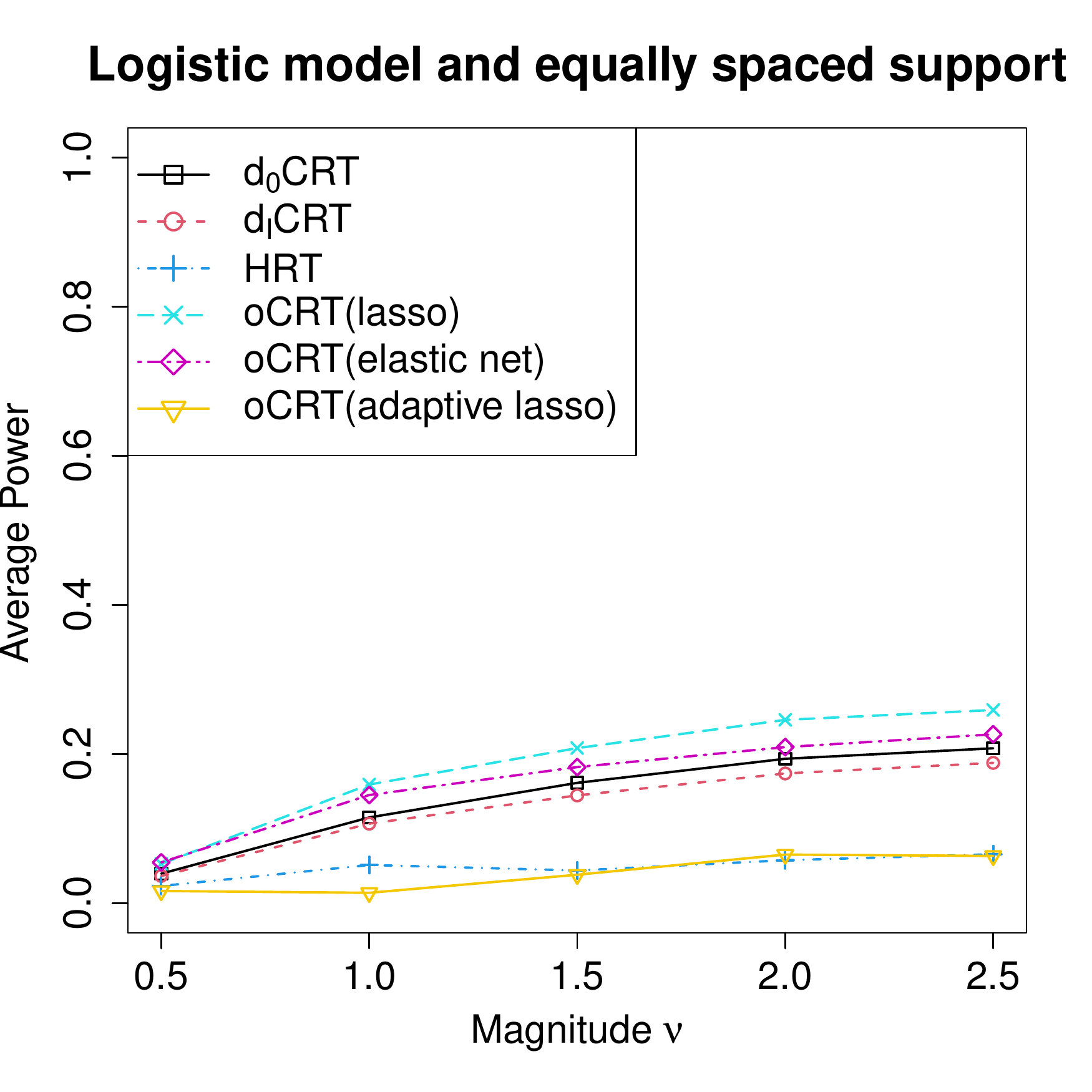}
\caption{\label{fig:smc:fwer} Average powers of family-wise error rate control of the $n=p=300$ simulation of Appendix~\ref{sec:sim:smc}. All standard errors are below $0.01$ and are hence excluded. The dCRT approaches are the most powerful for adjacent support but slightly worse than \rev{o}CRT methods for equally-spaced support.}
\end{figure}

To study and compare the computational efficiency of the methods, we present in Table~\ref{tab:com} the average computation time of the methods, with all algorithms implemented in \textsf{R}. Compared with the \rev{o}CRT procedures, the dCRTs drastically reduce the computation time and are thus much more user-friendly. Knockoffs and the HRT use less time than dCRT because they only fit a high dimensional regression once.
\begin{table}[htb!]
\centering
\begin{tabular}{l|c|c|c|c|c|c|c}
\multicolumn{8}{c}{{\bf Average computation times (minutes)}} \\
\hline
       & d$_0$CRT & d$_\mathrm{I}$CRT & knockoff    & HRT     & \rev{o}CRT (L) & \rev{o}CRT (EN) & \rev{o}CRT (AL) \\ \hline
Linear   & $0.6$  & $0.6$  & $0.2$ & $0.3$ & $355.9$     & $425.5$     & $378.9$       \\ \hline
Logistic & $1.7$  & $1.8$  & $0.1$ & $0.5$ & $309.0$    & $461.0$    & $391.8$      \\ \hline
\end{tabular}
\caption{\label{tab:com} Average computation times (in minutes) of the $n=p=300$ simulations of Appendix~\ref{sec:sim:smc}. dCRT is much more efficient than the \rev{o}CRT, and slightly more expensive than knockoffs and HRT. We used abbreviations for Lasso (L), Elastic Net (EN) and Adaptive Lasso (AL).}
\end{table}

\subsection{Large size data simulation}\label{sec:sim:hrt}
In this section, we conduct simulation studies of a scale beyond the \rev{o}CRT's computational feasibility, and hence focus on the remaining methods whose computation stays manageable. As a baseline, we set $n=p=800$, again use ${\rm AR}(1)$ covariates with autocorrelation $0.5$, generate $Y$ from Gaussian linear model with unit residual variance, and use a coefficient vector with $s=50$ nonzero entries of equal magnitude $\nu=0.175$ (chosen to make the power around 0.5) and random signs (each independently having equal probability of being positive or negative). Again, the adjacent and equally-spaced supports are studied separately and we pursue controlled variable selection with nominal false discovery rate or family-wise error rate $\alpha=0.1$.

Each of the four average power plots in Figure~\ref{fig:diffnp} varies one the parameters ($\nu$, $n$, $p$, or $s$) from this baseline simulation setup, with the ranges given by the x-axes. The two dCRTs have similar performance, both of them outperform the HRT in all cases for both false discovery rate and family-wise error rate control. They perform better than knockoffs for adjacent support but worse than knockoffs for equally-spaced support under most cases. When the sparsity level $s$ is below $10$, the power of knockoffs drops to $0$ because of the effect mentioned in Section~\ref{sec:sim}.

\begin{figure}[htpb!]
\centering
  \includegraphics[width=0.4\textwidth]{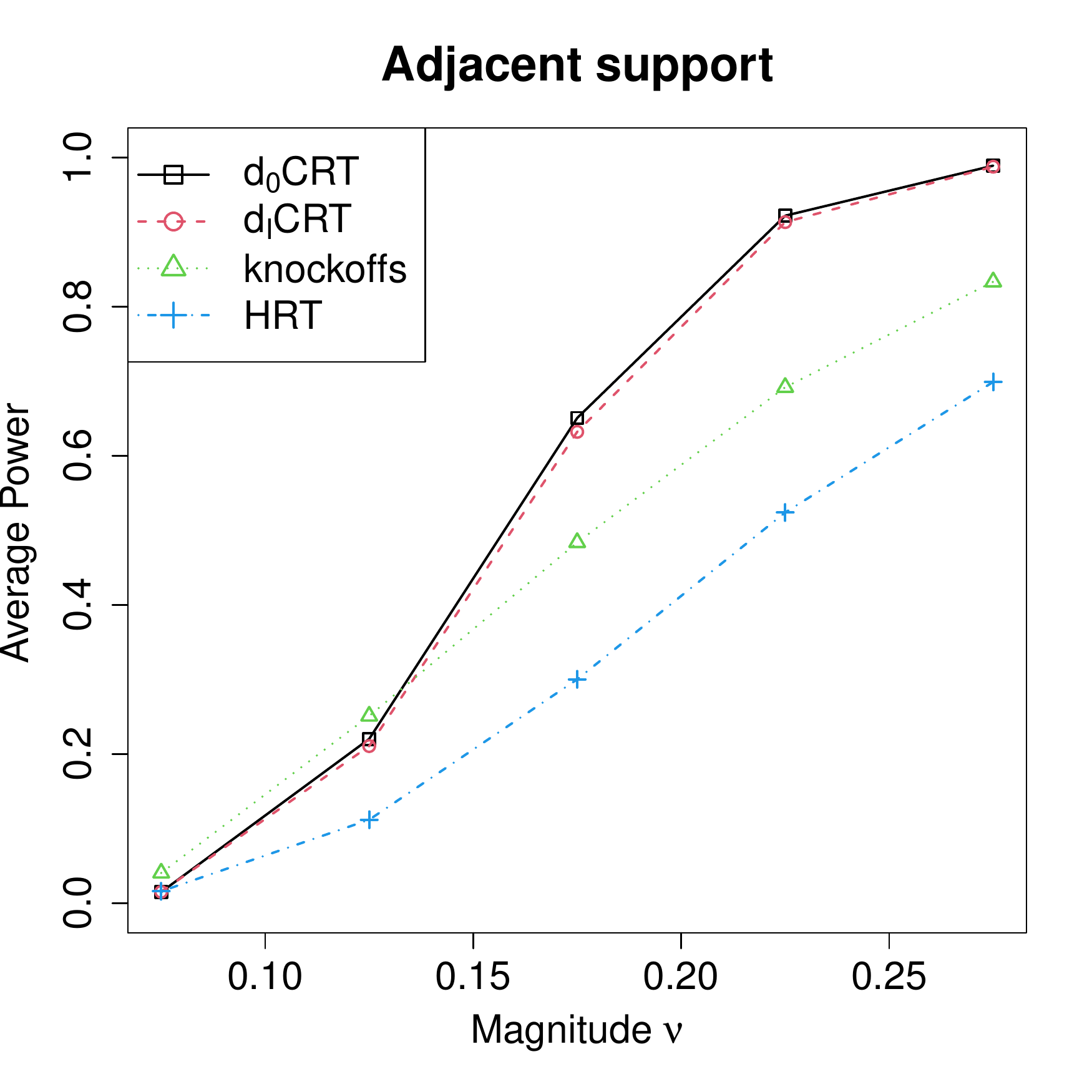}
  \includegraphics[width=0.4\textwidth]{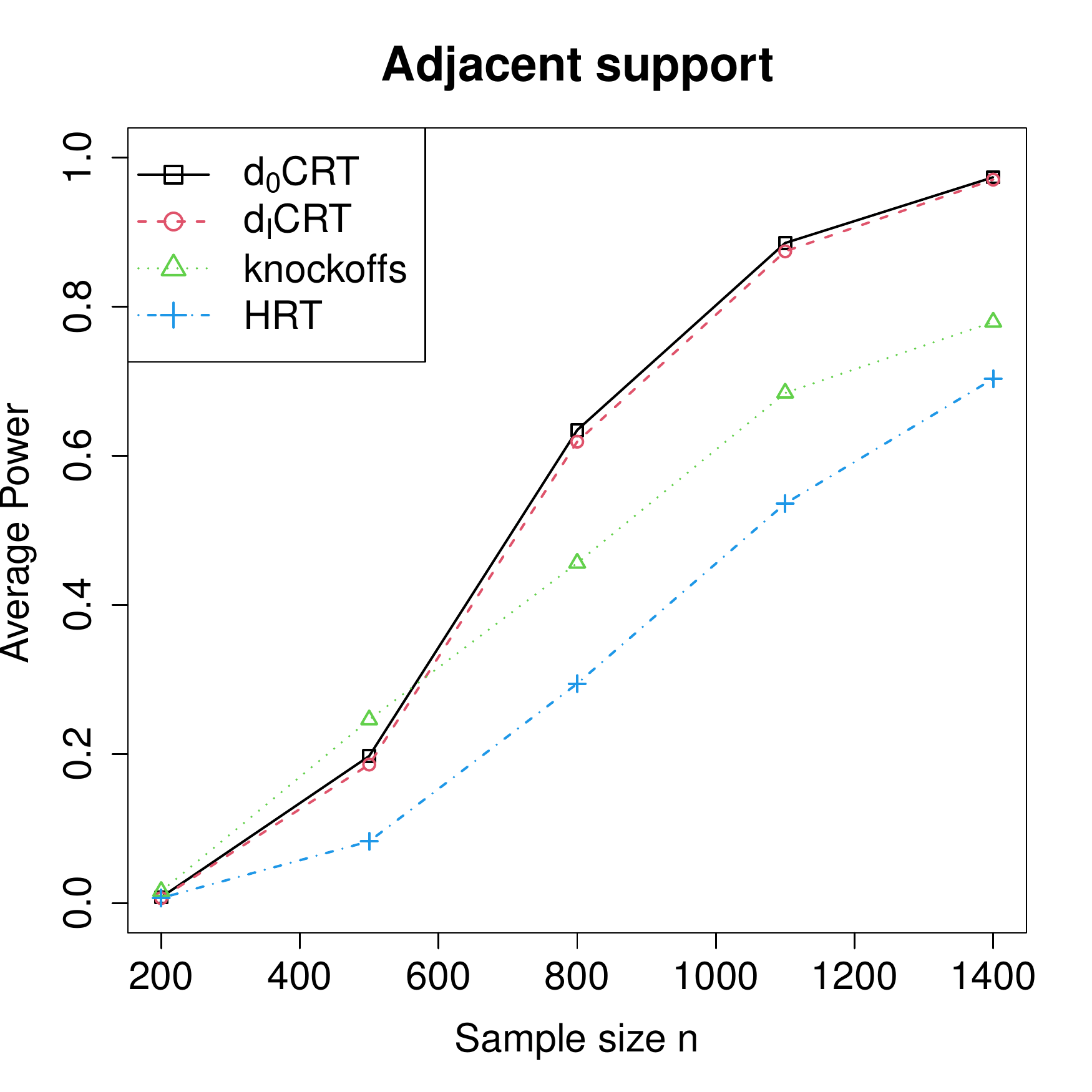}
  \includegraphics[width=0.4\textwidth]{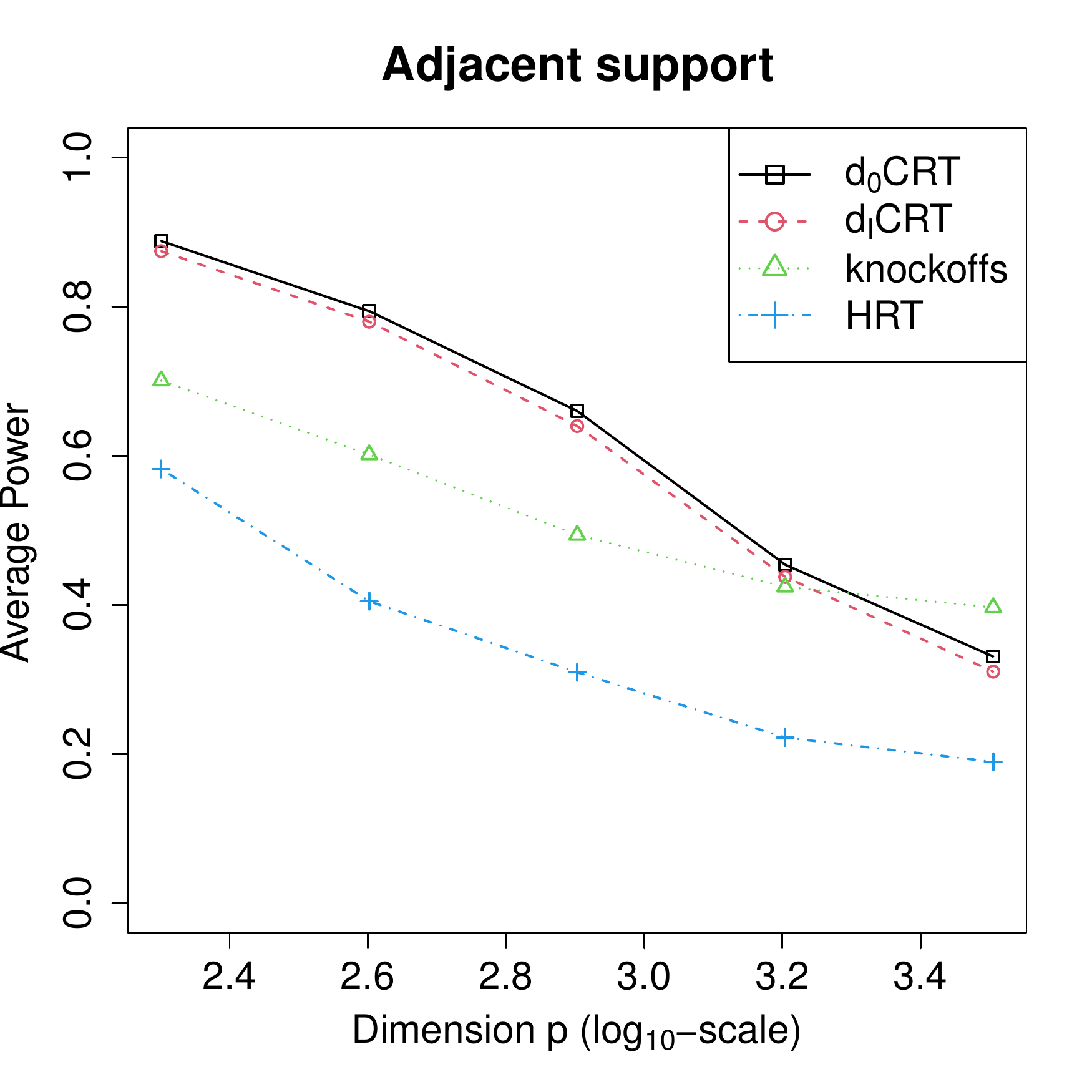}
  \includegraphics[width=0.4\textwidth]{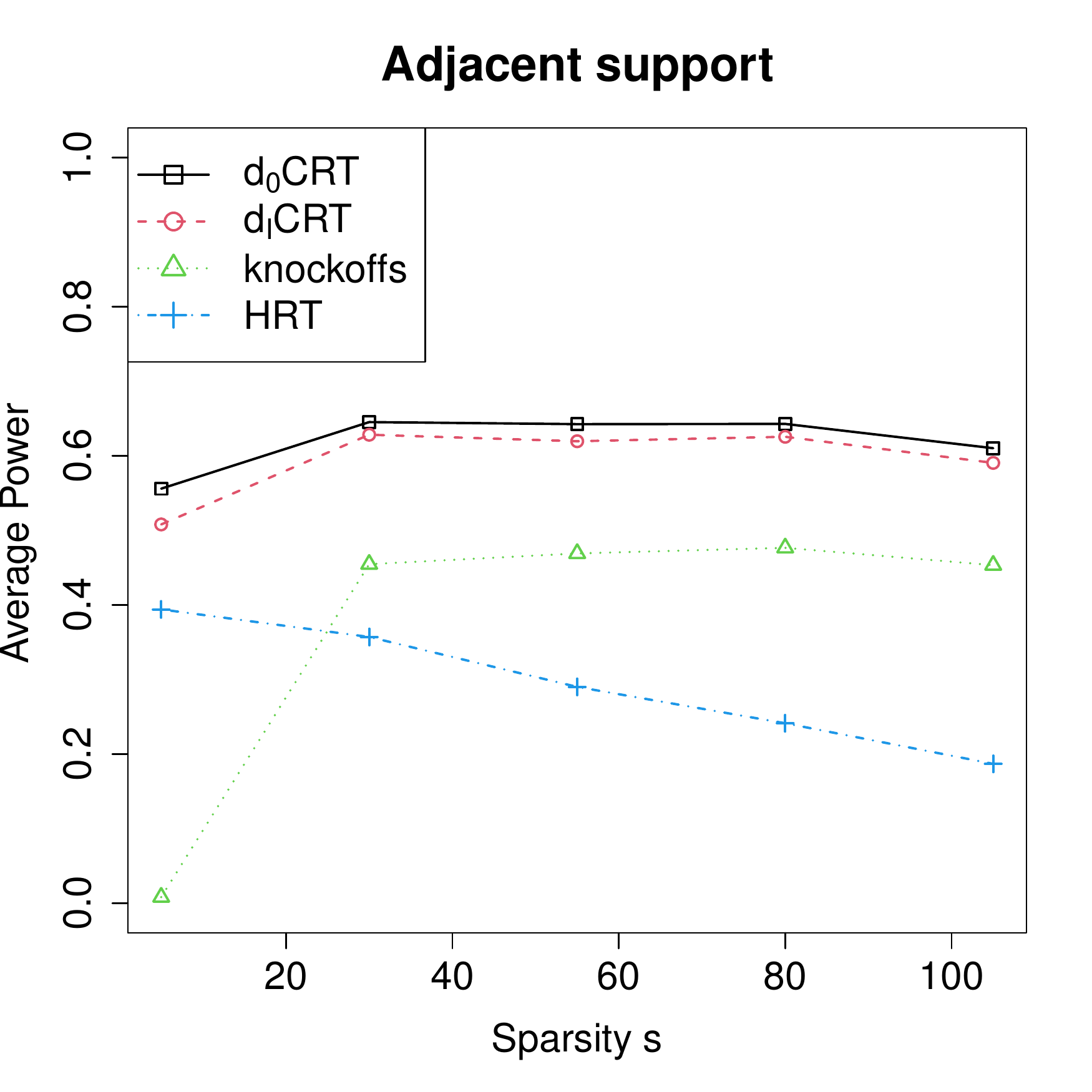}

\caption{\label{fig:diffnp} Average powers of false discovery rate control of the large scale simulations of Appendix~\ref{sec:sim:hrt} that vary the coefficient magnitude, sample size, dimension, and coefficient sparsity with adjacent support. All standard errors are below $0.01$. Our dCRT approaches are typically the most powerful.}
\end{figure}

\begin{figure}[htpb!]
\centering
  \includegraphics[width=0.4\textwidth]{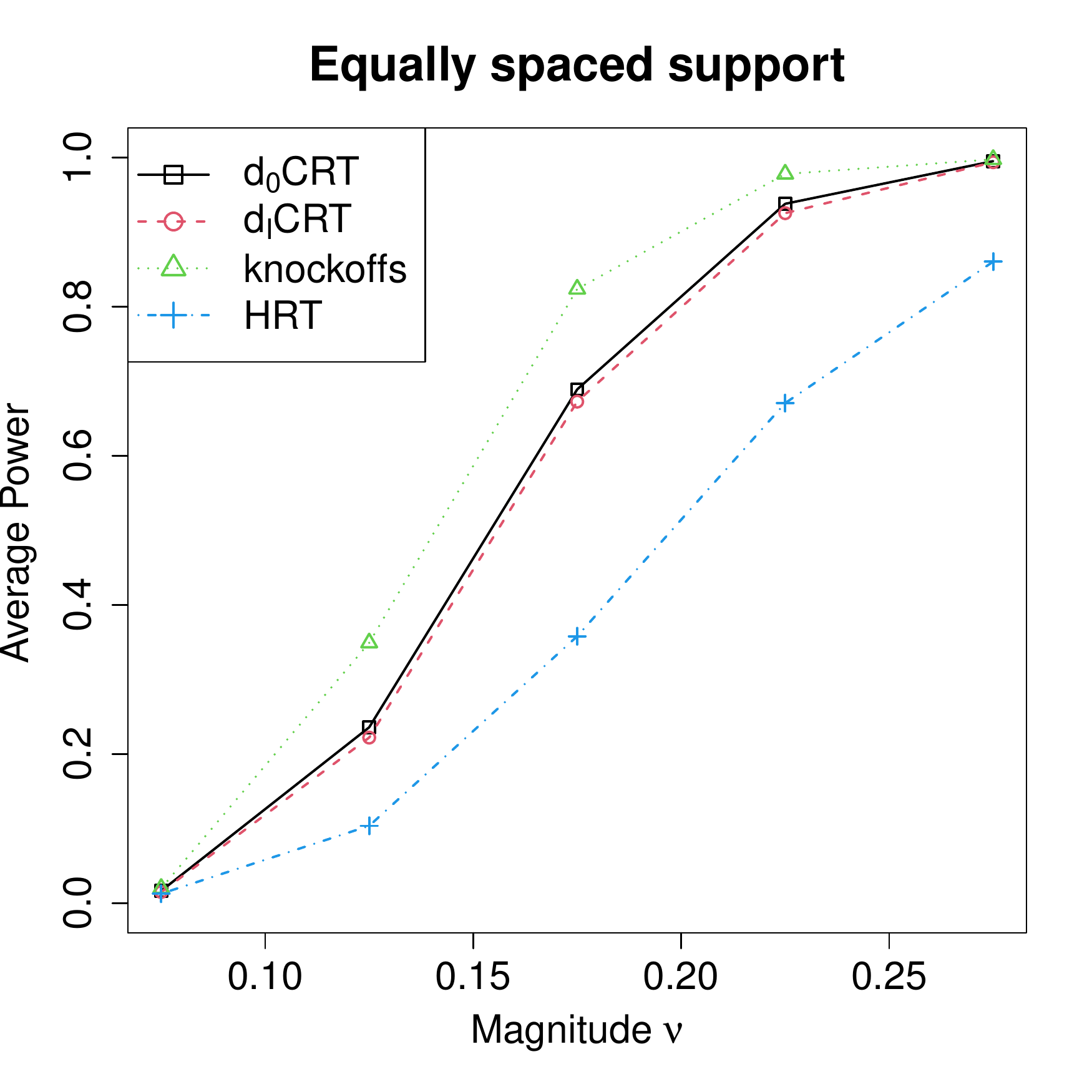}
  \includegraphics[width=0.4\textwidth]{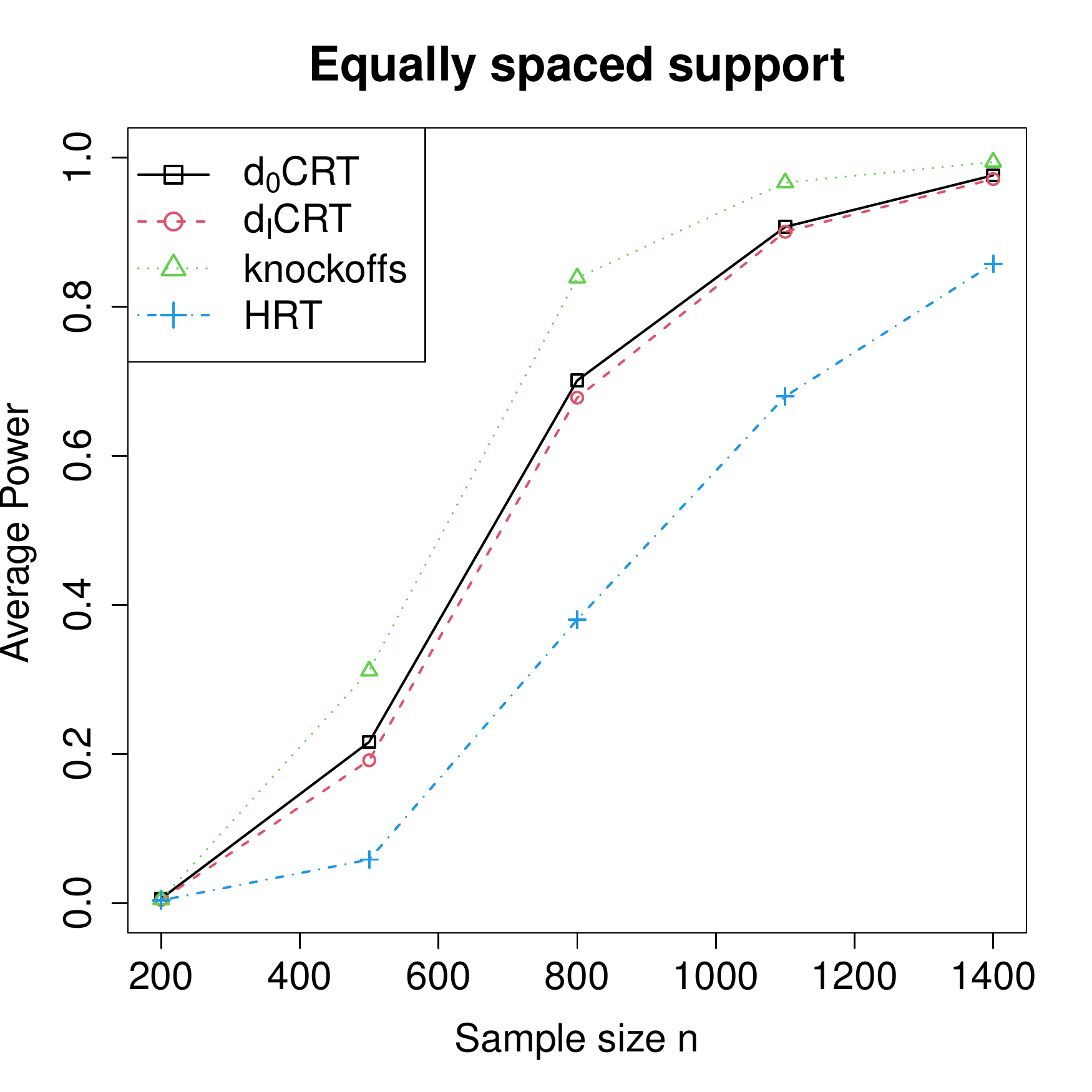}
  \includegraphics[width=0.4\textwidth]{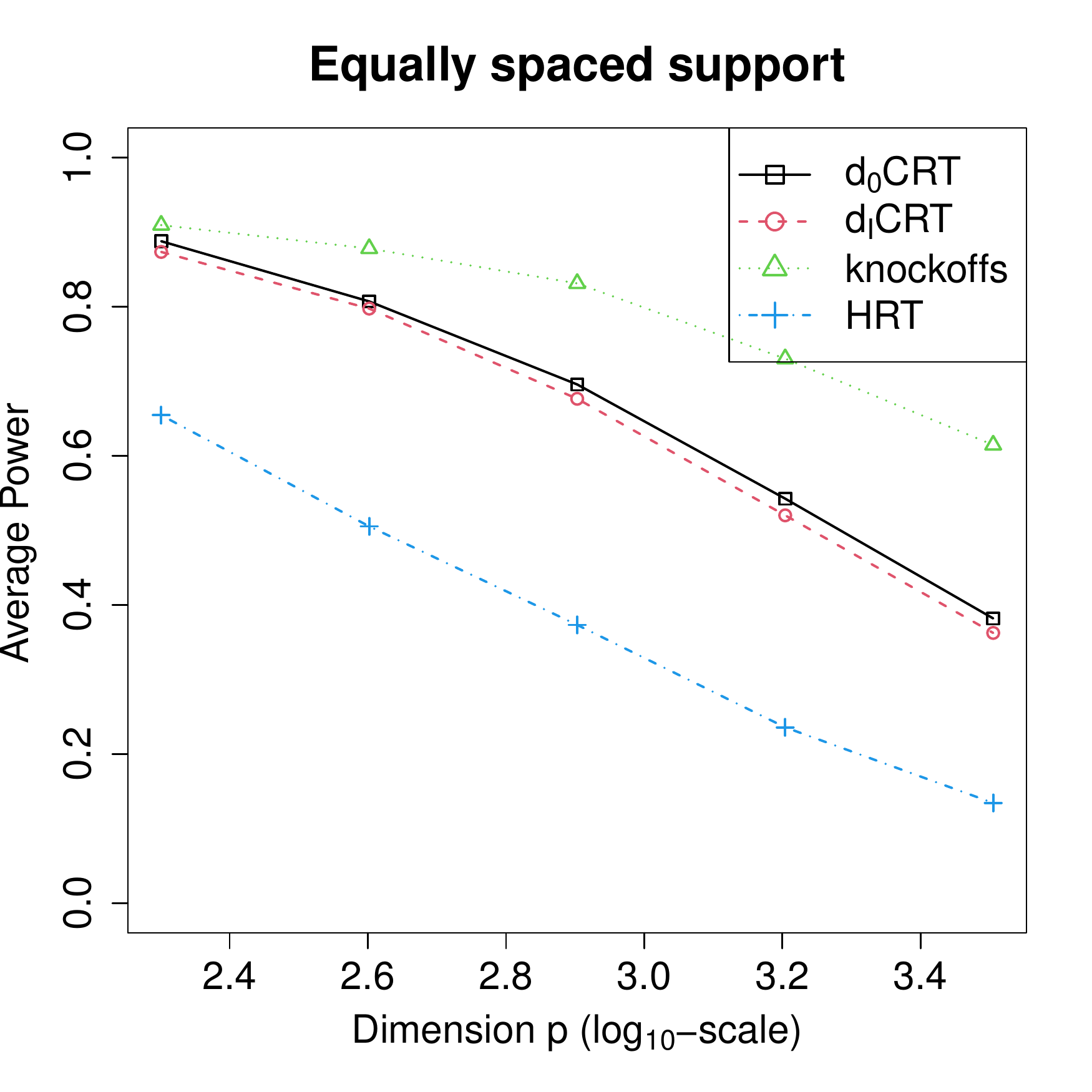}
  \includegraphics[width=0.4\textwidth]{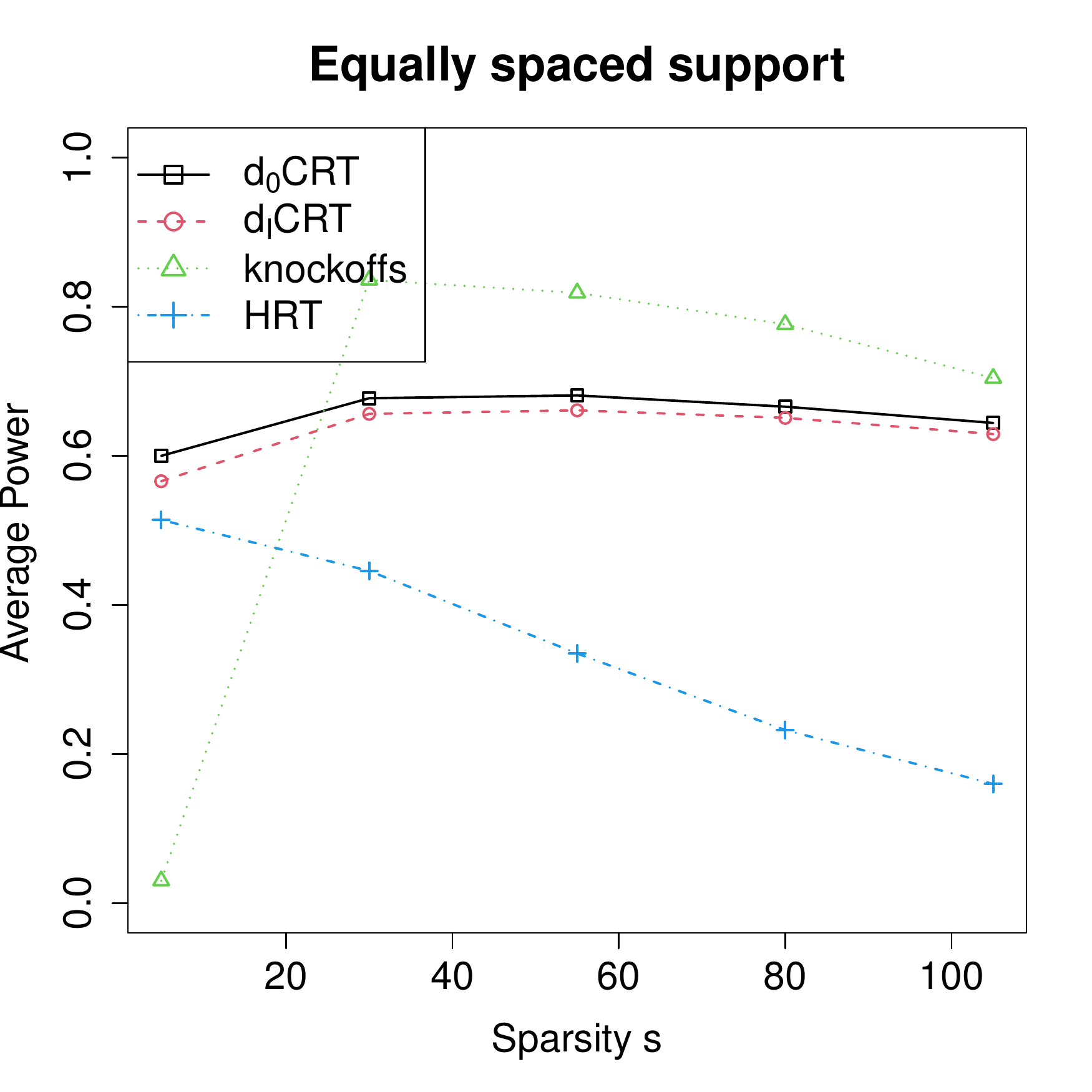}

\caption{\label{fig:diffnp:equal} Average powers of false discovery rate control of the large scale simulations of Appendix~\ref{sec:sim:hrt} that vary the coefficient magnitude, sample size, dimension, and coefficient sparsity with equally spaced support. All standard errors are below $0.01$. The dCRTs have lower power than knockoffs in most cases.}
\end{figure}

\begin{figure}[htpb!]
\centering
  \includegraphics[width=0.4\textwidth]{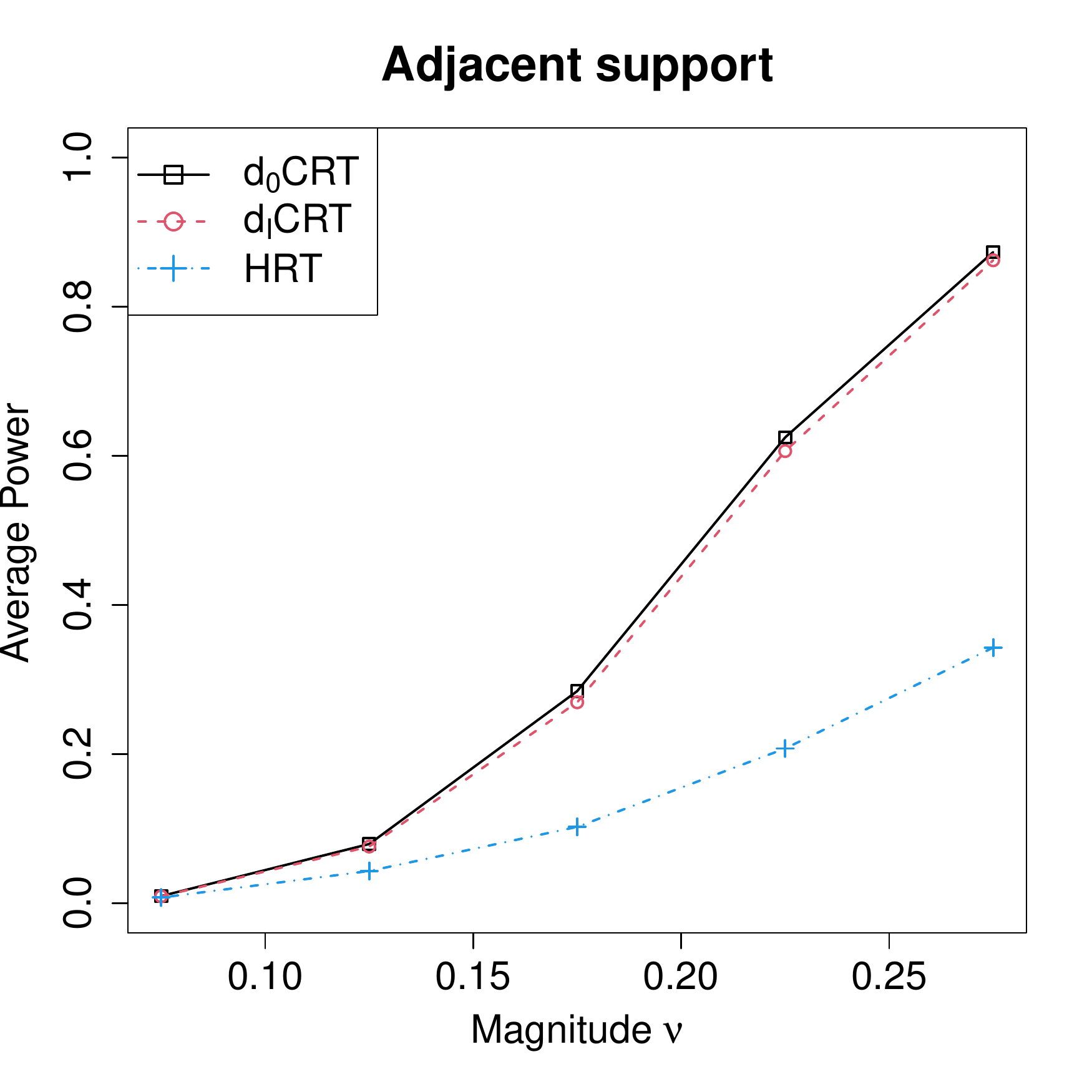}
  \includegraphics[width=0.4\textwidth]{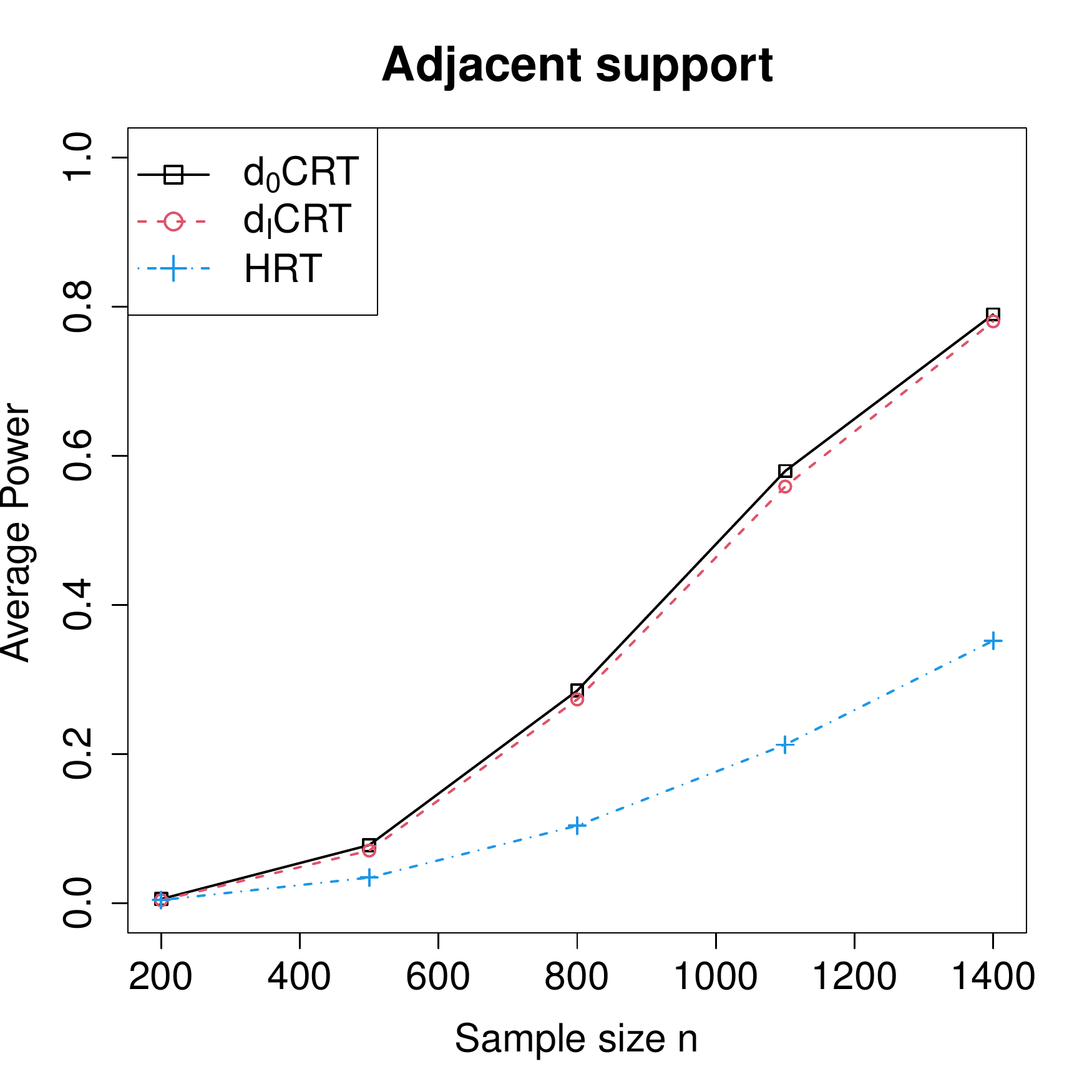}
  \includegraphics[width=0.4\textwidth]{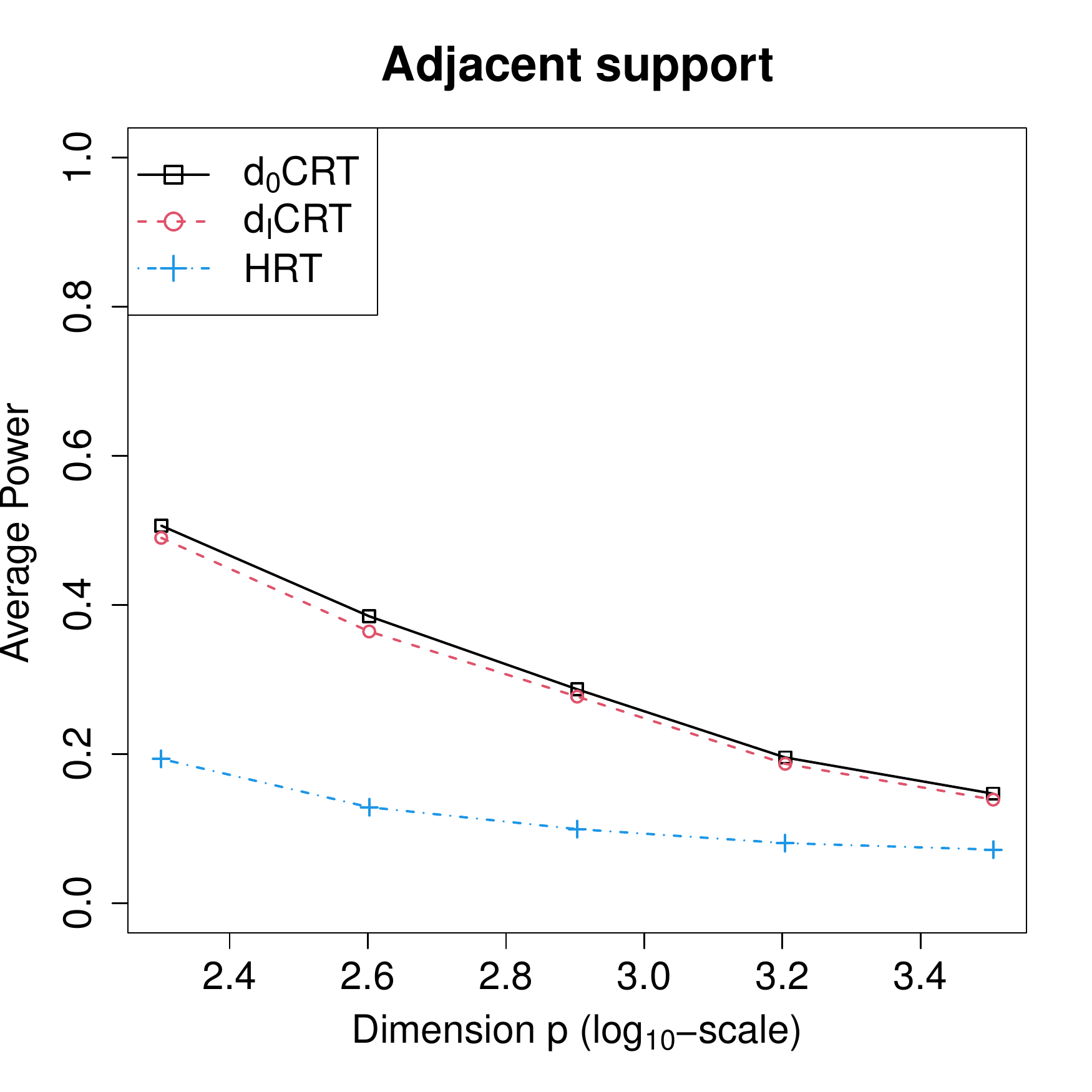}
  \includegraphics[width=0.4\textwidth]{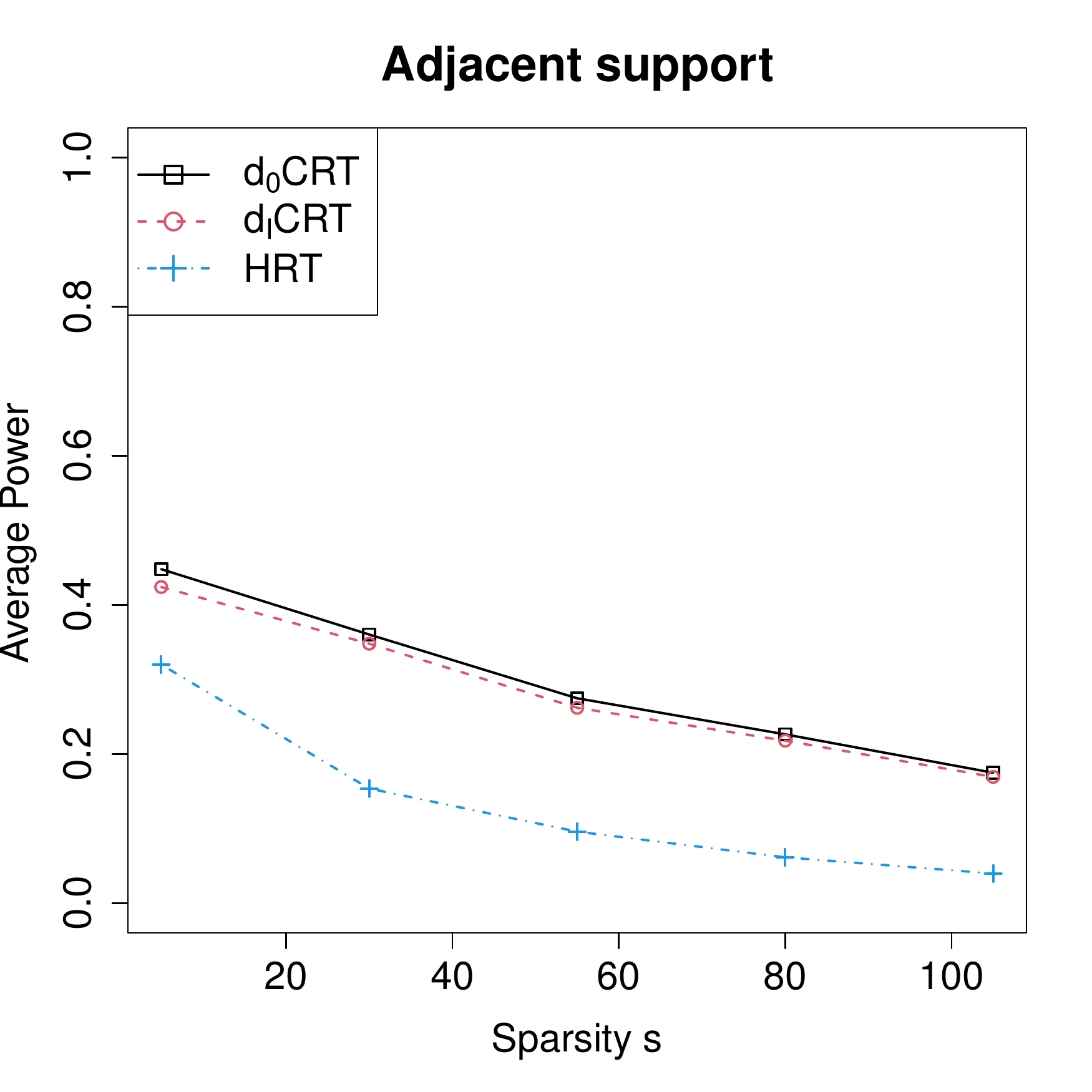}

\caption{\label{fig:diffnp:fwer} Average powers of family-wise error rate control of the large scale simulations of Appendix~\ref{sec:sim:hrt} that vary the coefficient magnitude, sample size, dimension, and coefficient sparsity with adjacent support. All standard errors are below $0.01$. The dCRTs are more powerful than the HRT.}
\end{figure}

\begin{figure}[htpb!]
\centering
  \includegraphics[width=0.4\textwidth]{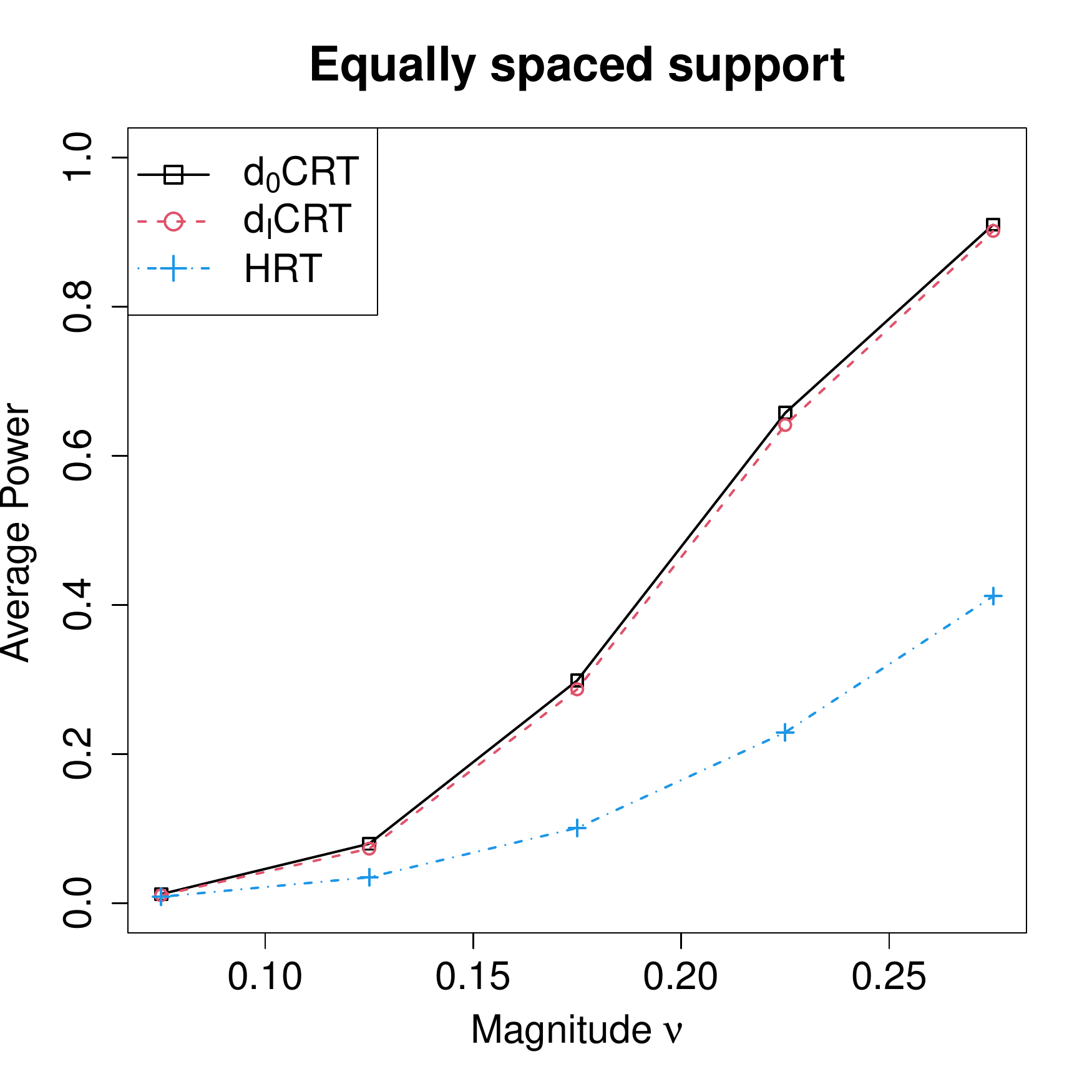}
  \includegraphics[width=0.4\textwidth]{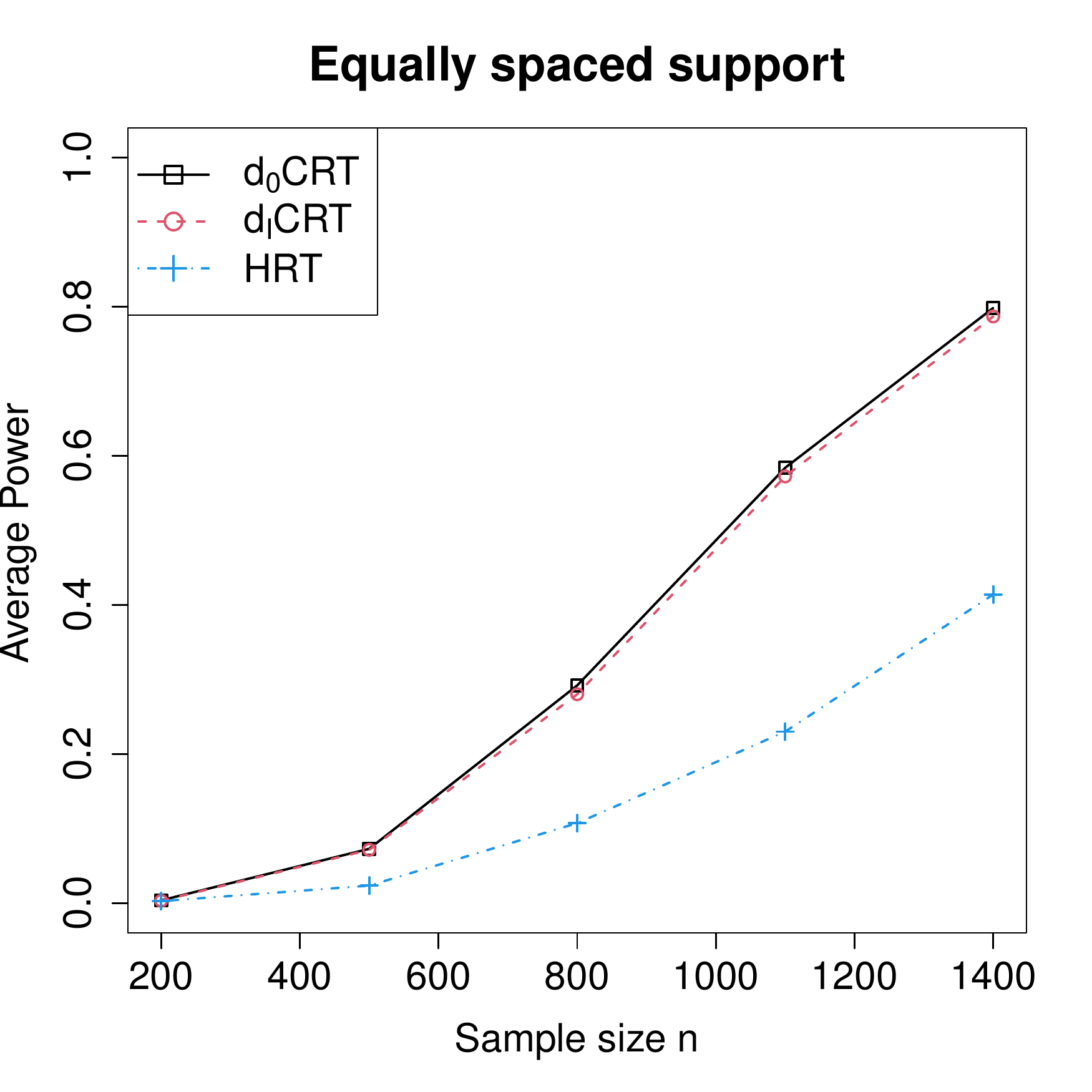}
  \includegraphics[width=0.4\textwidth]{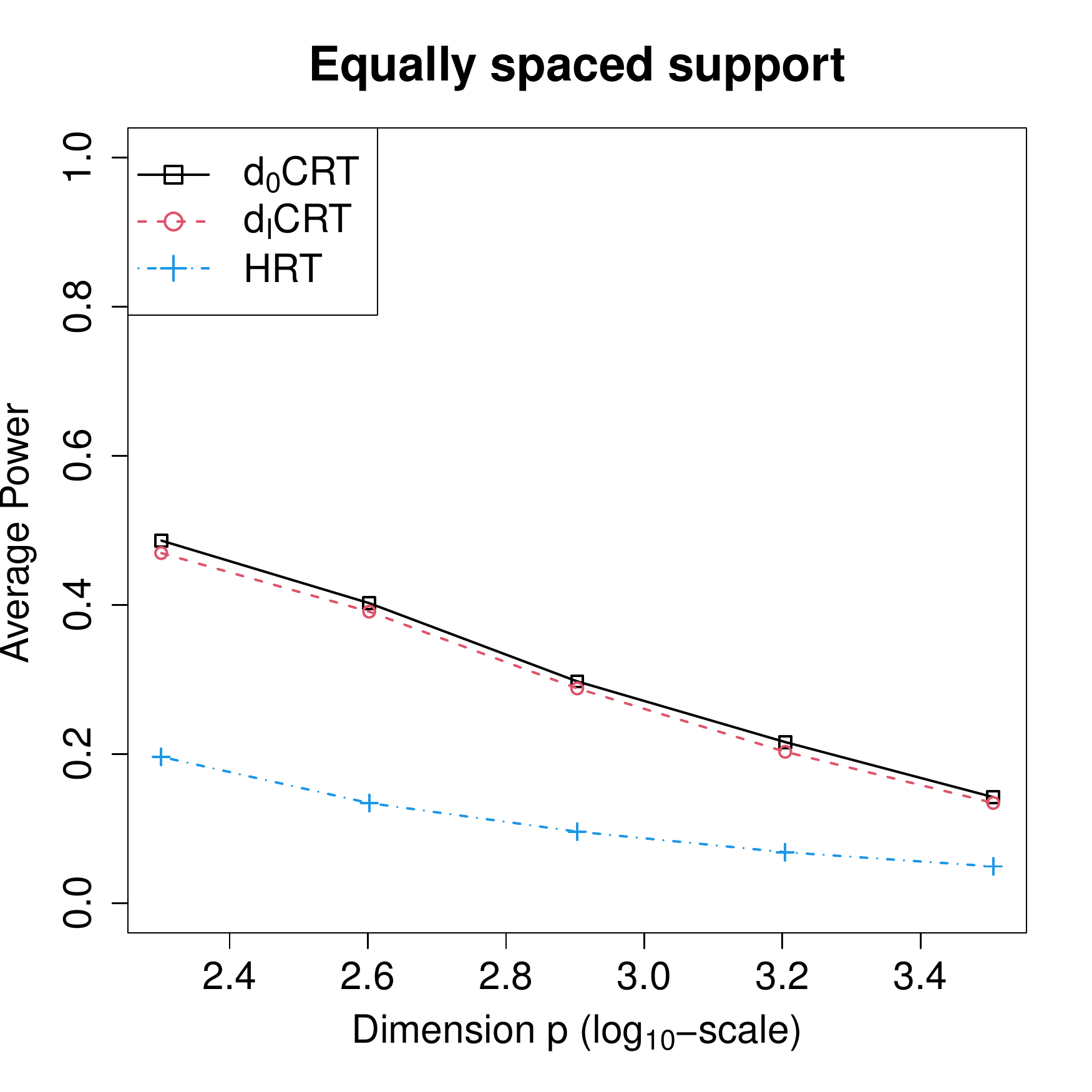}
  \includegraphics[width=0.4\textwidth]{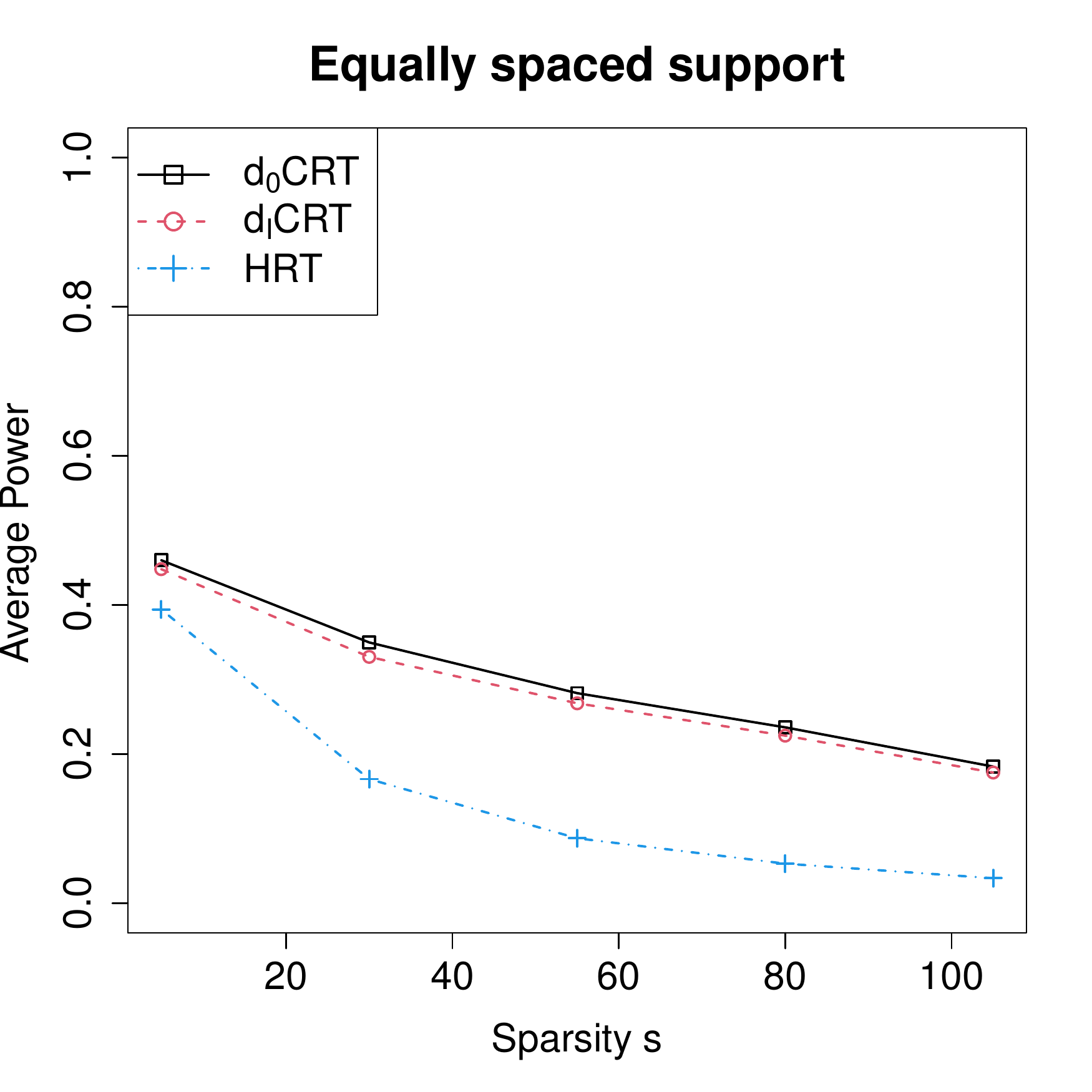}

\caption{\label{fig:diffnp:equal:fwer} Average powers of family-wise error rate control of the large scale simulations of Appendix~\ref{sec:sim:hrt} that vary the coefficient magnitude, sample size, dimension, and coefficient sparsity with equally spaced support. All standard errors are below $0.01$. The dCRTs are more powerful than the HRT.}
\end{figure}

We present the average computation times when $n=800$, $p=800$ and $s=50$ in Table~\ref{tab:com:large}. 
Knockoffs and HRT still run faster than the dCRT methods since they only fit high dimensional model once in the whole procedure. 

\begin{table}[htb!]
\centering
\begin{tabular}{c|c|c|c}
\multicolumn{4}{c}{{\bf Average computation times (minutes)}} \\
\hline
      \mbox{d$_0$CRT} &  d$_{\mathrm{I}}$CRT & knockoffs  & HRT  \\ \hline
     $9.8$  & $10.2$ & $1.8$ & $6.2$ \\ \hline
\end{tabular}
\caption{\label{tab:com:large} Average computation times (in minutes) of the $n=p=800$ simulations of Appendix~\ref{sec:sim:hrt}. As expected, knockoffs is the fastest, while dCRT is slightly slower than the HRT.}
\end{table}

We now study varying the covariate and response models from this same baseline simulation. 
First we generate Gaussian covariates with covariance structure set as ${\rm AR}(1)$ with autocorrelations 0.25, 0.5, and 0.75, and equicorrelated with correlations 0.15, 0.3, and 0.45. Second we generate $Y$ from three additional models given by:
\begin{enumerate}[(i)]
\item {\bf Poisson model:} $Y$ is generated from a Poisson generalized linear model with the same coefficient vector as the baseline model. 
\item {\bf Logistic model:} $Y$ is generated from a logistic regression with the same coefficient vector as the baseline model except $\nu=0.5$.

\item {\bf Polynomial model:} $Y$ is generated from a Gaussian model with conditional mean given by a polynomial that starts from the baseline model with $\nu=0.105$ and takes each covariate with a nonzero coefficient and adds a term equal to 0.3 times its cube.
\end{enumerate}
The signal magnitudes of each setting are chosen to make the powers of the main methods close to $0.5$, for convenience of comparison. Note that under (i) and (iii), we still fit linear lasso, though the model is wrong. The resulting average powers of both these simulations are plotted in Figure~\ref{fig:diff:design}. The d$_0$CRT and d$_{\mathrm{I}}$CRT have significantly higher power than HRT in all settings. The comparison to knockoffs is more complicated. Knockoffs generally performs worse under adjacent support and under spaced support with equicorrelated design, suggesting that knockoffs is sensitive to correlations among signal variables. On the other hand, dCRT performs worse under highly autocorrelated designs. These subtle differences are intriguing and we leave their further investigation to future work. Again, all the methods control false discovery rate properly with the desired level $0.1$ in the numerical studies corresponding to Figures~\ref{fig:diffnp} and~\ref{fig:diff:design}, and we present the false discovery rate plot in Appendix~\ref{sec:sim:fdr}.

\begin{figure}[htbp!]
\centering
\includegraphics[width=0.4\textwidth]{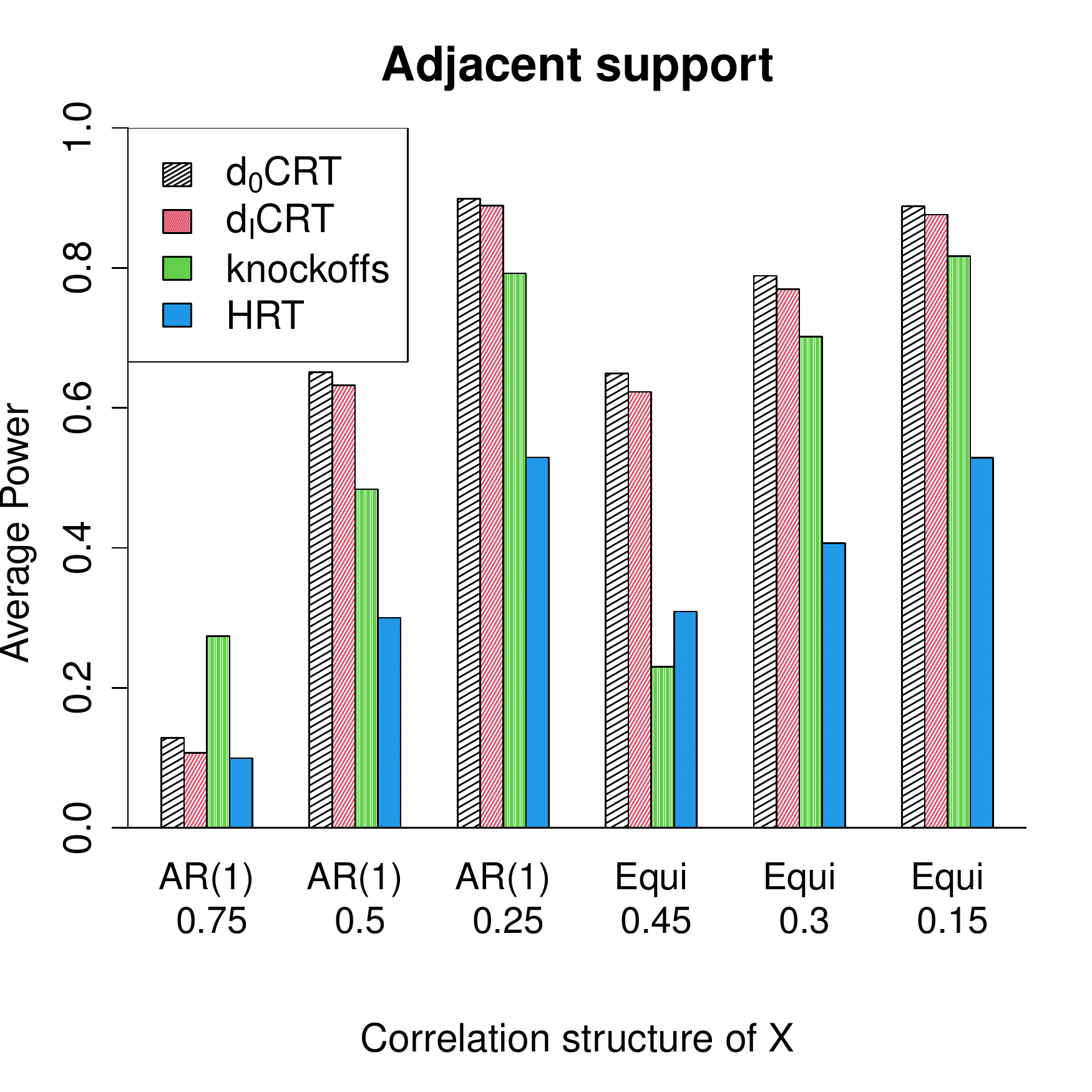}
 \includegraphics[width=0.4\textwidth]{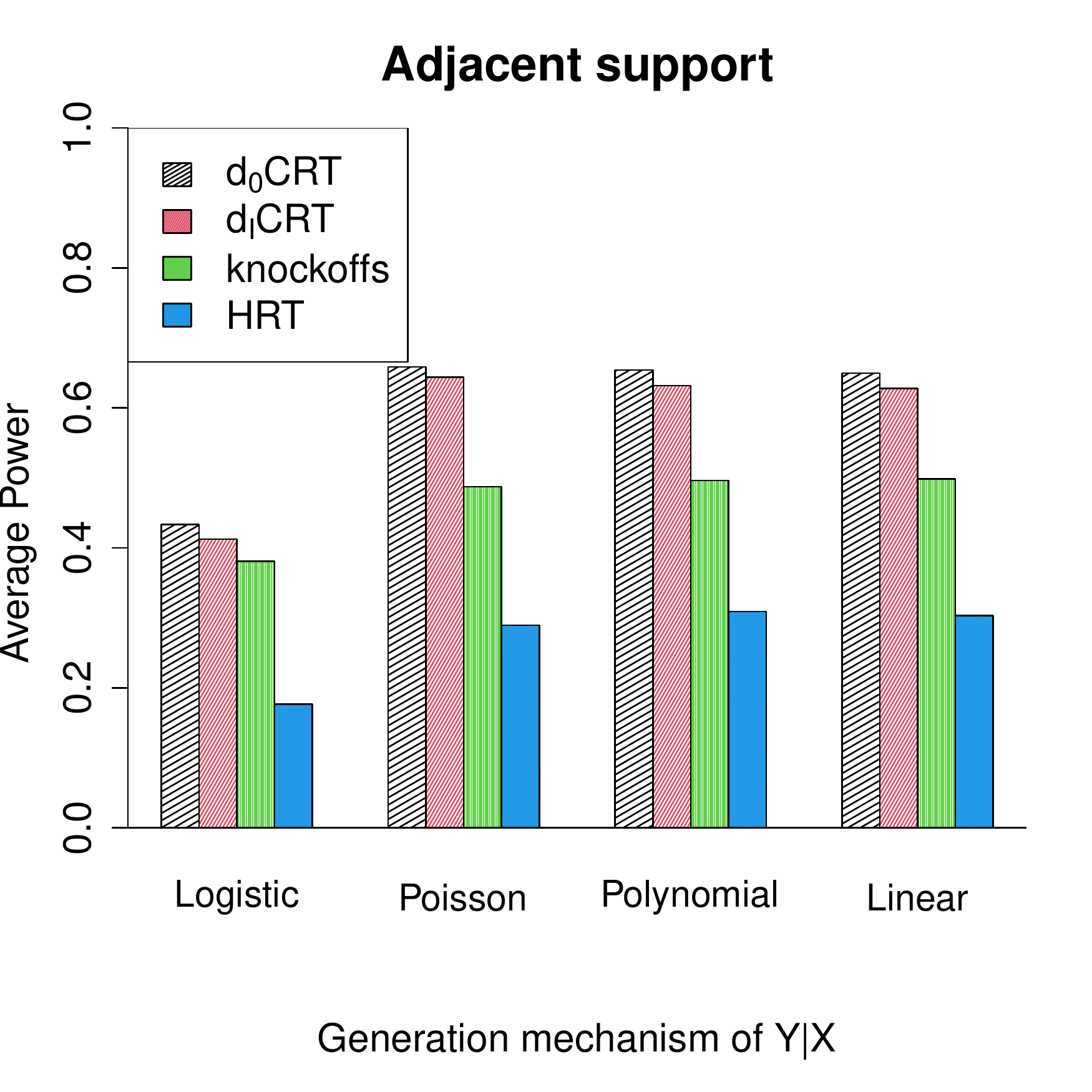}
\includegraphics[width=0.4\textwidth]{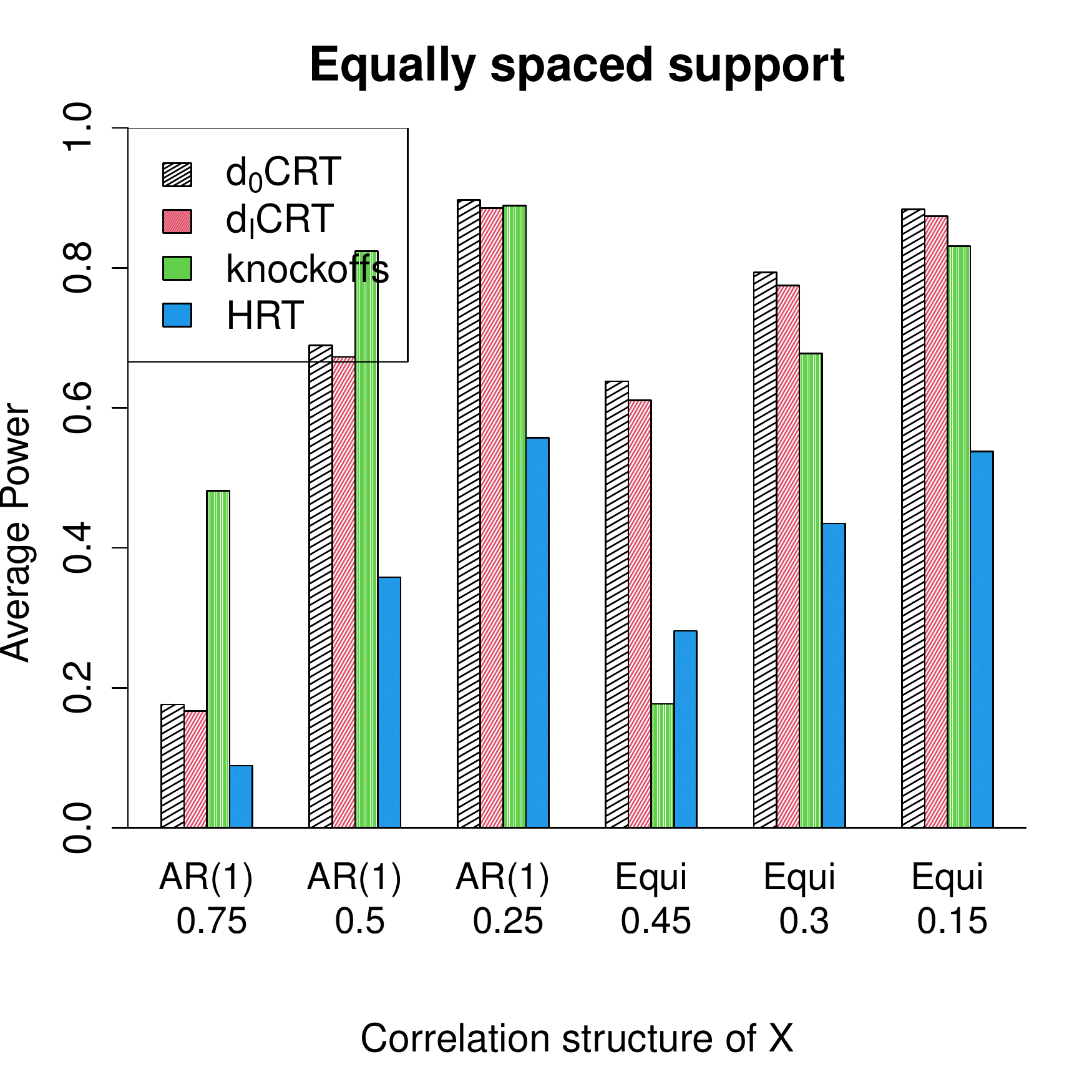}
  \includegraphics[width=0.4\textwidth]{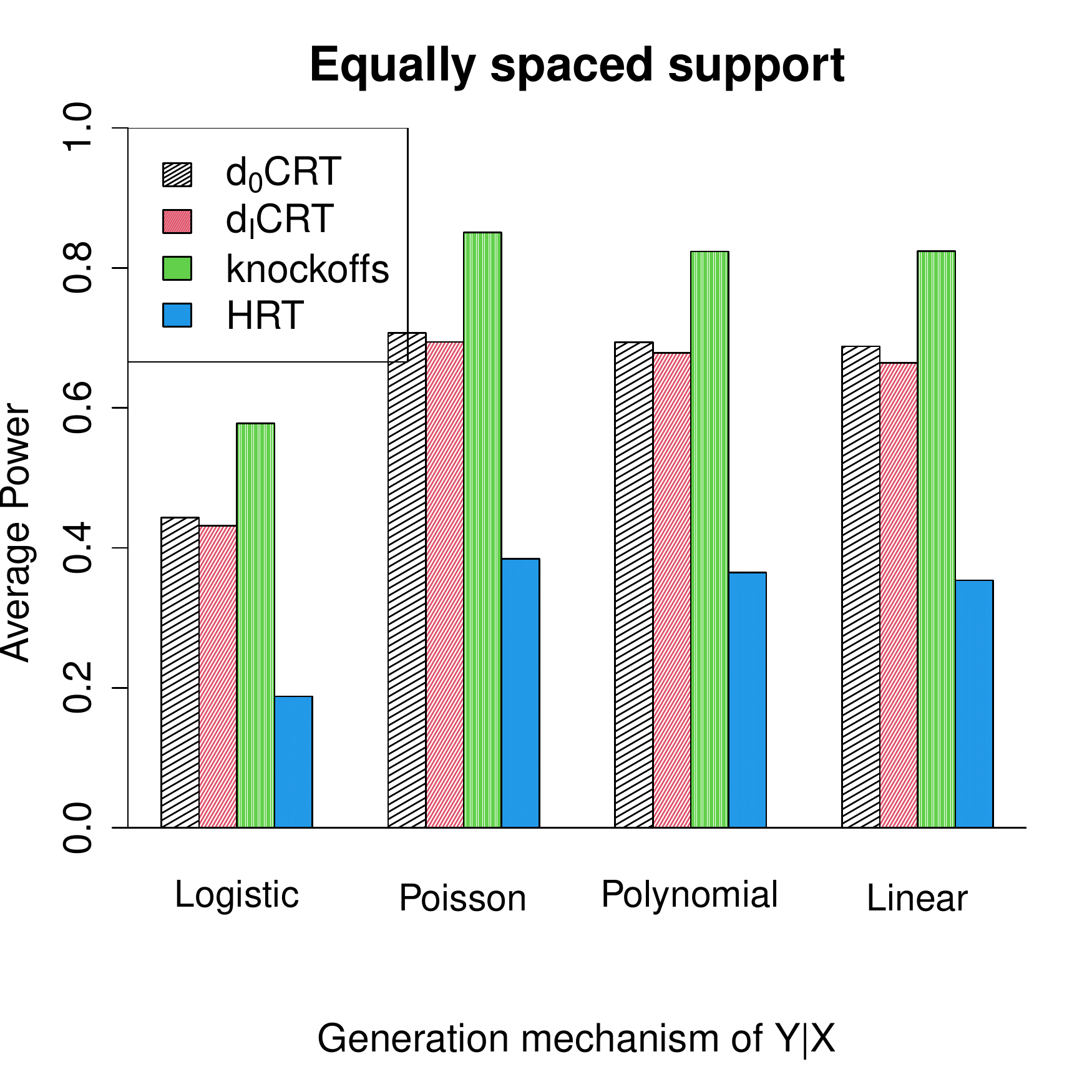}
\caption{\label{fig:diff:design} Average powers of the large scale simulations of Appendix~\ref{sec:sim:hrt} that vary the covariate and response models; all standard errors are below $0.01$. The dCRTs have higher power than knockoffs in most scenarios with adjacent support or equicorrelated design, lower power than knockoffs for equally-spaced support and autocorrelated design, and consistently higher power than HRT.}

\end{figure}

{
\subsection{Assessing algorithmic variability} \label{sec:knockoffs-variability}

In addition to assessing power, we also assessed \textit{algorithmic variability} of the methods compared. We define algorithmic variability as the degree to which the rejection set changes when rerunning a method on the same data with different random seeds. We quantified this notion using the following definition. Given data $\bX, \by$, suppose $\mathcal R^{(1)} = \mathcal R^{(1)}(\bm X, \bm y)$ and $\mathcal R^{(2)} = \mathcal R^{(2)}(\bm X, \bm y)$ are the resulting rejection sets based on two random seeds. Their similarity can be measured via the \textit{Jaccard Index}, defined
\begin{equation*}
J(\mathcal R^{(1)}, \mathcal R^{(2)}) \equiv \frac{|\mathcal R^{(1)} \cap \mathcal R^{(2)}|}{|\mathcal R^{(1)} \cup \mathcal R^{(2)}|}.
\end{equation*}
This quantity is between 0 and 1, with higher Jaccard Indices representing greater similarity. We define the \textit{stability} of the rejection set as the expectation of this quantity:
\begin{equation}
\text{stability}(\bX,\by) \equiv \mathbb E\left[\left.J\left(\mathcal R^{(1)}(\bm X, \bm y), \mathcal R^{(2)}(\bm X, \bm y)\right) \right| \bX, \by\right].
\label{stability}
\end{equation}

We assessed this measure of stability in the context of the large data size simulation from Section~\ref{sec:sim:hrt}. We reran each method for each realization of the data with 50 different seeds to obtain 50 different rejection sets $R^{(r)}(\bm X, \bm y)$ for $1 \leq r \leq 50$. Then, we approximated the stability via 
\begin{equation*}
\widehat{\text{stability}}(\bX,\by) \equiv \frac{1}{{50 \choose 2}}\sum_{1 \leq r_1 < r_2 \leq 50}J\left(\mathcal R^{(r_1)}(\bm X, \bm y), \mathcal R^{(r_2)}(\bm X, \bm y)\right).
\end{equation*}
Figure~\ref{fig:knockoff-variability} compares the average stability of each method for different values of the signal strength. We see that the dCRT yields the most stable rejection sets (with stability at least 0.9 for all parameter values tested), followed by the HRT, followed by knockoffs. The increase in algorithmic variability for HRT and knockoffs is greater for intermediate values of the signal strength, as one would expect. 

%\begin{figure}[h!]
%	\centering
%	\includegraphics[width = 0.7 \textwidth]{figure/knockoffs_variability_v_power_500.pdf}
%	\caption{Expected similarity---measured using the Jaccard Index---between knockoff rejection sets resulting from two independent knockoff samples (see definition~\eqref{stability}). Each point corresponds to a realized pair $(\bX,\by)$, generated from a sparse linear model with 1000 samples, 500 variables, $s$ non-nulls, auto-regressive of order 1 (AR(1)) Gaussian design with parameter $\rho$. There are 25 points (repetitions) each for six signal strengths, and the false discovery rate level is $\alpha = 0.1$. Locally estimated scatterplot smoother (gray) is added for visualization. Ideally, the Jaccard index would be around $(1-\alpha)/(1+\alpha) \approx 0.82$ (dashed line) if repeated runs lead to only the selected nulls changing with the selected non-nulls staying the same. The dip in the middle of each plot suggests high algorithmic variability in practical settings where the power is not very low or very high. This significant variability across re-runs is acknowledged to be a drawback of knockoffs from a reproducibility standpoint.}
%	\label{fig:knockoff-variability}
%\end{figure}

\begin{figure}[h!]
	\centering
	\includegraphics[width = 0.45\textwidth]{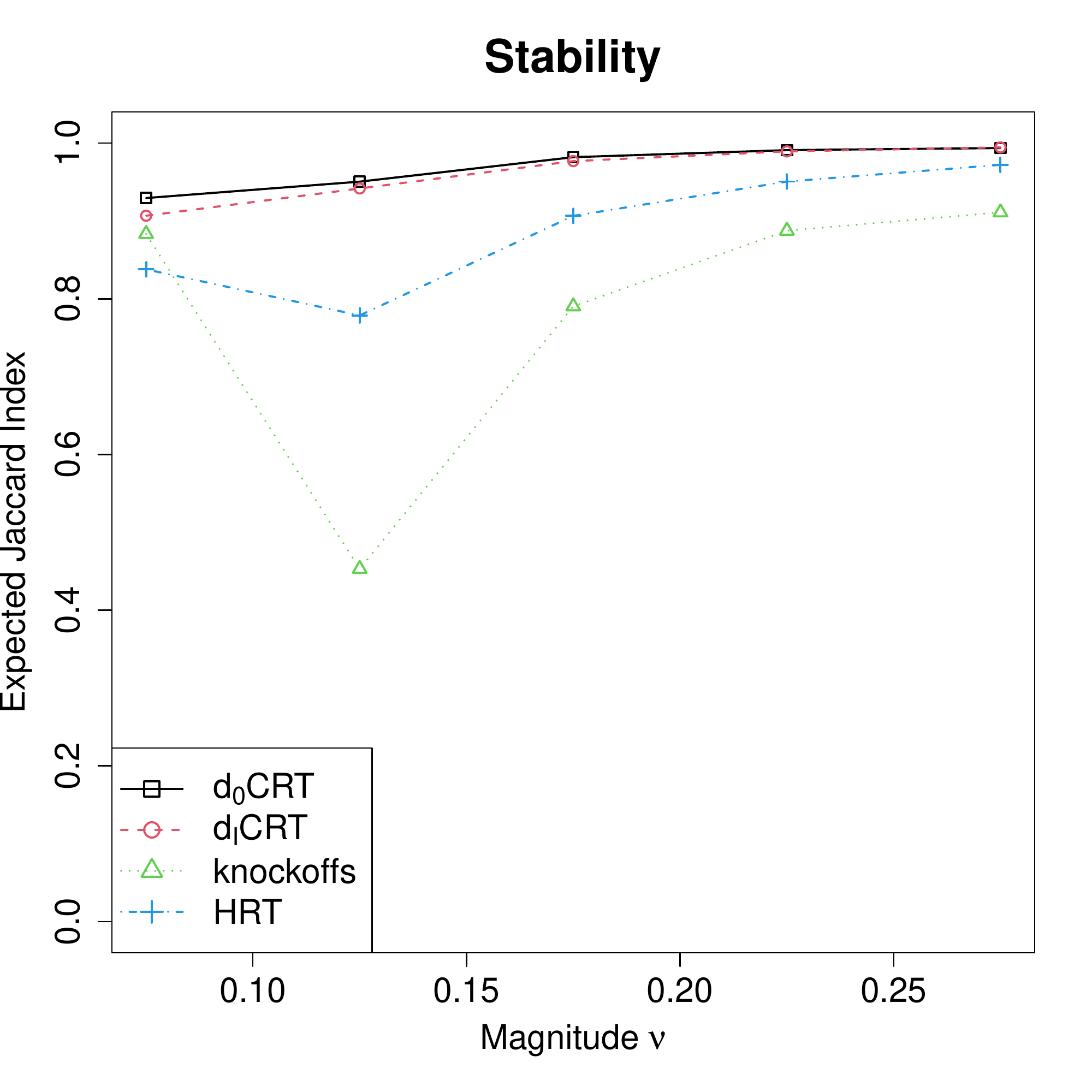}
	\caption{\label{fig:variability} Expected similarity---measured using the Jaccard Index---between the rejection sets resulting from two independent repetitions of each method (see definition~\eqref{stability}). We follow the large scale simulation setting in Appendix~\ref{sec:sim:hrt} with equally spaced signals and the coefficient magnitude varying.}
	\label{fig:knockoff-variability}
\end{figure}
}

\subsection{Comparing dCRT with double machine learning and generalized covariance measure}\label{sim:dml}
As mentioned in the introduction, the test statistic of the resampling-free d$_0$CRT is quite similar to that of double machine learning \citep{chernozhukov2016double} and generalized covariance measure \citep{shah2018hardness} when $X$ is Gaussian. However, we remind the reader that both double machine learning and generalized covariance measure rely on asymptotic normality to calibrate their tests, requiring large samples and well-behaved tails for validity. By comparison, CRT-based methods including the dCRT and HRT are valid in finite samples for any data distribution.

We demonstrate this difference by setting $n=30$, $p=100$ and drawing each covariate independently from a Laplace distribution with mean $0$ and variance $2/9$. We generate $Y$ from a linear model with the residual $\epsilon$ also drawn from the Laplace distribution of mean $0$ and variance $1/2$. Our target is again multiple testing with false discovery rate level $0.1$. We compare HRT, double machine learning, generalized covariance measure and dCRT. When implementing double machine learning and generalized covariance measure, we use our assumed exact model-X knowledge to construct the exact partial residual for each covariate, and use the lasso on $(\by,\bZ)$ to obtain the partial residuals for $Y$. For double machine learning, we use 8-fold cross-fitting. The resulting false discovery rates are presented in Figure~\ref{fig:dml}. Both double machine learning and generalized covariance measure have false discovery rate level substantially above the nominal $0.1$ under all magnitudes, while HRT and dCRT still control the false discovery rate below $0.1$ (as guaranteed by Theorem~\ref{thm:crt}). 

\begin{figure}[htpb]
\centering
\includegraphics[width=0.4\textwidth]{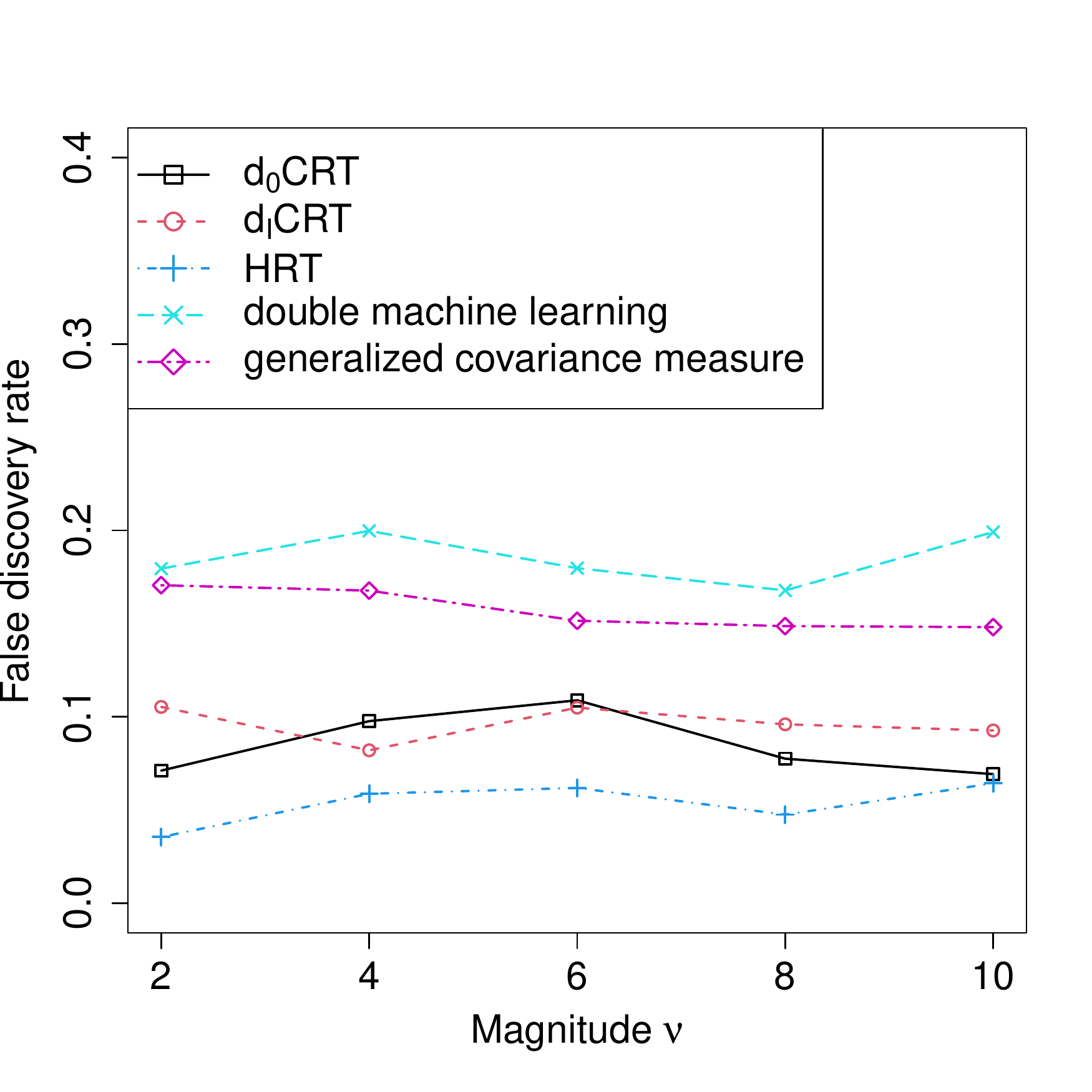}
\caption{\label{fig:dml} False discovery rates of the simulation in Appendix~\ref{sim:dml} comparing computationally efficient CRT methods to double machine learning and generalized covariance measure. All standard errors are below $0.01$. double machine learning and generalized covariance measure do not control the false discovery rate at the target level 0.1, but the other methods do.}
\end{figure}

\subsection{Power improvement of the d$_{\mathrm{I}}$CRT in the presence of interactions}\label{sec:sim:dk}

All previous simulations have shown similar, if slightly worse, performance for the d$_{\mathrm{I}}$CRT compared to the d$_0$CRT. This is because the models have all been additive (technically a logistic regression model is not additive, but the logistic-regression-derived statistics used by both dCRT methods fit to the logistic-transformed $Y$, which does follow an additive model). To demonstrate the benefits of the d$_{\mathrm{I}}$CRT to characterize more complex effects, we conduct here a non-additive simulation with first-order interactions that obey the hierarchy principle described in Section~\ref{sec:dIcrt}. We take $n=p=800$ and generate $(X,Z\trans)\trans$ from ${\rm AR}(1)$ with autocorrelation 0.5. Letting $\mu(X,Z)=\nu(X+\sum_{k=1}^5Z_{j_k}+1.5X\sum_{k=1}^5Z_{j_k})$ with $j_1,\ldots,j_5$ randomly picked from $\{1,2,\ldots,799\}$, we generate $Y$ either from a Gaussian model with conditional mean given by $\mu(X,Z)$ or from a Bernoulli model with $\log(\mathbb{P}(Y=1\,|\,X,Z)/\mathbb{P}(Y=0\,|\,X,Z))=\mu(X,Z)$. The target is to test the single hypothesis $Y\indp X\mid Z$ at level $0.05$ (hence knockoffs does not apply). Figure~\ref{fig:dglm} shows the powers of the d$_0$CRT, d$_{\mathrm{I}}$CRT, and HRT. As is expected, d$_{\mathrm{I}}$CRT has substantially higher power than d$_0$CRT.

\begin{figure}[htpb]
\centering
  \includegraphics[width=0.4\textwidth]{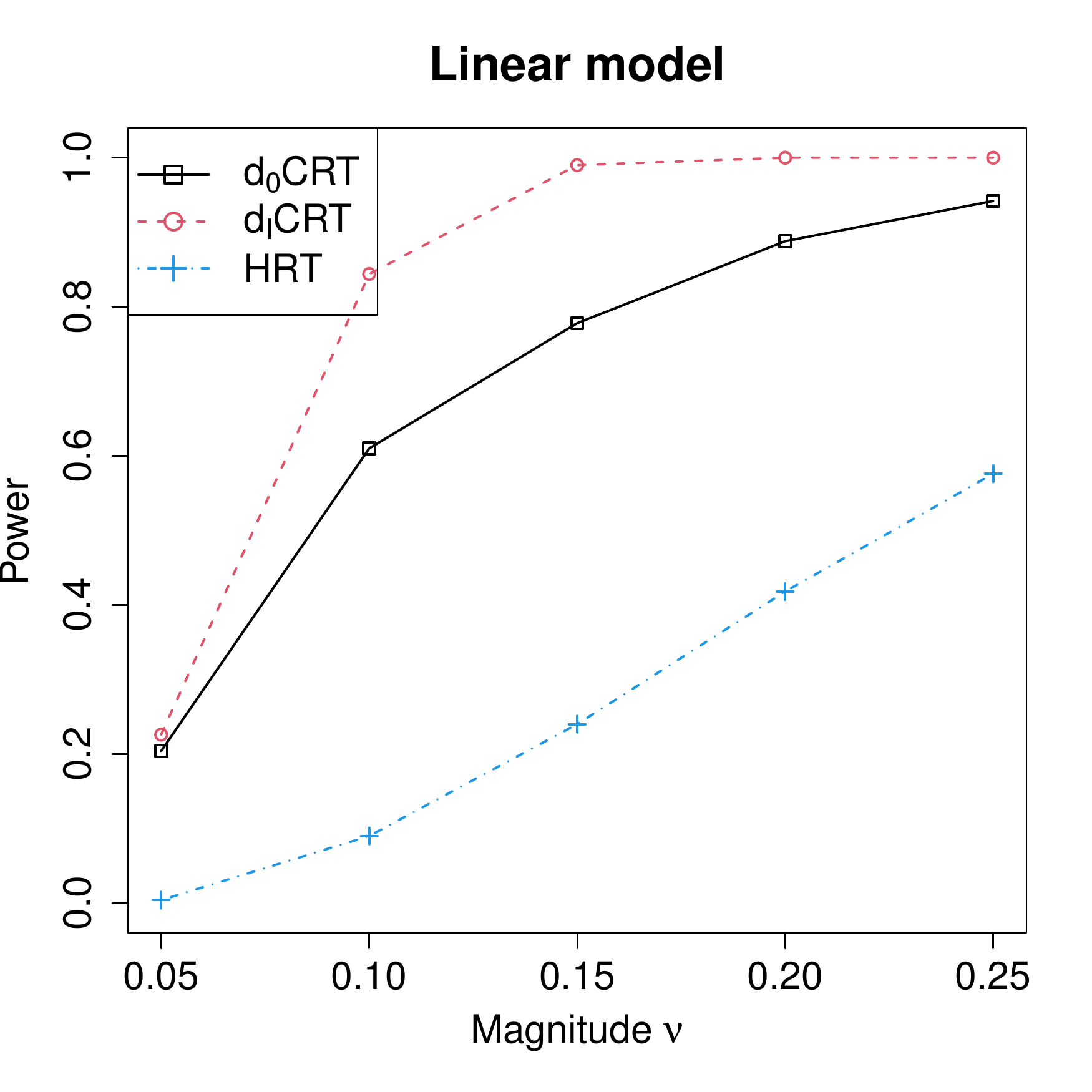}
  \includegraphics[width=0.4\textwidth]{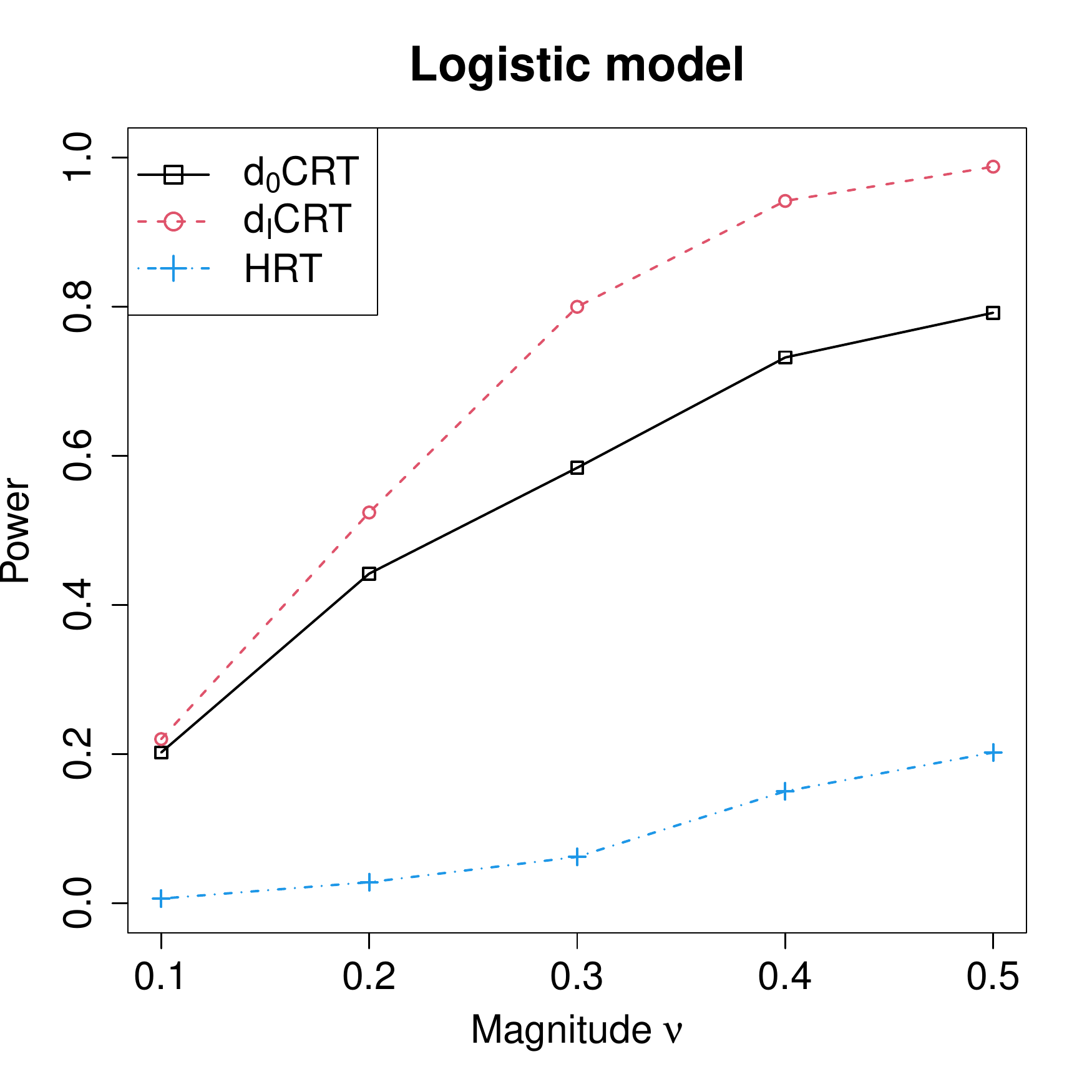}
\caption{\label{fig:dglm} Powers of the simulations in Appendix~\ref{sec:sim:dk} comparing methods in the presence of interactions; standard errors are below $0.03$. The dCRT is much more powerful than the HRT.}
\end{figure}

\subsection{Stability of the d$_{\mathrm{I}}$CRT to the choice of $k$}\label{app:dIcrt:k}
In this section, we study the sensitivity of d$_{\mathrm{I}}$CRT to the choice of $k$ defined in Example~\ref{ex:2}. We simulate the d$_\mathrm{I}$CRT in the baseline setting (linear model) of Section~\ref{sec:sim:hrt} and the linear interaction model setting of Section~\ref{sec:sim:dk} for varying choices of $k$ (in both settings, the default $k=2\log(p)\approx 13$). The results in Figure~\ref{fig:choose:k} show that the choice of $k$ has nearly no impact on the power of the d$_{\mathrm{I}}$CRT in the linear model setting.
In the interaction setting, the power of the d$_{\mathrm{I}}$CRT decreases with $k$ for $k>5$ since there are only $5$ true interactions in the model, but the trend is quite gradual and the d$_\mathrm{I}$CRT's power stays above that of d$_0$CRT through $k=22$.

\begin{figure}[htpb]
\centering
  \includegraphics[width=0.4\textwidth]{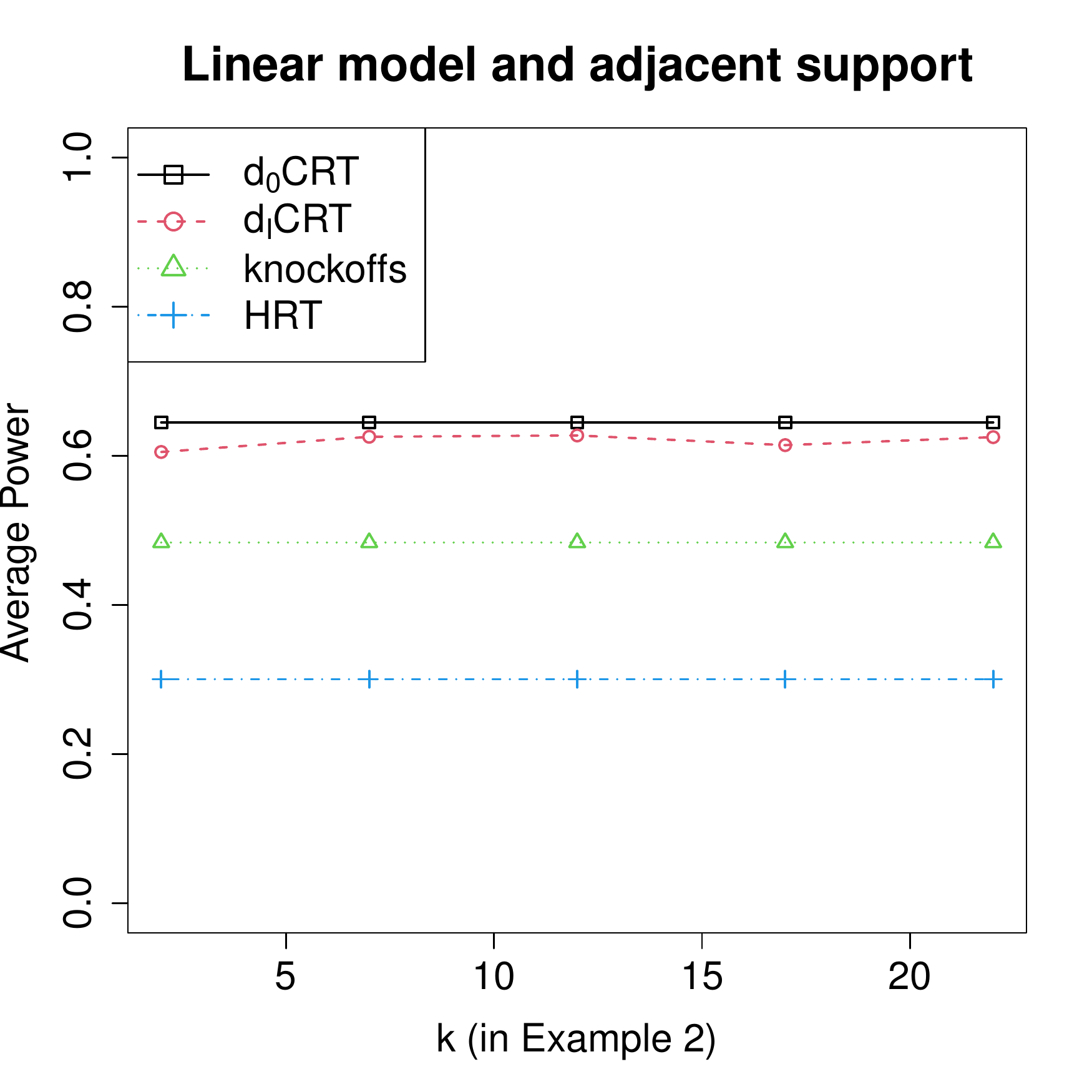}
  \includegraphics[width=0.4\textwidth]{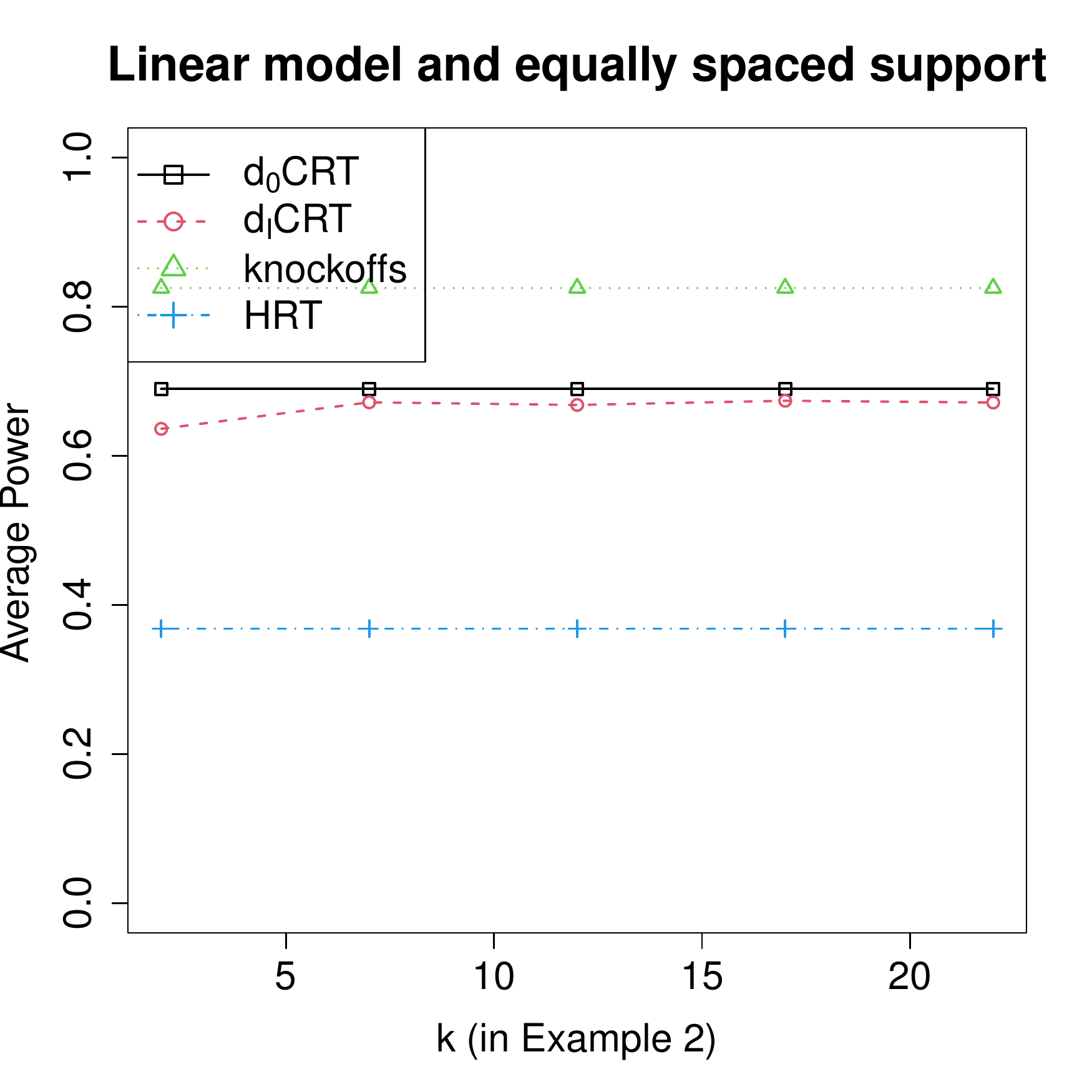}
  \includegraphics[width=0.4\textwidth]{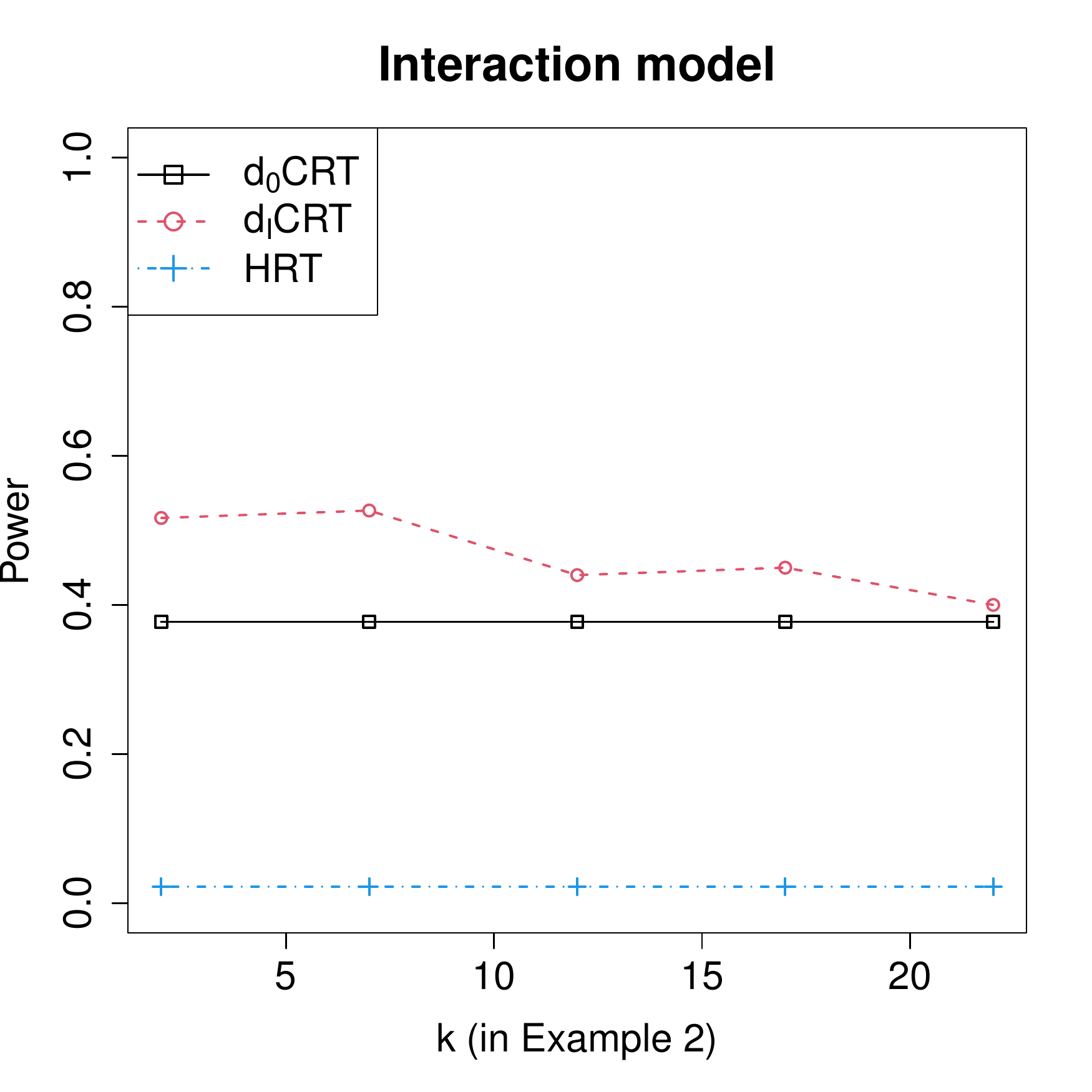}
\caption{\label{fig:choose:k} Powers of the simulation in Appendix~\ref{app:dIcrt:k} evaluating the stability of the d$_{\mathrm{I}}$CRT to the choice of $k$; all standard errors are below $0.03$. The d$_{\mathrm{I}}$CRT is fairly stable to the choice of $k$.}
\end{figure}

\subsection{A random-forest-based d$_\mathrm{I}$CRT}
\label{sec:rf}
Examples~\ref{ex:1} and~\ref{ex:2} are inherently rooted in generalized linear models, and we expect them to perform well in situations where a generalized linear model captures much of the interesting dependence between $Y$ and $X$. But there is nothing limiting the dCRT's application to such settings, and in this section we demonstrate the power of a random-forest-based d$_\mathrm{I}$CRT in a setting that is far from a generalized linear model. 
\begin{example}[Random-forest-based d$_\mathrm{I}$CRT]\label{ex:rf}
Let $\dyone$ be the fitted predictions from a random forest fitting $\by$ to $\bZ$, let $\dynone$ be the columns of $\bZ$ corresponding to the $k$ largest values of the default variable importance measure in the \textsf{R} package \textsf{randomForest}, and let $T(\by,\bx,\dy,\dx)$ fit a random forest of $\by$ on $\bx-\dx$ and $\dy$ and return the default variable importance measure for $\bx-\dx$.
\end{example}

We take $n=p=800$ and $(X,Z\trans)\trans$ as following an AR(1) model with autocorrelation 0.5. We choose a conditional model for $Y$ in which the magnitude of the effect of $X$ on $Y$ is heterogeneous and varies with $Z$: $\mu(X,Z)=\nu[0.5X^2+{\rm sin}(0.5\pi X)](0.3+\sum_{k=1}^5Z_{j_k})$, and $Y$ is standard normal noise added to $\mu(X,Z)$. We simulate tests of $Y\indp X \mid Z$ at significance 0.05 and plot the results in Figure~\ref{fig:rf}. The random-forest-based d$_\mathrm{I}$CRT is denoted by ``d$_\mathrm{I}$CRT (forest)" and uses 100 trees for distillation and 30 trees for computation of $T$. The additional function-approximation flexibility of random forests imparts a substantial gain in power compared to d$_0$CRT, d$_\mathrm{I}$CRT, HRT, which are all implemented based on generalized linear models.

\begin{figure}[htpb]
\centering
  \includegraphics[width=0.4\textwidth]{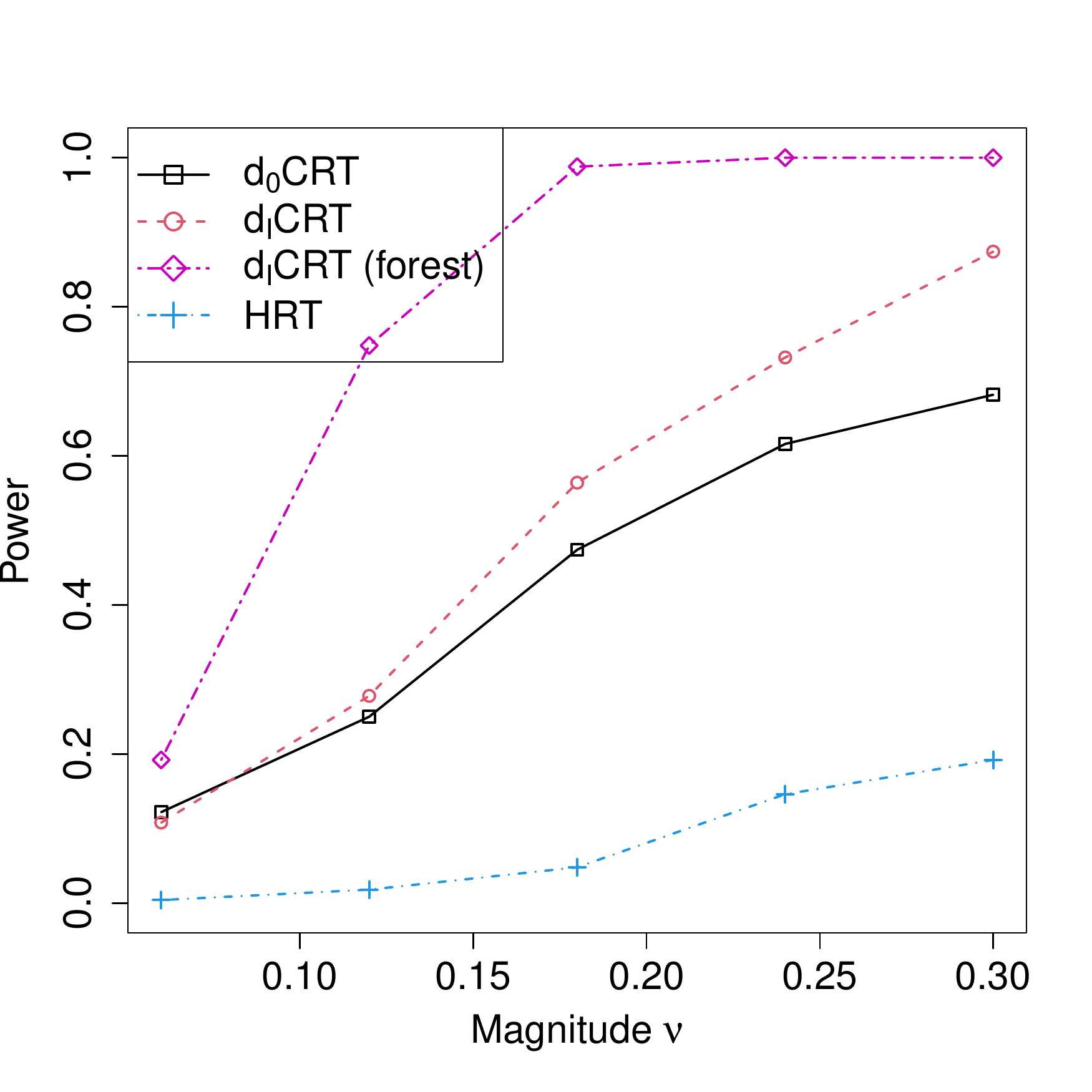}
\caption{\label{fig:rf} Powers of the simulations in Appendix~\ref{sec:rf} demonstrating a random-forest-based d$_\mathrm{I}$CRT on a complex ground truth model; all standard errors are below $0.03$. The flexibility of the random forest variant makes it much more powerful than the generalized linear model-based d$_0$CRT and d$_{\text{I}}$CRT.}

\end{figure}

\subsection{Robustness: Known first and second moments}\label{app:sim:robustfirstsecond}
We designed numerical experiments to study the robustness of the dCRTs, i.e., whether the methods still control Type-I error and have power when the $X\mid Z$ distribution is misspecified. We first consider the case when one has no knowledge of the conditional distribution of $X\mid Z$ except its first two moments, and simply treats $X\mid Z$ as conditionally Gaussian with matching moments. We let $n=p=800$ and generate $Z=(Z_1,Z_2,\ldots,Z_{799})\trans$ from a Gaussian ${\rm AR}(1)$ model with autocorrelation 0.5 and sample $X$ as conditionally Poisson:
\[
X=0.15\sum_{j=1}^{50}\varphi_j Z_j+\delta,\quad\mbox{where}\quad\delta=\frac{O-r}{\sqrt{r}}\quad\mbox{with}\quad O\mid Z\sim {\rm Poi}(r),
\]
where each $\varphi_j$ is independently and uniformly drawn from $\{-1,1\}$ and ${\rm Poi}(r)$ represents the Poisson distribution with mean $r$. When $r$ is small, $(O-r)/\sqrt{r}$ is quite skewed with its tail behaviour highly different from Gaussian while as $r$ becomes larger, $(O-r)/\sqrt{r}$ converges to a $\mathcal{N}(0,1)$. We run the dCRT as if $\delta\sim\N(0,1)$, and hence $r$ measures the degree of misspecification (lower $r$ corresponds to more misspecification). For $Y$, we use linear or logistic model linked with $\nu X+0.15\sum_{j=1}^{50}\psi_jZ_j$.

Our target is to test for $Y\indp X\mid Z$ with level $0.05$. To study the performance in Type-I error control we set $\nu=0$, while to study the power we let $\nu=0.1$ for linear model and $\nu=0.2$ for logistic model. We compare d$_0$CRT, d$_{\mathrm{I}}$CRT, and HRT, with the same specification as the previous section except that $X-\mathbb{E}[X|Z]$ is approximated as $\mathcal{N}(0,1)$ when modelling $X$. The resulting Type-I error and power versus $\log_2(r)$ are plotted in Figure~\ref{fig:robust:sec}. Even when $r$ is as small as $0.5$, the Type-I error of the dCRTs remain below their nominal level and their powers are relatively similar to the nearly-well-specified setting of $r=64$.

\begin{figure}[htpb]
\centering
  \includegraphics[width=0.4\textwidth]{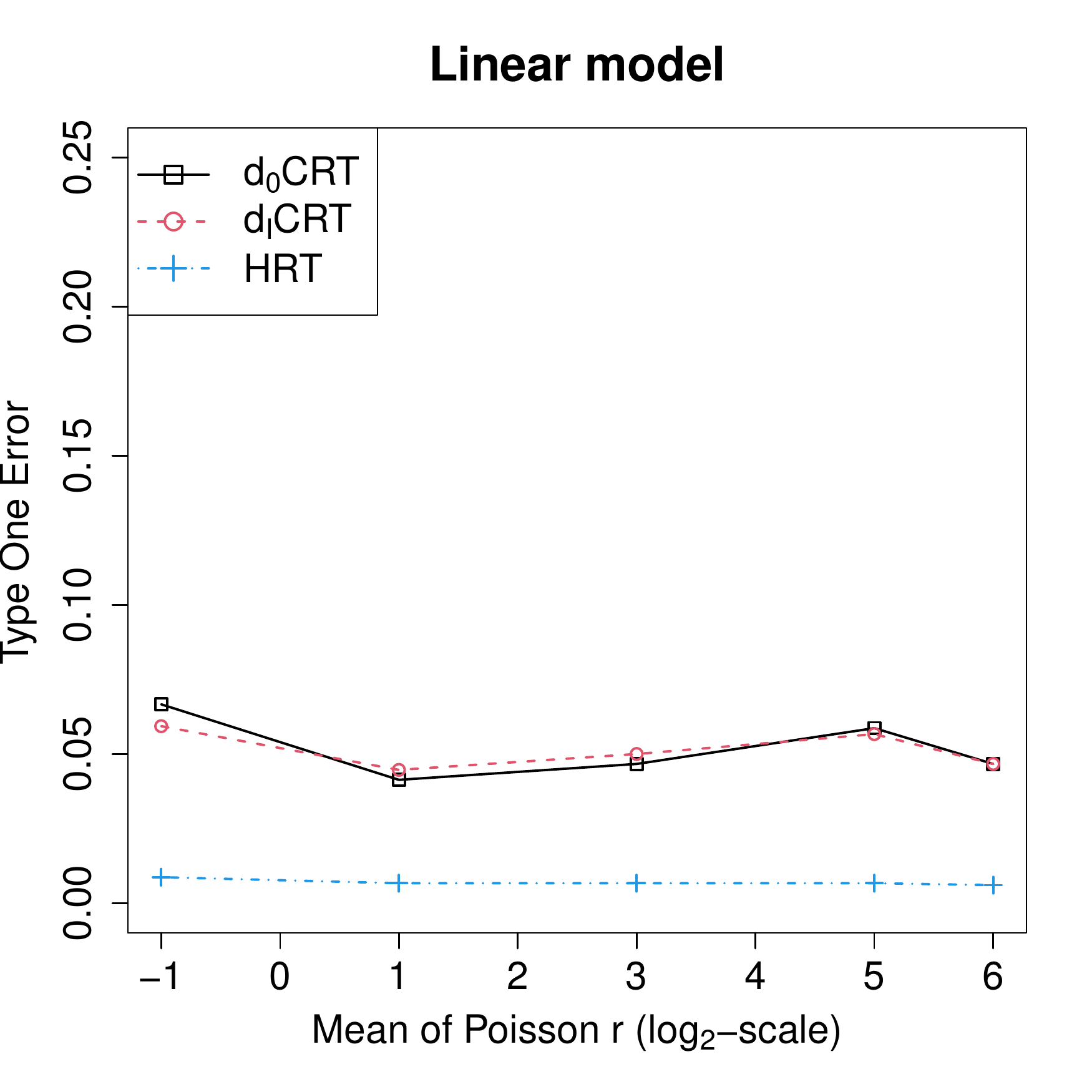}
  \includegraphics[width=0.4\textwidth]{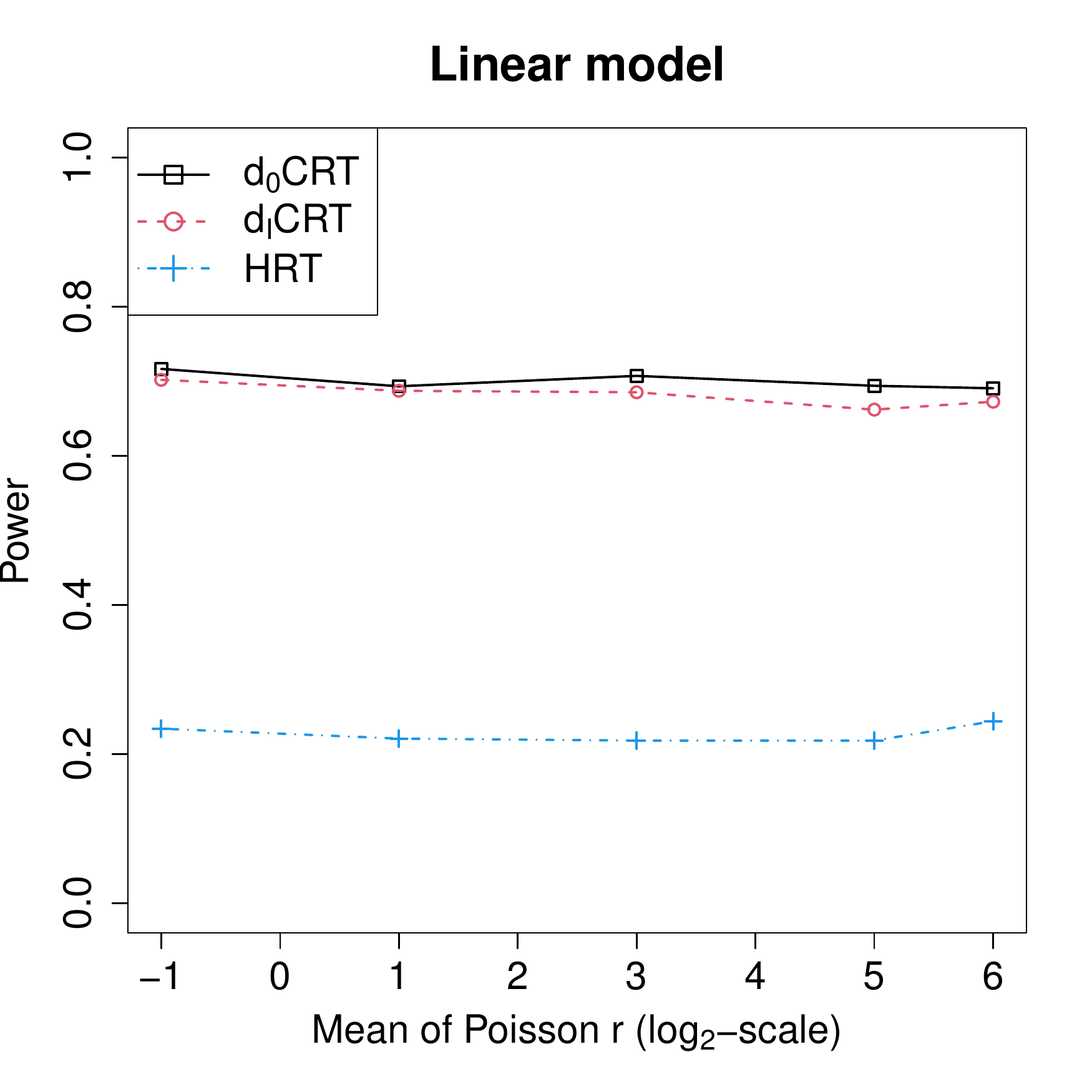}

\includegraphics[width=0.4\textwidth]{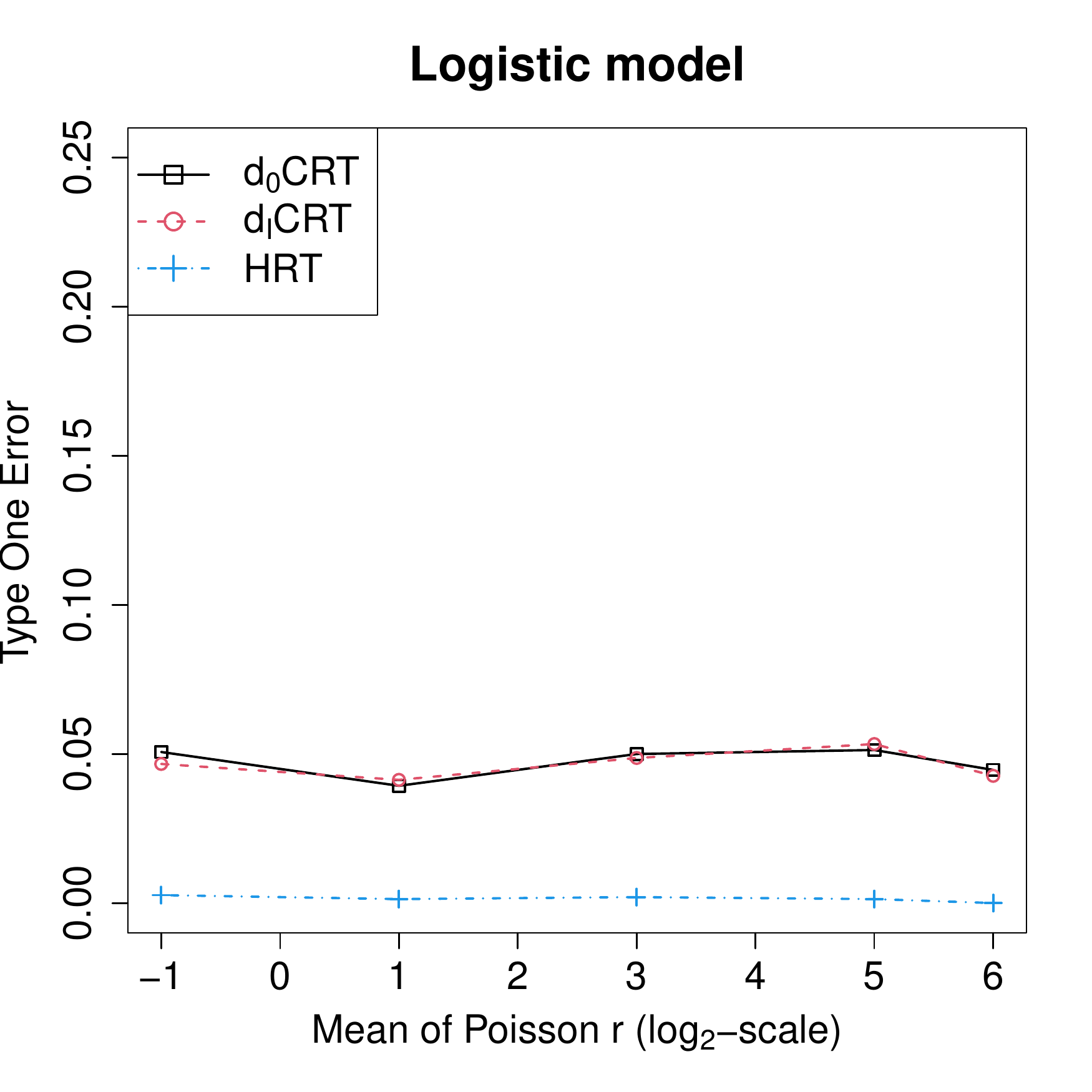}
\includegraphics[width=0.4\textwidth]{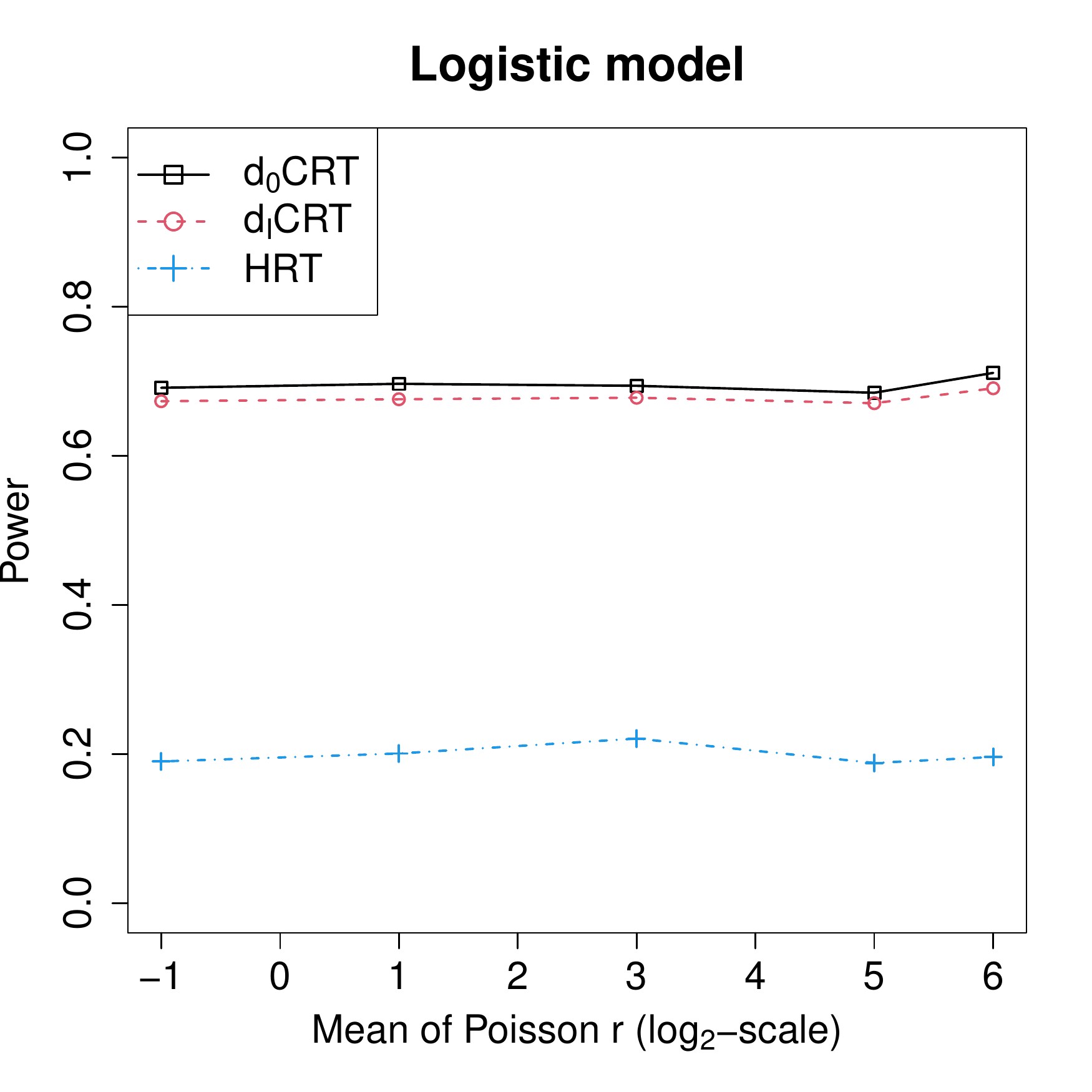}
\caption{\label{fig:robust:sec} Type-I error rates and powers of the simulations in Appendix~\ref{app:sim:robustfirstsecond} measuring robustness to misspecification in terms of the parameter $r$. All standard errors are below $0.03$. All methods control Type-I error, and the dCRTs are more powerful than HRT.}
\end{figure}

\subsection{In-sample-estimated moments}\label{app:sim:robustestmom}
Next, we study the case when one knows a model family for $X\mid Z$ but needs to estimate its parameters in-sample. Again, we set $n=800$, $p=800$, $s=50$, generate the covariates from a Gaussian ${\rm AR}(1)$ distribution with autocorrelation 0.5 and generate $Y$ from a linear model with magnitude $\nu=0.175$ or logistic model with magnitude $\nu=0.5$, which again makes the power roughly $0.5$. {Again, we study adjacent and equally spaced supported signals separately.} Then, as part of our dCRT procedures, we use the $n=800$ samples to estimate the {conditional distribution parameters} of the covariates. {For this purpose, we consider three commonly used approaches as the options:
\begin{enumerate}[(i)]
\item {\bf The Ledoit--Wolf estimator:} We follow \cite{ledoit2004well} to obtain an estimate of the covariance matrix of the covariates which is the optimal (in terms of mean square error) linear shrinkage of the sample empirical covariance to the identity matrix. Letting the resulting estimator be $\widehat\Sigma$, we actually use a rescaled version given by $D\widehat\Sigma D$, where $D={\rm diag}\{d_1,d_2,\cdots,d_p\}$ is a $p\times p$ diagonal matrix and $d_j$ is the ratio of the estimated conditional variance by inverting $\widehat\Sigma$ and that estimated using the mean squared residuals. Here, the multiplier $D$ serves to de-bias the estimated conditional variances of the covariates, and this was important since this determined the conditional variance of the resampled covariates.

\item {\bf Graphical lasso:} We implement the graphical lasso \citep{friedman2008sparse} tuned by cross-validation to estimate the precision matrix and invert it to estimate the covariance matrix of the covariates. Again, we rescale that estimate in the same way as in (i) using $D$.

\item {\bf Nodewise lasso:}  For each covariate $X$, we fit the lasso tuned by cross-validation to model its conditional mean given $Z$, and use the resulting mean squared regression residuals to estimate the conditional variance of $X\mid Z$. Knockoffs is not included in this case since this nodewide lasso, applied to each covariate in turn, does not in general provide a coherent covariance matrix, and a covariance matrix is required to generate Gaussian knockoffs.

\end{enumerate}
}

%implement the graphical lasso \citep{friedman2008sparse} tuned by cross-validation. Given the true sparsity of this particular covariate model, the graphical lasso does relatively well, so we also try mixing its covariance matrix estimate $\hat{\Sigma}_g$ with the sample covariance matrix $\hat{\Sigma}$, which is a very poor estimate since $n=p$. 
%\[
%\hat{\Sigma}_m=D\{\upsilon\hat{\Sigma}+(1-\upsilon)\hat{\Sigma}_g\}D,
%\]
%where $\upsilon\in [0,1)$ is a proportion parameter controlling the mixture of the two estimates and $D={\rm diag}\{d_1,d_2,\cdots,d_p\}$ is a $p\times p$ diagonal matrix, where $d_j$ is the ratio of the estimated conditional variance by inverting $\upsilon\hat{\Sigma}+(1-\upsilon)\hat{\Sigma}_g$ and that estimated using the mean squared residuals. Here $D$ is just used to ensure that the estimated conditional variance is close to its sample mean squared residual. By changing $\upsilon$, we are able to inspect the performance of the dCRT as the quality of our covariance estimation varies. 

The goal of the simulation is controlled variable selection with false discovery rate level $0.1$. {The resulting false discovery rate and average power under different models and the above estimation strategies of $X\mid Z$ are presented in Figure~\ref{fig:robust:X:fdr} and Figure~\ref{fig:robust:X:power}, respectively. When the conditional distribution parameters are estimated by graphical lasso or nodewise lasso, false discovery rates of all the methods are well controlled by the nominal level under all the model and signal space settings. Also, their powers are close to the ideal model-X case presented in the right panel of Figure~\ref{fig:diff:design}. With Ledoit--Wolf estimation, the false discovery rates of the two dCRT methods increase only very slightly above nominal. 
%While knockoffs and HRT still preserve proper false discovery rate control. We also notice that the power gap between dCRT and knockoffs become smaller, and the power of dCRT drops a little compared with Figure~\ref{fig:diff:design}. This should be due to that the Ledoit--Wolf estimation does not utilize the sparsity of $X\mid Z$ and shrink the estimators excessively. Even though, the inflation of false discovery rate and reduction of power are moderate. The methods still have a reasonably good performance in this case.
}

\begin{figure}[htpb!]
\centering
  \includegraphics[width=0.4\textwidth]{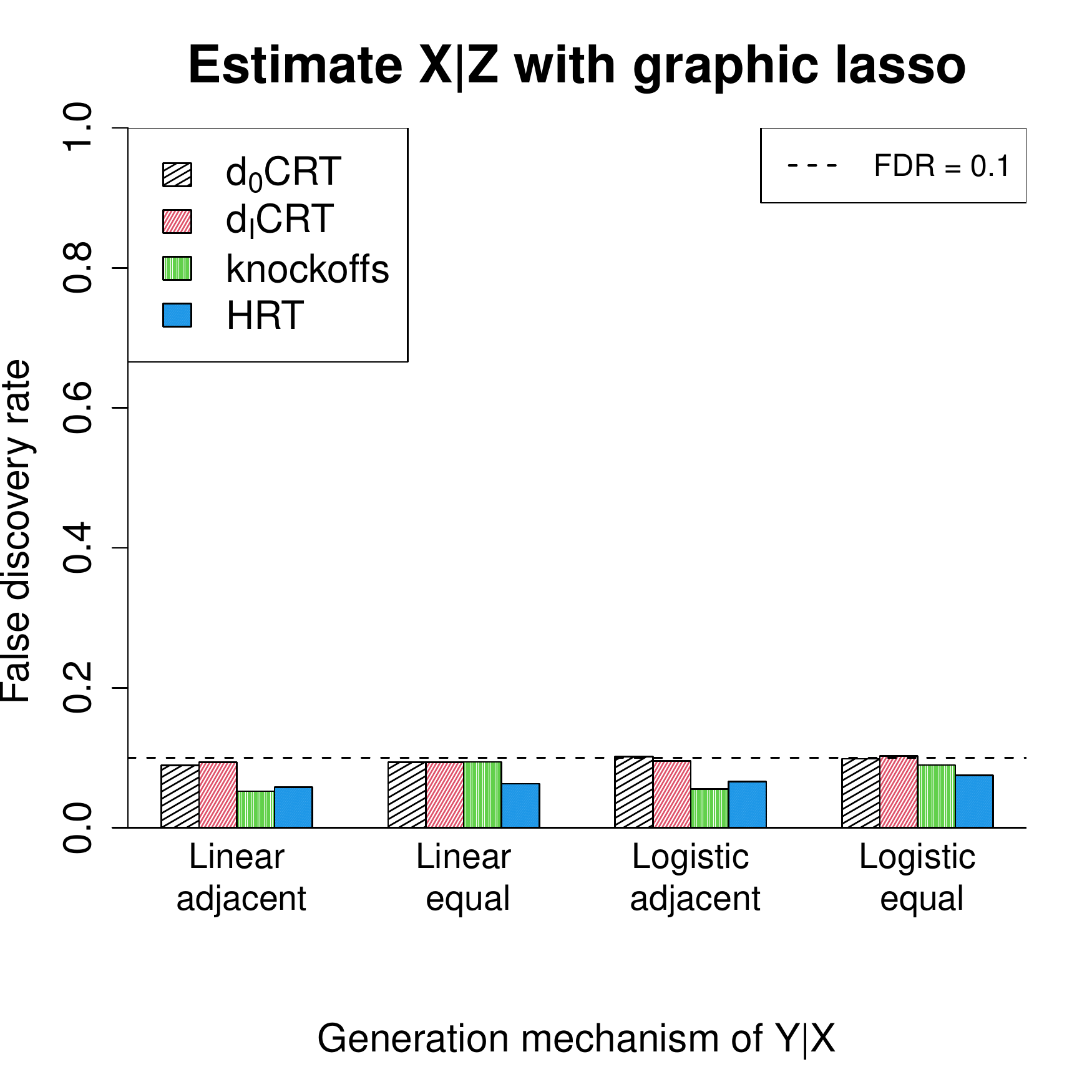}
  \includegraphics[width=0.4\textwidth]{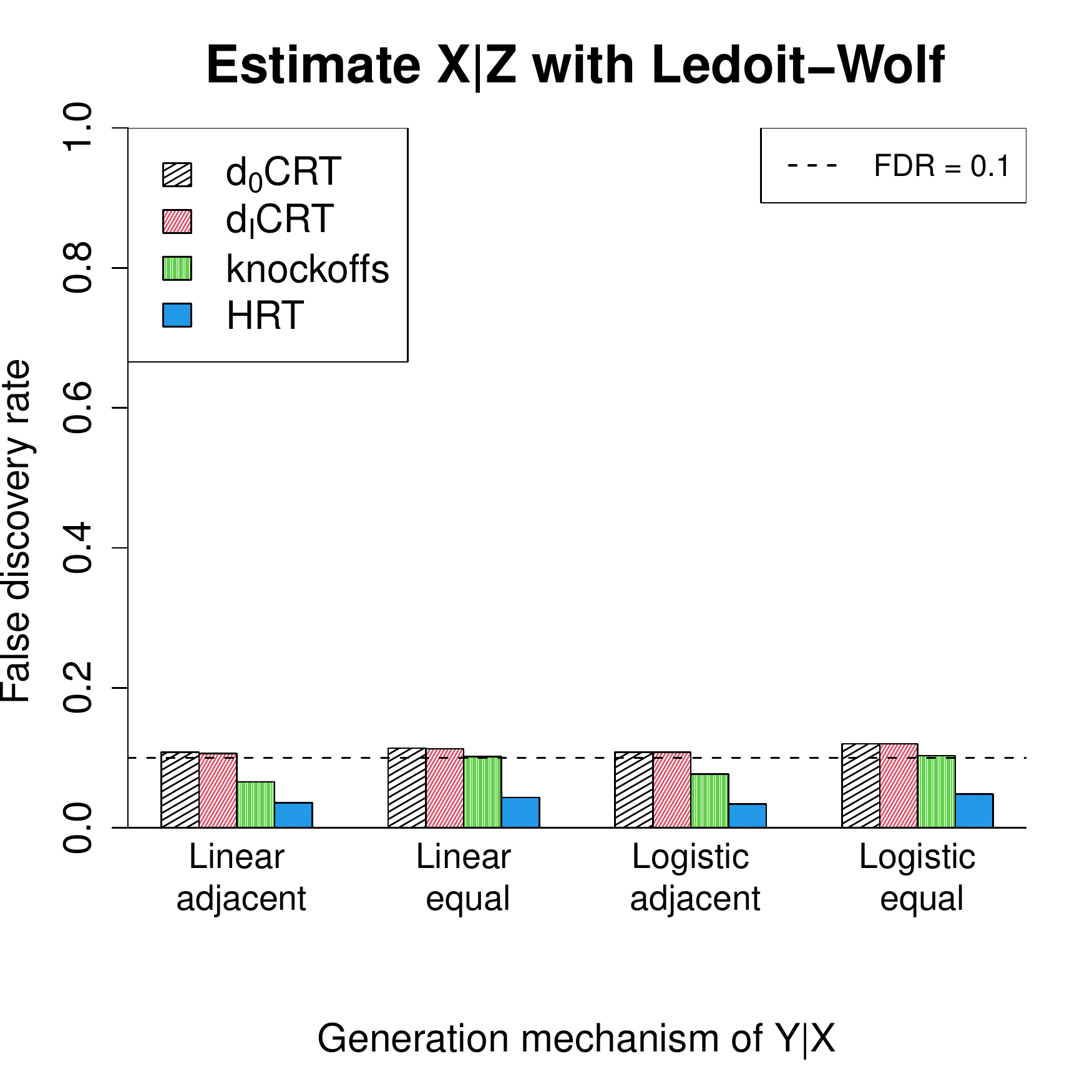}
\includegraphics[width=0.4\textwidth]{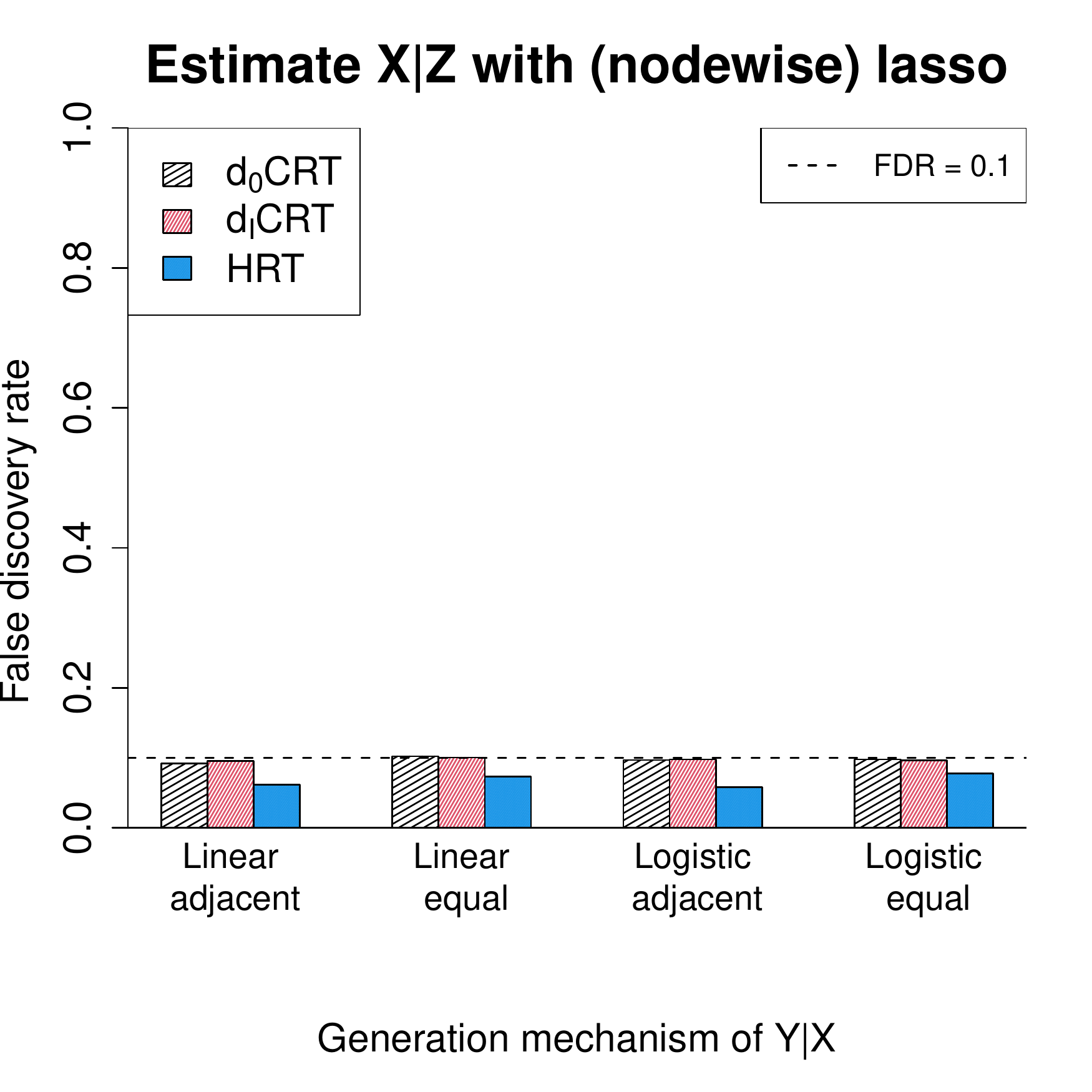}
\caption{{\label{fig:robust:X:fdr} False discovery rates of the simulations in Appendix~\ref{app:sim:robustestmom} measuring robustness to in-sample estimation of the covariate covariance matrix obtained via three common approaches, with the model for $Y$ varying. All standard errors are below $0.01$. The false discovery rates of the dCRT remain close to the nominal level 0.1 in all settings.}}
\end{figure}

\begin{figure}[htpb!]
\centering
  \includegraphics[width=0.4\textwidth]{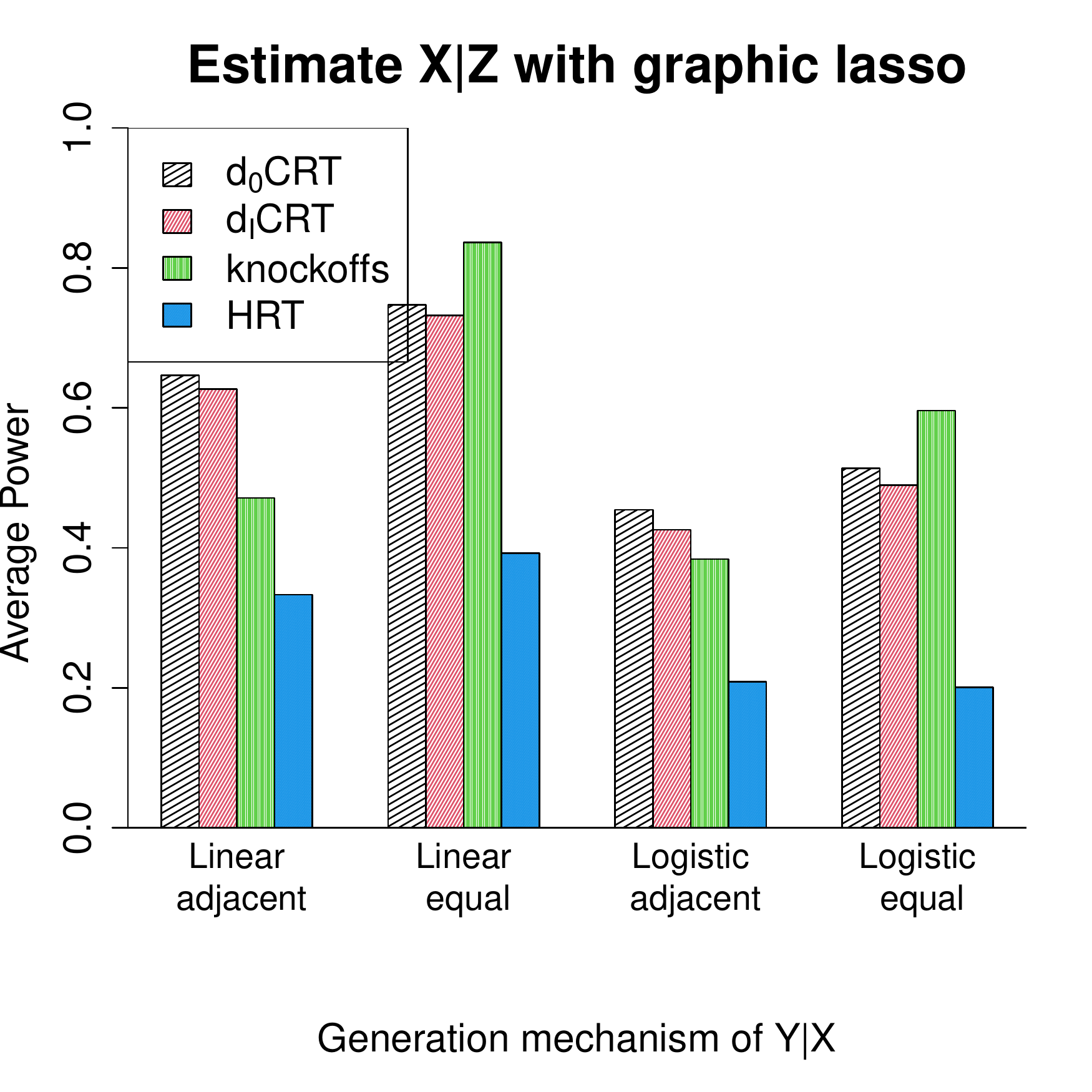}
  \includegraphics[width=0.4\textwidth]{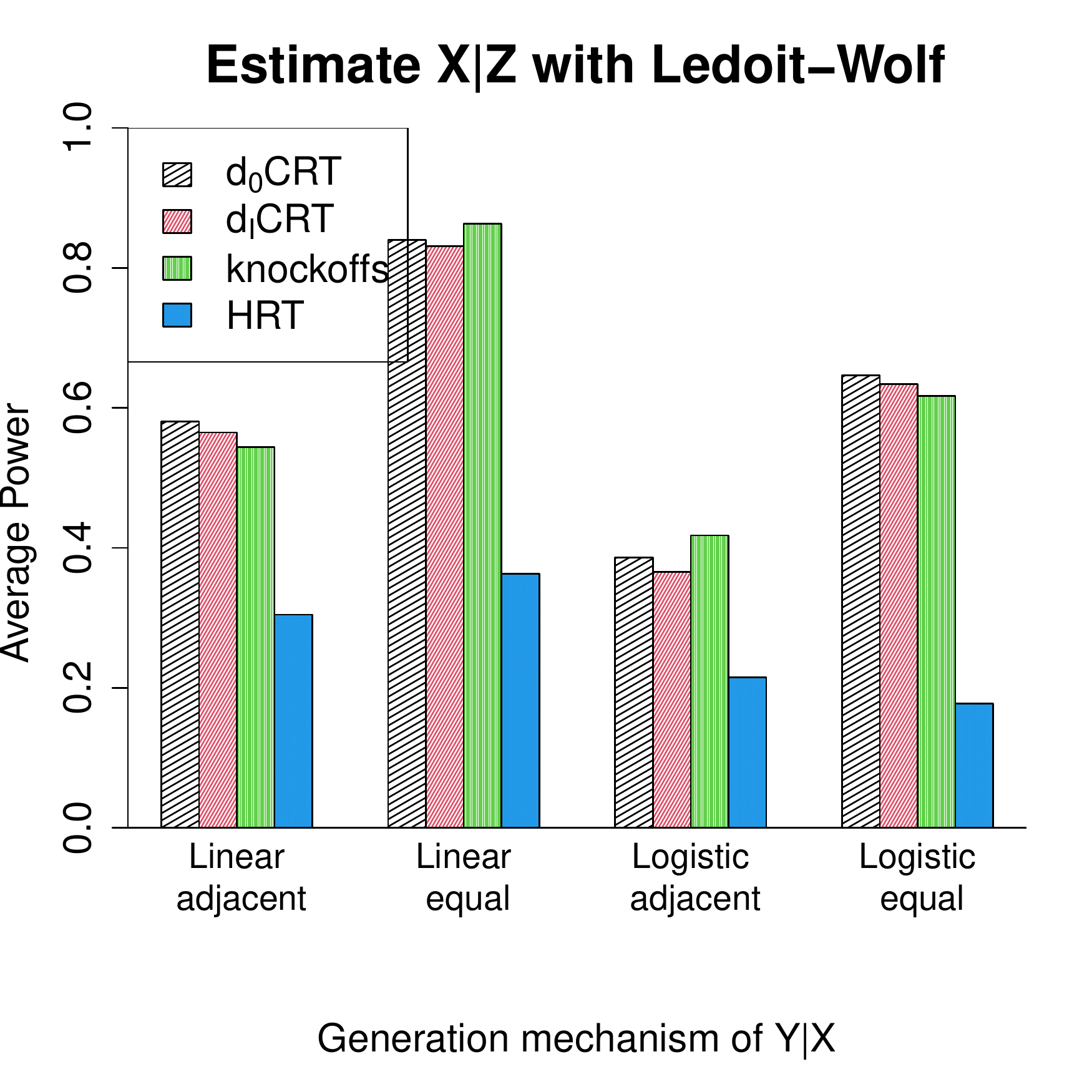}
\includegraphics[width=0.4\textwidth]{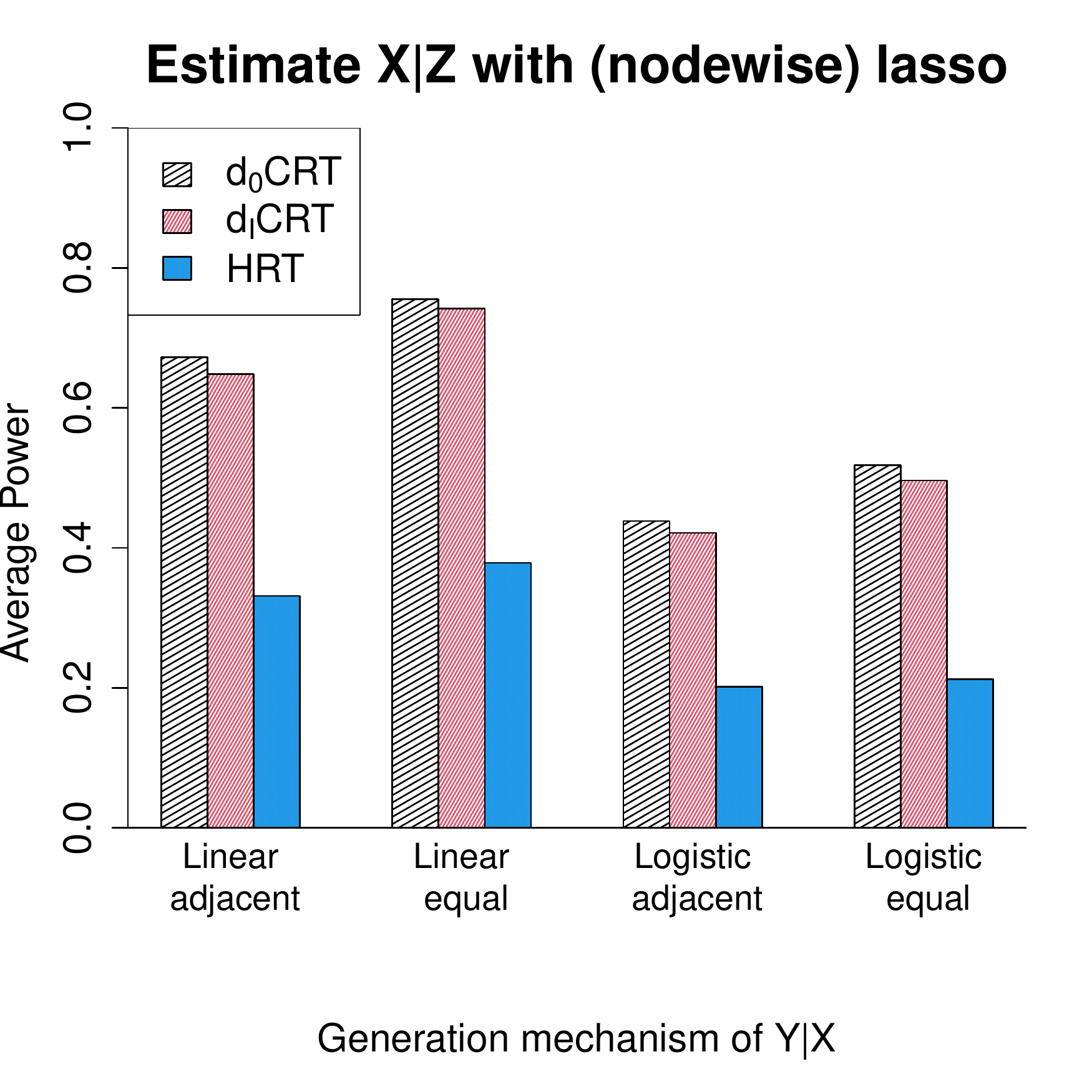}
\caption{{\label{fig:robust:X:power} Average power of the simulations in Appendix~\ref{app:sim:robustestmom} measuring robustness to in-sample estimation of the covariate covariance matrix obtained via three common approaches, with the model for $Y$ varying. All standard errors are below $0.01$.}}
\end{figure}

% \begin{figure}[htpb!]
% \centering
%   \includegraphics[width=0.4\textwidth]{figure_revision/robustX_LW_fdr-eps-converted-to.pdf}
%   \includegraphics[width=0.4\textwidth]{figure_revision/robustX_LW_power-eps-converted-to.pdf}
% \includegraphics[width=0.4\textwidth]{figure_revision/robustX_lasso_fdr-eps-converted-to.pdf}
% \includegraphics[width=0.4\textwidth]{figure_revision/robustX_lasso_power-eps-converted-to.pdf}
% \caption{\label{fig:robust:est:new} False discovery rates and average powers of the simulations in Appendix~\ref{app:sim:robustestmom} measuring robustness to in-sample estimation of the covariate conditional distribution; all standard errors are below $0.01$. The true coefficients are generated with equally spaced support. All standard errors are below $0.01$. The false discovery rate seems relatively unaffected by estimation of the conditional distribution of $X$.}
% \end{figure}

\subsection{Measuring the effect of the resampling-free modification: Gaussian covariates}\label{sim:rfgauss}
Section~\ref{sec:main:mcf} proposes a resampling-free version of d$_0$CRT and d$_{\mathrm{I}}$CRT requiring a small modification to their test statistics; we show here this modification does not affect their powers. Under the baseline setting of Section~\ref{sec:sim:hrt} and the setting with Gaussian covariates and interactions in Section~\ref{sec:sim:dk}, we compare the resampling-free dCRTs with their non-resampling-free versions in terms of average power. Figure~\ref{fig:resample:free:gauss} shows the resampling-free modification makes essentially no difference to their powers.

\begin{figure}[htpb!]
\centering
 \includegraphics[width=0.4\textwidth]{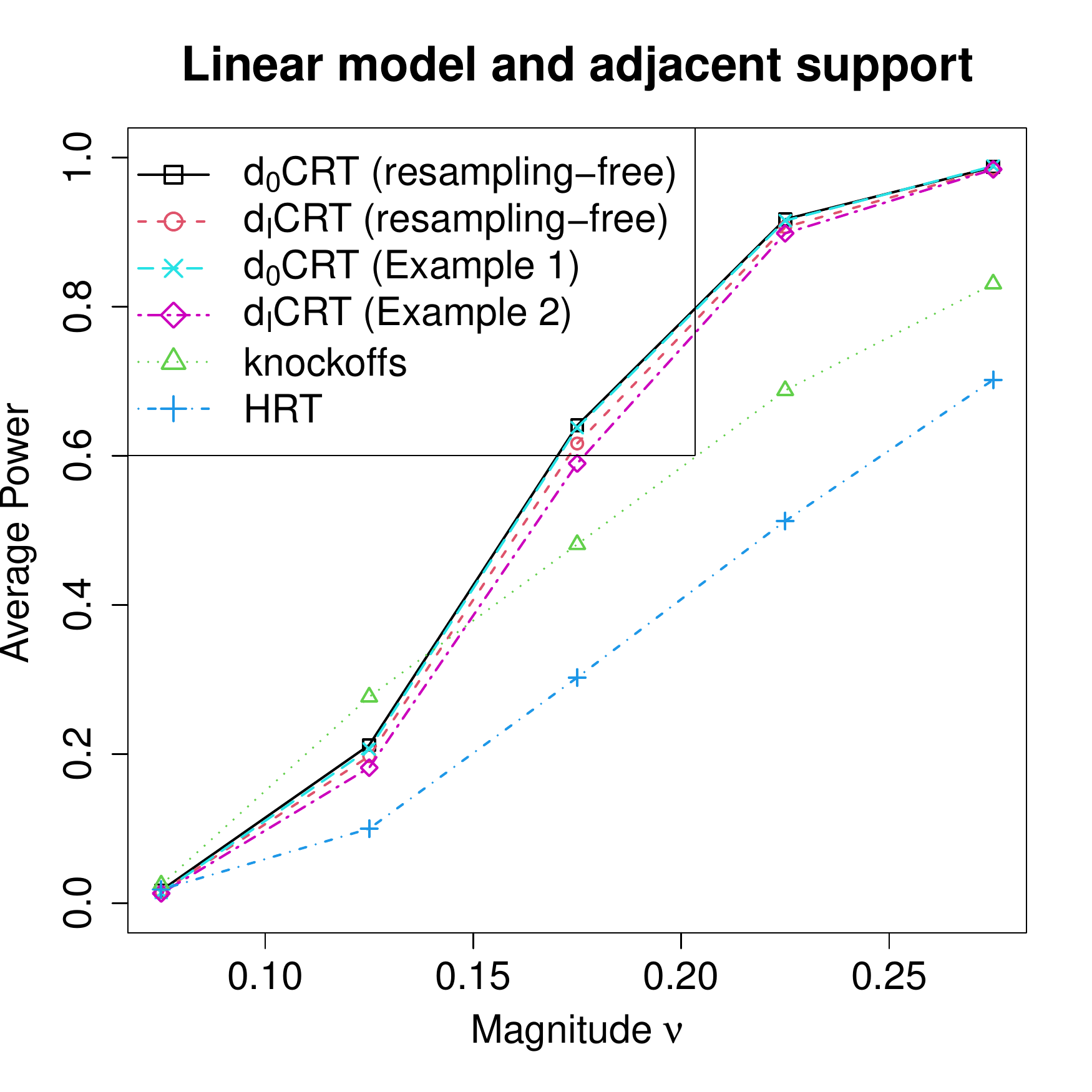}
\includegraphics[width=0.4\textwidth]{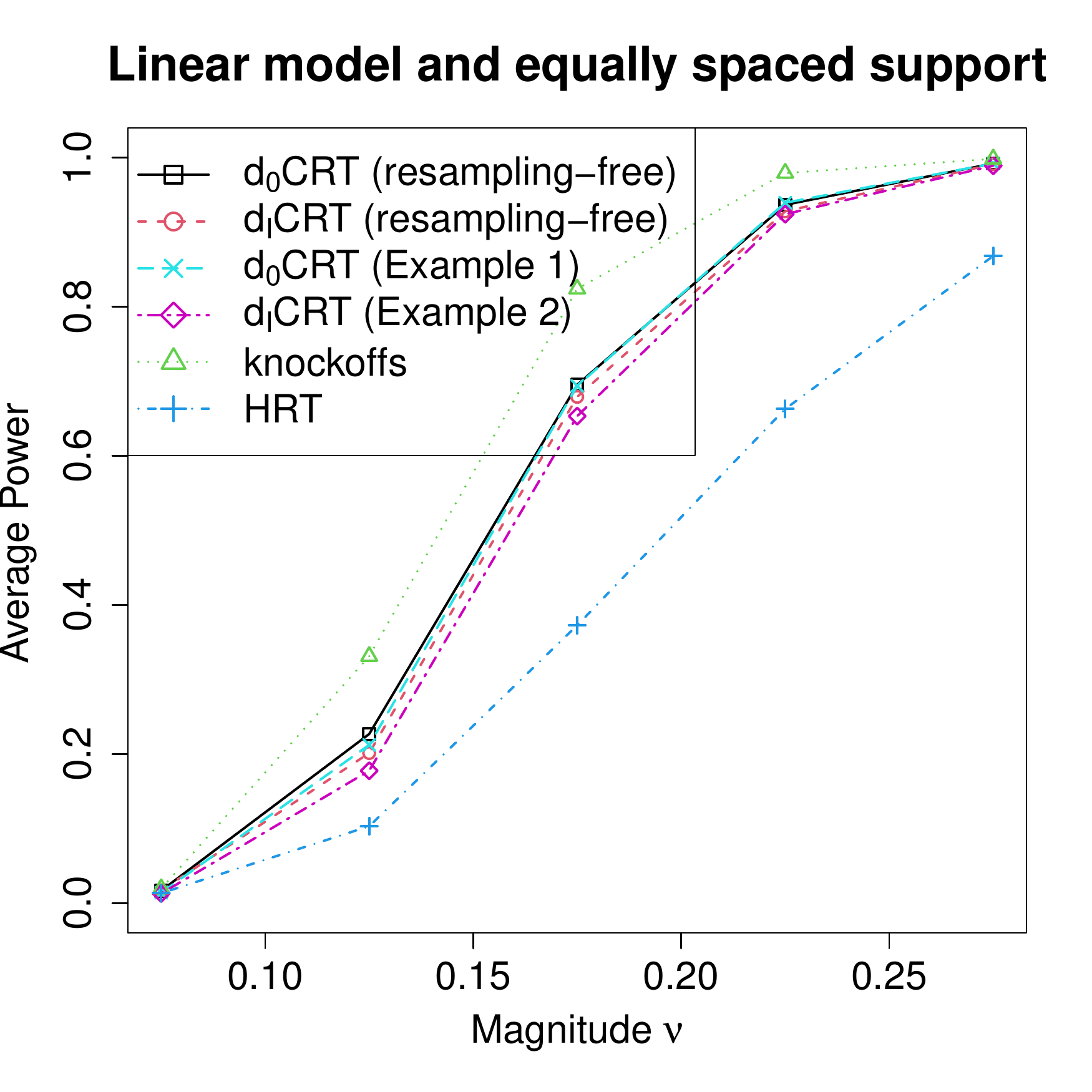}
 \includegraphics[width=0.4\textwidth]{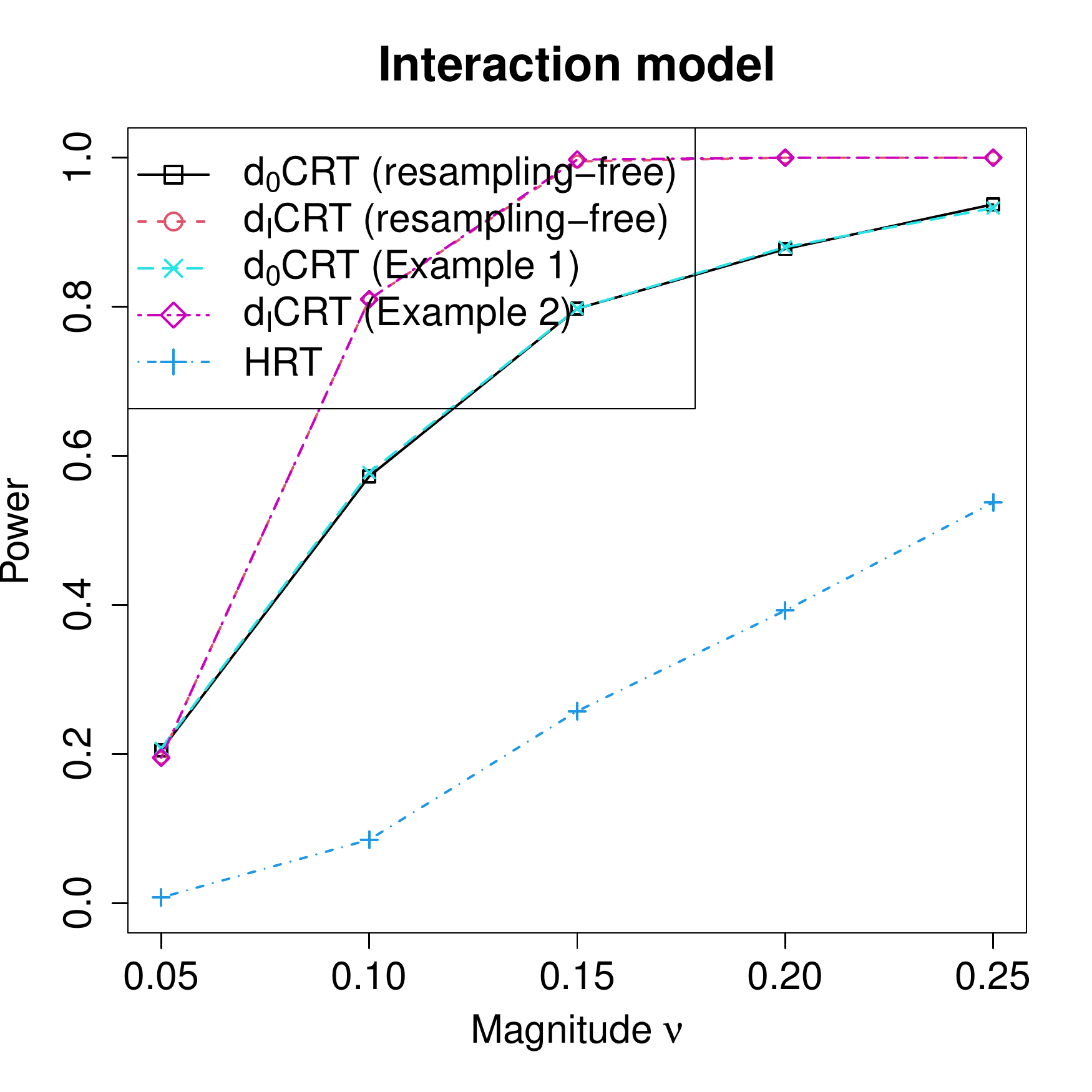}
\caption{\label{fig:resample:free:gauss} Powers of the simulation in Appendix~\ref{sim:rfgauss} measuring the effect of the resampling-free modification to the d$_0$CRT and d$_\mathrm{I}$CRT test statistics; all standard errors are below $0.03$. The resampling-free modifications of the dCRT have essentially the same power as the resampled versions.}
\end{figure}

We also include computation times for the baseline setting of Section~\ref{sec:sim:hrt} in Table~\ref{tab:com:resample}, showing that the resampling-free versions of the dCRTs confer a substantial computational savings.
\begin{table}[htb!]
\centering
\begin{tabular}{p{2.98cm}|p{2.1cm}|p{2.98cm}|p{2.1cm}|p{1.4cm}|p{1cm}}
\multicolumn{6}{c}{{\bf Average computation times (minutes)}}\\
\hline
      d$_0$CRT (resampling-free) &  d$_0$CRT\newline (Example 1) & d$_{\mathrm{I}}$CRT (resampling-free) &  d$_{\mathrm{I}}$CRT\newline (Example 2) & knockoff  & HRT  \\ \hline
     $9.8$  & $26.1$ & $10.2$ & $120.7$ & $1.8$ & $6.2$ \\ \hline
\end{tabular}
\caption{\label{tab:com:resample} Average computation times of the linear model simulations of Appendix~\ref{sim:rfgauss}.
The resampling-free modifications of the dCRT lead to considerable runtime savings.
}
\end{table}

\subsection{Measuring the effect of the resampling-free modification: Non-Gaussian covariates}\label{sim:rfnongauss}
As we introduced in Section~\ref{sec:main:mcf} and detailed in Appendix~\ref{sec:app:mcf}, when $X\mid Z$ is non-Gaussian, it must be transformed to Gaussian in order to apply the resampling-free speedup; we examine here the effect this transformation has on power. We generate covariates i.i.d. from two different distributions: (i) Gamma with shape $3$ and rate $0.5$ and (ii) Bernoulli with mean $0.5$. We took $n=p=800$, $s=50$ and $Y$ generated from linear (in the untransformed covariates) model and performed multiple testing for variable selection at false discovery rate level 0.1. Our main goal is to compare the d$_0$CRT of Example~\ref{ex:1} and the d$_\mathrm{I}$CRT of Example~\ref{ex:2} with their respective resampling-free counterparts, though we also run the HRT and knockoffs. The resulting average powers versus signal strength $\nu$ are shown in Figure~\ref{fig:nmc}. For Gamma $X$, the Gaussian transformation comes with almost no loss in power while for Bernoulli $X$, the resampling-free dCRTs lose substantial power but still outperform the HRT. This is due to the highly non-Gaussian nature of a Bernoulli(0.5) distribution and the need for substantial exogenous randomness to be added to $X$ to make it Gaussian. Knockoffs performs competitively with the dCRT methods in both simulations, and we attribute this to the covariate independence which allows very high-quality knockoffs to be used.

\begin{figure}[htpb]
\centering
  \includegraphics[width=0.4\textwidth]{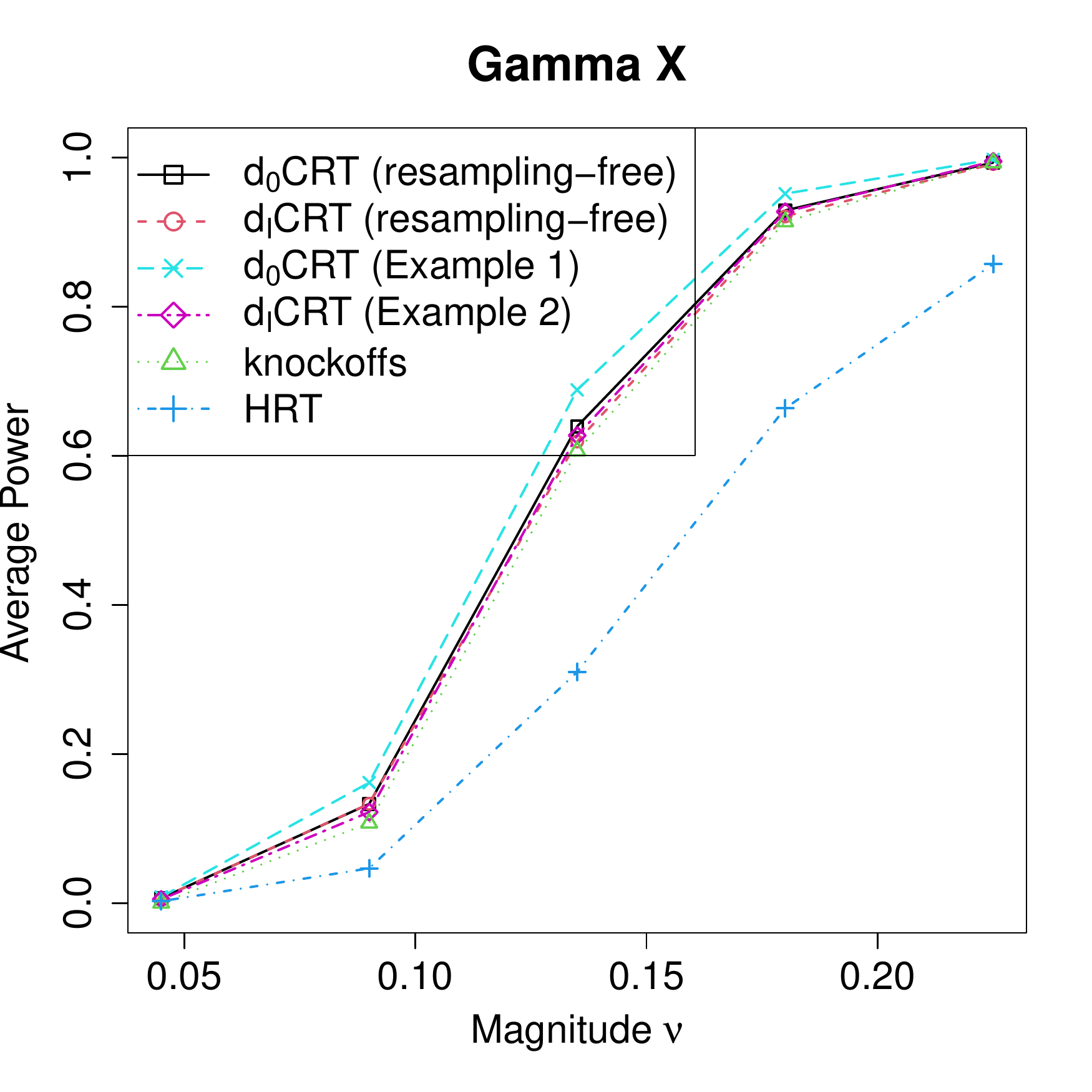}
% \end{subfigure}
% \begin{subfigure}[t]{0.48\textwidth}
% \centering
  \includegraphics[width=0.4\textwidth]{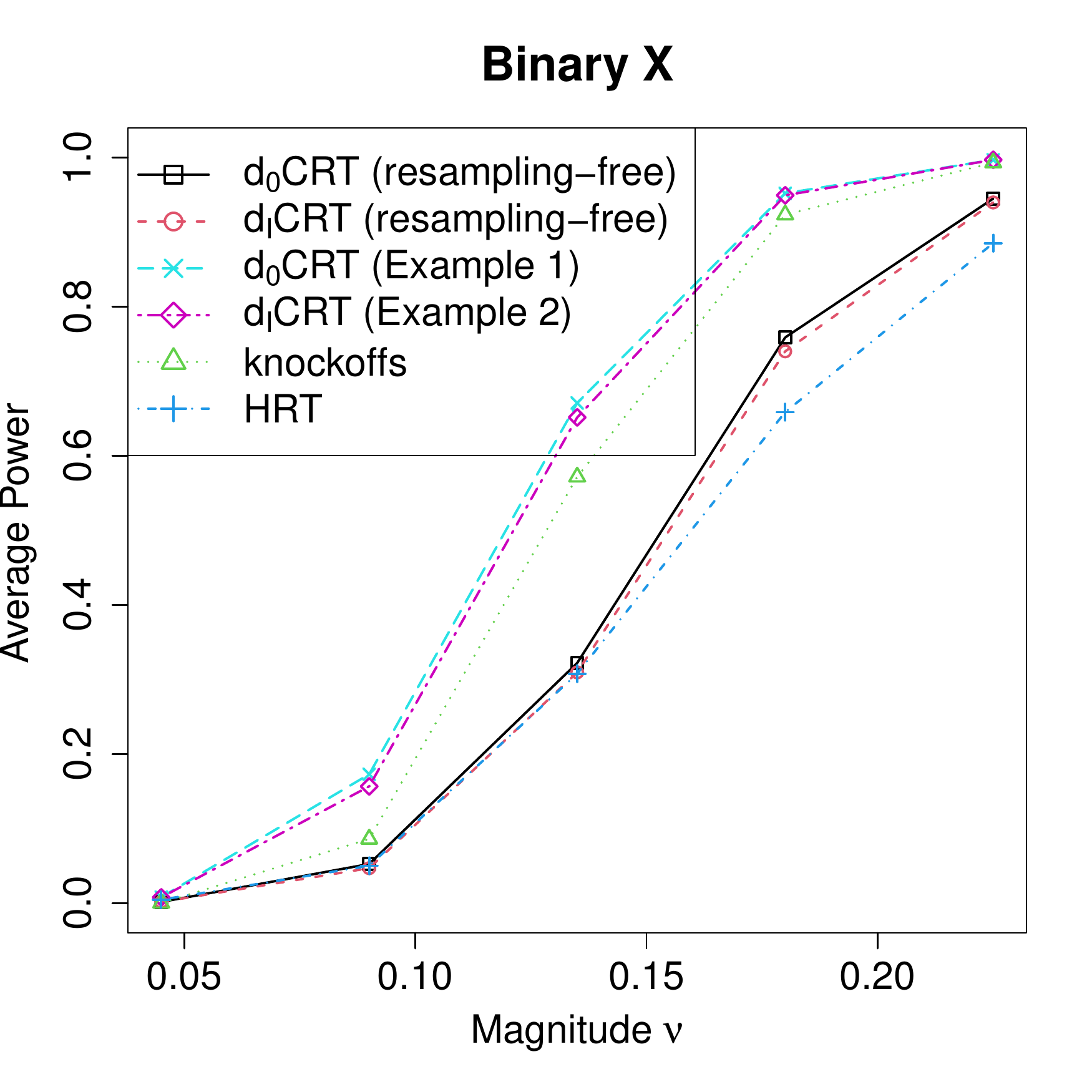}
\caption{\label{fig:nmc} Powers of the simulation in Appendix~\ref{sim:rfnongauss} measuring the effect of the Gaussian transformation in the resampling-free dCRTs. All standard errors are below $0.01$. The dCRT resampling-based approaches have the most power in non-Gaussian settings as well, but the resampling-free modifications have lower power in the binary case.}
\end{figure}

\subsection{Impact of screening on computation efficiency and power}\label{sec:sim:screening}
Here we demonstrate the effect of the screening modification introduced in Section~\ref{sec:screening} on computation time and power. We again simulate the baseline setting in Section~\ref{sec:sim:hrt} and compare the power and computation time of the dCRT methods with screening with the dCRT procedures without using screening. In Table~\ref{tab:com:screening} we present computation times demonstrating that screening can substantial improve the computational efficiency of dCRT. And the corresponding average powers are shown in Figure~\ref{fig:power:screening}, demonstrating that screening has nearly no impact on the power of the d$_0$CRT or d$_\mathrm{I}$CRT.

\begin{table}[htb!]
\centering
\begin{tabular}{c|c|c|c|c|c}
\multicolumn{6}{c}{{\bf Average computation times (minutes)}}\\
\hline
      \mbox{d$_0$CRT (screening)} &  d$_0$CRT (full) & d$_{\mathrm{I}}$CRT (screening) &  d$_{\mathrm{I}}$CRT (full) & knockoff  & HRT  \\ \hline
     $9.8$  & $46.1$ & $10.2$ & $47.1$ & $1.8$ & $6.2$ \\ \hline
\end{tabular}
\caption{\label{tab:com:screening} Average computation times (in minutes) of the simulations in Appendix~\ref{sec:sim:screening}. Screening leads to large computational savings at nearly no cost to the statistical power.}
\end{table}

\begin{figure}[htpb]
\centering
 \includegraphics[width=0.4\textwidth]{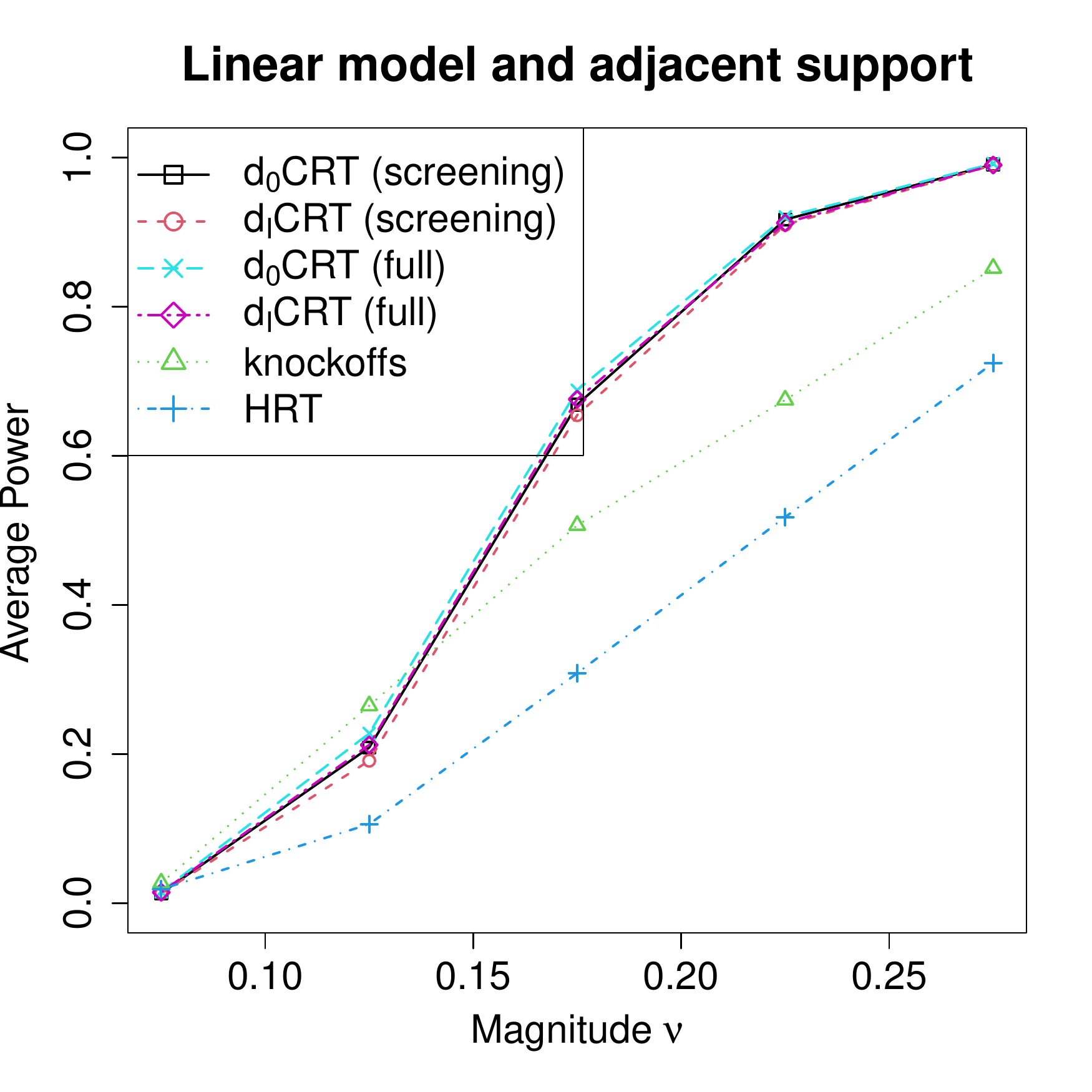}
\includegraphics[width=0.4\textwidth]{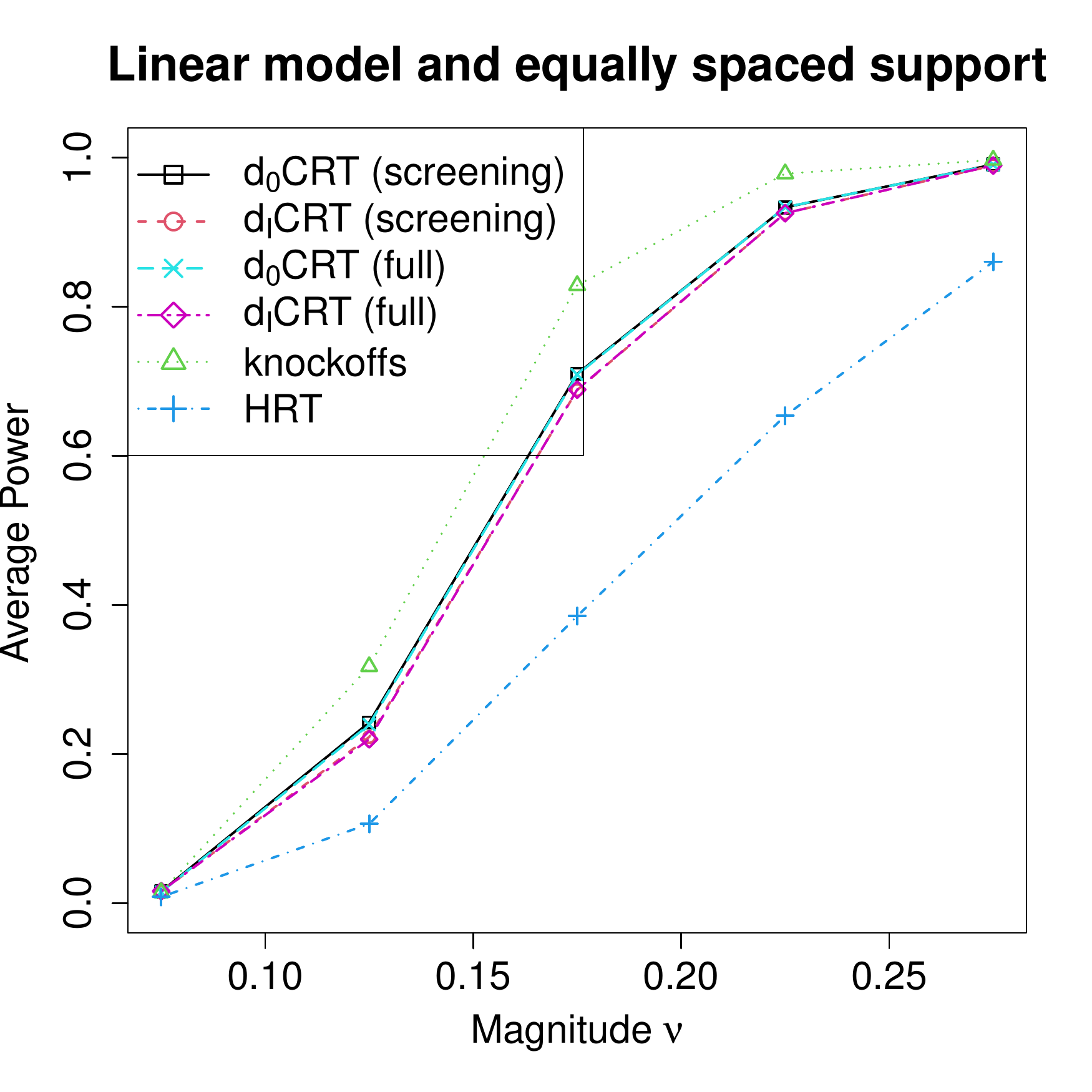}
\caption{\label{fig:power:screening} Average powers of the simulation in Appendix~\ref{sec:sim:screening} measuring the effect of the screening modification. All standard errors are below $0.01$. The screened dCRTs have identical power to the dCRTs without screening, as expected.}
\end{figure}

\subsection{Additional false discovery rate results}\label{sec:sim:fdr}
We compile false discovery rate results of our simulations with well-specified covariate distributions here. The false discovery rate is guaranteed to be controlled by knockoffs and the $p$-values of the CRT procedures including the dCRTs are guaranteed to be valid, but they do not satisfy the conditions for the Benjamini--Hochberg procedure to control the false discovery rate. In practice, they do control FDR, as Figures~\ref{fig:fdr:smc}--\ref{fig:diff:fdr:design} show. This adds further support to a widely acknowledged empirical observation that the BH procedure rarely (if ever) violates FDR control in practice, outside of truly adversarial simulation setups with specially designed dependence structures.

\begin{figure}[htpb!]
\centering
  \includegraphics[width=0.4\textwidth]{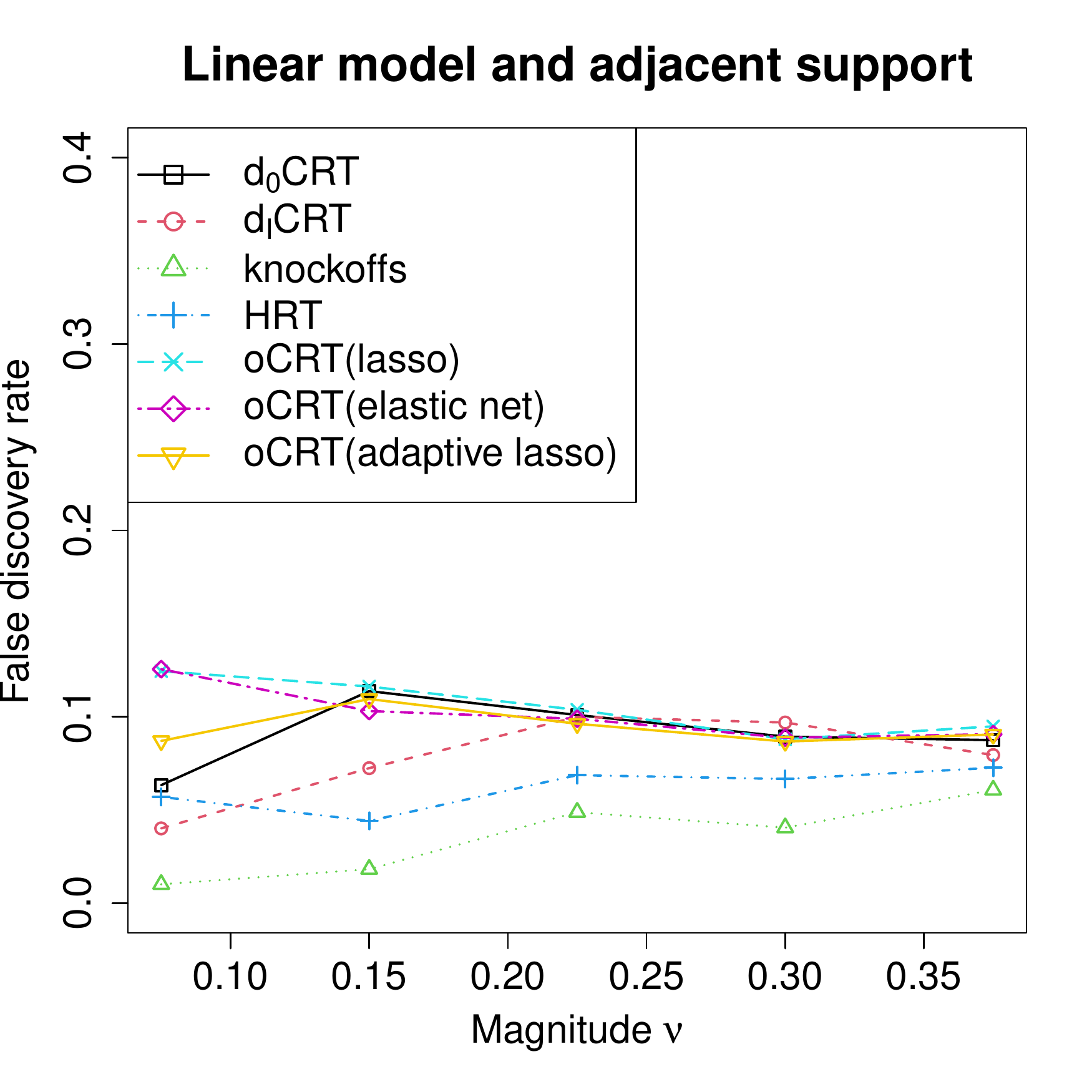}
  \includegraphics[width=0.4\textwidth]{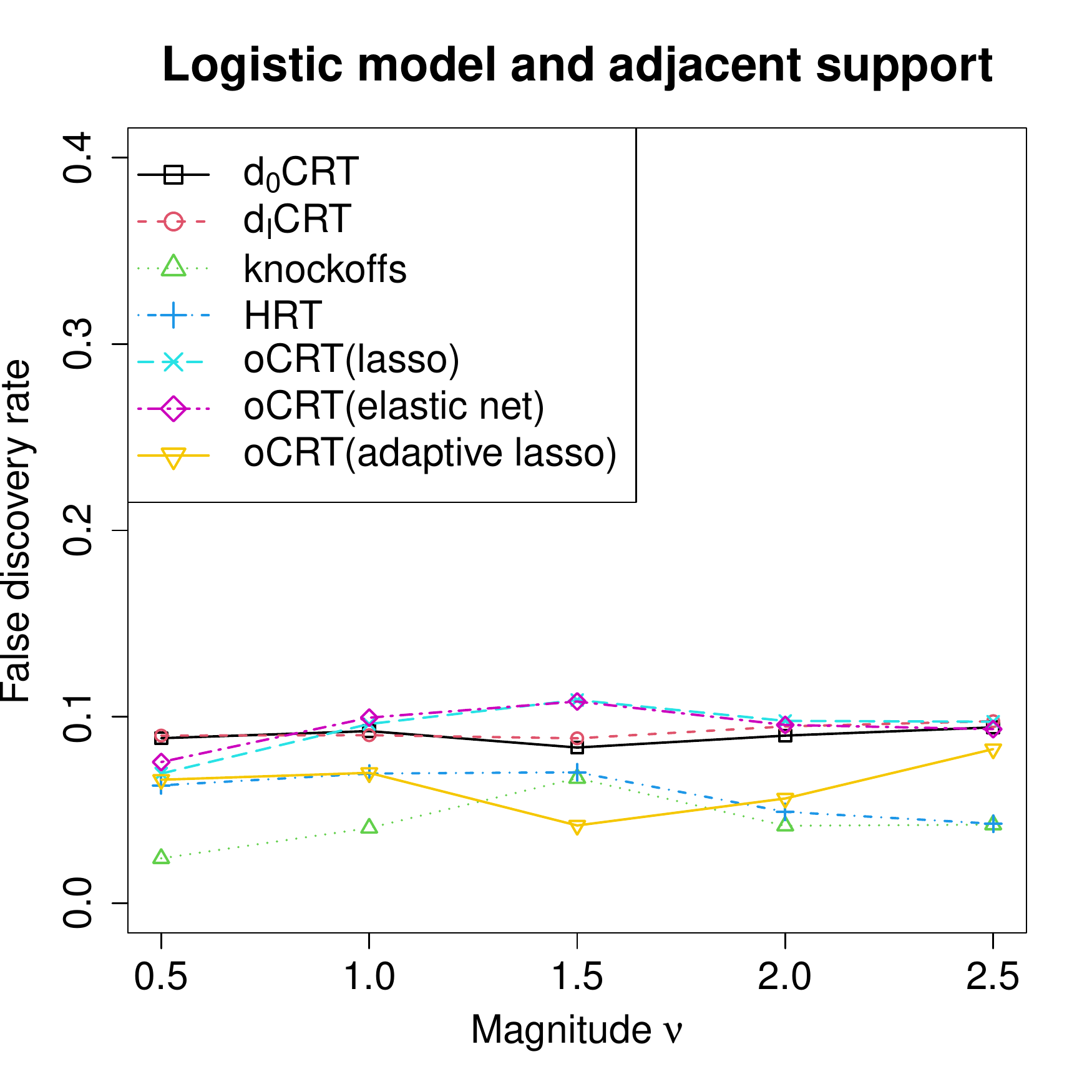}
  \includegraphics[width=0.4\textwidth]{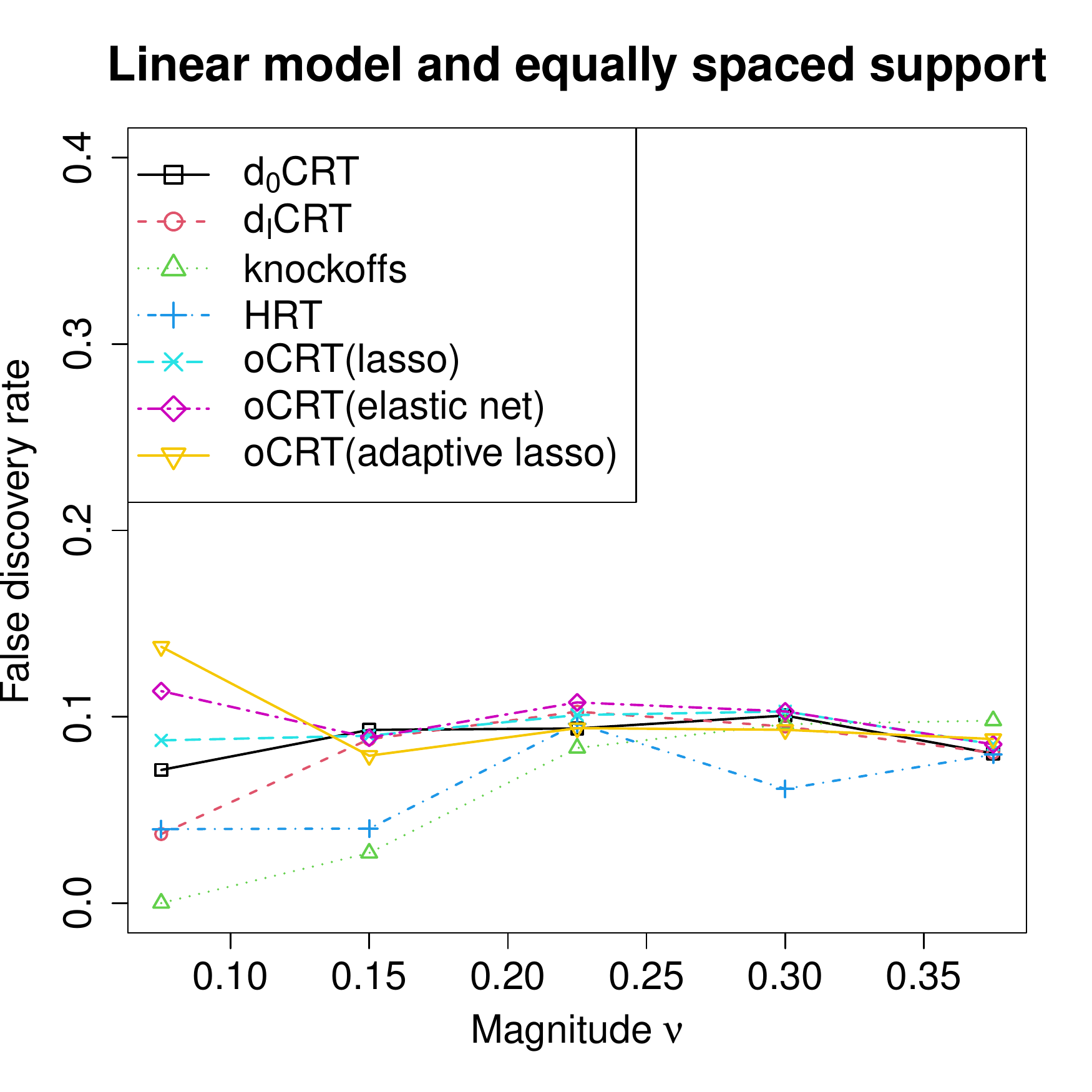}
  \includegraphics[width=0.4\textwidth]{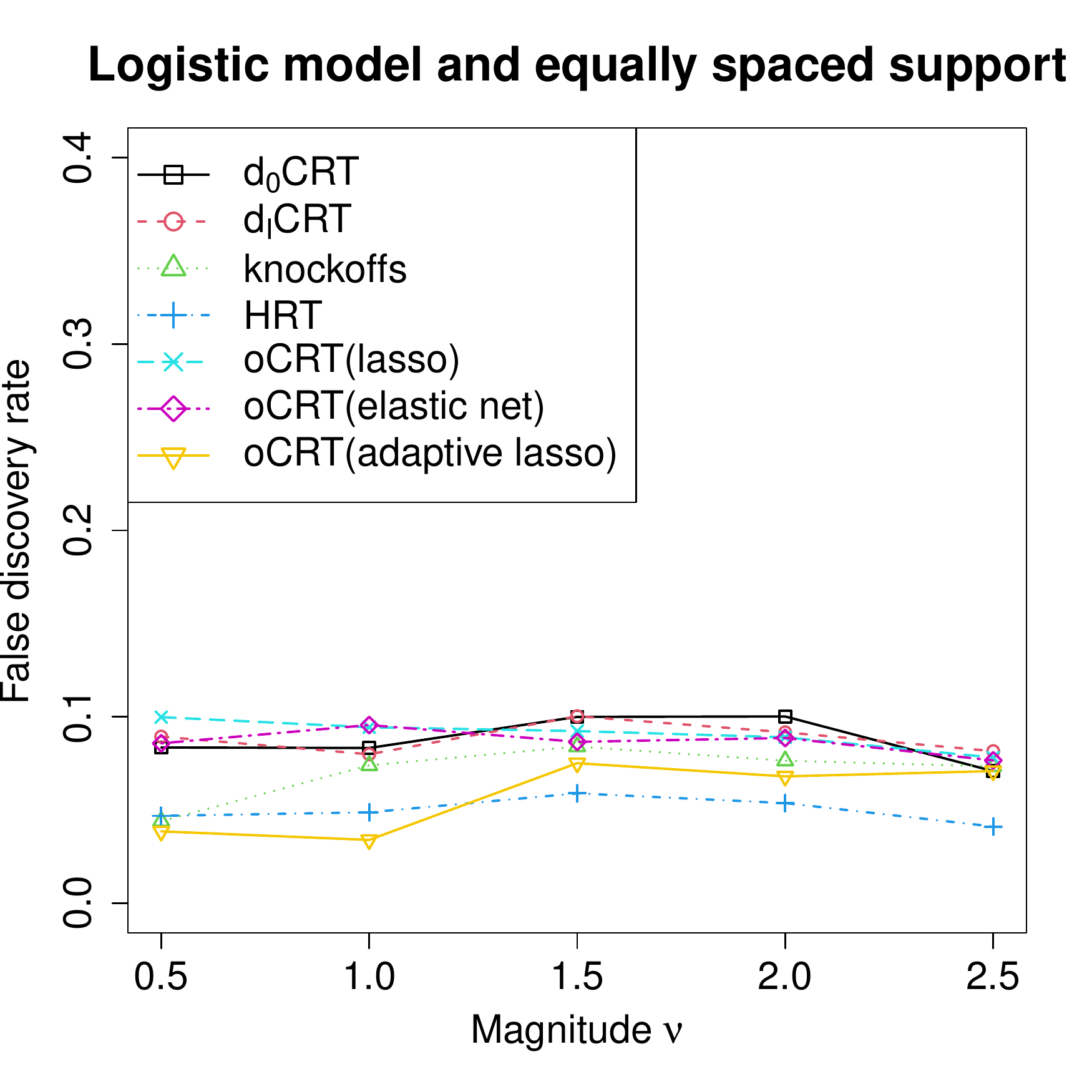}
\caption{\label{fig:fdr:smc} False discovery rates of the $n=p=300$ simulation of Appendix~\ref{sec:sim:smc}; standard errors are below 0.01. All methods control the false discovery rate at the target level 0.1, as desired.}
\end{figure}

\begin{figure}[htpb!]
\centering
  \includegraphics[width=0.4\textwidth]{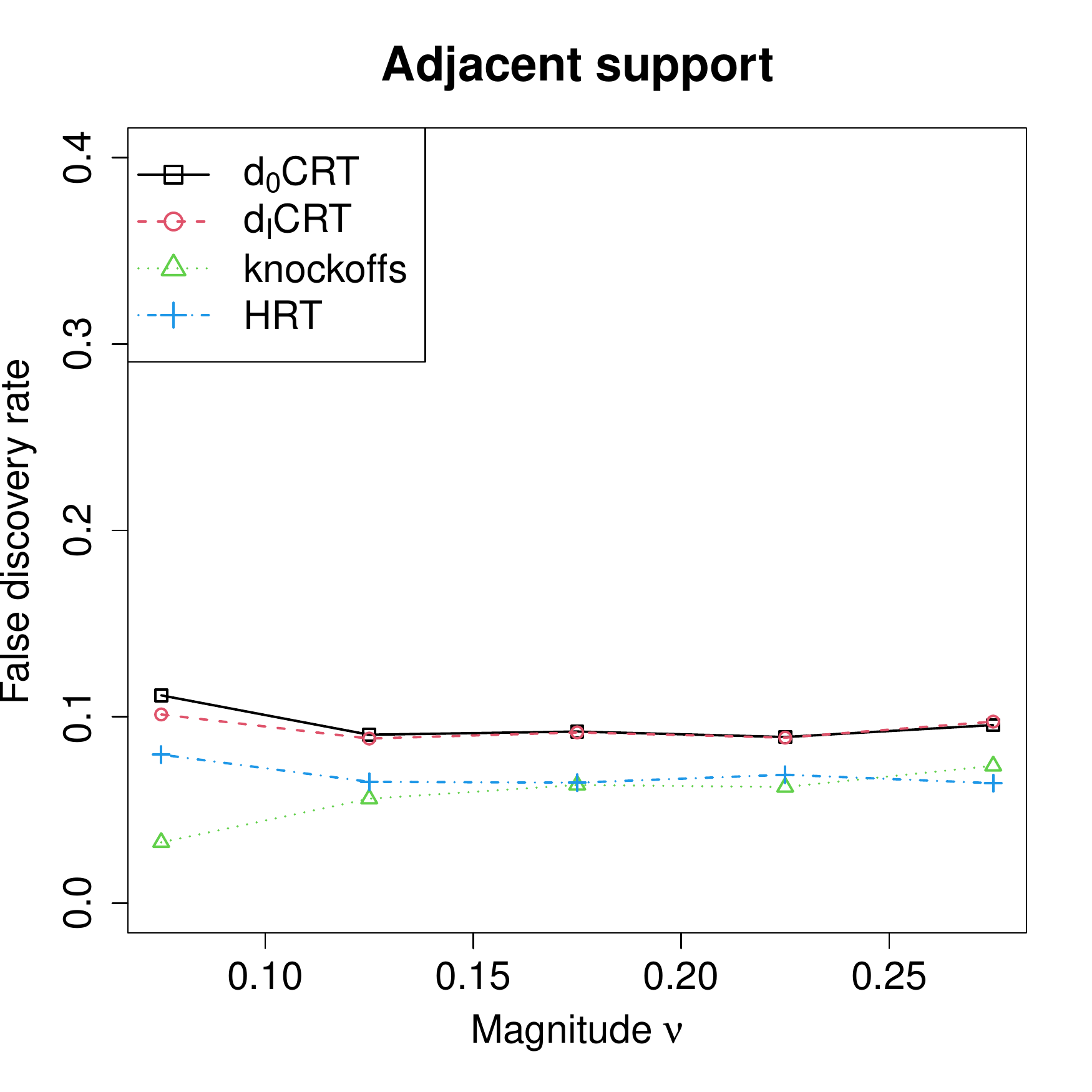}
  \includegraphics[width=0.4\textwidth]{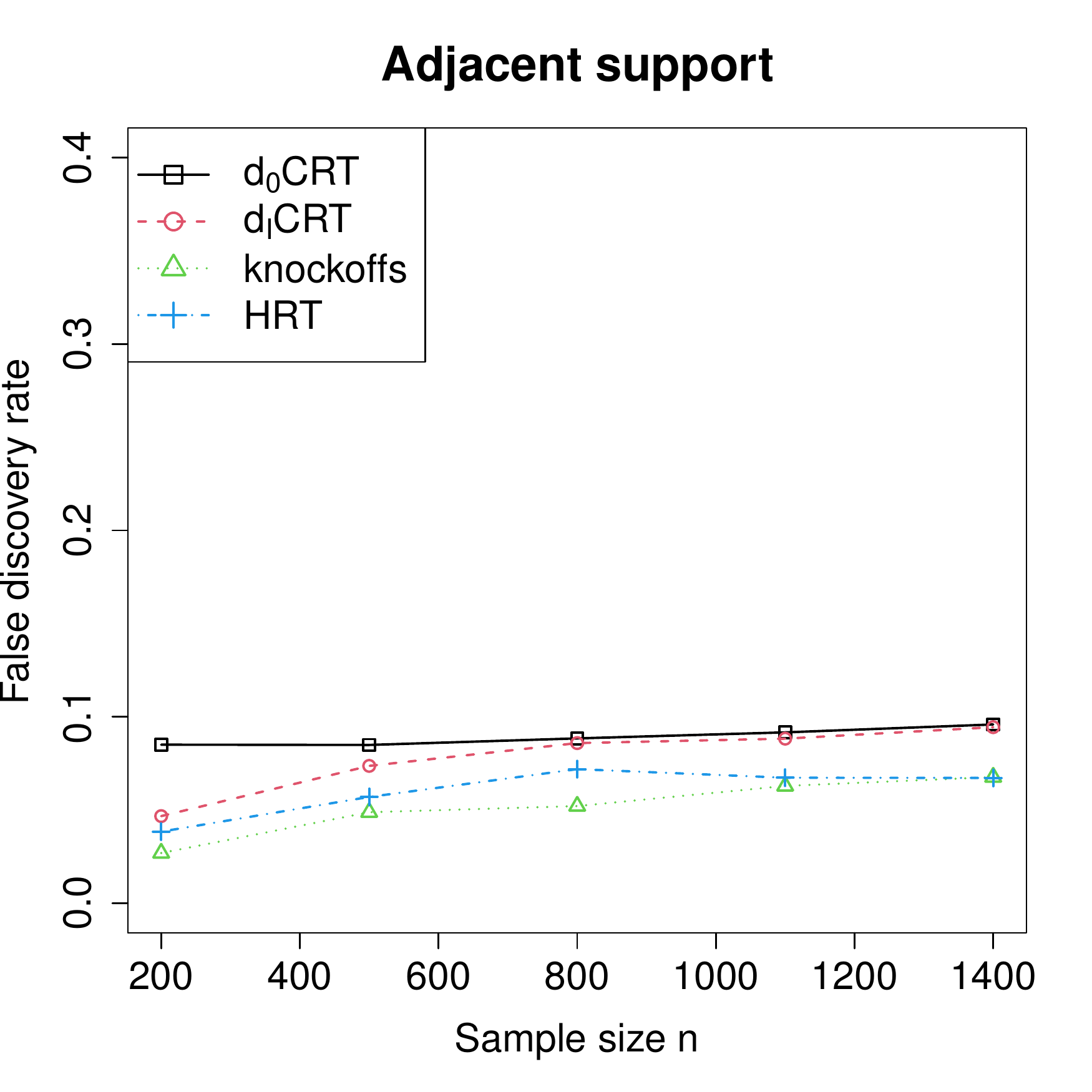}
  \includegraphics[width=0.4\textwidth]{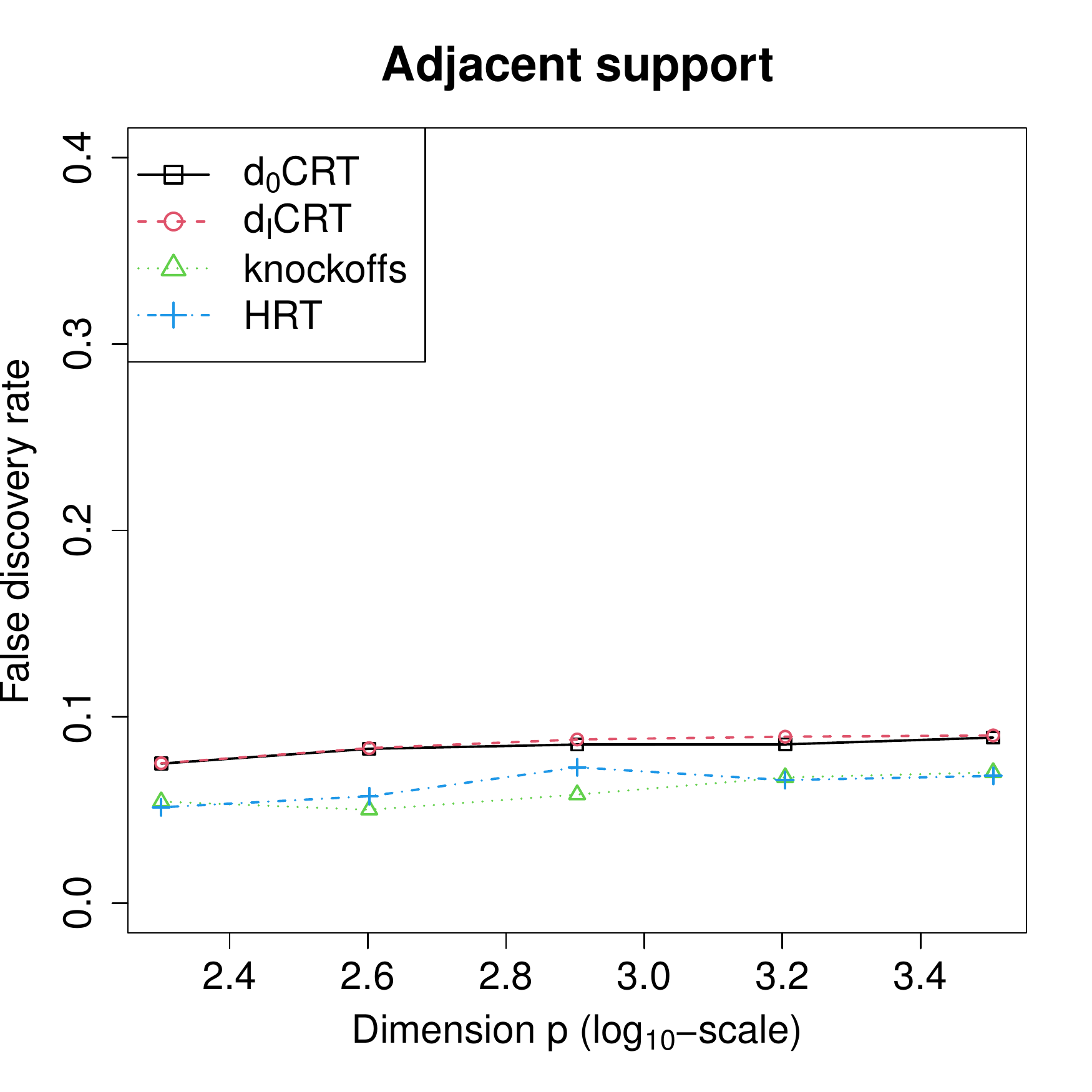}
  \includegraphics[width=0.4\textwidth]{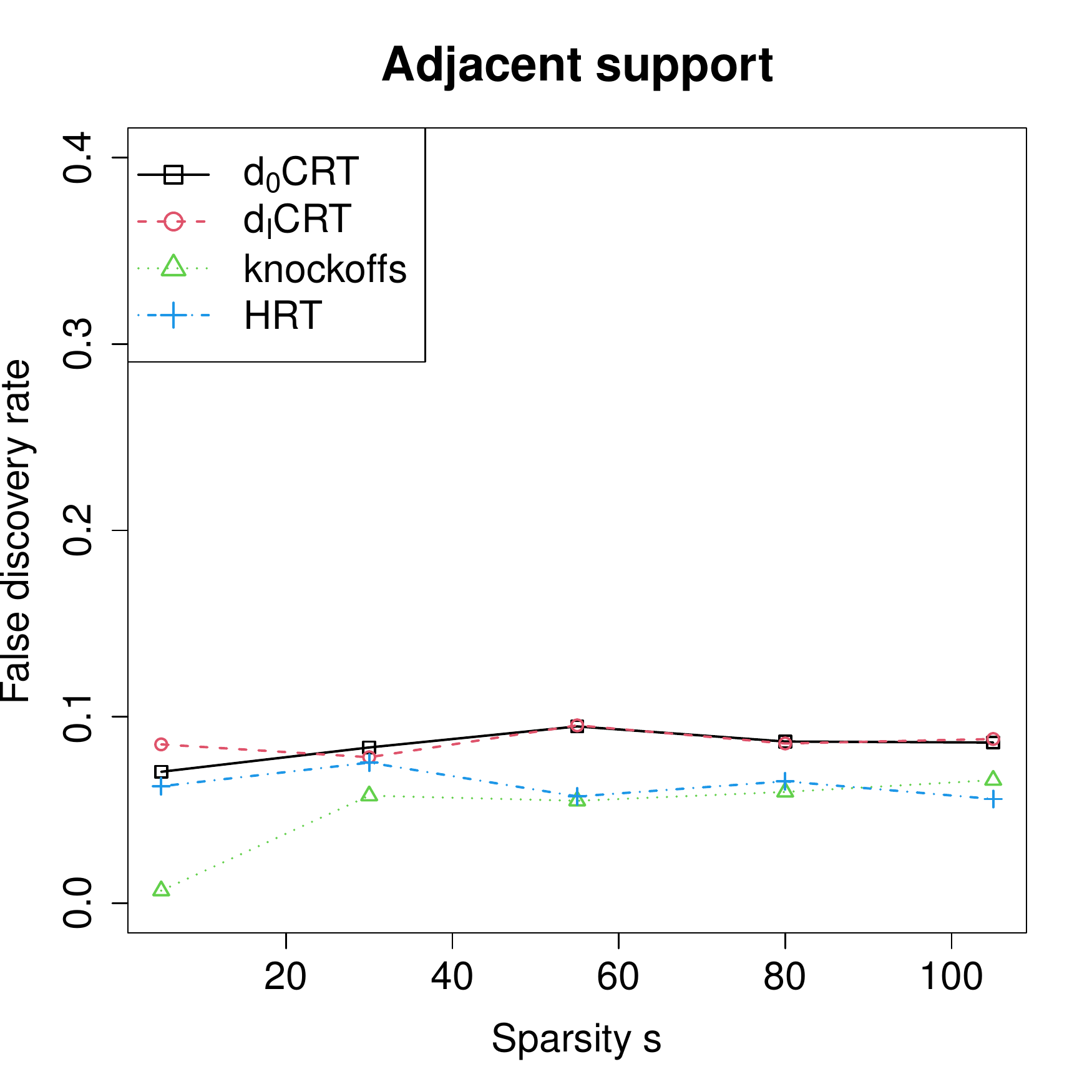}
\caption{\label{fig:fdr:diffnp} False discovery rates of the large scale simulations of Appendix~\ref{sec:sim:hrt} that vary the coefficient magnitude, sample size, dimension, and coefficient sparsity with adjacent support. All standard errors are below 0.01; all methods control false discovery rate across all settings.}
\end{figure}

\begin{figure}[htpb]
\centering
  \includegraphics[width=0.4\textwidth]{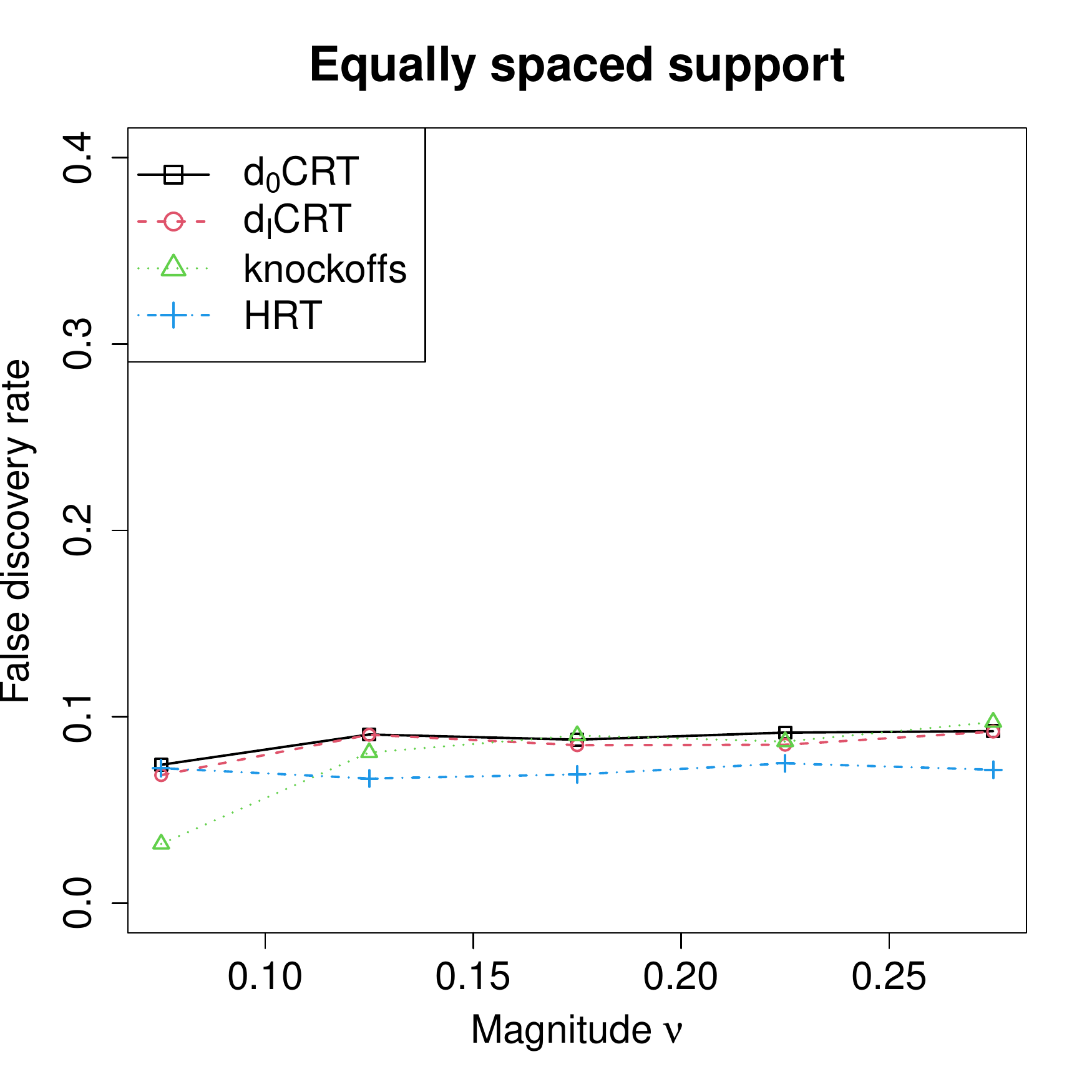}
  \includegraphics[width=0.4\textwidth]{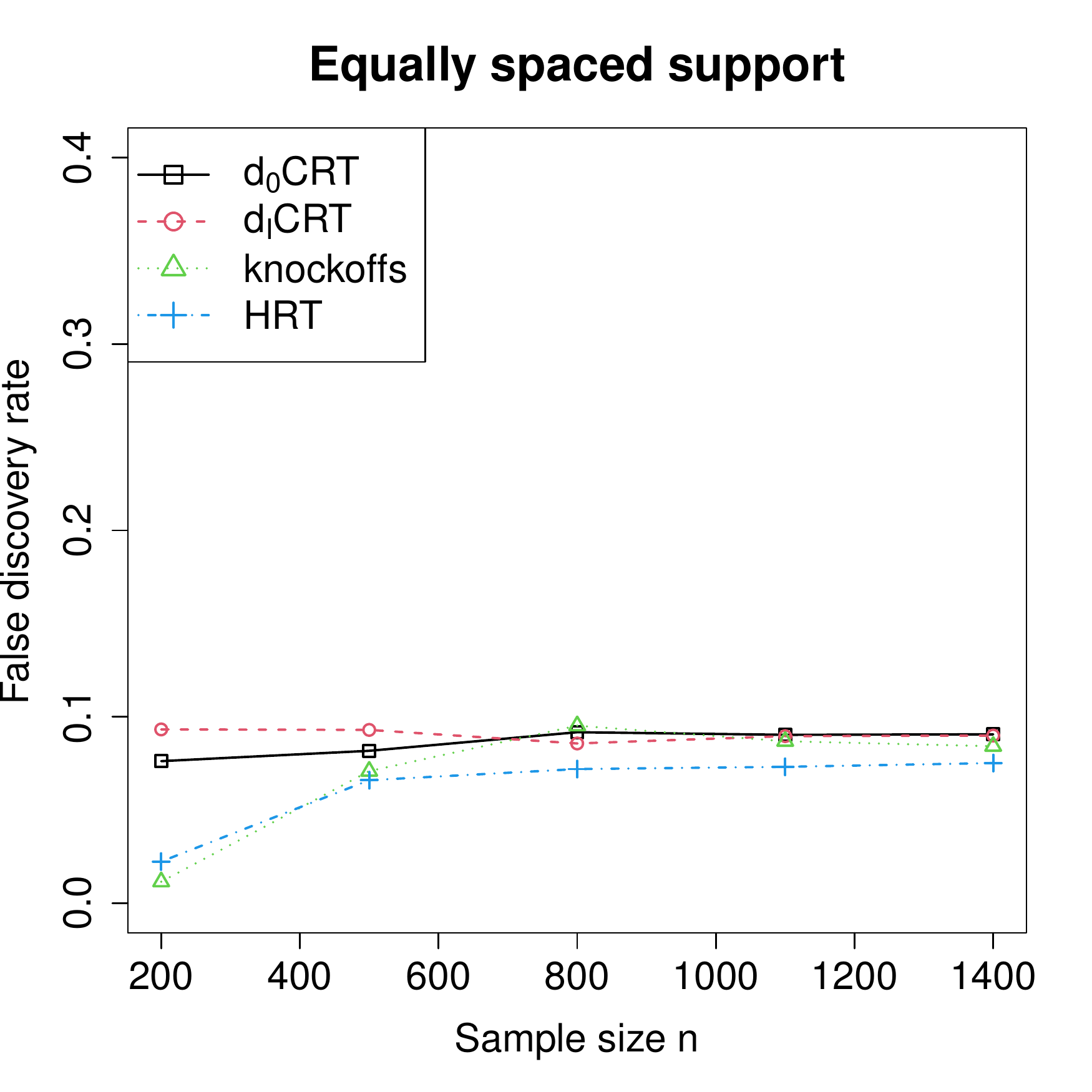}
  \includegraphics[width=0.4\textwidth]{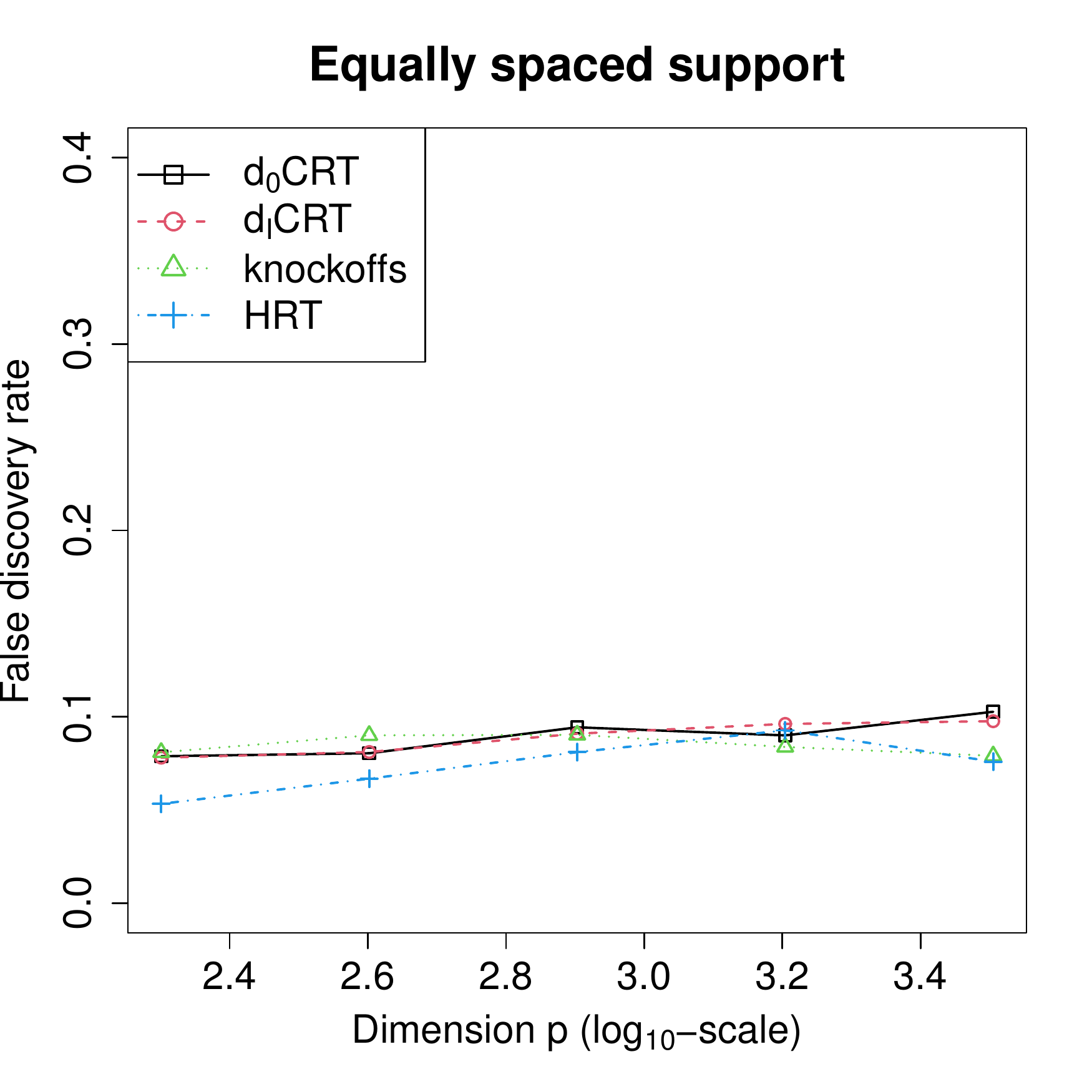}
  \includegraphics[width=0.4\textwidth]{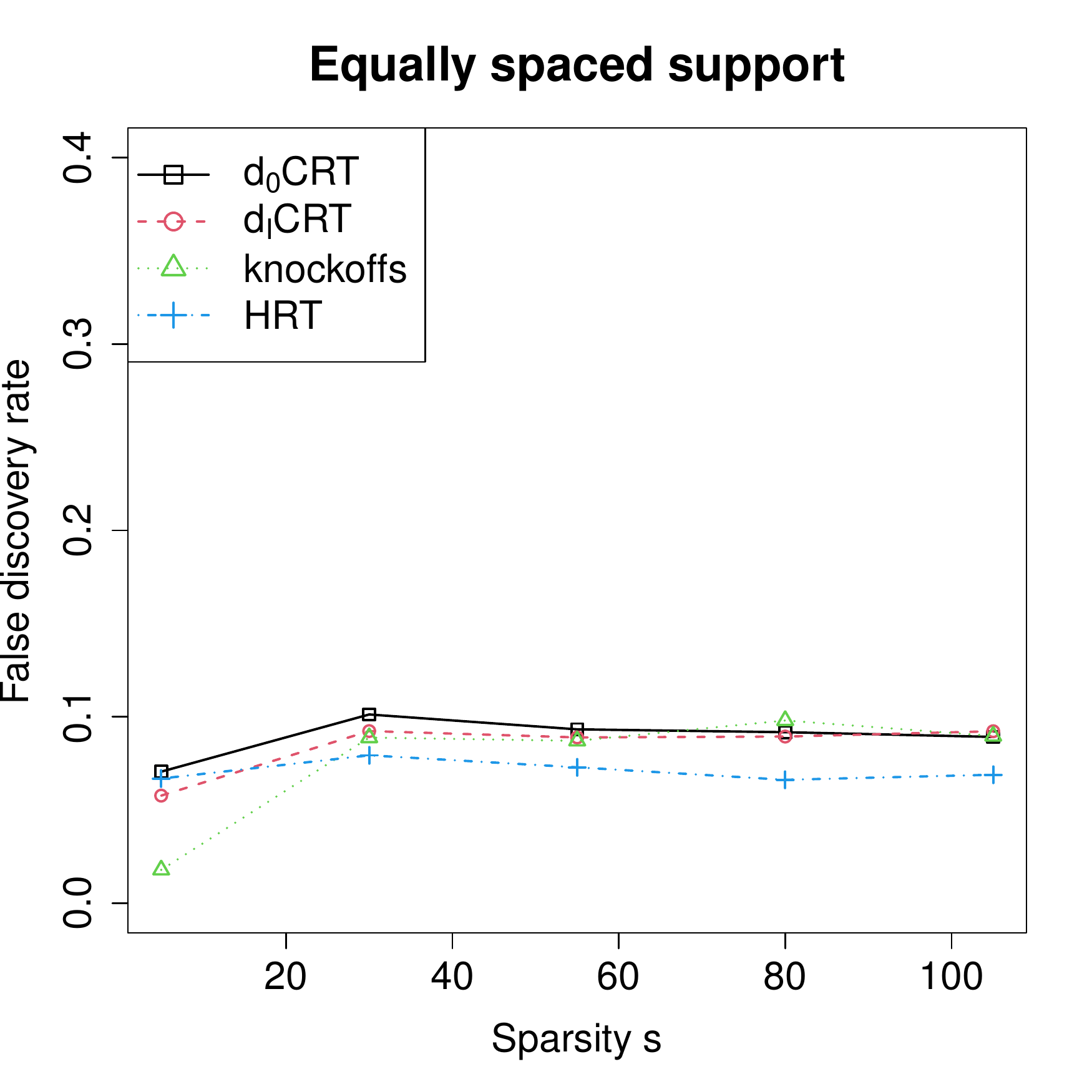}
\caption{\label{fig:fdr:diffnp:equal} False discovery rates of the large scale simulations of Appendix~\ref{sec:sim:hrt} that vary the coefficient magnitude, sample size, dimension, and coefficient sparsity with equally spaced support. All standard errors are below 0.01; all methods control the false discovery rate in all settings.}
\end{figure}

\begin{figure}[htbp]
\centering
  \includegraphics[width=0.4\textwidth]{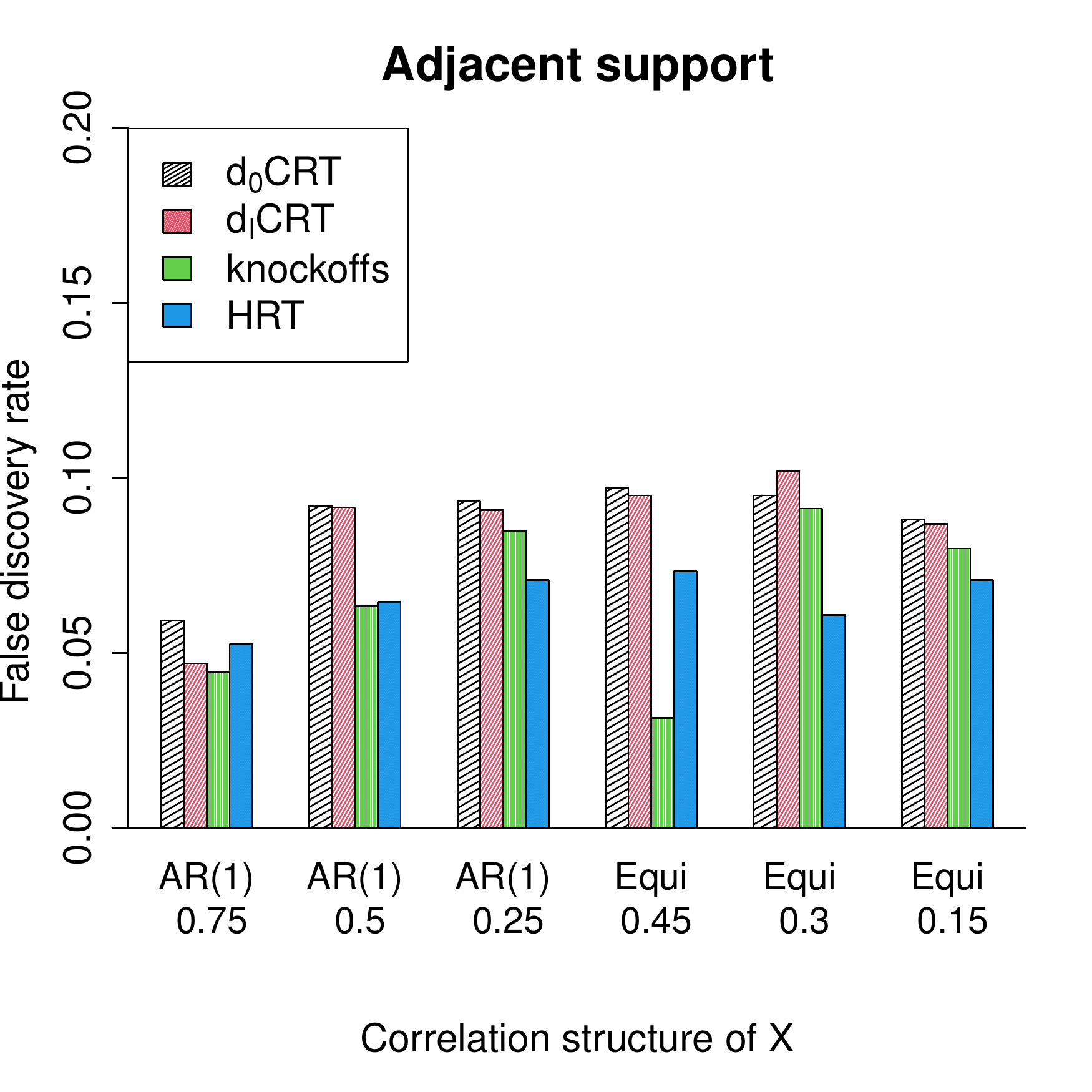}
  \includegraphics[width=0.4\textwidth]{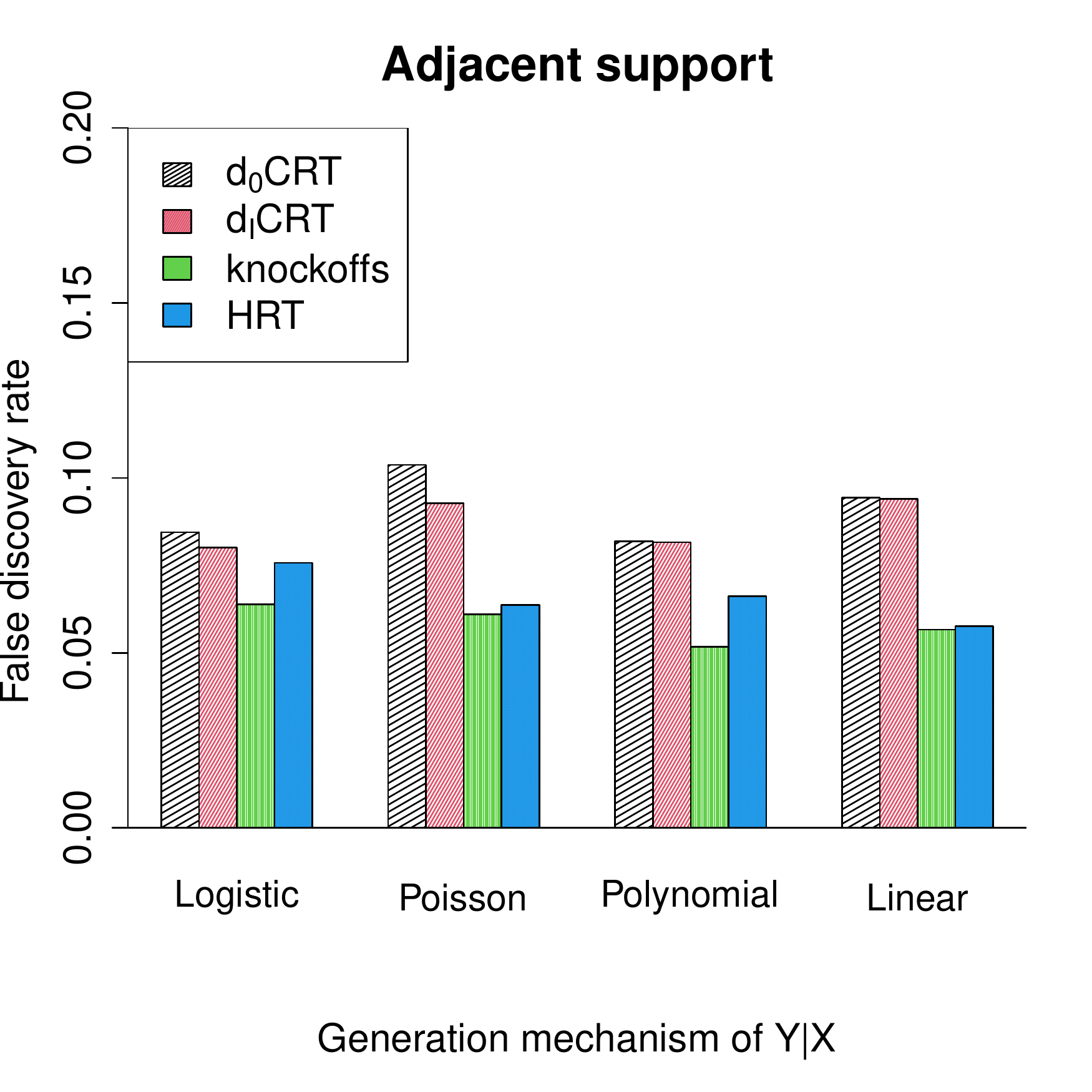}
  \includegraphics[width=0.4\textwidth]{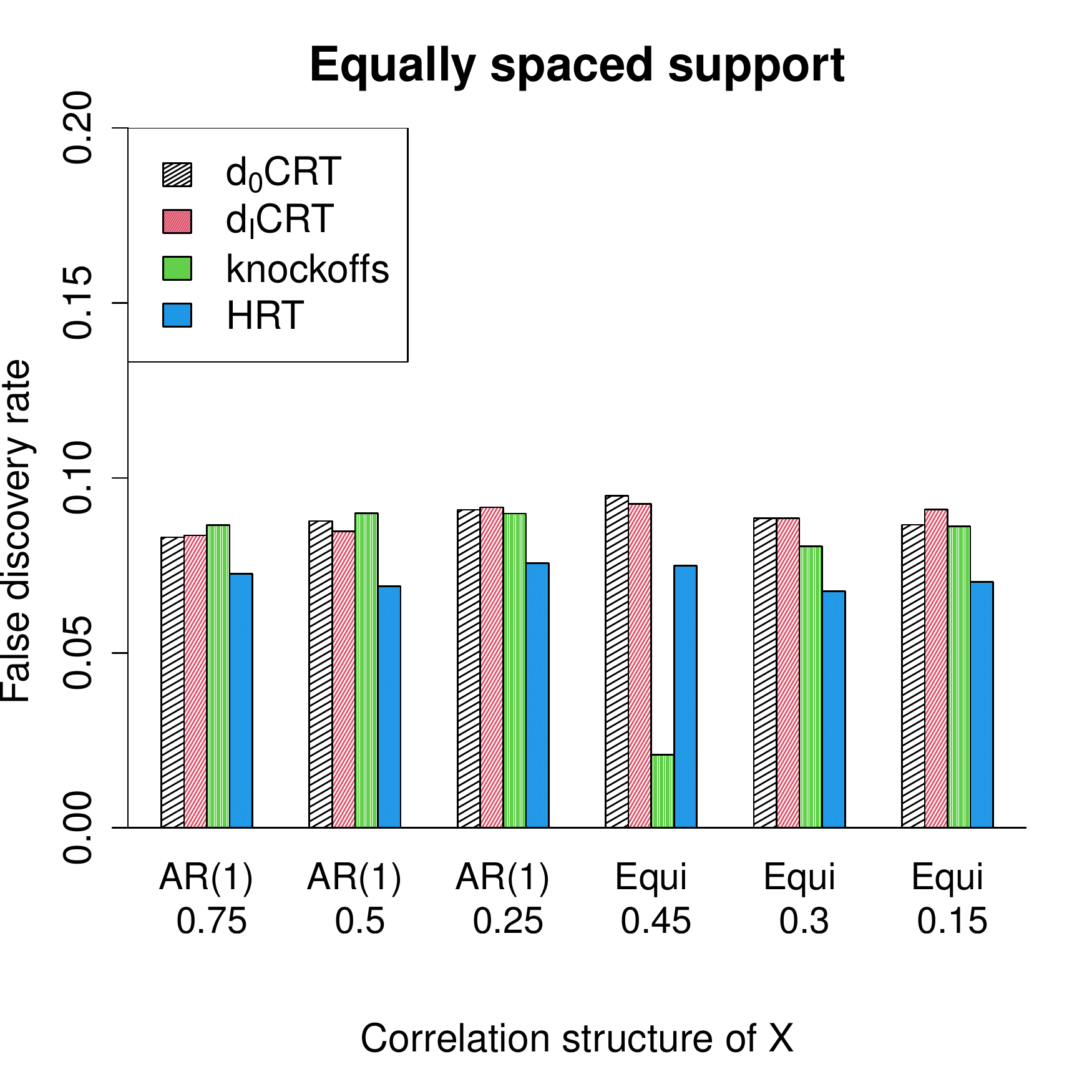}
  \includegraphics[width=0.4\textwidth]{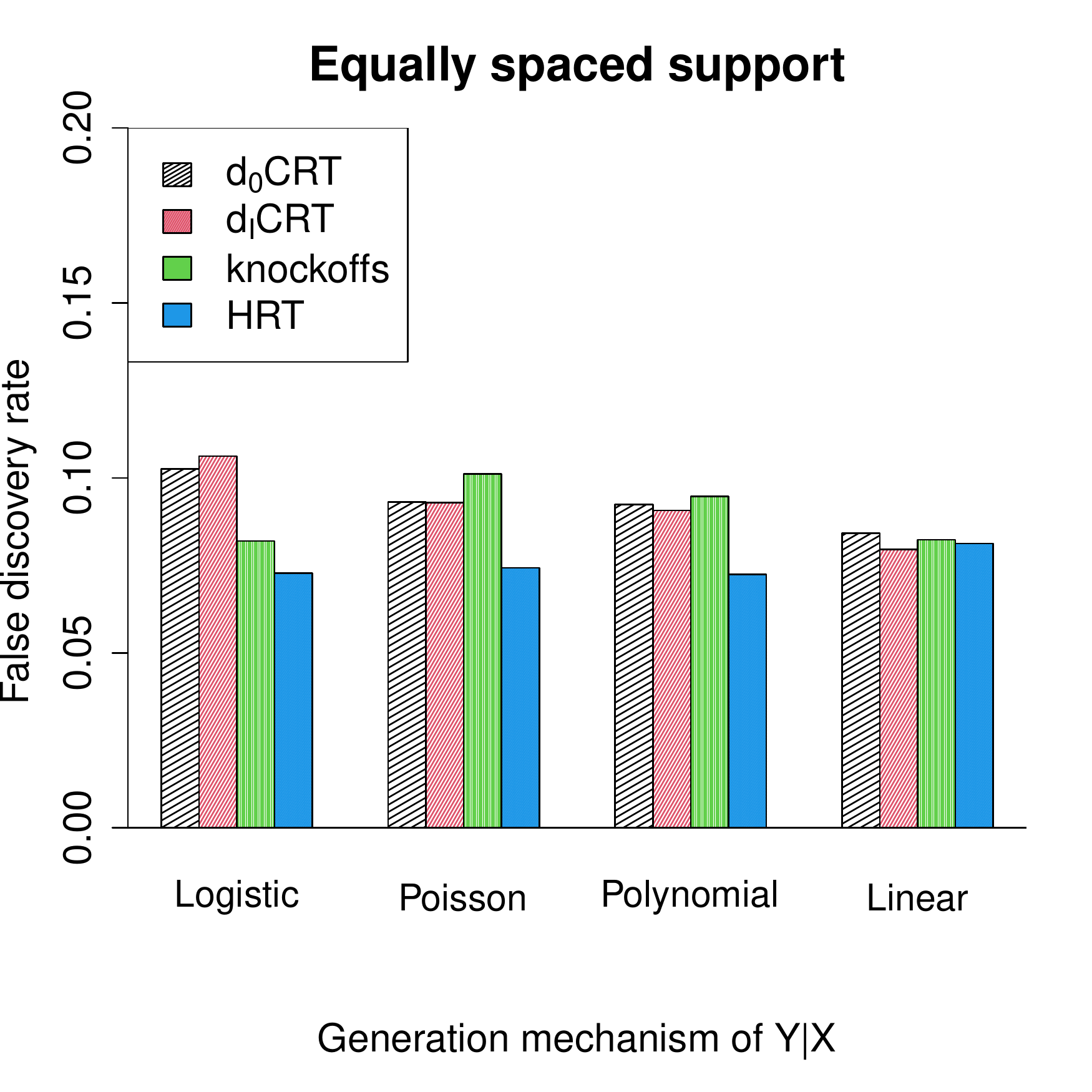}
\caption{\label{fig:diff:fdr:design} False discovery rates of the large scale simulations of Appendix~\ref{sec:sim:hrt} that vary the covariate and response models; all standard errors are below 0.01. All methods control false discovery rate in all settings.}
\end{figure}

\section{Breast cancer data analysis}\label{app:data}
In this section, we present the details of our analysis of the breast cancer data set in Section~\ref{sec:realdata}; our code for pre-processing and analyzing the data is available at \url{https://github.com/moleibobliu/Distillation-CRT}. The list of $p=164$ candidate genes is obtained as the set of genes among those measured by \cite{curtis2012genomic} that are the most frequently mutated according to \cite{pereira2016somatic} (see Supplementary Data 1 at \url{https://www.nature.com/articles/ncomms11479#Sec32}). The CNA, gene expression and clinical data itself is from {\sf cBioPortal} and can be downloaded from \url{https://www.cbioportal.org/study/summary?id=brca_metabric}. The raw cancer stage used in our response variable is from the column labeled $TUMOR\_STAGE$ in their table for clinical data. It consists of three categories, $1$, $2$ and $3$ that represent the progression stage of breast cancer. Since there were relatively few observations in category $1$, we merge categories $1$ and $2$ together, resulting in a binary response. And the samples for analysis were chosen as all the patients with ER+ given in the clinical table.

Now we introduce the procedures for modeling the covariates. To model the expression levels of each gene $G_j$ conditional on its corresponding CNA level $C_j$, we follow the methods proposed and discussed in \cite{solvang2011linear,lahti2012cancer,leday2013modeling} to fit a piecewise linear regression of each $G_j$ on each $C_j$. Denoting the fitted residuals as $\widetilde{G}_j$ and $\widetilde{\mathbf{G}}=(\widetilde{G}_1,\ldots,\widetilde{G}_p)$ we then model $\widetilde{\mathbf{G}}$ as mean-zero multivariate Gaussian (similar to \cite{shen2019false}) and estimate its precision matrix via a similar procedure as in Section~\ref{app:sim:robustestmom}. That is, we remove the mean of each $\widetilde{G}_j$, fit graphical lasso tuned with cross-validation to estimate the precision matrix, and finally take the inverted precision matrix estimate and multiply it by a diagonal matrix to match the conditional variance of each $\widetilde{G}_j$ with the mean square of its residuals. 

As in the simulations, the d$_0$CRT and d$_{\mathrm{I}}$CRT we use are the resampling-free logistic regression versions of Examples~\ref{ex:1} and~\ref{ex:2} along with screening with the logistic lasso. Again, we do not use a logistic regression test statistic to allow for the computational gains of the resampling-free modification. We also implement knockoffs, the HRT, and the \rev{o}CRT as in the simulations section with analogous logistic lasso statistics. The number of resamples for the HRT and the \rev{o}CRT is set as $M=25,000$, again satisfying $M/5>p/\alpha$ as the false discovery rate or family-wise error rate level $\alpha$ is set as $0.1$. 

We summarize the discovered genes and their average $p$-values (over the 300 repetitions) estimated by each method in Table~\ref{tab:1}, {their frequencies of being detected in terms of false discovery rate control in Table~\ref{tab:2}, and their frequencies of being detected in terms of family-wise error rate control in Table~\ref{tab:3}.} 

\begin{table}[htbp]

\centering
\begin{tabular}{l|lllllll}
\hline
Gene   & d$_0$CRT    & d$_{\mathrm{I}}$CRT & \rev{o}CRT(lasso)     & HRT     \\ \hline
{\em FBXW7}  & ${1.1\times 10^{-3}}$ & ${5.8\times 10^{-4}}$   & ${2.0\times 10^{-3}}$ & $2.1\times 10^{-2}$ \\ 
{\em GPS2}   & ${2.5\times 10^{-4}}$ & ${2.2\times 10^{-4}}$   & ${3.5\times 10^{-4}}$ & $1.0\times 10^{-2}$ \\ 
{\em HRAS}   & ${1.6\times 10^{-3}}$ & ${1.7\times 10^{-3}}$   & $3.0\times 10^{-3}$ & $6.1\times 10^{-4}$ \\ 
{\em MAP3K13} & $1.1\times 10^{-4}$ & ${8.0\times 10^{-5}}$   &  $3.3\times 10^{-3}$ & $1$        \\ 
{\em NRAS} & $6.1\times 10^{-3}$ & $9.5\times 10^{-3}$   &  $6.5\times 10^{-3}$ & $1.4\times 10^{-3}$        \\ 
{\em RUNX1}  & $2.4\times 10^{-4}$ & $2.0\times 10^{-4}$   & $5.8\times 10^{-4}$ & ${2.9\times 10^{-4}}$ \\ \hline
\end{tabular}

\caption{\label{tab:1} Significant genes and their corresponding average $p$-values over 300 repetitions.}

\end{table}

\begin{table}[htbp]

\centering
\begin{tabular}{l|cccccc}
\hline
Gene   &  d$_0$CRT    & d$_{\mathrm{I}}$CRT & \rev{o}CRT(lasso)     & HRT   & knockoffs \\ \hline
{\em FBXW7}  & 1 & 1   & 0.70 & 0 & 0.54 \\ 
{\em GPS2}   & 1 & 1   & 1 & 0.24 & 0.53 \\ 
{\em HRAS}   & 1 & 1   & 0.14 & 1 & 0.54\\ 
{\em MAP3K13} & 1 & 1   &  0.06 & 0 & 0.13 \\ 
{\em NRAS} & 0 & 0   &  0 & 1  & 0.54 \\ 
{\em RUNX1}  & 1 & 1   & 1 & 1 & 0.54 \\ \hline
\end{tabular}

\caption{\label{tab:2} Frequency of being detected over the 300 repetitions, with false discovery rate control at the level $0.1$.}

\end{table}

\begin{table}[htbp]

\centering
\begin{tabular}{l|ccccc}
\hline
Gene   &  d$_0$CRT    & d$_{\mathrm{I}}$CRT & \rev{o}CRT(lasso)     & HRT   \\ \hline
{\em FBXW7}  & 0 & 0.88   & 0 & 0 \\ 
{\em GPS2}   & 1 & 1   & 0.97 & 0\\ 
{\em HRAS}   & 0 & 0   & 0 &  1\\ 
{\em MAP3K13} & 1 & 1   &  0 & 0\\ 
{\em NRAS} & 0 & 0   &  0 & 1  \\ 
{\em RUNX1}  & 1 & 1  & 0.64 & 1\\ \hline
\end{tabular}

\caption{\label{tab:3} Frequency of being detected over the 300 repetitions, with family-wise error rate control at the level $0.1$.}

\end{table}

\end{document}